\documentclass[11pt]{article}
\usepackage[utf8]{inputenc}
\usepackage[toc,page]{appendix}
\usepackage{setspace}
\usepackage{comment}
\usepackage{subcaption}  

\usepackage{csquotes}
\usepackage{url}
\usepackage{amsmath,amsfonts}
\usepackage{tikz}
\usetikzlibrary{fit}
\usepackage{multirow}

 \usepackage{amsthm}
 {

\newtheorem{assumption}{Assumption}
 
\newtheorem{theorem}{Theorem}[section]
\newtheorem{lemma}[theorem]{Lemma}

\newtheorem{proposition}[theorem]{Proposition}
\newtheorem{corollary}[theorem]{Corollary}

 }
 \usepackage{amsmath,bm}

\usepackage{hyperref,nccmath}
\hypersetup{
    colorlinks=true,
    citecolor = blue,
    linkcolor=blue,
    filecolor=magenta,      
    urlcolor=cyan, 
    pdfpagemode=FullScreen,
    }

\usepackage{xcolor}

\usepackage{soul}

\usepackage{algorithm}
\usepackage{algorithmic}

\usepackage{marginnote}
\usepackage[normalem]{ulem}
\usepackage[shortlabels]{enumitem}
\usepackage{graphicx}
\usepackage{caption}
\usepackage{subcaption}

 \usepackage{cleveref}
\usepackage{lineno}

 \usepackage{blkarray}%

\input{macros.sty}
\numberwithin{equation}{section}

\author{
\begin{tabular}{ccccc}
Garud Iyengar\thanks{IEOR Department, Columbia University. Email: \texttt{garud@ieor.columbia.edu}}
&
Yuanzhe Ma\thanks{IEOR Department, Columbia University. Email: \texttt{ym2865@columbia.edu}} 
&
Jay Sethuraman\thanks{IEOR Department, Columbia University. Email: \texttt{jay@ieor.columbia.edu}}
\end{tabular}
}
\usepackage{cancel}

 \title{Assortment and Procurement Design in Dual-Mode Content Platforms}

\usepackage{booktabs}
\usepackage{fancyhdr}

\begin{document}

\singlespacing

\clearpage\maketitle
\thispagestyle{empty}

\begin{abstract}  
We study assortment and procurement design for a digital content platform offering both ad-supported and subscription access. 
Users are heterogeneous in content preferences and ad tolerance and self-select between the two modes or an outside option.
For a fixed common subscription price and ad load, 
the platform chooses assortment distributions specific to each user type and access mode, together with content-family-level buy-versus-rent decisions
to maximize profit.
Rental costs scale with realized consumption, whereas buying provides a reusable pool of titles whose cost depends on the largest induced requirement across user types and modes. 
We show that the resulting problem is NP-hard. 
We then develop a scalable approximation framework based on a candidate buy set, a relaxation of the procurement coupling, and a decomposition into linear programs with a single equality constraint. 
These subproblems are solved by dual bisection with cardinality-constrained assortment optimization, followed by restricted-master postprocessing to recover primal feasibility. 
The method yields computable optimality-gap bounds, an interpretable threshold-based procurement heuristic, and asymptotic optimality under proportional market scaling as market size and grid resolution increase. 
Numerical experiments show strong performance at moderate market scales and grid sizes.
\end{abstract}
\tableofcontents

\newpage
\setcounter{equation}{0}
\setcounter{footnote}{0}
\setcounter{page}{1}

\section{Introduction}
\label{sec:intro}

Digital content platforms  face two interdependent decisions: which assortments
to offer to each user segment, and how to procure content 
through fixed licensing fees and/or usage-based payments.
These  choices
are especially important   for 
platforms  that offer both an ad-supported mode and an ad-free subscription mode. Users self-select 
 based on the content they expect to  see.  Platform revenue and content costs are  jointly
 determined by  platform decisions and user responses.

 Our procurement model mirrors practice at many leading platforms, which
  report  content costs as a   combination
  of fixed commitments (e.g., licensed or produced content amortized over time) and usage-linked payouts (e.g., revenue sharing or royalties).
For example, the video platform \emph{Bilibili} reports ``revenue-sharing costs'' separately from ``content costs,'' which include amortized licensing costs and production spending~\citep{Bilibili25}.
Similarly, the music platform \emph{Spotify} explains that royalty costs depend on the product tier
and the mix of ad-supported versus premium usage, and that licensing arrangements may include fixed fees or minimum guarantees in addition to usage-based payments~\citep{Spotify25}.
 Many platforms also
 operate in a 
 \emph{dual monetization mode}, offering an ad-supported experience alongside an ad-free subscription option.
This ``pay with attention'' versus ``pay with money'' menu  segments heterogeneous consumers, improving engagement and increasing
platform revenue~\citep{FanKuWh07,Tag09,LambrechtMi17,DeValvePe22}.

Personalized assortments shape demand across content families and therefore change the platform's procurement incentives. 
Procurement and assortment design are therefore inseparable:
the value of including a content family in many tailored assortments depends on whether it is bought or rented, and procurement choices depend on the induced demand.
More broadly,   multisided media models
emphasize that content strategy interacts with both advertising and subscription revenues~\citep{AmaldossDuSh21,CarroniPa20}.

Motivated by these observations, we study a unified profit-maximization problem for a dual-mode content platform.
Given a  
subscription price $p$ and an 
ad load $\sigma$ (the number of ads per unit time),
we study
how the platform 
should
choose personalized stochastic assortments in each mode for each type of user, together with 
family-level buy-versus-rent procurement decisions.
Users  differ in 
content preferences and ad tolerance. 
They self-select among the ad-supported mode, the subscription mode, and an outside option based on expected utility under each mode's assortment; conditional on a mode, they choose from the displayed assortment according to a multinomial logit (MNL) model~\citep{TalluriRy04, RusmevichientongShSh10}.  
The platform can display at most $\kappa$ content families and may randomize over feasible assortments.
Even for fixed $(\sigma, p)$, the  problem combines discrete procurement choices,  
assortment randomization, and endogenous mode choice,
and is NP-hard
  for $\kappa = 1$. 
Given this complexity and the industry norm that pricing is slow-moving while assortment is fast-moving, we focus on the assortment and procurement problem
parameterized by $(\sigma, p)$.

Our goal is  
an approximation framework that is both interpretable and scalable. 
We start from a candidate buy set and relax the max-type coupling induced by purchased content families.
This relaxation decomposes across user types.
For each type, we show that the relaxed problem can be parameterized by an ad-side ratio, which reduces the search to two 
 linear programs 
with one equality constraint
over assortment distributions.
Dualizing this constraint yields
cardinality-constrained MNL assortment problems.
Building on the  classic 
 algorithm in \citep{RusmevichientongShSh10}, 
we develop a   dual-bisection
search procedure for these subproblems and aggregate   across a ratio grid.
Because the subproblems can be solved independently across user types and sampled ratios, the method  parallelizes naturally.
The resulting  guarantee is fully computable and decomposes
into a buy-set term, a dual-relaxation term, and a grid-discretization term.
In a market-scaling regime, the approximation becomes asymptotically
optimal, and the implementation 
runs in polynomial time in the number of user types and content families.
Numerically, we
 find strong performance at moderate grid sizes.

\subsection{Contributions}
\label{sec:contributions}

\begin{enumerate}

 \item \textbf{Modeling contributions.}
 Our \paper\   introduces a new assortment-optimization model class and an 
 approximation
 framework. 
 Where prior work typically optimizes a single assortment objective, we study a setting in which assortment choices interact with endogenous mode choice, ad/subscription monetization, and buy-versus-rent procurement.

\item \textbf{Hardness of joint assortment and procurement design (Theorem~\ref{thm:our-prob-np-hard}).} 
We  show
that the prize-collecting fractional set cover problem~\citep{DinitzGu13}
is NP-hard. 
We then use a reduction from this problem to prove that our joint assortment-and-procurement
problem is NP-hard when $\kappa = 1$.
\item \textbf{An efficient approximation framework with computable certificates (Section~\ref{sec:model-sol}).}
We develop an approximation framework
(with details of Algorithm~\ref{alg:mnl-solver-all} presented in Section~\ref{sec:aggregation})
that combines a buy-set relaxation, a ratio reduction, and a bisection-based dual search built around the 
cardinality-constrained
MNL algorithm of \citet{RusmevichientongShSh10}. 
We include an explicit primal-recovery step: we store assortments encountered during bisection and solve a restricted mixing problem over these actions, recovering primal feasibility together with a computable approximation certificate.
This yields a scalable algorithm with a finite-instance guarantee that decomposes into three interpretable terms: a buy-set gap, a dual-relaxation gap, and a grid-discretization error.
Let $J$ be the number of user types and $L$ 
the number of   content families.
Under market scaling, where the consumer measure scales proportionally with $N$, the returned solution is asymptotically nearly optimal. 
 For fixed   model data,
using $K_T$ points in the outer ratio grid yields total runtime $\widetilde{O}(JL^2K_T)$ (where $\widetilde{O}$ hides polylogarithmic factors),
and total relative error   $O(K_T^{-1}+N^{-1})$. 
The constants in this relative-error bound are instance-dependent and 
include the Lipschitz constants used to certify the ratio-grid discretization error.
\item \textbf{Managerial insights and the value of hybrid procurement.} 
Our solution
yields a simple and interpretable threshold buying rule: 
a content family is bought when its fixed licensing cost is small relative to the rental exposure it would generate across user types and modes, and rented otherwise.
Let \(\lambda\jj\) denote the baseline measure of type-\(j\) consumers.
When the  market size (consumer mass) is scaled
by a factor \(N\),   the threshold rule buys content family \(\ell\) whenever
\[
\frac{\eta_\ell}{\gamma_\ell} \le \frac{N}{2}\sum_{j=1}^J \lambda\jj .
\]
Here, \(\eta_\ell/\gamma_\ell\) is the buy-to-rent
cost ratio for content family \(\ell\). 
This characterization
clarifies how platform scale, cross-type reuse, and per-use rental costs jointly shape procurement decisions. 
Our numerical study shows that this hybrid policy can outperform both pure renting and pure buying by combining buying   high-reuse content with renting   niche content. 
\end{enumerate}

\subsection{Structure of the \paper.}
The remainder of the \paper\ is organized as follows.
Section~\ref{sec:related-work} discusses related work.
Section~\ref{sec:model} introduces the assortment-based platform model. 
Section~\ref{sec:model-sol}
develops the approximation framework for the procurement problem.
Section~\ref{sec:numerics}
reports numerical experiments on algorithmic performance and hybrid procurement insights.  
Section~\ref{sec:conclusion} concludes with key takeaways and directions for future work.

\section{Related work}
\label{sec:related-work}

Our work relates to
several streams of research, including two-sided pricing and content provision, hybrid monetization design, 
personalized content offerings and engagement,
procurement flexibility, and assortment optimization.  We review the literature most relevant to our contributions.

\paragraph{Two-sided pricing and content provision.}
\citet{AmaldossDuSh21} model interactions among consumers, advertisers, and content suppliers, highlighting how procurement costs shape user pricing and ad levels.
\citet{AmaldossDuSh24} extend this framework to competitive settings with heterogeneous ad tolerance and multi-homing.
\citet{Lin20} examines two-sided price discrimination, emphasizing the interplay between users' willingness to pay and advertiser demand.
These studies focus on equilibrium pricing but do not model content-level assortment decisions or the buy-versus-rent procurement trade-off that is  central to our work.

\paragraph{Balancing subscription and advertising revenues.}
A large literature studies how platforms segment users through differentiated access modes~\citep{KumarSe09,HalbheerStLe14,CarroniPa20}.
\citet{CaiSp23} and \citet{Sato19} analyze menus combining a basic ad-supported mode and a premium subscription, showing that self-selection   enables platforms to serve heterogeneous users while extracting surplus from premium segments.
\citet{DeValvePe22} study optimal price-and-advertising menus in two-sided markets, demonstrating that dual-mode access
aligned with consumer preferences
enhances profitability.
\citet{LiBaChPa24} analyze internal coordination between advertising and subscription divisions, uncovering a positive relationship between subscription price and ad intensity: higher subscription fees push more users to the ad mode, increasing advertising reach.
\citet{WeiNa14} study versioning of information goods under 
preference heterogeneity.
We extend this literature by jointly optimizing dual-mode access, personalized (possibly stochastic) assortments, and buy-versus-rent procurement, so that mode choice and within-mode content consumption are simultaneously endogenized.

\paragraph{Personalized content offerings and engagement.}
Recent work explores how platforms  use content offerings   to match users with preferred content and manage long-term engagement~\citep{DeanMo22, IyengarMaSe25}.
\citet{IyengarMaSe25} propose a dynamic framework in which platforms strategically introduce users to new content types, balancing exploration with churn risk.
Similar to their work, our model recognizes  that optimal content offerings need not maximize immediate user utility.
We differ by focusing on the interaction between personalized assortments and procurement mode:
assortments are decision variables that interact with endogenous mode choice of consumers and buy-versus-rent procurement choices.

\paragraph{Buy-versus-rent and procurement flexibility.}
A growing literature studies procurement contracts that trade off upfront commitment against ex-post flexibility. For video content,
\citet{Jiangetal25} analyze contract selection between usage-based royalties and upfront  fixed fees,
highlighting how contract form interacts with piracy surveillance incentives. Their  single-title model  centers on production cost and anti-piracy effort, abstracting from downstream revenue-management.
In operations management, related commitment-versus-flexibility trade-offs appear in studies of contracts combined with spot purchases \citep{PeiSiTu11}, 
option contracts with upfront premia and future exercise rights  \citep{BarnesSchusterBaAn02}, 
and quantity-flexibility contracts that couple baseline commitments with ex-post adjustment \citep{Tsay99}.  
Our content procurement setting shares this canonical trade-off.
First, the commitment level of a bought 
content family is not exogenously chosen but endogenously induced by the platform's personalized content offerings across user types and access modes, entering procurement costs through a maximum-exposure term. 
Second, demand is endogenous. The platform jointly designs content offerings and the ad-supported versus subscription menu, so usage flows (driving rental payments) and peak exposure (driving buying costs) are both decision-dependent. Together, these features turn the classical commitment--flexibility trade-off into a joint design problem in which procurement, mode design, and personalized assortments must be optimized simultaneously.

\paragraph{Assortment optimization.}
Our work  
relates to the assortment optimization literature under discrete-choice models.  
Classical work studies how sellers choose product subsets to maximize revenue when customer choice follows the multinomial logit (MNL) model, with extensions to dynamic settings and capacity or cardinality constraints~\citep{TalluriRy04, RusmevichientongShSh10}.
Recent work incorporates richer objectives and operational constraints, including revenue-utility trade-offs and structured feasible sets~\citep{SumidaGaRuToDa21}.
In contrast,  
we do not optimize a single assortment against an exogenous choice model; instead, assortment distributions are personalized by user type and monetization mode. Moreover, assortments jointly affect within-mode choice, users’ endogenous selection between ad-supported and subscription access, and 
 procurement decisions.
 This coupling across user types and modes, through both endogenous mode choice and buy-versus-rent costs, falls outside standard assortment models and motivates our decomposition-based approximation approach. We build directly on the cardinality-constrained MNL assortment algorithm of \citet{RusmevichientongShSh10} as a subroutine inside our dual-bisection and primal-recovery framework.

\paragraph{General optimization methodology.}

Our method is related to the scalar dual-bisection framework of~\citet{ManieriFaPr25}, but differs in how primal feasibility is recovered. In their setting, feasibility can be preserved by retaining a single feasible primal iterate along the dual search. In our setting, however, an individual action typically has nonzero residual and is therefore not primal-feasible on its own.
 This is because feasibility requires matching certain coupling constraints (e.g., ratio constraints) in expectation over a distribution of assortments, not pointwise for a single assortment. 
We therefore explicitly introduce a postprocessing step that mixes the generated primal actions by solving a restricted primal problem over the actions seen so far.
The proposed postprocessing step is similar to
column generation: the
dual-search   generates candidate primal actions one by one, and a restricted master problem is then solved over the generated set. 
Unlike standard column generation, however, the generation mechanism here is driven by scalar dual bisection rather than by repeated pricing from a master linear program.

\section{Model}
\label{sec:model}

We study a digital content platform that serves $J$  types of users. Let $\lambda\jj > 0$, $j \in [J] := \{1,\ldots,J\}$, denote the mass of type-$j$ users. The platform offers two access modes: an ad-supported mode and a subscription mode. It sets a uniform subscription fee $p \ge 0$ and a uniform ad load $\sigma \ge 0$. 
The uniform-$(p,\sigma)$ restriction reflects practical and regulatory considerations:
personalized prices or ad loads are often harder to justify and implement than personalized content offerings, while uniform monetary and ad-time charges are more transparent to users and more consistent with common industry practice~\citep{GoldfarbTu11, Weyl10}.

\paragraph{User choice.}
Users face two nested decisions.
In stage 1, a type-$j$ user chooses between the ad-supported mode, 
the subscription mode, or
leaving the platform,
based on the expected utility of these options
to be specified below. 
We refer to this stage-1 exit as the platform outside option. 
If a user chooses one of the two modes,
she will repeatedly visit the platform over $T \gg 1$ periods, which we denote as stage 2.
In this stage,  at every visit,
the platform draws an
assortment $A$, and the user makes 
a 
choice from $A \cup \set{0}$, where alternative $0$ denotes that no displayed content family  is selected on that visit.
To avoid ambiguity, we call $0$ the no-item alternative.
Separately, $A = \emptyset$ denotes the empty assortment, which is a platform action rather than a user choice; if $A = \emptyset$, the no-item alternative is chosen with probability one.

\paragraph{Content offerings via assortment distributions.}
 At each visit $\tau \in \set{1,\dots,T}$, the platform's content offering is modeled as a (possibly randomized) assortment of titles, where the distribution of the assortment 
 depends on the user type and
 the mode chosen by the user.
We model the user's within-visit choice from the displayed assortment using the standard Multinomial Logit   (MNL) choice model. 
Under this model, the user either selects one of the displayed content families or   chooses the
no-item alternative. 
There are $L$ content families
indexed by $\ell \in [L] := \{1,\ldots,L\}$. 
For each type $j$ and content family $\ell$, let $\alpha_\ell^{(j)} \ge 0$ denote the MNL attraction parameter and let $u_\ell^{(j)} \ge 0$ denote the utility delivered if content family $\ell$ is chosen.
We interpret each modeled   object
$\ell \in [L]$ as a \emph{content family}
containing a large pool of exchangeable titles, rather than a single fixed title.
At each visit, conditional on the chosen mode, the platform displays an assortment
of content families. If the user selects family $\ell$, one fresh title from that
family is consumed; if the user does not select $\ell$, no title from that family
is consumed. 
Accordingly,
$\alpha_{\ell}^{(j)}$ and $u_{\ell}^{(j)}$ should be interpreted as family-level
attraction and expected-utility parameters for a type-$j$ user.
Within the MNL model, the 
no-item alternative
is indexed by $0$ and has an
attraction parameter $\alpha_0^{(j)} = 1$. 
We emphasize that $\alpha_{\ell}^{(j)}$ and $u_{\ell}^{(j)}$ play different roles. The parameter
$\alpha_{\ell}^{(j)}$ governs the user's within-assortment selection probability under the MNL
model and captures immediate attraction or click propensity. By contrast, $u_{\ell}^{(j)}$ denotes
the user's perceived expected utility conditional on consuming content family $\ell$. 
This distinction is natural in digital content settings: for example, long-form content may deliver high value once consumed, but may be selected
less often than shorter content.
Empirically, contextual variables can move attention and utility differently~\citep{TeruiBaAl11, KawaguchiUeWa21}.
Thus, a content family may have high expected value conditional on consumption
even if it has relatively low selection propensity, or vice versa.
Further, we assume users can
form such expectations from prior experience, metadata, creator or genre familiarity, and learned
platform behavior. 
We do not require $\alpha_\ell^{(j)}$ and $u_\ell^{(j)}$ to vary independently across content families: for example,
the model also accommodates structured cases in which $u_\ell^{(j)} = g\jj(\alpha_\ell^{(j)})$ for some increasing function $g\jj$  (e.g., $u_\ell^{(j)} = e^{\alpha_\ell^{(j)}}$). Even then, the   users' mode choice still depends on a nontrivial trade-off between click propensity and delivered utility.

Our solution approach and insights extend to the more general model
in which the utility parameters depend on the mode; however, we choose to let them be independent of the mode for notational simplicity.
We impose a cardinality cap $\kappa \in \{1,\ldots,L\}$ and define the family of feasible assortments by
\begin{align*}
\mc{A}_\kappa := \{A \subseteq [L] : |A| \le \kappa\}.
\end{align*}
 The platform can also choose not to offer any content in a given interaction, corresponding to the choice $A = \emptyset \in \mc{A}_\kappa$.
 Our approach naturally adapts to the model without a   cardinality constraint, i.e., \(\kappa=L\), 
 by replacing Algorithm~\ref{alg:mnl-oracle} with the unconstrained MNL assortment algorithm
 of \citep{TalluriRy04}.  
 We adopt the  MNL
 model as our maintained within-mode choice model because it is standard, parsimonious, and interpretable, 
 and because it permits efficient exact solution of the cardinality-constrained weighted-assortment subproblems arising throughout 
 our dual-bisection procedure.
 In our analysis, the MNL structure also enables 
 an
 exact solution of the subscription boundary case
  (Algorithm~\ref{alg:mnl-sub-boundary})
 and the derivation of explicit, computable 
 Lipschitz constants   
 used in the finite-instance performance guarantees (Theorem~\ref{thm:mnl-inner-separated-gap}). 
 At the same time, the main economic and algorithmic ideas, including the buy-set relaxation, the ratio reduction, and the restricted-master recovery step, 
 are not inherently tied to MNL. 
 These components can potentially be generalized to other regular discrete-choice models that admit an
 efficient solution of the corresponding weighted cardinality-constrained assortment problem; relevant examples include the generalized attraction model and the nested logit model under a global cardinality constraint~\citep{Wang13, FeldmanTo15, GallegoRaSh15}. 
 A complete end-to-end extension to a particular choice model may, however, require model-specific treatment of boundary cases and computable Lipschitz constants. 
 We therefore focus on MNL to provide a complete, transparent, and computationally implementable framework while highlighting the broader applicability of the underlying methodology.

For a deterministic assortment $A \in \mc{A}_\kappa$, the MNL choice probabilities for a type-$j$ user are
\begin{align*}
 p_\ell^{(j)}(A)
 :=
 \frac{\alpha_\ell^{(j)}}{1 + \sum_{k \in A} \alpha_k^{(j)}}  \I{\ell \in A},
 \qquad
 p_0^{(j)}(A)
 :=
 \frac{1}{1 + \sum_{k \in A} \alpha_k^{(j)}}.
\end{align*}
We define the expected click probability 
 and the expected utility generated by $A$ as
\begin{align*}
 x^{(j)}(A)
 &:=
 \sum_{\ell \in A} p_\ell^{(j)}(A)
 =
 \frac{\sum_{\ell \in A} \alpha_\ell^{(j)}}{1 + \sum_{k \in A} \alpha_k^{(j)}},
 \\
 u^{(j)}(A)
 &:=
 \sum_{\ell \in A} u_\ell^{(j)} p_\ell^{(j)}(A)
 =
 \frac{\sum_{\ell \in A} \alpha_\ell^{(j)} u_\ell^{(j)}}{1 + \sum_{k \in A} \alpha_k^{(j)}}.
\end{align*}

We also define the family-level extreme utilities and the largest attainable assortment utility by
\begin{align}
 \umaxj &:= \max_{\ell \in [L]: \alpha_\ell^{(j)} > 0} u_\ell^{(j)},
 &
\uminj &:= \min_{\ell \in [L]: \alpha_\ell^{(j)} > 0} u_\ell^{(j)},
 &
 U_{\max}^{(j)} &:= \max_{A \in \mc{A}_\kappa} u^{(j)}(A).
 \label{eqn:u-a-max-min-def}
\end{align}
Note that the maximum and minimum are
defined only over content families
$\ell$
such that  $\alphajl > 0$.
In addition,
\begin{align*}
\max_{\ell} \frac{\alphajl \ujl}{1 + \alphajl} \le \uumaxj \le \umaxj.
\end{align*}

A mode-$m \in \{a,s\}$ offering policy for type $j$ is a distribution
\begin{align*}
 \bm{y}^{(m,j)} = \bigl(y_A^{(m,j)}\bigr)_{A \in \mc{A}_\kappa} \in \Delta(\mc{A}_\kappa)
\end{align*}
over feasible assortments,
where \(\Delta(\mc{A}_\kappa)\) denotes the probability simplex over the feasible assortment family \(\mc{A}_\kappa\).
The induced   click probability and 
expected utility, conditioned on 
the mode $m$, are:
\begin{align}
 x^{(m,j)} &:= \sum_{A \in \mc{A}_\kappa} y_A^{(m,j)} x^{(j)}(A),
 &
 u^{(m,j)} &:= \sum_{A \in \mc{A}_\kappa} y_A^{(m,j)} u^{(j)}(A).
 \label{eqn:x-u-def}
\end{align}
For later use, define the  consumption flow into   $\ell$ by
\begin{align}
f_\ell^{(m,j)}(\bm{y})
&:=
\sum_{A \in \mc{A}_\kappa : \ell \in A} y_A^{(m,j)} p_\ell^{(j)}(A).
\label{eqn:y-flow-def} 
\end{align}
By construction,
\begin{align}
\sum_{\ell \in [L]} f_\ell^{(m,j)}(\bm{y})
= x^{(m,j)}
\le 1.
\label{eqn:y-sum-less-than-kappa}
\end{align}
In what follows, $\bm{y}$ denotes the collection $\{\yadj, \ysubj\}_{j \in [J]}$.

\paragraph{Mode choice.}
We assume users can preview or infer the mode-dependent content offerings,
for example, through browsing interfaces, trial experiences, or learned platform behavior,
before choosing a mode. 
We model post-selection ads; pre-roll/display ads independent of selection are outside the scope of this model.
Let $\chij>0$ be a random ad-tolerance parameter with cumulative distribution function $F\jj$.
Conditional on ad tolerance $\chij > 0$, the per-selection ad disutility is $\sigma/\chij$.
Thus, under the ad-supported assortment distribution $y^{(a,j)}$, the expected ad disutility is
$(\sigma/\chij) \sum_{A \in \mc{A}_\kappa} y_A^{(a,j)} x^{(j)}(A)$.
This stochastic specification captures heterogeneity in ad aversion and is consistent with  the literature~\citep{Sato19,DeValvePe22}. 
We write $\bar F\jj := 1 - F\jj$.
We adopt the tie-breaking convention that subscription wins ties against the outside option. 
This convention makes the subscription-active region \(\{\usubj\ge p\}\) closed ($\usubj$ is defined in~\eqref{eqn:u-sub-def}); without it, the subscription-active region would be \(\{\usubj>p\}\), and an optimal subscription assortment distribution may fail to be attained at the threshold.
For a type-$j$ user, the expected net utility from the ad-supported mode under the assortment distribution $\yadj$ is
\begin{align}
 \underbrace{\sum_{A \in \mc{A}_\kappa} y_A^{(a,j)} u^{(j)}(A)}_{:= u^{(a,j)}}
 -
 \frac{\sigma}{\chij}
 \underbrace{\sum_{A \in \mc{A}_\kappa} y_A^{(a,j)} x^{(j)}(A)}_{:= \xadj}.
 \label{eqn:u-ads-def}
\end{align}
The net utility from the subscription mode under $\ysubj$ is
\begin{align}
 \underbrace{\sum_{A \in \mc{A}_\kappa} y_A^{(s,j)} u^{(j)}(A)}_{:= \usubj} - p.
 \label{eqn:u-sub-def}
\end{align}
We assume the outside option yields zero utility. 
We further adopt the convention that the ad-supported mode loses ties against both the subscription mode and the outside option. Under 
$\sigma > 0$
and $\xadj > 0$, such ties occur with probability zero by continuity of  
$F\jj$; the convention also covers degenerate cases such as $\sigma = 0$.
Thus,
a type-$j$ user chooses the ad mode if
\begin{align*}
u^{(a,j)} - \frac{\sigma \xadj}{\chij}
>
\max\{\usubj - p, 0\}
= (\usubj - p)_+.
\end{align*}
Under the stated tie-breaking convention, equality does not result in selection of the ad-supported mode.
Therefore, the probability that a type-$j$ user chooses the ad mode is 
\begin{align}
\padj(\bm{y}) = \bfj{h\jj(\bm{y})}, \quad
h\jj(\bm{y})
:=
\begin{cases}
\frac{\sigma \xadj}{u^{(a,j)} - (\usubj - p)_+}
  & \text{ if } u^{(a,j)} - (\usubj - p)_+ > 0, \\
+ \infty & \text{otherwise},
\end{cases}
\label{eqn:padj}
\end{align} 
Although \(\padj\) in~\eqref{eqn:padj} depends only on \(\bm{y}^{(j)} := (\yadj,\ysubj)\), we suppress the type superscript for notational simplicity; the same convention applies below.
Similarly, the probability that a type-$j$ user chooses the subscription mode is
\begin{align}
\psubj(\bm{y})
:=
\I{\usubj \ge p}
F\jj \Paran{h\jj(\bm{y})}.
\label{eqn:psubj}
\end{align}
We define $\ynullS$
as the assortment vector that assigns 
probability one to the empty assortment.

\paragraph{Revenue.}
In the ad-supported mode, the platform earns advertising revenue \(\rads\sigma\) per visit whenever a type-$j$ user chooses a
content family $\ell \in [L]$,
where \(\rads \ge 0 \) is the revenue rate per unit of ad load.
For a given assortment distribution $\bm{y}$
and quantities
  $\xadj,\padj,\psubj$ 
in~\eqref{eqn:x-u-def},~\eqref{eqn:padj},~\eqref{eqn:psubj}, respectively,
the per-type expected revenue contribution
from type-$j$ users is given by
\begin{align}
 \mathsf{Rev}\jj(\bm{y})
 :=
 \rads\sigma \xadj \padj(\bm{y})
 +
 p \psubj(\bm{y}).
 \label{eqn:R-j-def}
\end{align}

Figure~\ref{fig:user_flow} summarizes the platform-user interaction. Users first choose between the ad-supported mode, the subscription mode, and the outside option. Conditional on the chosen mode, the platform draws an assortment from the corresponding assortment distribution, the user selects a content family
from that assortment according to the MNL model, and the platform earns either advertising revenue or subscription revenue.

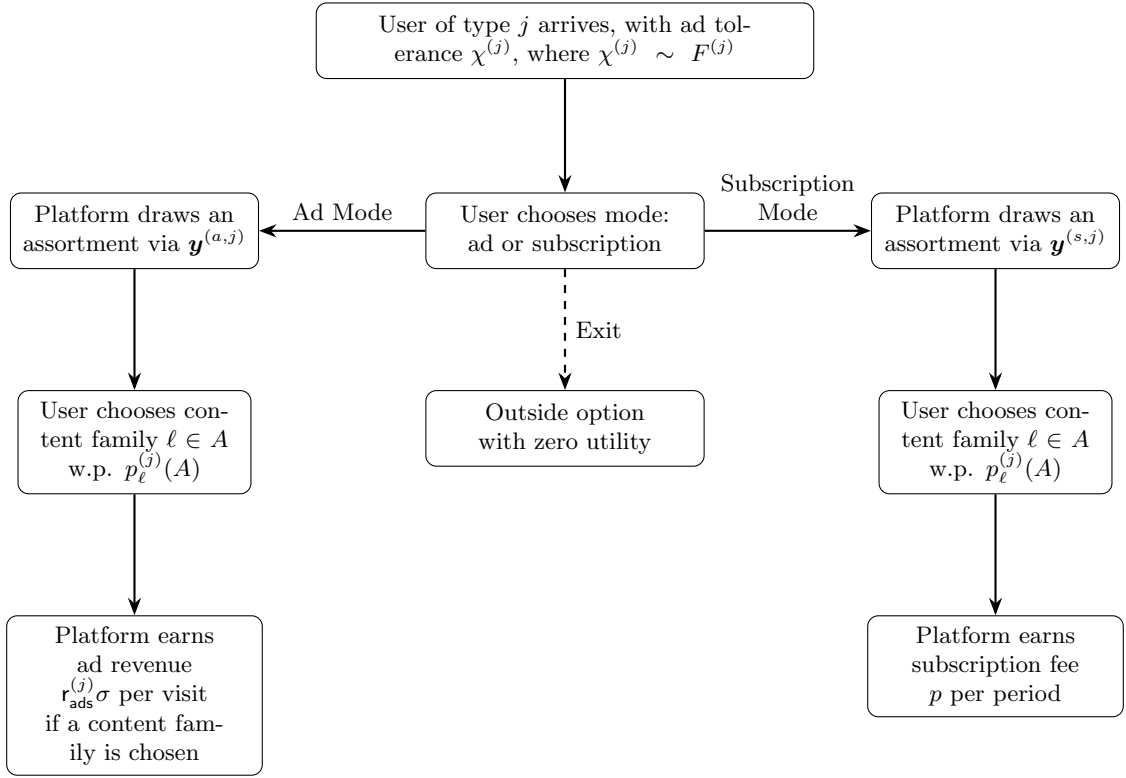
\begin{figure}[h]
\centering
\begin{tikzpicture}[every node/.style={font=\footnotesize}]
\tikzstyle{block} = [rectangle, draw, rounded corners, align=center, minimum height=1.0cm, inner sep=4pt, text width=3.0cm]
\tikzstyle{arrow} = [thick, ->, >=Stealth]

\node (start)    [block, text width=6.3cm] {User of type $j$ arrives, with ad tolerance $\chij$, where $\chij \sim F\jj$};
\node (decision) [block, below=1.5cm of start, text width=3.4cm] {User chooses   mode:\ ad or subscription};
\node (outside)  [block, below=1.6cm of decision, text width=3.4cm] {Outside option with zero utility};
\node (adroute)  [block, left=2.2cm of decision, text width=3.0cm]  
{Platform draws an assortment  via $\yadj$};
\node (subroute) [block, right=2.2cm of decision, text width=3.0cm] 
{Platform draws an assortment  via $\ysubj$};
\node (adaccept) [block, below=1.6cm of adroute, text width=2.8cm]  {User chooses content family $\ell \in A$\\ w.p. $p_\ell^{(j)}(A)$};
\node (subaccept)[block, below=1.6cm of subroute, text width=2.8cm]{User chooses content family $\ell \in A$\\ w.p. $p_\ell^{(j)}(A)$};
\node (adrevenue)[block, below=1.6cm of adaccept, text width=3.1cm]{Platform earns\\ ad revenue $\rads\sigma$ per visit\\ if a content family is chosen};
\node (subrevenue)[block, below=1.6cm of subaccept, text width=3.1cm]{Platform earns\\ subscription fee $p$ per period};

\draw [arrow] (start) -- (decision);
\draw [arrow] (decision) -- node[midway, above] {Ad Mode} (adroute);
\draw [arrow] (decision) -- node[midway, above] {\shortstack{Subscription\\Mode}} (subroute);
\draw [arrow, dashed] (decision) -- node[midway, right] {Exit} (outside);
\draw [arrow] (adroute)   -- (adaccept);
\draw [arrow] (subroute)  -- (subaccept);
\draw [arrow] (adaccept)  -- (adrevenue);
\draw [arrow] (subaccept) -- (subrevenue);
\end{tikzpicture}
\caption{User and platform
decision flow under assortment distributions: users first choose a mode, then face a mode-dependent assortment, and the platform earns revenue from either advertising or subscription.}
\label{fig:user_flow}
\end{figure}

\paragraph{Procurement and profit.}
Each content family $\ell$ can either be rented on a usage/click
basis at rate $\gamma_\ell \ge 
0$ or 
  bought at unit cost $\eta_\ell  
  \ge
  0$ per required copy
that is shareable across different users.
We define
\begin{align*}
\eta_{\max} = \max_{\ell \in [L]} \eta_\ell,  \quad
\gamma_{\max} = \max_{\ell \in [L]}\gamma_\ell.
\end{align*}

The procurement cost induced by $\bm{y}$ is
\begin{align}
 C(\bm{y})
 :=
 \sum_{\ell \in [L]}
 \min\Bigg\{
 \eta_\ell \max_{j \in [J], m \in \{a,s\}} f_\ell^{(m,j)}(\bm{y}),
 \gamma_\ell \sum_{j \in [J]} \lambda\jj
 \Bigl(
 \padj(\bm{y}) f_\ell^{(a,j)}(\bm{y})
 +
 \psubj(\bm{y}) f_\ell^{(s,j)}(\bm{y})
 \Bigr)
 \Bigg\}.
 \label{eqn:cost-content}
\end{align}

The first term of~\eqref{eqn:cost-content}
corresponds to buying: if content family $\ell$ is licensed, the platform needs enough distinct copies to support the largest 
expected
  requirement 
across user types and modes. 
This term intentionally does not scale with $\lambda\jj$ 
or with endogenous mode-choice probabilities, because a   content family that is owned is  reusable across users.
The second term of~\eqref{eqn:cost-content}
corresponds to renting: the platform pays   for the total usage across user types.
We consider profit in per-visit units; equivalently, the common factor $T$ from the repeated-visit horizon has been divided out.  
The platform's    problem is
\begin{align}
R\opt  =  \max_{\bm{y} \in \simplexprod}
 \set{
 \sum_{j \in [J]} \lambda\jj \mathsf{Rev}\jj(\bm{y}) - C(\bm{y})}.
 \label{eqn:problem-buy-or-rent}
\end{align}

For the algorithmic development in Section~\ref{sec:model-sol},
it is convenient to fix a buy set $S \subseteq [L]$ and treat the remaining content families as rented. The resulting inner objective is
\begin{align}
 \begin{split}
 R(S, \bm{y})
 :=
 \sum_{j \in [J]}
 &\Bigg[
 \lambda\jj\Bigl(\rads\sigma \xadj \padj(\bm{y}) + p \psubj(\bm{y})\Bigr)
 \\
 &\qquad
 - \lambda\jj \padj(\bm{y})
 \sum_{\ell \in [L] \setminus S} \gamma_\ell f_\ell^{(a,j)}(\bm{y})
 \\
 &\qquad
- \lambda\jj \psubj(\bm{y})
\sum_{\ell \in [L] \setminus S} \gamma_\ell f_\ell^{(s,j)}(\bm{y})
\Bigg]
\\
&
- \sum_{\ell \in S} \eta_\ell \max_{j \in [J], m \in \{a,s\}} f_\ell^{(m,j)}(\bm{y}).
\end{split}
\label{eqn:R-S-def}
\end{align}
For fixed $(S,\sigma,p)$, the platform chooses assortment distributions  $\bm{y}$
to solve
\begin{align}
R(S)
 =
\max_{\bm{y}  \in \simplexprod} R(S, \bm{y}).
 \label{eqn:R-S-problem-def}
\end{align}
The original procurement problem 
$R\opt$~\eqref{eqn:problem-buy-or-rent}
can equivalently be written as
\begin{align*}
 R\opt := \max_{S \subseteq [L]} R(S).
\end{align*}
Since we have 
$ R\opt  = \max_{S \subseteq [L]}
\max_{y \in \simplexprod} R(S,\bm{y})
= \max_{y \in \simplexprod} \max_{S \subseteq [L]}
 R(S,\bm{y})$, 
we let $R(\bm{y}) = \max _{S \subseteq [L]} R(S, \bm{y})$ 
so that $R\opt = \max_{\bm{y} \in \simplexprod} R(\bm{y})$.

We impose the following regularity conditions throughout the \paper.

\begin{assumption}
\label{assump:MNL}
\begin{enumerate}[label=\alph*), ref=\alph*] 
\item
\label{assm:alpha-implies-u-support}
For every type $j$, $\max_{\ell \in [L]} \alphajl > 0$ and
\begin{align}
 \alpha_\ell^{(j)} > 0 \Rightarrow  u_\ell^{(j)} > 0.
 \label{eqn:alpha-implies-u-support}
\end{align}
\item  \label{item:opt-profit-positive}
The optimal value  $R\opt$ defined  
in~\eqref{eqn:problem-buy-or-rent} 
satisfies
$R\opt > 0$.
\item
\label{item:assump-Lip}
For each type $j$, the cdf $F\jj$ is absolutely
continuous and admits a density $\phi_F\jj$ with $|\phi_F\jj(\cdot)| \le W_{F\jj}$. 
Thus,
$F\jj$ is Lipschitz: 
\begin{align*}
 |F\jj(z) - F\jj(z')| \le W_{F\jj}|z-z'| \qquad \text{for all } z,z'.
\end{align*}
\end{enumerate}
\end{assumption} 

\paragraph{Interpretation of Assumption~\ref{assump:MNL}.} 
Assumption~\ref{assump:MNL}(\ref{assm:alpha-implies-u-support}) requires that, for each type $j$, the content families that are choice-relevant under the MNL model   deliver a strictly positive utility.  
Assumption~\ref{assump:MNL}(\ref{item:opt-profit-positive}) 
is a 
mild nondegeneracy condition requiring that the platform’s optimal profit is strictly positive
for the given user-mass vector  $\bm \lambda$.
Assumption~\ref{assump:MNL}(\ref{item:assump-Lip}) is a smoothness condition on ad-tolerance heterogeneity. 
It is mainly used to obtain a clean Lipschitz error bound for the
outer ratio grid. 
Our approach in Section~\ref{sec:model-sol}
can also handle discrete ad-tolerance distributions if $\sigma > 0$:
one can add to the ratio grid in~\eqref{eqn:outer-grid-def} the breakpoints 
induced
by the jump points of the cdf $F\jj$: 
$\{\tau/\sigma \pm \epsilon:\tau \text{ is a jump point of }F\jj\}$ for some small enough $\epsilon > 0$.
When  $\sigma = 0$, the mode-choice probabilities do not vary with $t$, so no jump-induced ratio points are required.
We impose the continuous-cdf condition only to
keep the notation simple.

The platform's joint  assortment  and procurement problem~\eqref{eqn:problem-buy-or-rent} 
combines continuous optimization over $\bm{y} \in \simplexprod$ with a combinatorial decision over which 
content families to buy versus rent,
and even when the buying set is fixed, the max operator also introduces computational challenges.
This combination makes the problem computationally intractable in general.
We establish NP-hardness (Theorem~\ref{thm:our-prob-np-hard}) 
of the problem  in  
the special case with assortment
capacity $\kappa = 1$.
See Appendix~\ref{sec:proof-hardneess}
for the
  proof.

\begin{theorem}
The decision version of   problem~\eqref{eqn:problem-buy-or-rent} is NP-hard.
Specifically, 
under Assumption~\ref{assump:MNL},
given rational input data, 
$\kappa = 1$,
fixed rational
$(\sigma,p)$, and
a rational threshold   $B \in \mathbb{Q}$, deciding whether there exists
$\bm{y} \in \simplexprod$ such that $R(\bm{y}) \ge B$
is NP-hard.
\label{thm:our-prob-np-hard}
\end{theorem}

\section{Our approach} %
\label{sec:model-sol}

In this section, we develop a framework for 
approximately
computing
$R(S)$ in~\eqref{eqn:R-S-problem-def}. 
The main challenge is that the problem contains several layers of coupling that prevent a direct application of standard assortment-optimization methods. 
First, the content families 
that are purchased
create a max-type procurement coupling: the cost of a purchased
content family
depends on the largest induced consumption flow across all user types and both modes, through the term
$\max_{j\in[J],m\in\{a,s\}} f_\ell^{(m,j)}(\bm{y})$.
Second, the mode-choice probabilities $P^{(a,j)}(\bm{y})$ and $P^{(s,j)}(\bm{y})$ introduce 
coupling between the
assortments for the
ad-supported and subscription modes, since both probabilities depend on the endogenous screening threshold $h^{(j)}(\bm{y})$. 
Third, the buy set $S$ is itself a discrete procurement decision, determining which content families are treated as reusable fixed-cost assets and which remain usage-based rentals. 
Our approach separates these difficulties in layers. We fix a candidate buy set $S$, relax the max-type purchase cost using   weights $\bm{\beta}$, and thereby obtain an upper-bounding objective that decomposes across user types.  
For each user type, we handle 
probability coupling by introducing another parameter $t=x^{(a,j)}/u^{(a,j)}$. 
Once $t$ is fixed, the endogenous mode-choice probabilities no longer couple the
assortment decisions for the
ad-supported and subscription modes: 
for each user type, the
problem decomposes
into an ads 
linear program
(LP)
and a subscription LP, each with one scalar equality constraint.
Each LP is solved by dual bisection using a
cardinality-constrained
MNL assortment algorithm, and the generated assortments are mixed in a postprocessing LP to recover primal feasibility.

Define
\begin{equation}
  \mc{B}_S
  :=
  \left\{
    \bm{\beta} = (\beta^{(m,j)}_\ell)_{\ell \in S, m \in \set{a,s}, j \in[J]} \in \R_+^{2 J |S|} :
    \sum_{j \in [J]}\sum_{m \in \{a,s\}} \beta^{(m,j)}_\ell = 1
    \text{ for every } \ell \in S
  \right\}. \label{eqn:B-S-def}
\end{equation}
For $\bm{\beta} \in \mc{B}_S$, define the $\bm{\beta}$-relaxed objective as follows:
\begin{align}
  G(S,\bm{\beta}, \bm{y})
  :=
  \sum_{j \in [J]}
  &\Bigg[
    \lambda\jj\bigl(\rads\sigma x\adj \padj(\bm{y}) + p \psubj(\bm{y})\bigr)
      \nonumber
    \\
    &\qquad
    - \lambda\jj \padj(\bm{y})
      \sum_{\ell \in [L] \setminus S} \gamma_\ell f_\ell\adj(\bm{y})
    - \lambda\jj \psubj(\bm{y})
      \sum_{\ell \in [L] \setminus S} \gamma_\ell f_\ell\subj(\bm{y})
  \Bigg] 
  \nonumber
  \\
  &
  - \sum_{\ell \in S} \eta_\ell
    \sum_{j \in [J]}\sum_{m \in \{a,s\}} \beta_\ell^{(m,j)} f_\ell^{(m,j)}(\bm{y}).
    \label{eqn:G-S-beta-def}
\end{align}
It is easy to check that
\begin{align}
  R(S, \bm{y}) 
\le G(S,\bm{\beta}, \bm{y})
\qquad \forall \bm{y}, \forall \bm{\beta} \in \mc{B}_S. \label{eqn:beta-relax-ineq}
\end{align}
After fixing $(S,\bm{\beta})$,
we often suppress 
the fixed dependence on
$(S,\bm{\beta})$,
and simply write $G(\bm{y})$.
By collecting terms belonging to each user type $j$, the relaxed objective decomposes across user types:
\begin{equation*}
  G(\bm{y})
  =
  \sum_{j \in [J]} G^{(j)}\bigl(\yadj, \ysubj\bigr).
\end{equation*}
We define 
\begin{align}
G\jj(S, \bm{\beta}) 
:= \max_{\yadj,\ysubj \in \simplex} \set{ G^{(j)}\bigl(\yadj, \ysubj\bigr)},
\label{eqn:G-j-def}
\end{align}
and $G(S, \bm{\beta}) = \sum_{j \in [J]}G\jj(S, \bm{\beta}) $.
Next, we discuss how to approximately compute 
$G\jj(S, \bm{\beta})$.

\paragraph{Ratio reduction.}
We first show in Lemma~\ref{lem:sub-zero-p}
that the subscription side admits a boundary reduction: 
If
the utility level $p$ is attainable (i.e., there exists $\ysubj$ such that $\usubj(\ysubj) \ge p$, or equivalently,
$\uumaxj = \max_{\bm{y} \in \simplex} \usubj(\bm{y}) \ge p$), 
then there exists an optimal solution in which the subscription utility is either $0$ or exactly $p$.
If $p$ is unattainable, the optimal subscription utility is $0$.

\begin{lemma} 
\label{lem:sub-zero-p}
Consider the problem
$ \max_{\yadj,\ysubj \in \simplex} \set{ G^{(j)}\bigl(\yadj, \ysubj\bigr)}.$
Suppose  $\uumaxj \ge p$. Then there exists an optimal solution with
\begin{equation*}
u^{(s,j)}_\star \in \{0,p\}.
\end{equation*}
If  $\uumaxj  < p$, then there exists an optimal solution with
\begin{equation*}
u^{(s,j)}_\star = 0.
\end{equation*}
\end{lemma}
By Lemma~\ref{lem:sub-zero-p}, when evaluating  $G\jj(S,\bm{\beta})$, we
may restrict attention to solutions where
the subscription-side utility
is either zero or   equal to $p$. 
Hence, on the subscription-active
branch, the term 
 $(\usubj - p)_+$ 
 vanishes. In this case, whenever the ad-supported assortment is active,
 i.e.,  $\uadj > 0$, the endogenous cutoff in the mode-choice probability
 reduces to 
\begin{align*}
    h\jj(\bm{y}) = \frac{\sigma \xadj(\yadj)}{\uadj(\yadj)}.
\end{align*}
Thus, %
Lemma~\ref{lem:sub-zero-p}\ implies that  the interaction between the
ad-supported assortment and the subscription assortment is captured by the
ad-side ratio 
\begin{align*}
t = \frac{\xadj(\yadj)}{\uadj(\yadj)}.
\end{align*}
This observation motivates a parametric programming approach
indexed by $t$. 
Once $t$ is fixed, the cutoff  $h\jj$
and 
 $\padj = \bfj{\sigma t}$ are fixed; 
 $\psubj$ is then fixed conditional on whether the subscription branch is active.
Conditional on this cutoff, the ad-supported and subscription
assortment decisions can then be optimized separately. 
  We next identify the range of ratios that can arise from an ad-active solution. 
  If  $\uadj = 0$ then Assumption~\ref{assump:MNL}(\ref{assm:alpha-implies-u-support})
  implies $\xadj = 0$, so the ad-supported mode is inactive. We represent
  this branch by  $t = + \infty$. If instead  $\uadj > 0$, then  
  it is easy to show that
  \begin{align*}
1/t = \frac{\uadj}{\xadj} \in [\uminj,\umaxj],
\end{align*}
or equivalently,
\begin{align}
t \in 
I^{(j)} := [t_{\min}^{(j)}, t_{\max}^{(j)}],
~t_{\min}^{(j)} := \frac{1}{\umaxj},
~t_{\max}^{(j)} := \frac{1}{\uminj}.
\label{eqn:t-I-def}
\end{align}
Based on the previous discussion,
we consider the following
parametric program: 
\begin{align} 
\max_{t \in I\jj \cup \set{+\infty}}
\set{    \max_{\yadS,\ysubS \in 
\Delta(\mc{A}_\kappa)} G\jj(\yadS,\ysubS) \mathrm{~s.t.~ }
  \frac{\xadj(\yadS)}{\uadj(\yadS)} = t, \usubj(\ysubS) \in \set{0,p}}, 
\label{eqn:ratio-reduction-step-1}
\end{align}
where 
in~\eqref{eqn:ratio-reduction-step-1}
we interpret the %
ratio
$\frac{\xadj(\yadS)}{\uadj(\yadS)} = + \infty$
when $\uadj(\yadS) = 0$.
Collecting terms, we see that
\begin{align*}
G\jj(\yadS,\ysubS) =
& \underbrace{\lambda\jj \padj(t)  \bigl(\rads\sigma x\adj(\yadS)
-  \sum_{\ell \in [L] \setminus S} \gamma_\ell
f_\ell\adj(\yadS)\bigr)
- \sum_{\ell \in S} \eta_\ell\beta_\ell^{(a,j)}
f_\ell\adj(\yadS)}_{:=\
G\adj(t,\yadS)}\\
& \mbox{} +  \underbrace{\lambda\jj \psubj(t, \ysubS)  \bigl(p
-  \sum_{\ell \in [L] \setminus S} \gamma_\ell
f_\ell\subj(\ysubS)\bigr)
- \sum_{\ell \in S} 
\eta_\ell\beta_\ell^{(s,j)}
f_\ell\subj(\ysubS)}_{:=\
G\subj(t,\ysubS)},
\end{align*}
where $\padj(t) = \bfj{\sigma t}$
and $\psubj(t, \ysubS) = F\jj(\sigma t) \I{\usubj(\ysubS) = p}$.
Accordingly, we consider
\begin{equation}
\label{eq:ratio-reduction-step-2}
\max_{t \in I\jj \cup \set{+\infty}}
G\adj(t) + G\subj(t),
\end{equation}
where for a finite $t$,
\begin{equation}
G\adj(t) :=
\begin{array}[t]{rl}
\max_{\bm{y} \in 
\Delta(\mc{A}_\kappa)}  &  G\adj(t,\bm{y}) \\
\mbox{ s.t. } & d\adj(t,\bm{y}) := \xadj(\bm{y}) - t \uadj(\bm{y}) = 0,
\end{array}
\label{eqn:ads-LP-def}
\end{equation}
and
\begin{align}
G\subj(t)
:=
\max\set{
\max_{\bm{z} \in \simplex} 
\set{ G\subj(t,\bm{z}): d\subj(p, \bm{z}) := u\subj(\bm{z}) - p= 0},0}.
\label{eqn:sub-LP-def}
\end{align}
The outer max in~\eqref{eqn:sub-LP-def}
follows from the fact that  the subscription profit is
zero when $u\subj = 0$. We follow the usual convention of setting the
value of the inner
linear program in~\eqref{eqn:sub-LP-def} to $-\infty$ if it 
is infeasible.
Since $\padj(\infty) = 0$, there is no ad revenue when $t = \infty$.
Therefore, it is optimal to set $\yadj  = \bm{y}_{\emptyset}$  yielding
$G\adj(\infty) = 0$.
We define   $G\subj(\infty)$  by~\eqref{eqn:sub-LP-def},
where we set $\psubj(\infty, \ysubS)= \I{\usubj(\ysubS) =  p}$.
Next, we show that the optimization problem
\eqref{eq:ratio-reduction-step-2} recovers the value of $G\jj(S, \bm{\beta})$. %
\begin{lemma}
\label{lemma:ratio-reduction-in-G}
Let
\begin{align}
\gej(t) =     \geadj(t) +    \gesubj(t). 
\label{eqn:G-total-e-j-def}
\end{align}
 Then
\begin{align}
G\jj(S, \bm{\beta})
= 
\max_{t\in I^{(j)} \cup \set{\infty}} \gej(t).
\label{eqn:G-can-be-solved-in-ratio}
\end{align}
\end{lemma}

Although~\eqref{eqn:G-can-be-solved-in-ratio} resembles a direct
reparameterization, showing $G\jj(S, \bm{\beta}) \ge \max_{t\in I^{(j)} \cup \set{\infty}} \gej(t)$
is not immediate because the
ad and subscription problems at a fixed $t$ are optimized separately.
In particular, the empty ad assortment satisfies the balance constraint
$x^{(a,j)}-t u^{(a,j)}=0$ for every finite $t$, even though it corresponds
economically to the ad-inactive branch $t=\infty$. The proof shows that
whenever the ad-side value is positive, an optimal ad solution must have
positive utility and therefore genuinely induces the ratio
$t=x^{(a,j)}/u^{(a,j)}$, so it can be combined with an optimal subscription
solution at the same cutoff. When the ad-side value is zero, any positive
subscription value obtained at a finite $t$ is weakly dominated by the
subscription-only $t=\infty$ branch. Thus,
\eqref{eqn:G-can-be-solved-in-ratio} gives an exact decomposition of the
original user type-level problem rather than merely an upper bound obtained by
optimizing the two branches separately.

The proof of Lemma~\ref{lemma:ratio-reduction-in-G}
is deferred to Appendix~\ref{sec:proof-main-lemma-thm}.
The resulting value function $\gej(\cdot)$ 
in~\eqref{eqn:G-total-e-j-def}
is generally nonconvex, nonsmooth, and lacks exploitable shape structure.  
Therefore, without imposing additional assumptions on the ad-tolerance distribution or content-family primitives,
we approximate the outer maximization using a finite
grid and control the resulting discretization loss through a Lipschitz bound.

The ratio $t$ has a direct economic meaning: on the active ads branch, the ad-choice cutoff becomes $h^{(j)}=\sigma t$, so $t=x^{(a,j)}/u^{(a,j)}$ is the value-adjusted ad burden induced by the ad assortment. 
Equivalently, $1/t=u^{(a,j)}/x^{(a,j)}$ is the  
expected utility per click under the ads mode,
so a smaller $t$ means that the ad mode delivers more value for a given ad burden and is therefore acceptable to more users. 
This is the key difference from a classical MNL assortment model, where the assortment only affects within-assortment choice through a single objective; here it also affects cross-mode self-selection through the endogenous screening value $t$. Moreover, for any deterministic assortment $A$ with $u^{(j)}(A)>0$, the ratio $x^{(j)}(A)/u^{(j)}(A)$ is a convex combination of singleton ratios $1/u^{(j)}_\ell$, so the extreme feasible ratios are attained by singleton assortments, even though the optimal fixed-$t$ assortment itself need not be a singleton.

We   describe how  Algorithm~\ref{alg:mnl-solver-all},
with details
presented in Section~\ref{sec:aggregation},
approximately computes $G(S,\bm{\beta})$.
Algorithm~\ref{alg:mnl-solver-all}
is motivated by~\eqref{eqn:G-can-be-solved-in-ratio} and
uses  two   subroutines.
The first subroutine described in
Section~\ref{sec:abstract-LP}
is
Algorithm~\ref{alg:lp-postprocessing}: given a finite candidate action set, it computes the optimal feasible mixture
over that set by solving the corresponding restricted primal LP.
The second subroutine, 
Algorithm~\ref{alg:mnl-oracle} in Section~\ref{sec:MNL-oracle}, 
computes an optimal cardinality-constrained assortment
using the algorithm in~\citep{RusmevichientongShSh10}.
We then combine the MNL algorithm with a 
bisection procedure to generate
candidate assortments for the ads branch~\eqref{eqn:ads-LP-def}
and the subscription branch~\eqref{eqn:sub-LP-def},
and apply Algorithm~\ref{alg:lp-postprocessing}
to 
generate   near-optimal solutions using the candidate assortments.
This yields the fixed-$(j,t)$ solvers in
Algorithm~\ref{alg:mnl-fixed-jt-ads}
and Algorithm~\ref{alg:mnl-fixed-jt-sub}.
Finally, Algorithm~\ref{alg:mnl-solver-all}
aggregates these fixed-$(j,t)$ solutions over the ratio grid and across user types to produce an approximation to
$G(S,\bm{\beta})$. We next discuss the performance guarantees.

Before presenting 
Theorem~\ref{thm:mnl-inner-separated-gap}, we 
define some useful constants.
For a %
vector $\bm{v}=(v_\ell)_{\ell \in [L]}$ and an integer $m \ge 1$, let
\begin{equation}
\TOP_m \set{\bm{v}}
:=
\sum_{q=1}^{\min\{m,L\}} v_{(q)}, 
\label{eqn:top-def}
\end{equation}
where $v_{(1)} \ge \cdots \ge v_{(L)}$ are the entries of $\bm{v}$ sorted in nonincreasing order.
Define $\bm{q}^{S,\bm{\beta}}$ as
\begin{align}
    \begin{split}
         q^{S,\bm{\beta}}_\ell
  :=
  \begin{cases}
    \eta_\ell \left(1-\min_{j \in [J],m \in \set{a,s}}\beta_\ell^{(m,j)}\right), & \ell \in S,\\
    0, & \ell \notin S,
  \end{cases} 
    \end{split} \label{eqn:q-S-def}
\end{align} 
and
the per-family   buy-versus-rent penalty vector $\bm{w}^S$ with entries
\begin{equation}
\label{eq:wS-mnl-def}
w_\ell^{S}
:=
\begin{cases}
\eta_\ell, &
\ell \in S,\\[0.4em]
\Bigl(
\gamma_\ell \sum_{j=1}^J \lambda\jj   - \eta_\ell
\Bigr)_+, &
\ell \notin S.
\end{cases}
\end{equation}

\begin{theorem} 
\label{thm:mnl-inner-separated-gap}
Fix $(S,\bm{\beta},\sigma,p)$ and  $K_T \ge 2$.  
Assume 
for every type $j$ and every
$t,t'\in I^{(j)}$,
\begin{align}
    \ABS{G\jj(S,\bm{\beta}, t)-G\jj(S,\bm{\beta}, t')}
& \le
\lipsig |t-t'|, 
\label{eqn:lip-cond-holds}
\end{align}
where $\lipsig$  is nondecreasing   in the vector $\bm{\eta}$, and does
not depend on $(S, \bm{\beta})$.
In~\eqref{eqn:lip-cond-holds},
we write $G\jj(S,\bm{\beta}, t)$
for
$\gej(t)$ in~\eqref{eqn:G-total-e-j-def}
to emphasize its dependence on $S$ and $\bm{\beta}$. \\
Let $(\rhat(S,\bm{\beta},K_T), \ghat(S,\bm{\beta},K_T), \yhat(S,\bm{\beta},K_T) )\gets$~\ref{alg:mnl-solver-all}($S,\bm{\beta},K_T$), where
$\rhat(S,\bm{\beta},K_T)$ is the objective value of the solution
returned by~\ref{alg:mnl-solver-all}.
Recall that $R\opt$ in~\eqref{eqn:problem-buy-or-rent}
is the platform's optimal profit.
Then,  
\begin{align}
\begin{split}
   R\opt-\rhat(S,\bm{\beta},K_T)  
& \leq\ \TOP_{2 J}
\set{\bm{w}^S}
+
\TOP_{2 J}\set{\bm{q}^{S,\bm{\beta}}} \\
& \hspace*{2em} 
+
\frac{1}{2(K_T-1)}
\sum_{j=1}^J
\lipsig
\bigl(t_{\max}^{(j)}-t_{\min}^{(j)}\bigr) 
+
\frac{1}{K_T - 1}
\sum_{j=1}^J (\Gamma\adj +\Gamma\subj ),
\end{split}
\label{eqn:thm-original-gap}
\end{align}
where $\Gamma\adj, \Gamma\subj$ are
defined in~\eqref{eqn:Gamma-j-a-def} when
we introduce Algorithm~\ref{alg:mnl-solver-all} 
in Section~\ref{sec:aggregation}.
\end{theorem}

Note that in~\eqref{eqn:lip-cond-holds}, we assume the existence of a Lipschitz constant~$\lipsig$ for
the map $t \mapsto G\jj(S,\bm{\beta}, t)$.
In Section~\ref{sec:lip-consant}, we discuss how to construct such
an explicit   constant
in~\eqref{eqn:L-ratio-def}.
Algorithm~\ref{alg:mnl-solver-all} solves the fixed-$(j,t)$ subproblems independently across user types and sampled ratio points, 
and only then performs a final maximization over $t$ and summation over $j$. 
Consequently, the outer grid-search layer is naturally parallelizable and therefore well suited to large-scale implementations.

Theorem~\ref{thm:mnl-inner-separated-gap} decomposes the error into three interpretable terms in~\eqref{eqn:thm-original-gap}.
The first term is the buy-set gap, which captures the loss from committing to a candidate buy set \(S\). 
The second term is the \(\bm{\beta}\)-relaxation gap, which captures the loss from relaxing the shared procurement coupling and then recovering a feasible procurement solution. 
The third term is the algorithmic approximation error, which arises from the ratio-grid discretization and the inner bisection routines. 
The first two terms are instance-dependent and do not vanish by refining the ratio grid, whereas,
for fixed primitive model data,
the algorithmic term is \(O(1/K_T)\) under the grid and bisection 
parameters used in the algorithm.

By Lemma~\ref{lem:complexity-thm-callcount-order}, Algorithm~\ref{alg:mnl-solver-all}
makes 
$\noracle(K_T)=O(JK_T\log K_T)$ 
calls to 
Algorithms~\ref{alg:mnl-oracle} and~\ref{alg:exact-card-mnl}, 
including calls made within the boundary solvers.
As
discussed in 
Sections~\ref{sec:MNL-oracle},~\ref{subsec:sub-boundary}, and~\ref{sec:complexity}, each such call has worst-case
complexity $O(L^2\log L)$, yielding total runtime
$\widetilde O(JK_TL^2)$.
For simplicity of notation, we use $\toracle = O(L^2\log L)$ to represent the complexity of those boundary solvers as well.
In addition, define the heuristic buy set $S\opt$ by
\begin{align}
\I{\ell \in S\opt}
=
\I{
2\eta_\ell
\le
\gamma_\ell \sum_{j=1}^J \lambda\jj
}.
\label{eqn:S-opt-def}
\end{align}
By the optimality-gap bound in~\eqref{eqn:thm-original-gap} and the definition of
$\bm w^S$ in~\eqref{eq:wS-mnl-def}, this choice satisfies
\begin{align*}
S\opt
\in
\argmin_{S\subseteq [L]}
\left\{
\TOP_{2J}\set{\bm w^S}
\right\}.
\end{align*}
Thus, $S\opt$ minimizes the component of the gap bound that depends on
$\bm w^S$.
In addition, it is clear that the uniform $\bm{\beta}$, i.e.,
\begin{align}
    (\bm{\beta}\opt)_{\ell}\adj
=
(\bm{\beta}\opt)_{\ell}\subj
=
\frac{1}{2J},
\qquad
j\in[J], \ell\in S\opt,
\label{eqn:uniform-beta-general}
\end{align} satisfies $\bm{\beta}\opt \in \argmin_{\bm{\beta} \in \mc{B}_{S\opt}} \set{\TOP_{2 J}\set{\bm{q}^{S\opt,\bm{\beta}}}}$,
where $\mc{B}_S$ is defined
in~\eqref{eqn:B-S-def}.

We now study a regime in which 
 consumer masses scale
proportionally
while other parameters %
remain fixed.
This helps us understand the performance of our approach in the typical regimes of modern platforms.
Starting from a base mass vector $\bm{\lambda}$, the size-$N$
market has masses $\bm{\lambda}^{(N)} := N\bm{\lambda}.$ 
Let $\RL$   
denote the optimal value of problem~\eqref{eqn:problem-buy-or-rent} under the base mass vector $\bm{\lambda}$
so 
that $\RL > 0$ by Assumption~\ref{assump:MNL}(\ref{item:opt-profit-positive}).
Let  
$\lipsig$ in~\eqref{eqn:lip-cond-holds},
$\Gamma\adj$ and $\Gamma\subj$ in~\eqref{eqn:Gamma-j-a-def} denote the constants evaluated at  the
 base-mass vector $\bm{\lambda}$,
and define
\begin{align}
    A_{\sigma,p}
& :=
\frac{1}{\RL}
\Paran{
\frac12 \sum_{j=1}^J
\lipsig
\bigl(t_{\max}^{(j)}-t_{\min}^{(j)}\bigr)
+
\sum_{j=1}^J
\bigl(
\Gamma\adj
+
\Gamma\subj
\bigr)}, \label{eqn:A-sigma-p-def} \\ 
B_{\sigma,p}
& :=
\frac{2\TOP_{2J}\{(\eta_\ell)_{\ell\in[L]}\}}{\RL}.  \label{eqn:B-sigma-p-def}
\end{align}

For each $N \ge 1$, 
let $S_{N}\opt$ be the heuristic buy set~\eqref{eqn:S-opt-def}
of this 
instance with consumer measure $\bm{\lambda}^{(N)}$ as
\begin{align}
   \I{ \ell \in S_{N}\opt}  = \I{ 
   2 \eta_\ell \le
\gamma_\ell \sum_{j=1}^J \lambda\jj   N
}.  \label{eqn:S-N-opt-def}
\end{align} 
Following~\eqref{eqn:uniform-beta-general},
we choose the uniform relaxation weights
\begin{align}
    (\bm{\beta}_N\opt)_{\ell}\adj
=
(\bm{\beta}_N\opt)_{\ell}\subj
=
\frac{1}{2J},
\qquad
j\in[J],\ \ell\in S_{N}\opt. \label{eqn:uniform-beta}
\end{align}

\begin{theorem}
\label{thm:asymptotic-opt}
Fix a base consumer-mass vector $\bm{\lambda} \in \mathbb{R}_{++}^J$, i.e., $\lambda\jj >
0$ for all $j$.
For each $N \ge 1$, let $\bm{\lambda}^{(N)} := N\bm{\lambda},$
and let $R_N$ denote the optimal value of the
problem~\eqref{eqn:problem-buy-or-rent}
under $\bm{\lambda}^{(N)}$.
Assume~\eqref{eqn:lip-cond-holds}
holds with constants  $W\jj$
that are coordinatewise nondecreasing in 
$\bm\eta$
and independent of
 $(S,\bm \beta)$, as required in Theorem~\ref{thm:mnl-inner-separated-gap}.
 Recall
$S_{N}\opt$  defined in~\eqref{eqn:S-N-opt-def}
and $\bm{\beta}_N\opt$  defined in~\eqref{eqn:uniform-beta}.
Let
$(\rhat_N(K_T), 
\ghat_N(K_T),
\widehat{\bm{y}}_N(K_T)) \gets$~\ref{alg:mnl-solver-all}($S_{N}\opt,\bm{\beta}_N\opt,K_T$).
Then, for every $N \ge 1$ and every $K_T \ge 2$,
\[
\frac{\rhat_N(K_T)}
{R_N}
\ge
1
-
\frac{A_{\sigma,p}}
{K_T-1}
-
\frac{B_{\sigma,p}}{N},
\]
where $A_{\sigma,p}$ and $B_{\sigma,p}$
are defined in~\eqref{eqn:A-sigma-p-def} and~\eqref{eqn:B-sigma-p-def}.
\end{theorem}

Our procedure yields a closed-form buying set $S_N^\star$~\eqref{eqn:S-N-opt-def}
with an intuitive
managerial interpretation: content family $\ell$ is bought once market scale is large enough relative to its buy-versus-rent cost ratio, namely when
\[
2\eta_\ell \le \gamma_\ell \sum_{j=1}^J (\lambda^{(N)})\jj = N\gamma_\ell \sum_{j=1}^J \lambda\jj.
\]
Thus, content families
with low fixed licensing fees or high per-use rental costs are bought first, and as the market-scaling factor $N$ increases, additional content families cross the buy threshold and should be bought rather than rented.
Theorem~\ref{thm:asymptotic-opt} shows that
the choice $S_N^\star$~\eqref{eqn:S-N-opt-def}
and our heuristic lead
to an asymptotically optimal solution.
Moreover,   the   total runtime is still $\widetilde{O}(J K_T L^2)$.
\section{Numerical Experiments}
\label{sec:numerics}

To complement the asymptotic guarantee in Theorem~\ref{thm:asymptotic-opt},
Section~\ref{sec:exp-convergence} evaluates
Algorithm~\ref{alg:mnl-solver-all} on practical
instances and shows that it is
already competitive under moderate market-scaling factors \(N\) and moderate
outer ratio-grid resolutions \(K_T\).
The experiments show that the returned 
solutions stabilize as \(K_T\) increases
and that the threshold buy set~\eqref{eqn:S-N-opt-def}
provides a useful hybrid procurement structure relative
to pure rental $(S = \emptyset)$
and pure buying $(S = [L])$ of all content. 
We then design a preference-concentration case study
in Section~\ref{sec:preference-concentration}, where demand shifts toward some shared
content, and show how our
heuristic adapts its planned buy set to this change.

\subsection{Performance of Algorithm~\ref{alg:mnl-solver-all}}
\label{sec:exp-convergence}

We first evaluate the performance of Algorithm~\ref{alg:mnl-solver-all} in a structured synthetic setting. We also develop a mixed-integer programming (MIP) formulation of the problem in Section~\ref{sec:MIP}, which serves as a benchmark. To ensure computational feasibility, we focus on relatively small instances.
We use \textsf{Gurobi}
with a solver time limit of 300 seconds.
We generated multiple independent baseline instances according to the specification below and observed qualitatively similar results across them. For brevity, we report one baseline instance in detail; Section~\ref{sec:preference-concentration}
uses this same instance for the preference-concentration experiment.

Each generated instance has \(J=5\) user types, \(L=10\) content families, and
assortment capacity \(\kappa=3\). Content families \(1,\ldots,5\) are niche content
families, one primarily targeted to each user type. Content families \(6,\ldots,10\) are
reusable content families with moderate attraction across user types. This baseline
catalog is intentionally simple.
The design isolates the main procurement trade-off in the model: niche content
is valuable for personalization but has limited reuse, whereas reusable content
can serve multiple user groups and monetization modes and may therefore justify
buying. We visualize the setting in Figure~\ref{fig:synthetic-catalog}.

 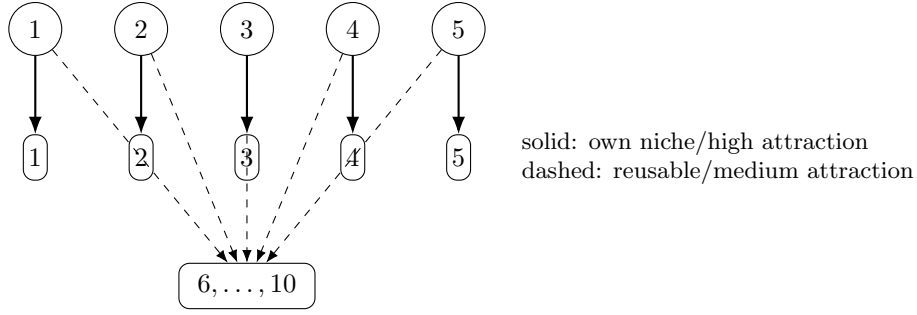
\begin{figure}[t]
\centering
\begin{tikzpicture}[
    >=Latex,
    user/.style={circle, draw, minimum size=7mm, inner sep=1pt},
    item/.style={rectangle, draw, rounded corners, minimum height=6mm, inner sep=2pt},
    high/.style={->, thick},
    med/.style={->, dashed},
    every node/.style={font=\small}
]
\node[user] (u1) at (0,2.2) {$1$};
\node[user] (u2) at (1.4,2.2) {$2$};
\node[user] (u3) at (2.8,2.2) {$3$};
\node[user] (u4) at (4.2,2.2) {$4$};
\node[user] (u5) at (5.6,2.2) {$5$};

\node[item] (p1) at (0,0.5) {$1$};
\node[item] (p2) at (1.4,0.5) {$2$};
\node[item] (p3) at (2.8,0.5) {$3$};
\node[item] (p4) at (4.2,0.5) {$4$};
\node[item] (p5) at (5.6,0.5) {$5$};

\node[item, minimum width=18mm] (r) at (2.8,-1.2) {$6,\ldots,10$};

\draw[high] (u1) -- (p1);
\draw[high] (u2) -- (p2);
\draw[high] (u3) -- (p3);
\draw[high] (u4) -- (p4);
\draw[high] (u5) -- (p5);

\draw[med] (u1) -- (r);
\draw[med] (u3) -- (r);
\draw[med] (u5) -- (r);
\draw[med] (u2) -- (r);
\draw[med] (u4) -- (r);

\node[align=left, font=\footnotesize, anchor=west] at (6.3,0.5)
{solid: own niche/high attraction\\ dashed: reusable/medium attraction};
\end{tikzpicture}
\caption{Content families \(1,\ldots,5\) are niche
content families, one primarily targeted to each user type. Content families \(6,\ldots,10\)
are reusable content families with moderate attraction across user types. Solid edges
mark each user type's own high-attraction niche family. Dashed edges indicate positive
attraction for reusable content. For visual simplicity, the figure omits many
positive but lower-attraction edges, including each user type's positive attraction for other
user types' niche families and additional positive links to reusable content. The
realized baseline attraction values are reported in
Table~\ref{tab:appendix-baseline-alpha} in
Appendix~\ref{app:additional-experimental-details}; under the numerical
specification $u_\ell^{(j)}=\alpha_\ell^{(j)}$, these are also the utility values.}
\label{fig:synthetic-catalog}
\end{figure}

For each instance, all user masses are set to one: $\lambda\jj = 1$ for all $j$.
User ad tolerance is drawn
from the common distribution
$
\chi\jj\sim \mathrm{Unif}[0.5,3.0], j=1,\ldots,5,
$
so heterogeneity in this experiment comes from content preferences rather than from
differences in ad-tolerance distributions. We set the advertising revenue parameter to
\(\rads =3.0\) for all $j$. 
Throughout the experiments, we set
$u_\ell^{(j)}=\alpha_\ell^{(j)}$ for all $(j,\ell)$.
This restriction is used only in the numerical study: a content family with higher
MNL attraction also delivers higher utility when consumed. The analytical results do
not rely on this restriction and continue to allow $\alpha_\ell^{(j)}$ and
$u_\ell^{(j)}$ to differ.
We generate preferences  
from three role-based distributions:
$
D_{\mathrm{low}}=\mathrm{Unif}[0.60,0.80], 
D_{\mathrm{med}}=\mathrm{Unif}[1.15,1.45], 
D_{\mathrm{high}}=\mathrm{Unif}[2.80,3.40].
$
For each user type \(j\), its own niche family receives an independent draw from
\(D_{\mathrm{high}}\), while other users' niche families receive independent draws from
\(D_{\mathrm{low}}\). 
All reusable families, i.e., content families \(6,\ldots,10\), receive independent draws
from \(D_{\mathrm{med}}\).
All generated primitives are
rounded to two decimal places. 
Procurement costs are generated from the attraction matrix. For each content family
\(\ell\), the rental cost is
$
\gamma_\ell
=
0.7\cdot \frac{1}{J}\sum_{j=1}^J \alpha_\ell^{(j)}.
$
Thus, content families with higher average attractiveness are more expensive to rent.
Buy costs are role-specific multiples of rental costs:
$
\eta_\ell=\mathrm{ratio}_\ell\gamma_\ell,
$
where
\begin{align}
    \mathrm{ratio}_\ell=
\begin{cases}
5.50, & \ell\in\{1,\ldots,5\},\\
3.00, & \ell\in\{6,7,8,9\},\\
2.30, & \ell=10.
\end{cases} \label{eqn:buy-cost-ratio}
\end{align}
These ratios place the generated instances near the buy-versus-rent threshold of our heuristic. 
At
\(N=1\), the total user mass is $\sum_{j=1}^5 \lambda\jj = 5$, so the threshold condition
\(2\eta_\ell\le \gamma_\ell \sum_{j=1}^5 \lambda\jj N\)   reduces to \(\mathrm{ratio}_\ell\le 2.5\).
Thus,  content family \(10\) is  naturally attractive for the platform to buy
at \(N=1\) because it has the
lowest buy-cost multiplier, content families \(6,\ldots,9\) are close to the threshold, and
niche families require larger market scale before buying becomes attractive.
We fix \((\sigma,p)=(1.50,1.80)\). 
For each generated
instance and each market scale
$
N\in\{1,2,5,10,100\},
$
we scale user masses according to \(\bm \lambda^{(N)}=N \bm \lambda\), holding all other primitives
\(\alpha,u,\gamma,\eta\) fixed. 
For every
\(N\), the planned buy set is selected by the threshold rule~\eqref{eqn:S-N-opt-def}:
$
\ell\in S_N\opt
\quad\Longleftrightarrow\quad
2\eta_\ell\le \gamma_\ell \sum_{j=1}^J \lambda\jj N.
$
Given \(S_N\opt\), we use the uniform relaxation weights in~\eqref{eqn:uniform-beta}:
$
(\bm{\beta}_N\opt)_{\ell}^{(a,j)}
=
(\bm{\beta}_N\opt)_{\ell}^{(s,j)}
=
\frac{1}{2J},
 j\in[J],\ \ell\in S_N\opt.
$
For each   \(N\),  we run
Algorithm~\ref{alg:mnl-solver-all} over
$
K_T\in\{5,9,17,33,65,129\}.
$
These values satisfy \(K_T-1\in\{4,8,16,32,64,128\}\), so each refinement approximately
halves the outer ratio-grid mesh.
We summarize our results in Table~\ref{tab:sec51-main}.

\begin{table}[t]
\centering
\scriptsize
\resizebox{\textwidth}{!}{
\begin{tabular}{c c c c c c c c}
\hline
\(N\) & \(K_T\) & \(S_N\opt\) & \(R(S_N\opt, \yhat)\) (gap to MIP full)
& MIP full & MIP\((S_N\opt)\)
& MIP\(([L])\) & MIP\((\emptyset)\) \\
\hline
1 & 5   & \(\{10\}\) & 13.89 (11.37\%) & 15.67 & 15.37 & 11.19 & 15.01 \\
1 & 9   & \(\{10\}\) & 14.73 (5.98\%) & 15.67 & 15.37 & 11.19 & 15.01 \\
1 & 17  & \(\{10\}\) & 15.17 (3.19\%) & 15.67 & 15.37 & 11.19 & 15.01 \\
1 & 33  & \(\{10\}\) & 15.29 (2.43\%) & 15.67 & 15.37 & 11.19 & 15.01 \\
1 & 65  & \(\{10\}\) & 15.31 (2.26\%) & 15.67 & 15.37 & 11.19 & 15.01 \\
1 & 129 & \(\{10\}\) & 15.32 (2.20\%) & 15.67 & 15.37 & 11.19 & 15.01 \\
\hline
2 & 5   & \(\{6,7,8,9,10\}\) & 27.75 (14.03\%) & 32.28 & 32.28 & 26.31 & 30.01 \\
2 & 9   & \(\{6,7,8,9,10\}\) & 30.30 (6.13\%) & 32.28 & 32.28 & 26.31 & 30.01 \\
2 & 17  & \(\{6,7,8,9,10\}\) & 30.31 (6.12\%) & 32.28 & 32.28 & 26.31 & 30.01 \\
2 & 33  & \(\{6,7,8,9,10\}\) & 30.69 (4.92\%) & 32.28 & 32.28 & 26.31 & 30.01 \\
2 & 65  & \(\{6,7,8,9,10\}\) & 30.69 (4.95\%) & 32.28 & 32.28 & 26.31 & 30.01 \\
2 & 129 & \(\{6,7,8,9,10\}\) & 30.71 (4.86\%) & 32.28 & 32.28 & 26.31 & 30.01 \\
\hline
5 & 5   & \(\{1,\ldots,10\}\) & 67.82 (17.48\%) & 82.19 & 81.12 & 81.12 & 75.03 \\
5 & 9   & \(\{1,\ldots,10\}\) & 74.73 (9.08\%) & 82.19 & 81.12 & 81.12 & 75.03 \\
5 & 17  & \(\{1,\ldots,10\}\) & 75.73 (7.86\%) & 82.19 & 81.12 & 81.12 & 75.03 \\
5 & 33  & \(\{1,\ldots,10\}\) & 77.13 (6.16\%) & 82.19 & 81.12 & 81.12 & 75.03 \\
5 & 65  & \(\{1,\ldots,10\}\) & 77.22 (6.04\%) & 82.19 & 81.12 & 81.12 & 75.03 \\
5 & 129 & \(\{1,\ldots,10\}\) & 77.41 (5.81\%) & 82.19 & 81.12 & 81.12 & 75.03 \\
\hline
10 & 5   & \(\{1,\ldots,10\}\) & 153.38 (11.92\%) & 174.14 & 174.14 & 174.14 & 150.06 \\
10 & 9   & \(\{1,\ldots,10\}\) & 165.47 (4.98\%) & 174.14 & 174.14 & 174.14 & 150.06 \\
10 & 17  & \(\{1,\ldots,10\}\) & 167.71 (3.69\%) & 174.14 & 174.14 & 174.14 & 150.06 \\
10 & 33  & \(\{1,\ldots,10\}\) & 170.06 (2.34\%) & 174.14 & 174.14 & 174.14 & 150.06 \\
10 & 65  & \(\{1,\ldots,10\}\) & 170.26 (2.23\%) & 174.14 & 174.14 & 174.14 & 150.06 \\
10 & 129 & \(\{1,\ldots,10\}\) & 170.55 (2.06\%) & 174.14 & 174.14 & 174.14 & 150.06 \\
\hline
100 & 5   & \(\{1,\ldots,10\}\) & 1693.49 (8.60\%) & 1852.74 & 1852.74 & 1852.74 & 1500.62 \\
100 & 9   & \(\{1,\ldots,10\}\) & 1798.81 (2.91\%) & 1852.74 & 1852.74 & 1852.74 & 1500.62 \\
100 & 17  & \(\{1,\ldots,10\}\) & 1816.88 (1.94\%) & 1852.74 & 1852.74 & 1852.74 & 1500.62 \\
100 & 33  & \(\{1,\ldots,10\}\) & 1843.02 (0.52\%) & 1852.74 & 1852.74 & 1852.74 & 1500.62 \\
100 & 65  & \(\{1,\ldots,10\}\) & 1845.02 (0.42\%) & 1852.74 & 1852.74 & 1852.74 & 1500.62 \\
100 & 129 & \(\{1,\ldots,10\}\) & 1847.14 (0.30\%) & 1852.74 & 1852.74 & 1852.74 & 1500.62 \\
\hline
\end{tabular}}
\caption{Performance of Algorithm~\ref{alg:mnl-solver-all}
and relative gap to the full-model MIP benchmark under market scaling for one
generated instance. For each market scale \(N\) and ratio-grid size \(K_T\),
\(R(S_N\opt,\yhat)\) is the value returned by
Algorithm~\ref{alg:mnl-solver-all}. The MIP columns report benchmarks from the
explicit formulation: the full endogenous buy-versus-rent problem, the problem with
the buy set fixed to \(S_N\opt\), the pure-buying benchmark \(S=[L]\), and the
pure-rental benchmark \(S=\emptyset\). Percentages in parentheses are relative gaps
to the reported full-model MIP incumbent. The full-model MIP reached the
300-second limit for \(N=1\) and \(N=5\), with final solver gaps of \(3.06\%\) and
\(0.99\%\), respectively; it was solved to certified optimality for
\(N\in\{2,10,100\}\). All fixed-buy-set and pure-procurement benchmarks were solved
to certified optimality. Under \(u_\ell^{(j)}=\alpha_\ell^{(j)}\), the unnormalized
grid-and-inner-search numerator in~\eqref{eqn:A-sigma-p-def} is
\(3.1312\times10^7\).
Using the feasible lower bound
$\underline{\RL}=15.6687\le \RL$,
we obtain the conservative bounds
$A_{\sigma,p}\le 1.9984\times 10^6$ and
$B_{\sigma,p}\le 4.6156$. 
Although this worst-case certificate is
conservative, the empirical results show strong performance at moderate values of $K_T$.
}
\label{tab:sec51-main}
\end{table}

Table~\ref{tab:sec51-main} highlights two main findings. First, our algorithm exhibits strong
empirical performance
as the outer ratio grid is refined, especially for $N=100$:
for $K_T \ge 33$, we observe that
the relative optimality gap of our approach is less than $1\%$;
the gap then stabilizes for
\(K_T \ge 65\).
The results
are consistent with the convergence logic of
Theorem~\ref{thm:asymptotic-opt}, which predicts that the grid-search component of the error
decays with \(1/(K_T-1)\). 
The theorem certificate is conservative in magnitude because it is based
on worst-case constants, but the empirical values show that moderate grid sizes already deliver
high-quality primal solutions.
Small nonmonotonicities can occur even though the ratio grids~\eqref{eqn:outer-grid-def}
used here are nested. The reason is that Table~\ref{tab:sec51-main} reports the original profit 
$R(S_N\opt, \yhat)$ 
of the   solution returned by Algorithm~\ref{alg:mnl-solver-all} at each value of  $K_T$, whereas the algorithm 
attempts to compute the relaxed objective $G(S,\bm{\beta})$  and then evaluates the returned solution
under the original objective.

Second, the MIP benchmarks demonstrate the importance of procurement design. The columns
MIP\((S_N\opt)\), MIP\(([L])\), and MIP\((\emptyset)\) evaluate the same assortment-design
problem under different procurement restrictions: the threshold buy set $S_N\opt$, pure buying, and pure
rental. 
These comparisons show that the procurement choice is  meaningful:
in all cases, we observe that $\mathrm{{MIP}}(S_N\opt) \ge \max\set{\mathrm{{MIP}}([L]), \mathrm{{MIP}}(\emptyset)}$.
At \(N=1\), the threshold rule buys only content family \(10\), which has the lowest
buy-cost multiplier~\eqref{eqn:buy-cost-ratio}, and the resulting fixed-set benchmark improves on both
pure rental and pure buying. At \(N=2\), the threshold set expands to include
all reusable families \(6,\ldots,10\), again outperforming both pure procurement
extremes.
For larger market scales, all content families cross the threshold, so MIP\((S_N\opt)\) coincides
with the pure-buying benchmark and remains above pure rental. Thus, the threshold rule adapts to
market scale in the expected way: it preserves rental flexibility at small scale, selectively buys
high-reuse content at intermediate scale, and moves toward buying all content when aggregate
rental exposure becomes sufficiently large.

The returned assortment distributions are reported in
Table~\ref{tab:appendix-market-scaling-distributions} of
Appendix~\ref{app:additional-experimental-details}. 
They provide 
 additional insight into
the   structure of the solution. 
On the ad side, the returned
assortments typically use multiple content families, often combining a user
type's
own niche family with reusable families from content families \(6,\ldots,10\). This is
intuitive because ad revenue is click-based: adding attractive reusable families
can increase the inside-choice probability and therefore advertising revenue.
On the subscription side, the returned assortments generally have smaller sizes.
This reflects the different economics of the subscription branch: once expected
utility reaches the subscription threshold \(p\), additional utility does not
increase subscription revenue but can increase consumption and procurement
costs. Thus, the recovered subscription assortment tends to use a more targeted
assortment to attain the subscription boundary, while the ad assortment more
often uses larger assortments to generate clicks.

In addition to profit and the procurement decision, we report two postprocessing metrics that describe how the returned solution serves users and content demand.
For an assortment solution $\bm y$, define the admitted fraction as
\begin{equation}
\operatorname{Adm}(\bm y)
:=
\frac{
\sum_{j=1}^{J}\lambda^{(j)}
\left(
P^{(a,j)}(\bm y)+P^{(s,j)}(\bm y)
\right)}
{\sum_{j=1}^{J}\lambda^{(j)}}.
\label{eqn:admission-fraction}
\end{equation}
For a fixed buy set $S$,
define the realized rental flow into content family $\ell$ as
\begin{equation}
Q_{\ell}^{\mathrm{rent}}(\bm y;S)
:=
\I{\ell\notin S}
\sum_{j=1}^{J}\lambda^{(j)}
\left[
P^{(a,j)}(\bm y)f_{\ell}^{(a,j)}(\bm y)
+
P^{(s,j)}(\bm y)f_{\ell}^{(s,j)}(\bm y)
\right],
\label{eqn:realized-rental-flow}
\end{equation}
and let
\begin{equation}
Q^{\mathrm{rent}}(\bm y;S)
:=
\sum_{\ell=1}^{L}Q_{\ell}^{\mathrm{rent}}(\bm y;S),
\qquad
\bar Q^{\mathrm{rent}}(\bm y;S)
:=
\frac{Q^{\mathrm{rent}}(\bm y;S)}
{\sum_{j=1}^{J}\lambda^{(j)}} \in [0,1],
\label{eqn:total-rental-flow}
\end{equation}
denote total rental flow and rental flow per arrival, respectively.
We compute these quantities,
\eqref{eqn:admission-fraction}--\eqref{eqn:total-rental-flow}, for the returned policies at $K_T=129$. 
Detailed family-level rental flows are reported in 
Table~\ref{tab:appendix-market-scaling-rental-flow}.
In all reported market-scaling instances, each user type's subscription offering 
reaches the utility threshold $p$, so every user chooses either the 
ad-supported or subscription mode, and the admitted fraction is   one. 
Rental usage, however, changes substantially with market scale.
At $N=1$, total rental flow is approximately $3.27$, of which about $74\%$
is directed to the five niche families. At $N=2$, absolute rental flow rises to
approximately $4.59$ as the arriving user mass increases, while rental flow per
arrival falls from $0.653$ to $0.459$ because all reusable families
$6,\ldots,10$ enter the buy set. At $N\geq 5$, the threshold rule purchases the
entire catalog, so realized rental flow is zero.

\subsection{Preference Concentration and Procurement Adaptation}
\label{sec:preference-concentration}

We next study how our
threshold buy set adapts when preferences become more concentrated
on shared content.
This experiment simulates an exogenous popularity shock in which demand becomes
more concentrated on a subset of shared content families.
This is consistent
with evidence   that social influence and popularity
signals can increase concentration in product success \citep{SalganikDoWa06}.
For the base setting, we use the same generated instance and the same fixed  
parameters as in Section~\ref{sec:exp-convergence}. 
We fix \(N=1\) and use the high-resolution grid \(K_T=129\). Thus, the case \(\rho=0\)
coincides exactly with the instance in Section~\ref{sec:exp-convergence}.
The experiment varies a concentration parameter
$
\rho\in\{0,0.1,0.2,0.3,0.4,0.5\}.
$

For this experiment, we introduce a clustered preference-concentration structure
on top of the baseline catalog in Figure~\ref{fig:synthetic-catalog}. We partition
the five user types into two taste clusters,
\(C_1=\{1,2,3\}\) and \(C_2=\{4,5\}\), and partition four of the reusable content
families into two shared-content portfolios,
\[
    \mathcal H_1=\{6,7\}, \qquad \mathcal H_2=\{8,9\}.
\]
Along the concentration path, user
types in cluster \(C_c\) become increasingly
attracted to the corresponding portfolio \(\mathcal H_c\), while content family \(10\) and
the niche families remain fixed. 
For each user type \(j\), let \(c(j)\in\{1,2\}\) denote the cluster
such that \(j\in C_{c(j)}\). Only the shared families in its own cluster
portfolio \(\mathcal H_{c(j)}\) move with \(\rho\).
The effective
concentration level is
\[
    \rho_{\mathrm{eff}}^{(j)}(\rho)
    =\min\{1,s_{c(j)}\rho\},
\]
where \(s_1=1.60\) and \(s_2=1.00\).
Accordingly, the first cluster moves faster along
the concentration path than the second cluster.

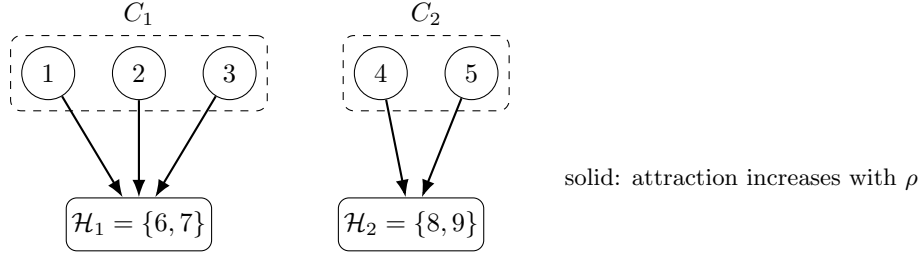
\begin{figure}[t]
\centering
\begin{tikzpicture}[
    >=Latex,
    user/.style={circle, draw, minimum size=7mm, inner sep=1pt},
    item/.style={rectangle, draw, rounded corners, minimum height=7mm, minimum width=16mm, inner sep=2pt},
    group/.style={rectangle, draw, dashed, rounded corners, inner sep=4pt},
    increase/.style={->, thick},
    fixed/.style={->, dashed},
    every node/.style={font=\small}
]

\node[user] (u1) at (0.0,2.4) {$1$};
\node[user] (u2) at (1.2,2.4) {$2$};
\node[user] (u3) at (2.4,2.4) {$3$};
\node[user] (u4) at (4.4,2.4) {$4$};
\node[user] (u5) at (5.6,2.4) {$5$};

\node[group, fit=(u1)(u2)(u3), label=above:\(C_1\)] {};
\node[group, fit=(u4)(u5), label=above:\(C_2\)] {};

\node[item] (h1) at (1.2,0.4) {$\mathcal{H}_1=\{6,7\}$};
\node[item] (h2) at (4.8,0.4) {$\mathcal{H}_2=\{8,9\}$};

\draw[increase] (u1) -- (h1);
\draw[increase] (u2) -- (h1);
\draw[increase] (u3) -- (h1);

\draw[increase] (u4) -- (h2);
\draw[increase] (u5) -- (h2);

\node[align=left, font=\footnotesize, anchor=west] at (6.7,1.0)
{solid: attraction increases with \(\rho\)
};

\end{tikzpicture}
\caption{Clustered preference-concentration structure used in
Section~\ref{sec:preference-concentration}. 
The base case $(\rho = 0)$ is the same as Figure~\ref{fig:synthetic-catalog}.
For clarity, we draw an edge \((j,\ell)\) only if \(\alpha_\ell^{(j)}\) depends on \(\rho\).
Starting from the baseline catalog in
Figure~\ref{fig:synthetic-catalog}, we partition the user types into two taste
clusters, \(C_1=\{1,2,3\}\) and \(C_2=\{4,5\}\), and partition four reusable
families into two shared-content portfolios, \(\mathcal{H}_1=\{6,7\}\) and
\(\mathcal{H}_2=\{8,9\}\). As \(\rho\) increases, user types in cluster \(C_c\)
become increasingly attracted to the corresponding portfolio \(\mathcal{H}_c\).
Content family
\(10\) remains fixed throughout the concentration path.}
\label{fig:preference-concentration-structure}
\end{figure}

For \(\ell\in\mathcal H_c\), we draw one family-specific high-attraction target
$h_\ell\sim D_{\mathrm{high}}=\mathrm{Unif}[2.80,3.40]$ and keep this realized
value fixed throughout the concentration experiment. Next, we update attraction by
linear interpolation and impose the same numerical relation $u=\alpha$ used in
Section~\ref{sec:exp-convergence}:
\[
\alpha_\ell^{(j)}(\rho)
=
\left(1-\rho_{\mathrm{eff}}^{(j)}(\rho)\right)\alpha_\ell^{(j)}(0)
+
\rho_{\mathrm{eff}}^{(j)}(\rho)h_\ell,
\qquad
u_\ell^{(j)}(\rho)=\alpha_\ell^{(j)}(\rho).
\]
All other attraction and utility entries remain fixed. In particular, niche families,
reusable families outside the user type's own cluster portfolio, and content family \(10\)
do not move along the concentration path.
We visualize this in Figure~\ref{fig:preference-concentration-structure}.
Along this path, rental costs are recomputed from the updated average attractiveness,
while buy costs remain fixed at their \(\rho=0\) values:
$
\gamma_\ell(\rho)
=
0.7\cdot \frac{1}{J}\sum_{j=1}^J \alpha_\ell^{(j)}(\rho),
\eta_\ell(\rho)=\eta_\ell(0).
$
This captures a setting in which usage-based royalties respond to content popularity,
while upfront licensing costs are fixed in advance. As preferences concentrate toward
shared families, those families become more expensive to rent, whereas their fixed
buying costs do not change. Therefore, shared families can cross the threshold 
and
enter the planned buy set:
$
\ell\in S\opt(\rho)
\quad\Longleftrightarrow\quad
2\eta_\ell(0)\le \gamma_\ell(\rho)\sum_{j=1}^J \lambda\jj N .
$

For each value of \(\rho\), we compute \(S\opt(\rho)\), run
Algorithm~\ref{alg:mnl-solver-all} with the uniform relaxation weights from~\eqref{eqn:uniform-beta}, and compare the returned value with the same procurement
benchmarks used in Section~\ref{sec:exp-convergence}. 
Table~\ref{tab:pref-concentration}
reports the value \(R(S\opt(\rho), \yhat)\) returned by the heuristic, the full
MIP benchmark, the MIP benchmark with the buy set fixed to \(S\opt(\rho)\), the
pure-buying benchmark \(S=[L]\), and the pure-rental benchmark \(S=\emptyset\).
The corresponding returned assortment distributions are reported in
Appendix~\ref{app:additional-experimental-details},
Table~\ref{tab:appendix-preference-concentration-distributions}. The table shows that,
as demand concentrates on reusable shared content, the ad-side assortments shift
toward assortments containing the cluster-shared families: content families \(6,7\) for
cluster \(C_1\) and content families \(8,9\) for cluster \(C_2\).

\begin{table}[t]
\centering
\caption{Preference concentration and procurement adaptation. The experiment uses
the same generated instance and fixed design parameters
\((\sigma,p)=(1.50,1.80)\) as Section~\ref{sec:exp-convergence}, with \(N=1\) and
\(K_T=129\). The row \(\rho=0\) coincides with the corresponding base-case row in
Table~\ref{tab:sec51-main}. Percentages in parentheses are relative gaps to the
reported full-model MIP incumbent. The full-model MIP reached the 300-second limit in
all rows; its final solver gaps for \(\rho=0,0.1,\ldots,0.5\) were
\(3.06\%\), \(2.57\%\), \(2.93\%\), \(2.26\%\), \(3.55\%\), and \(2.08\%\),
respectively. For \(\rho=0\), we reuse the base-case MIP run reported in
Table~\ref{tab:sec51-main}. The fixed-\(S\opt(\rho)\) and pure-buying benchmarks were solved to
certified optimality in every row. The pure-rental benchmark was also certified except
at \(\rho=0.4\), where the final solver gap was below \(0.04\%\).}
\label{tab:pref-concentration}
\scriptsize
\resizebox{\textwidth}{!}{
\begin{tabular}{c c c c c c c}
\hline
\(\rho\) & \(S\opt(\rho)\) & \(R(S\opt(\rho), \yhat)\) (gap to MIP full)
& MIP full & MIP\((S\opt(\rho))\)
& MIP\(([L])\) & MIP\((\emptyset)\) \\
\hline
0.0 & \(\{10\}\) & 15.32 (2.20\%) & 15.67 & 15.37 & 11.19 & 15.01 \\
0.1 & \(\{10\}\) & 15.31 (2.84\%) & 15.76 & 15.36 & 12.05 & 14.99 \\
0.2 & \(\{6,7,10\}\) & 15.41 (3.05\%) & 15.90 & 15.90 & 12.86 & 14.98 \\
0.3 & \(\{6,7,10\}\) & 15.58 (3.04\%) & 16.07 & 16.07 & 13.78 & 14.95 \\
0.4 & \(\{6,7,8,9,10\}\) & 15.15 (6.69\%) & 16.24 & 16.24 & 14.62 & 14.92 \\
0.5 & \(\{6,7,8,9,10\}\) & 15.14 (7.61\%) & 16.39 & 16.39 & 15.01 & 14.92 \\
\hline
\end{tabular}}

\end{table}

Table~\ref{tab:pref-concentration} illustrates how the threshold procurement rule responds to
a change in demand composition. At \(\rho=0\), the instance is exactly the base case from
Section~\ref{sec:exp-convergence}. As \(\rho\) increases, demand shifts toward reusable shared
families. 
Since 
rental costs are recomputed from average attractiveness while fixed buying
costs remain unchanged, shared content becomes increasingly attractive 
 for the platform
 to   buy.
Consequently, the planned buy set expands from content family \(10\) to the reusable
shared portfolios whose attractiveness increases along the concentration path.
The results 
in Table~\ref{tab:pref-concentration} also highlight the role of procurement
design. The fixed-\(S\opt(\rho)\) benchmark evaluates the procurement structure selected
by the threshold rule, while the pure-buying and pure-rental benchmarks represent the two
procurement extremes. When 
\(S\opt(\rho)\) is   hybrid, comparing these columns
shows whether the platform benefits from selectively licensing reusable shared content while
retaining rental flexibility for the rest of the catalog. 
The main message is that the planned procurement structure changes in an
economically interpretable way as demand concentrates on reusable content.
In addition, 
Table~\ref{tab:pref-concentration}
shows that 
our heuristic \(S\opt(\rho)\) is competitive: we see 
that 
MIP full and MIP\((S\opt(\rho))\) are close  for each value
of $\rho$.

As in the market-scaling experiment, the admission fraction~\eqref{eqn:admission-fraction}
remains one throughout
the preference-concentration path, while
Table~\ref{tab:appendix-preference-concentration-rental-flow} shows that realized
rental usage changes substantially. Total rental flow declines from approximately
$3.27$ at $\rho=0$ to $1.07$ at $\rho=0.5$. Before the buy set adjusts, the shift
toward still-rented shared content causes a slight increase from $3.27$ at $\rho=0$
to $3.32$ at $\rho=0.1$. Content families $6$ and $7$ enter the buy set at $\rho=0.2$,
after which realized rental flow to reusable content is zero; content families $8$ and $9$
enter at $\rho=0.4$. From $\rho=0.2$ onward, all realized rental usage is
concentrated on the niche families and declines as consumption shifts toward the
purchased shared-content portfolios.

\section{Conclusion}

\label{sec:conclusion}

This \paper\ 
studies the fixed-$(\sigma,p)$ 
assortment-and-procurement problem faced by dual-mode content platforms: how to tailor mode-specific content offerings while deciding which content families to license upfront and which to procure on a usage basis. 
The central difficulty is that assortment design, user self-selection, and procurement are jointly endogenous. Rental costs depend on realized usage flows, whereas buying induces a max-exposure coupling across user types and monetization modes. 
We show that this joint problem is NP-hard   in a   special case, which motivates an approximation-based solution approach.

Our framework combines a buy-set relaxation, a ratio reduction, scalar-dual bisection with an MNL assortment algorithm, and a restricted-master recovery step that restores primal feasibility. 
The result is a scalable and naturally parallelizable algorithm with computable finite-instance performance guarantees. 
The analysis also yields a transparent procurement insight: 
the threshold heuristic minimizes the buy-set component of the performance bound and favors content families with sufficiently
low buy-to-rent cost ratios at the prevailing market scale.
Under the market-scaling regime studied in the \paper, the relative error of the method vanishes as consumer scale and grid resolution increase. An important next step is to endogenize the common ad load and subscription price, which are treated as fixed here. 
Other natural extensions include dynamic learning of user preferences and ad tolerance, richer temporal variation in content value, and competitive platform settings.

\newpage
\begin{appendices}

The appendix is organized as follows. Section~\ref{sec:our_algo}   discusses the main algorithms and all ancillary subroutines.
Section~\ref{sec:proof-hardneess} proves the NP-hardness result. Section~\ref{sec:proof-of-thm} proves the main approximation and asymptotic guarantees. Sections~\ref{sec:generic-LP}--\ref{subsec:andrew-monotone-chain} provide the algorithmic details, including the LP primitive, ads and subscription subproblems, boundary solver, and postprocessing step. The complexity analysis follows next in Section~\ref{sec:complexity}, while Section~\ref{sec:MIP} gives the exact MIP benchmark formulation. Finally, Section~\ref{app:additional-experimental-details} reports additional experimental details and returned assortment distributions.

Table~\ref{tab:notation} summarizes the key notation used throughout the \paper.

\begin{table}[h!]
\centering
\small
\caption{Notation Summary}
\begin{tabular}{p{3.4cm} p{9.2cm}}
\toprule
\textbf{Symbol} & \textbf{Description} \\
\midrule
$j \in [J]$ & User type index \\
$\ell \in [L]$ & Content family index \\
$\lambda\jj$ & Mass of type-$j$ users \\
$F\jj$ & Ad-tolerance distribution of type-$j$ users \\
$\kappa$ & Maximum assortment size \\
$\mc{A}_\kappa$ & Family of feasible assortments $\{A \subseteq [L] : |A| \le \kappa\}$ \\
$\alpha_\ell^{(j)}$ & MNL attraction parameter of content family $\ell$ for type $j$ \\
$u_\ell^{(j)}$ & Utility delivered if type $j$ chooses content family $\ell$ \\
$p_\ell^{(j)}(A)$ & MNL probability that type $j$ chooses content family $\ell$ from assortment $A$ \\
$x^{(j)}(A)$ & Inside-choice probability under deterministic assortment $A$ \\
$u^{(j)}(A)$ & Expected utility under deterministic assortment $A$ \\
$y_A^{(m,j)}$ & Probability that assortment $A$ is offered to type $j$ in mode $m \in \{a,s\}$ \\
$x^{(m,j)}, u^{(m,j)}$ & Induced inside-choice probability and expected utility under mode $m$ \\
$f_\ell^{(m,j)}(\bm{y})$ & Expected choice flow into content family $\ell$ for type $j$ in mode $m$ \\
 $\padj(\bm{y})$ & Probability type $j$ chooses the ad-supported mode \\
$\psubj(\bm{y})$ & Probability type $j$ chooses the subscription mode \\
$p$ & Subscription price \\
$\sigma$ & Ad load \\
$\rads$ & Ad revenue per unit ad time \\
$\gamma_\ell$ & Rental cost of content family $\ell$ \\
$\eta_\ell$ & Purchase cost of content family $\ell$ \\
$S \subseteq [L]$ & Buy set \\
$R(S, \bm{y})$ & Fixed-$(S,\sigma,p)$ objective \\
$R(S)$ & Optimal fixed-buy-set value \\
\bottomrule
\end{tabular}
\label{tab:notation}
\end{table}

\section{Main solution algorithm and other helper algorithms}
\label{sec:our_algo} 

\subsection{Algorithm~\ref{alg:lp-postprocessing}} 
\label{sec:abstract-LP}
Consider  
an LP of the form
\begin{align}
    F(\theta)
:=
\max_{\bm{y} \in \Delta(\mathcal A)}
\left\{
\sum_{a \in \mathcal A} y_a \phi_\theta(a)
:
\sum_{a \in \mathcal A} y_a d_\theta(a)=0
\right\}. \label{eqn:F-theta-def}
\end{align}
Assuming the primal problem is feasible, its dual is
\begin{align}
    g_\theta(\nu)
=
\max_{a \in \mathcal A}
\bigl\{
\phi_\theta(a)+\nu d_\theta(a)
\bigr\},
\qquad
F(\theta) = \min_{\nu \in \R} g_\theta(\nu).
\label{eqn:generic-LP-dual}
\end{align}

Here we interpret $\mc{A}$ as the action set, which can be large.
To approximately solve $F(\theta)$, we proceed in two steps.
First, we construct a finite candidate set $C(\theta) \subseteq \mc{A}$ using the problem-specific procedure introduced later.
Second, we apply Algorithm~\ref{alg:lp-postprocessing}
 in Appendix~\ref{subsec:andrew-monotone-chain}
to compute the optimal feasible mixture over $C(\theta)$ for the restricted one-equality LP.
We assume $C(\theta)$ contains either an action with 
$d_\theta(a) = 0$
or two actions 
$a^-,a^+$
with $d_\theta(a^-) <  0 <  d_\theta(a^+)$, so that the restricted LP is feasible.
We interpret $\phi_\theta(a) \in \R$ as the reward of action $a$, and
$d_\theta(a) \in \R$ as the residual of the single linear balance constraint.
Given the candidate set $C(\theta)$, we restrict the original problem to actions in $C(\theta)$.

\subsection{Optimal Cardinality-Constrained Assortment}
 \label{sec:MNL-oracle}

For any attraction vector $\bm\alpha$, reduced-revenue vector $\bm R$,  and capacity
$\kappa$,  
we use 
Algorithm~\ref{alg:mnl-oracle} due to \citep{RusmevichientongShSh10}
to   solve
the problem
$\max_{A \in \mc{A}_\kappa} 
\frac{\sum_{\ell \in A} \alpha_\ell R_\ell}
{1+\sum_{k \in A} \alpha_k}$.
Let $\toracle(L)$ denote the cost of one call to Algorithm~\ref{alg:mnl-oracle}
on $L$ content families.
Theorem 2.2  of 
  \citep{RusmevichientongShSh10}   shows that
$\toracle(L)=O(L^2\log L)$.

\begin{algorithm}[H]
\renewcommand{\thealgorithm}{\texttt{StaticMNL}}
\caption{Cardinality-Constrained Assortment}
\label{alg:mnl-oracle} 
\begin{algorithmic}[1]
\STATE \textbf{Input:} Attraction vector $\bm\alpha$;  revenue vector $\bm R$;
 cardinality cap $\kappa$ 
\STATE Define the objective
$
H(A,\bm\alpha,\bm R)
=
\frac{\sum_{\ell \in A} \alpha_\ell R_\ell}
{1+\sum_{k \in A} \alpha_k}$,
for $A \in \mc{A}_\kappa
$
\STATE $A\opt \leftarrow$ %
Static MNL Algorithm of \citet{RusmevichientongShSh10}
\STATE %
$H\opt \gets H(A\opt,\bm\alpha,\bm R)$
\STATE \textbf{Return:} $H\opt,  A\opt$
\end{algorithmic}
\end{algorithm}

\subsection{The fixed-$(j,t)$ solvers for~\eqref{eqn:ads-LP-def} and~\eqref{eqn:sub-LP-def}}
 \label{sec:fixed-t-solver}

We now specialize 
the approach
to the two one-equality LPs that arise at a fixed
user type $j$ and fixed ratio $t$.

\subsubsection{The fixed-$(j,t)$ solver for     the ad problem~\eqref{eqn:ads-LP-def}}
We introduce
the dual variable $\mu$ for the constraint 
$d\adj(t, \bm{y}) = 0$, so that
\begin{align}
    \max_{\bm{y} \in \simplex}
\set{G\adj(t,\bm{y}): 
d\adj(t, \bm{y}) = 0
}
= \min_{\mu \in \R}
 \max_{A \in \mc{A}_\kappa}
 \set{ G\adj(t,A) + \mu d\adj(t, A)},
\label{eqn:duality-ads}
\end{align}
where $ G\adj(t,A) :=  G\adj(t,e_A)$.

Each fixed-$\mu$ evaluation $G\adj(t, \mu)$
is a cardinality-constrained MNL assortment problem:
\begin{align}
\label{eq:ads-fixed-mu}
G\adj(t, \mu)
=
\max_{A \in \mc{A}_\kappa}
\left\{
\frac{\sum_{\ell \in A} \alpha_\ell^{(j)} R_\ell\adj(t,\mu)}
 {1 + \sum_{k \in A} \alpha_k^{(j)}} 
\right\},
\end{align} with
\begin{align}
R_\ell\adj(t,\mu)
:=
\lambda\jj \bar F\jj(\sigma t)\bigl(\rads\sigma-\gamma_\ell \I{\ell \notin S}
\bigr)
 -\eta_\ell \beta\adjl \I{\ell \in S}
+
\mu\bigl(1-t u_\ell^{(j)}\bigr).
\label{eqn:R-C-def-ads}
\end{align}

Thus, 
\begin{align}
G\adj(t) = 
\min_{\mu \in \R} G\adj(t, \mu).
\label{eqn:ads-duality}
\end{align}

Recall $t \in   [t_{\min}^{(j)}, t_{\max}^{(j)}]$ in~\eqref{eqn:t-I-def}.
For the ads branch, we use Algorithm~\ref{alg:mnl-fixed-jt-ads}
to   compute $G\adj(t)$.
The procedure involves choosing two content families defined below:
\begin{align}
\lplus\jj \in\argmax_{\ell: \alpha_\ell^{(j)}>0} u_\ell^{(j)},
\qquad
\lminus\jj \in
\argmin_{\ell: \alpha_\ell^{(j)}>0} u_\ell^{(j)}.
\label{eqn:two-prod-pair-def}
\end{align}

Within   Algorithm~\ref{alg:mnl-fixed-jt-ads},
if $t \in \{t_{\min}^{(j)}, t_{\max}^{(j)}\}$, we use 
Algorithm~\ref{alg:mnl-ads-boundary}, which considers all assortments  formed from the 
content families in
$\mc{M}\jj(t)   = \set{\ell \in [L] : u_\ell^{(j)} = 1/t, \alpha_\ell^{(j)} > 0}
$ (note that 
$ \alpha_\ell^{(j)} > 0 \Rightarrow  u_\ell^{(j)} > 0$ by~\eqref{eqn:alpha-implies-u-support})
and
applies   Algorithm~\ref{alg:mnl-oracle} to the
content families in $\mc{M}\jj(t)$.
For $t\in (t_{\min}^{(j)},t_{\max}^{(j)})$,
we apply
a  bisection   procedure, which we denote by 
Algorithm~\ref{alg:dual-bisection}. In this case, we
start with an initial dual interval (to be introduced in~\eqref{eqn:ads-dual-interval-algo})
that contains an optimal dual solution, together with two certificate assortments whose residuals have opposite signs. At each iteration, we evaluate the midpoint $m_n=(L_n+U_n)/2$ by solving the corresponding MNL subproblem using  Algorithm~\ref{alg:mnl-oracle}. 
The sign of the residual of the returned assortment tells us which half of the interval can be discarded: if the residual is negative,
no dual minimizer lies strictly to the left of $m_n$, while a positive residual means the opposite. Hence, the interval shrinks monotonically while still bracketing the optimal dual value, and all assortments visited during the search are kept as candidate actions for the final restricted primal LP.
Since the residuals are uniformly bounded,  the bisection procedure yields an optimality-gap guarantee;
see Lemma~\ref{lem:ads-bisection-guarantee} in the Appendix.
In addition, in line~\ref{line:init-ads-solver} of  Algorithm~\ref{alg:mnl-fixed-jt-ads}, we initialize 
$\mathcal C_0  \gets \{\emptyset, A_t^-,A_t^+\}$ since the empty assortment is also a feasible action, i.e., $d\adj(t, \ynullS) = 0$.

\begin{algorithm}[h]
\renewcommand{\thealgorithm}{\texttt{DualBisection}}
\caption{Dual bisection for candidate assortment generation}
\label{alg:dual-bisection}
\begin{algorithmic}[1]
\STATE \textbf{Input:} $\bm\alpha$, $\bm R(\cdot)$, residual function $d(\cdot)$,
dual interval $\bar{\mc J}$, initial candidate set $\mc C_0$,
cardinality cap $\kappa$, iteration budget $K$

\STATE $l \gets \inf \bar{\mc J}$, $r \gets \sup \bar{\mc J}$,
and $\mc C \gets \mc C_0$

\FOR{$n=0,\ldots,K-1$}
    \STATE $m \gets \frac{l+r}{2}$
    \STATE $(R,A) \gets
    \text{~\ref{alg:mnl-oracle} }
    \Paran{\bm\alpha,\bm R(m),\kappa}$
    \STATE $\mc C \gets \mc C\cup\set{A}$
    \IF{$d(A)=0$}
        \STATE \textbf{break}
    \ELSIF{$d(A)<0$}
        \STATE $l\gets m$
    \ELSE
        \STATE $r\gets m$
    \ENDIF
\ENDFOR

\STATE \textbf{Return:} $\mc C$
\end{algorithmic}
\end{algorithm}

\begin{algorithm}[h]
\renewcommand{\thealgorithm}{\texttt{AdsBoundarySolver}}
\caption{Exact ads boundary
problem~\eqref{eqn:ads-LP-def}
solver at 
$t \in \{t_{\min}^{(j)}, t_{\max}^{(j)}\}$}
\label{alg:mnl-ads-boundary}
\begin{algorithmic}[1]
\STATE \textbf{Input:} $j,t,S,\bm{\beta}$
\STATE $
\mathcal M\jj  \gets \set{\ell \in [L] : u_\ell^{(j)} = 1/t,  \alpha_\ell^{(j)} > 0}$.
\STATE  $
    (\ghat, \hat A)  \gets \text{~\ref{alg:mnl-oracle} }\Paran{(\alpha_\ell^{(j)})_{\ell\in \mc M\jj}, 
\Paran{\bm R\adj_\ell(t,0)}_{\ell\in \mc M\jj}, \kappa} $, $\bm R\adj(t,0)$ defined in~\eqref{eqn:R-C-def-ads}.

\STATE  $\hat{\bm y} \gets \bm e_{\hat A}$,
where $\bm e_{\hat A}$ denotes the unit vector in $\mc{A}_\kappa$ corresponding to $\hat A$.
\STATE \textbf{Return:} $\ghat, \hat{\bm y}$
\end{algorithmic}
\end{algorithm}

\begin{algorithm}[h]
\renewcommand{\thealgorithm}{\texttt{\texttt{AdsSolver}}}
\caption{Fixed-$(j,t)$ ads
problem~\eqref{eqn:ads-LP-def}
solver}
\label{alg:mnl-fixed-jt-ads} 
\begin{algorithmic}[1]
\STATE \textbf{Input:} $j,t,S,\bm{\beta}, \bar{\mc{J}}, K$
\IF{$t \in \{t_{\min}^{(j)}, t_{\max}^{(j)}\}$}
    \STATE 
    $(\ghat, \hat{\bm y}) \gets$~\ref{alg:mnl-ads-boundary}$(j,t,S,\bm{\beta})$
\ELSIF{$t \in (t_{\min}^{(j)}, t_{\max}^{(j)})$}
    \STATE  
    $A^- \gets \{\lplus \}, A^+ \gets \{\lminus\}$,
    where $\lplus,\lminus$
    are defined in~\eqref{eqn:two-prod-pair-def}
\STATE$\phi(\cdot) \gets 
G\adj(t,\cdot)$ and
$d(\cdot) \gets x^{(j)}(\cdot)-t u^{(j)}(\cdot)$ for all $A$.
\STATE Set
$\mathcal C_0 \gets \{\emptyset, A^-,A^+\}$ \label{line:init-ads-solver}
\STATE
$\mathcal C \gets$~\ref{alg:dual-bisection}
$\Paran{
\bm\alpha^{(j)},
\bm R\adj(t,\cdot),
d,
\bar{\mc J},
\mathcal C_0,
\kappa,
K
}$, $\bm R\adj(t,\cdot)$ defined in~\eqref{eqn:R-C-def-ads}
\STATE $(\ghat, \hat{\bm y})
\gets$~\ref{alg:lp-postprocessing}$(\mc{C},\phi,d)$
\label{line:ads-solver-bisection-terminate} 
\ELSIF{$t  = \infty$}
\STATE $(\ghat, \hat{\bm y}) \gets (0,\ynullS)$
\ENDIF

\STATE \textbf{Return:} $\ghat, \hat{\bm y}$
\end{algorithmic}
\end{algorithm}

\subsubsection{The fixed-$(j,t)$ solver for   the subscription problem~\eqref{eqn:sub-LP-def}}

We introduce
the dual variable $\xi$ for the constraint 
$d\subj(p, \bm{z}) = 0$, so that
\begin{align}
\max_{\bm{z} \in \simplex}
\set{G\subj(t, \bm{z}): 
d\subj(p, \bm{z})  = 0
}
= \min_{\xi \in \R}
\max_{B \in \mc{A}_\kappa}
\set{ G\subj(t,B) + \xi  d\subj(p, B)},
\label{eqn:duality-sub}
\end{align}
where $  G\subj(t,B)  :=   G\subj(t,e_B)$.

Expanding the terms yields the decomposition
\begin{align}
\label{eq:subs-dual-split}
G\subj(t, \xi)
:=
\max_{B \in \mc{A}_\kappa}
\set{ G\subj(t,B) + \xi  d\subj(p, B)}
=
c\subj(t,\xi) + H\subj(t, \xi),
\end{align}
where the affine term is
\begin{equation*}
  c\subj(t,\xi) := \lambda\jj p F\jj(\sigma t) - \xi p,
\end{equation*}
and the  $H\subj(t, \xi)$ term is
\begin{align}\label{eq:subs-reduced-H}
H\subj(t, \xi)
=
\max_{B \in \mc{A}_\kappa}
\left\{
\frac{\sum_{\ell \in B} \alpha_\ell^{(j)} R_\ell\subj(t,\xi)}
{1 + \sum_{k \in B} \alpha_k^{(j)}} 
\right\},
\end{align}
with
\begin{align}
    R_\ell\subj(t,\xi)
:=
\xi u_\ell^{(j)} 
-  \eta_\ell  \beta\subjl \I{\ell \in S}
- \lambda\jj F\jj(\sigma t)\gamma_\ell \I{\ell \notin S}.
\label{eqn:R-C-def-sub}
\end{align}
In addition, 
for $p \in (0,\uumaxj)$,
we have
\begin{align}
G\subj(t) = 
\max\set{\min_{\xi \in \R} G\subj(t, \xi),0}.
\label{eqn:sub-duality}
\end{align}

For the subscription problem~\eqref{eqn:sub-LP-def}, we use
Algorithm~\ref{alg:mnl-fixed-jt-sub} to approximately
compute $G\subj(t)$.
For interior values $p\in(0,U_{\max}^{(j)})$, we apply a bisection
procedure similar to that in
Algorithm~\ref{alg:mnl-fixed-jt-ads}. The boundary case
$p=U_{\max}^{(j)}$  can be solved
exactly rather than through the dual search. 
Indeed, since every
deterministic assortment $B$ satisfies $u^{(j)}(B)\le U_{\max}^{(j)}$,
any feasible solution of the boundary LP must be supported on
assortments attaining $U_{\max}^{(j)}$. Unlike the ads boundary case,
however, this does not immediately reduce to a direct restriction of the
content family space: we must still characterize the family of
utility-maximizing assortments. We do so by rewriting the condition
$u^{(j)}(B)=U_{\max}^{(j)}$ as a subset-maximization problem, which lets
us search this family exactly using Algorithm~\ref{alg:mnl-oracle} 
or 
Algorithm~\ref{alg:exact-card-mnl},
depending on the resulting structure. We
defer the details to
 Algorithm~\ref{alg:mnl-sub-boundary} in
Section~\ref{subsec:sub-boundary}, with its
correctness  provided in  Lemma~\ref{lem:sub-boundary-poly-simple}.

  \begin{algorithm}[h]
\renewcommand{\thealgorithm}{\texttt{SubscriptionSolver}}
\caption{Fixed-$(j,t)$  subscription problem~\eqref{eqn:sub-LP-def} solver}
\label{alg:mnl-fixed-jt-sub}
\begin{algorithmic}[1]
\STATE \textbf{Input:} 
$j,t,S,\bm{\beta},\bar{\mc{J}}, K, B^{+},\uumaxS$, where
$t \in I^{(j)} \cup \{\infty\}$ \label{line:t-infy-sub-algo}
 \IF{$p > U_{\max}$ or $p = 0$}
\STATE $(\gtilde, \ytilde) \gets (0,\ynullS)$ 
\ELSIF{$p = U_{\max}$}
\STATE
 $(\gtilde, \ytilde) \gets$~\ref{alg:mnl-sub-boundary}$(j,t,S,\sigma,\bm\beta,U_{\max})$
\ELSE
\STATE Set $B^{-} \gets \emptyset$ and set 
$B^{+}$ as given 
\STATE   
$\phi(\cdot) \gets 
G\subj(t,\cdot)$ and
$d(\cdot) \gets u^{(j)}(\cdot)-p$  
\label{line:bpplus}
\STATE Set
    $\mathcal C_0 \gets \{B^{-},B^{+}\}$ \label{line:sub-C-t-init}
\STATE
$\mc C \gets$~\ref{alg:dual-bisection}
$(\bm\alpha^{(j)},\bm R\subj(t,\cdot),d,\bar{\mc J},
\mc C_0,\kappa,K)$,
  $\bm R\subj(t,\cdot)$   defined in~\eqref{eqn:R-C-def-sub}
\STATE    $(\gtilde, \ytilde)
\gets~\ref{alg:lp-postprocessing}(\mathcal  C,\phi,d)$
\label{line:sub-solver-bisection-terminate} 
\label{line:V-hat-before-max}
\ENDIF 
\STATE   $(\ghat, \yhat) \gets (\gtilde, \ytilde)$ if $\gtilde > 0$,
and   $(\ghat, \yhat) \gets (0,\ynullS)$ otherwise \label{line:V-hat-after-max}
\STATE \textbf{Return:} $(\ghat, \yhat)$
\end{algorithmic}
\end{algorithm}

\subsection{Algorithm~\ref{alg:mnl-solver-all}}
 \label{sec:aggregation}

Before presenting 
details  of
Algorithm~\ref{alg:mnl-solver-all} introduced in Section~\ref{sec:model-sol},  we discuss some necessary ingredients.

\paragraph{Degenerate ratio case.}
If \(t_{\min}^{(j)}=t_{\max}^{(j)}\), then the feasible ratio interval
\(I^{(j)} = \set{t_{\min}^{(j)}}\) in~\eqref{eqn:t-I-def}
is a singleton. In this case, there are no interior ratio
values for the ads-side dual search.
We therefore use the single ratio
\(t_{\min}^{(j)}\), solve the ads branch using
Algorithm~\ref{alg:mnl-ads-boundary} in Algorithm~\ref{alg:mnl-fixed-jt-ads}, set the ads-side ratio-discretization
error for type \(j\) to zero, and skip the definitions of the ads
dual-search interval discussed below.
The following discussion related to the ad
problem assumes
that \(t_{\min}^{(j)} < t_{\max}^{(j)}\).

\paragraph{Ratio-discretization constants.} 

Define the outer ratio grid size by $K_T\ge 2$. 
For each type $j\in[J]$, let
\begin{align}
T\jj_{K_T}= \set{t_{\min}^{(j)}+\frac{r-1}{K_T-1}\bigl(t_{\max}^{(j)}-t_{\min}^{(j)}\bigr):r=1,\dots,K_T}, 
\label{eqn:outer-grid-def}
\end{align} be the uniform grid points.
Also define
\begin{align}  
\Delta_{t,K_T}^{(j)} 
=\frac{t_{\max}^{(j)}-t_{\min}^{(j)}}{2(K_T-1)}. 
\label{eqn:Delta-t-K-def}
\end{align}  
We use $\Delta_{t,K_T}^{(j)} $ to denote the maximum distance from any  $t \in I\jj$ to its nearest grid point.

\paragraph{Inner dual-search constants.}

We now define the initial dual-search intervals
$\bar{\mc{J}}\adj(t)$~\eqref{eqn:ads-dual-interval-algo}
and $\bar{\mc{J}}\subj$~\eqref{eqn:sub-dual-interval-algo},
together with the bisection budgets used by
Algorithms~\ref{alg:mnl-fixed-jt-ads}
and~\ref{alg:mnl-fixed-jt-sub}.
The definitions below all come from the same one-equality LP principle:
if a valid initial dual interval is $\bar{\mc{J}}$, the residual magnitude is bounded
by $D$, and $K$ bisection iterations are run, then the   
optimality gap
is bounded by
\[
    \frac{D}{2}\,|\bar{\mc{J}}|\,2^{-K}.
\]
Thus, the constants below serve two purposes. First, they construct initial
intervals that contain an optimal dual solution. Second, they choose $K$ so
that each inner-search error is of order $(K_T-1)^{-1}$, matching the order
of the outer ratio-grid error.

We begin with branchwise bounds on reward spread and residual magnitude:
\begin{align}
\bar\Phi\adj& :=
\lambda\jj(\rads\sigma+\gamma_{\max})+\eta_{\max},
\qquad
\bar\Phi\subj:=\lambda\jj\gamma_{\max}+\eta_{\max},
\label{eqn:Phi-ads-sub-def} \\
D\adj &:=1+\tmaxj\uumaxj,
\qquad
D\subj:=\uumaxj.
\label{eqn:D-j-a-sdef}
\end{align}
Here $\bar\Phi\adj$ and $\bar\Phi\subj$ upper bound the spread of the ads
and subscription rewards over deterministic assortments, while $D\adj$ and
$D\subj$ upper bound the corresponding equality residuals on the branches
where bisection is used. We also define the branchwise error scales
\begin{equation}
\begin{aligned}
\Gamma\adj
&:=D\adj\bar\Phi\adj
=
\Bigl(\lambda\jj(\rads\sigma+\gamma_{\max})+\eta_{\max}\Bigr)
\Bigl(1+\tmaxj\uumaxj\Bigr), \\
\Gamma\subj
&:=D\subj\bar\Phi\subj \I{p\in(0,\uumaxj)}
=
\Bigl(\lambda\jj\gamma_{\max}+\eta_{\max}\Bigr)
\uumaxj \I{p\in(0,\uumaxj)} .
\end{aligned}
\label{eqn:Gamma-j-a-def}
\end{equation}

For the ads branch, the equality residual
$d\adj(t, \cdot)$ depends on the 
ratio $t$.
The two singleton initial
assortments are
$\{\lplus^{(j)}\}$ and $\{\lminus^{(j)}\}$, which correspond to the
endpoints $\tminj$ and $\tmaxj$. Define
\begin{align}
\Theta_-\adj
&:=
\frac{1+\alpha_{\lplus^{(j)}}^{(j)}}
{\alpha_{\lplus^{(j)}}^{(j)}u_{\lplus^{(j)}}^{(j)}},
\qquad
\Theta_+\adj
:=
\frac{1+\alpha_{\lminus^{(j)}}^{(j)}}
{\alpha_{\lminus^{(j)}}^{(j)}u_{\lminus^{(j)}}^{(j)}}.
\label{eqn:Theta-ads-def}
\end{align}
These constants are inverse residual slopes. Indeed, for any interior
$t\in(\tminj,\tmaxj)$,
\[
-d\adj(t, \{\lplus^{(j)}\})
=
\frac{t-\tminj}{\Theta_-\adj},
\qquad
d\adj(t, \{\lminus^{(j)}\})
=
\frac{\tmaxj-t}{\Theta_+\adj}.
\]
Therefore, the initial ads dual interval can be written compactly as
\begin{align}
\bar{\mc{J}}\adj(t)
:=
\left[
-\bar\Phi\adj\frac{\Theta_-\adj}{t-\tminj},
\;
\bar\Phi\adj\frac{\Theta_+\adj}{\tmaxj-t}
\right],
\qquad
t\in\set{t_2^{(j)},\ldots,t_{K_T-1}^{(j)}}.
\label{eqn:ads-dual-interval-algo}
\end{align}
The interval widens near the endpoints because one of the  
residuals approaches zero. For an interior grid point $t_r^{(j)}$ in~\eqref{eqn:outer-grid-def},
$r=2,\ldots,K_T-1$, its width is
\begin{align}
    \bigl|\bar{\mc{J}}\adj(t_r^{(j)})\bigr|
=
\bar\Phi\adj
\frac{K_T-1}{\tmaxj-\tminj}
\left(
\frac{\Theta_-\adj}{r-1}
+
\frac{\Theta_+\adj}{K_T-r}
\right).
\label{eqn:ads-dual-interval-width}
\end{align}

At the two ratio endpoints, the ads branch is solved by the boundary routine,
so no ads bisection is needed:
\[
K\adj(t\jj_1) = 
K\adj(t\jj_{K_T}):=0.
\]
For each interior index $r=2,\ldots,K_T-1$, we choose $K\adj(t\jj_r)$ so that
\[
\frac{D\adj}{2}
\bigl|\bar{\mc{J}}\adj(t_r^{(j)})\bigr|2^{-K\adj(t_r^{(j)})}
\le
\frac{\Gamma\adj}{K_T-1}.
\]
Using the displayed width of $\bar{\mc{J}}\adj(t_r^{(j)})$, this gives the explicit
budget
\begin{align}
K\adj(t_r^{(j)})
:=
\max\left\{
0,
\left\lceil
\log_2\!\left(
\frac{(K_T-1)^2}{2(\tmaxj-\tminj)}
\left(
\frac{\Theta_-\adj}{r-1}
+
\frac{\Theta_+\adj}{K_T-r}
\right)
\right)
\right\rceil
\right\}.
\label{eqn:inner-grid-size-ads}
\end{align}

For the subscription branch, the equality residual is
$d\subj(p, B)=u^{(j)}(B)-p$, which does not depend on $t$.
Thus, one dual interval suffices for all 
$t$ values. When $p\in(0,\uumaxj)$, the  residuals are
$-p$ from the empty assortment and $\uumaxj-p$ from a utility-maximizing
assortment. Define
\begin{align}
\Theta\subj
:=
\begin{cases}
\displaystyle
\frac{1}{p}+\frac{1}{\uumaxj-p},
& p\in(0,\uumaxj), \\[1.5ex]
0,
& p\notin(0,\uumaxj).
\end{cases}
\label{eqn:Theta-sub-def}
\end{align}
The subscription dual interval is
\begin{align}
\bar{\mc{J}}\subj
:=
\begin{cases}
\displaystyle
\left[
-\frac{\bar\Phi\subj}{p},
\frac{\bar\Phi\subj}{\uumaxj-p}
\right],
& p\in(0,\uumaxj), \\[2.0ex]
[0,0],
& p\notin(0,\uumaxj),
\end{cases}
\label{eqn:sub-dual-interval-algo}
\end{align}
and, therefore,
\begin{align}
\bigl|\bar{\mc{J}}\subj\bigr|
=
\bar\Phi\subj\Theta\subj. 
\label{eqn:sub-dual-interval-width}
\end{align}
When $p\in(0,\uumaxj)$, we choose $K\subj$ so that
\[
\frac{D\subj}{2}
\bigl|\bar{\mc{J}}\subj\bigr|2^{-K\subj}
\le
\frac{\Gamma\subj}{K_T-1}.
\]
Equivalently,
\begin{align}
K\subj
:=
\begin{cases}
\displaystyle
\max\left\{
0,
\left\lceil
\log_2\!\left(
\frac{(K_T-1)\Theta\subj}{2}
\right)
\right\rceil
\right\},
& p\in(0,\uumaxj), \\[2.0ex]
0,
& p\notin(0,\uumaxj).
\end{cases}
\label{eqn:inner-grid-size-sub}
\end{align}
With these choices, the ads and subscription inner-search errors are bounded
by $\Gamma\adj/(K_T-1)$ and $\Gamma\subj/(K_T-1)$, respectively. Therefore,
the inner dual-search error is balanced with the outer ratio-discretization
error at the same $O((K_T-1)^{-1})$ scale.

We now present
Algorithm~\ref{alg:mnl-solver-all}, which
  approximately solves $G\jj(S,\bm{\beta})$ in~\eqref{eqn:G-can-be-solved-in-ratio} for all $j$.
For each type $j$, it evaluates the fixed-$(j,t)$ solver on the ratio grid~\eqref{eqn:outer-grid-def},
selects the best ratio, and then aggregates over user types. 
To streamline the presentation of Algorithm~\ref{alg:mnl-solver-all}, 
we first collect the construction of the ratio grid, dual-search intervals, and bisection budgets into a separate setup subroutine
in Algorithm~\ref{alg:generate-grid-budgets}.
Here, we recall 
the discussion of the degenerate case \(t_{\min}^{(j)}=t_{\max}^{(j)}\)
at the beginning of 
Section~\ref{sec:aggregation}.

\begin{algorithm}[H]
\renewcommand{\thealgorithm}{\texttt{GenerateGridAndBudgets}}
\caption{Grid and dual-search setup}
\label{alg:generate-grid-budgets}
\begin{algorithmic}[1]
\STATE \textbf{Input:} $j,K_T,\uumaxS$
\IF{$t_{\min}^{(j)}=t_{\max}^{(j)}$}
    \STATE $\mc T^{(j)} \gets \set{t_{\min}^{(j)}}$
    \STATE $\bar{\mc J}\adS(t_{\min}^{(j)}) \gets [0,0]$
    \STATE $\bar{\mc J}\adS(\infty) \gets [0,0]$
    \STATE $K\adS(t_{\min}^{(j)}) \gets 0$
    \STATE Set $\bar{\mc J}\subS$ by~\eqref{eqn:sub-dual-interval-algo}
    and $K\subS$ by~\eqref{eqn:inner-grid-size-sub}.
    \STATE \textbf{Return:}
    $\mc T^{(j)},
    \bar{\mc J}\adS(\cdot),
    \bar{\mc J}\subS,
    K\adS(\cdot),
    K\subS$
\ENDIF
\STATE Set $\set{t_r}_{r=1}^{K_T}$ by~\eqref{eqn:outer-grid-def}.
\STATE Set
$\bar{\mc J}\adS(t_1),
\bar{\mc J}\adS(t_{K_T}),
\bar{\mc J}\adS(\infty)
\gets [0,0]$,
and set $\bar{\mc J}\adS(t_r)$
for $r=2,\ldots,K_T-1$
by~\eqref{eqn:ads-dual-interval-algo};
set $\bar{\mc J}\subS$
by~\eqref{eqn:sub-dual-interval-algo}.

\STATE Set $K\adS(t_r)$ for $r=2,\ldots,K_T-1$
by~\eqref{eqn:inner-grid-size-ads},
with $K\adS(t_1)=K\adS(t_{K_T})=0$;
set $K\subS$ by~\eqref{eqn:inner-grid-size-sub}.

\STATE \textbf{Return:}
$\set{t_r}_{r=1}^{K_T},
\bar{\mc J}\adS(\cdot),
\bar{\mc J}\subS,
K\adS(\cdot),
K\subS$
\end{algorithmic}
\end{algorithm}

\begin{algorithm}[H]
\renewcommand{\thealgorithm}{\texttt{OverallSolver}}
\caption{ }
\label{alg:mnl-solver-all}
\begin{algorithmic}[1]
\STATE \textbf{Input:} $S,\bm{\beta}$, outer-grid size $K_T \ge 2$

\FOR{each $j \in [J]$} 
\STATE   $\uumaxS, B^{+}
\gets$~\ref{alg:mnl-oracle}   $(\bm\alpha\jj, \bm{u}\jj, \kappa)$ 
\label{line:one-time-U-max-compute}
\STATE
$\left(
\mc{T}\jj,
\bar{\mc J}\adS(\cdot),
\bar{\mc J}\subS,
K\adS(\cdot),
K\subS
\right)
\gets~\ref{alg:generate-grid-budgets}
(j,K_T,\uumaxS)$
\STATE $(\ghat_\infty,\bm y_\infty\subS) \gets$~\ref{alg:mnl-fixed-jt-sub}$(j,\infty,S,\bm{\beta},\bar{\mc{J}}\subS, K\subS,B^{+},\uumaxS)$
\label{line:G-off-computation}
\FOR{$t \in \mc{T}\jj$}
   
\STATE $(\ghat_{t}\adS, \bm y_t\adS) \gets$~\ref{alg:mnl-fixed-jt-ads}$(j,t,S,\bm{\beta},\bar{\mc{J}}\adS(t),K\adS(t))$

\STATE $(\ghat_{t}\subS, \bm y_t\subS) \gets$~\ref{alg:mnl-fixed-jt-sub}$(j,t,S,\bm{\beta},\bar{\mc{J}}\subS, K\subS,B^{+},\uumaxS)$

\STATE $\ghat_t \gets \ghat_{t}\adS + \ghat_{t}\subS$
\ENDFOR

\STATE Let $t\opt \in \argmax_{t \in \mc{T}\jj \cup \set{\infty}} \ghat_t$
\IF{$t\opt = \infty$}
\STATE
$ (\yadopt,\ysubopt)
\gets
(\bm y_\emptyset,\bm y_\infty\subS)$ \COMMENT{$t=\infty$ is the exact ad-inactive branch}
\ELSE
\STATE
$ (\yadopt,\ysubopt)
\gets
(\bm y_{t\opt}\adS,\bm y_{t\opt}\subS)$
\ENDIF
\ENDFOR
\STATE
$\yhat
\gets
\{(\bm y_\star\adj, \bm y_\star\subj)\}_{j\in[J]}, \rhat
\gets
R(S, \yhat), \ghat
\gets
G(S,\bm{\beta} ,\yhat)$
\STATE \textbf{Return:}
$(\rhat, \ghat, \yhat)$
\end{algorithmic}
\end{algorithm}

\section{Proof of Theorem~\ref{thm:our-prob-np-hard}}
\label{sec:proof-hardneess}

In this section, we present the proof of
the hardness result, Theorem~\ref{thm:our-prob-np-hard}.
We use 
a reduction from the \emph{prize-collecting fractional set cover} problem~\eqref{eqn:prize-collection-fractional-set-cover}, a variant of set cover introduced by~\citet{DinitzGu13}.
The reduction proceeds in two steps: we first show that any instance of~\eqref{eqn:prize-collection-fractional-set-cover} can be embedded as a special case of our problem (Lemma~\ref{lemma-our-prob-harder-than-fractional-cover}), and then establish that the decision version of~\eqref{eqn:prize-collection-fractional-set-cover} is itself NP-hard (Lemma~\ref{lemma:pcfsc-nphard}), even under uniform penalties and uniform set costs.
To the best of our  knowledge, the NP-hardness of the prize-collecting fractional set cover problem has not been previously established in the literature; 
therefore,  Lemma~\ref{lemma:pcfsc-nphard}
is of independent interest.

The \emph{prize-collecting fractional set cover} problem with uniform penalty~\citep{DinitzGu13} is defined as follows.
Given a collection of sets $\mc{S}$ over a universe of elements $U$, a nonnegative cost function
$c : \mc{S} \to \mathbb{R}_{+}$, and penalties $p(e)$ for $e \in U$, each element must either be covered or incur its
penalty, and the goal is to minimize the total cost. We are allowed to cover an element
fractionally, i.e., by fractionally buying sets
that  cover the
element; however,
the decision of whether to cover an element or pay its penalty is integral.
Assuming $p(e) \equiv 1$, 
this is formalized as follows:

\begin{align}
\optpcfsc
:=
\begin{split}
\min_{\bm{x}, \bm{z}} \quad & \sum_{S \in \mathcal{S}} c(S)x_S + \sum_{e \in U}  z_e, \\
\text{s.t.} \quad & \sum_{S \ni e} x_S + z_e \geq 1 ~~ \forall e \in U, \\
& x_S \geq 0 ~~ \forall S \in \mathcal{S}, \\
& z_e \in \{0,1\} ~~ \forall e \in U.
\end{split}
\label{eqn:prize-collection-fractional-set-cover}
\end{align}

In Lemma~\ref{lemma-our-prob-harder-than-fractional-cover}  in Section~\ref{sec:proof-our-problem-harder-than-fractional},
we  show our problem~\eqref{eqn:problem-buy-or-rent}
 generalizes problem~\eqref{eqn:prize-collection-fractional-set-cover}.
 Next, in Lemma \ref{lemma:pcfsc-nphard}, 
we show the decision version of \eqref{eqn:prize-collection-fractional-set-cover} 
is NP-hard, with the proof in Section \ref{subsec:pcfsc-nphard-proof}.
We present the proof of Theorem~\ref{thm:our-prob-np-hard} in
Section~\ref{sec:proof-np-hard-using-lemma}.

\subsection{Lemma~\ref{lemma-our-prob-harder-than-fractional-cover} and its proof}
\label{sec:proof-our-problem-harder-than-fractional}

\paragraph{Preprocessing.}
First, we   assume without loss of 
generality that there is one element $\Tilde{e}_1$
and   a singleton set $\Tilde{S}_1 = \set{\Tilde{e}_1}$
with $c(\Tilde{S}_1) = c(\set{\Tilde{e}_1}) = 0$.
Indeed, if such an element-set pair does not exist, add a new element and a new
zero-cost singleton set, and relabel them as \(\Tilde{e}_1\) and \(\Tilde{S}_1\).
This
does not change \(\optpcfsc\), because
\(\Tilde{e}_1\) can be covered at zero cost.
Delete every element $e\in U$ with
$e \not\in \cup_{S \in \mathcal S} S$.
Any such element is necessarily assigned $z_e=1$; thus deleting $k$ such
elements decreases both $|U|$ and $\optpcfsc$ by exactly $k$ (and, for the
decision version, decreases the bound by $k$), thereby preserving both the
decision problem and the quantity $|U|-\optpcfsc$.
After this preprocessing, every element
$e\in U$ belongs to at least one set $S \in \mathcal S$.
Throughout Section~\ref{sec:proof-our-problem-harder-than-fractional}, we consider the PCFSC 
instance after 
preprocessing.

\begin{lemma} 
\label{lemma-our-prob-harder-than-fractional-cover} 
From any   instance of~\eqref{eqn:prize-collection-fractional-set-cover},
one can construct, in polynomial time,
an instance of the
problem~\eqref{eqn:problem-buy-or-rent}
with $\kappa=1$ that  
satisfies Assumption~\ref{assump:MNL}
and such that
\[
\optp = |U| - \optpcfsc.
\] 
\end{lemma}

\paragraph{Construction of the platform instance in Figure~\ref{fig:np-hard-instance-construction}.}
Create one   user type for each $e\in U$. %
Create one content family  
$\ell_S$ for each set $S\in\mathcal S$.
Set
\[
\kappa = 1,
\qquad
p = 1, \qquad
\sigma=3, 
\qquad
\radse =  0,
\qquad
\lambda^{(e)} = 1 \ \forall e\in U.
\]
Here $p$ is the   subscription price, not the PCFSC penalty function $p(e)$.
For each user type $e\in U$ and each
$S\in\mathcal S$, set
\[
\alpha_{\ell_S}^{(e)} = \I{e\in S},
\qquad
u_{\ell_S}^{(e)} = 2\,\I{e\in S},
\qquad
\eta_{\ell_S} = 2c(S),
\qquad
\gamma_{\ell_S} = 2c(S).
\]

For every other pair $(j,\ell)$ not specified above, set
\[
\alpha_\ell^{(j)} = u_\ell^{(j)} = 0.
\]
Finally, for every user type $j\in U$, define the ad-tolerance cdf
\[
F\jj(\chi)
=
\begin{cases}
0, & \chi \le \frac{1}{2},\\[2mm]
2\chi - 1, & \chi \in \left[\frac{1}{2},1\right],\\[2mm]
1, & \chi \ge 1.
\end{cases}
\]
Because $\kappa=1$, every nonempty assortment is a singleton.

For later use, for each user type $e\in U$ and each $S\in\mathcal S$, write
\begin{align}
    y_S^{(e)} := y_{\{\ell_S\}}^{(s,e)},
\qquad
y_{\emptyset}^{(e)} := y_{\emptyset}^{(s,e)}. \label{eqn:y-S-def}
\end{align}
Let $\Pi(\bm{y})$ denote the objective value of the constructed platform instance under assortment distribution $\bm{y}$.
We provide a visualization of the construction in 
Figure~\ref{fig:np-hard-instance-construction}.

\usetikzlibrary{arrows.meta,positioning,calc,fit,backgrounds}

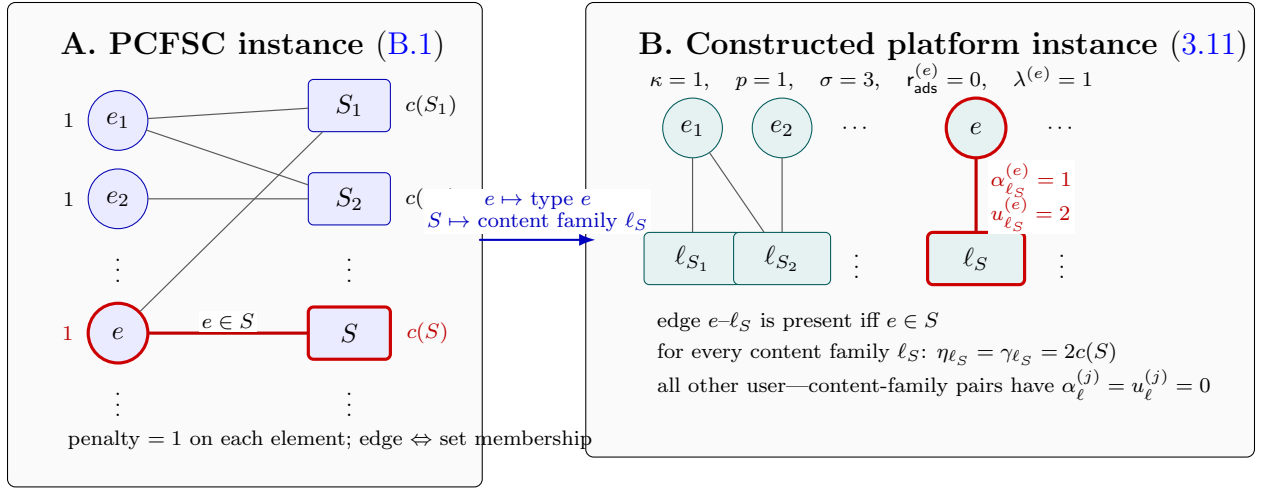
\begin{figure}[t]
\centering
\resizebox{\textwidth}{!}{%
\begin{tikzpicture}[
    >=Latex,
    font=\small,
    line join=round,
    line cap=round,
    every node/.style={inner sep=2pt},
    panel/.style={draw, rounded corners=3pt, fill=gray!4, inner sep=8pt},
    title/.style={font=\bfseries\large},
    note/.style={font=\scriptsize, align=left},
    elem/.style={circle, draw=blue!70!black, fill=blue!8, minimum size=8mm},
    set/.style={rectangle, rounded corners=2pt, draw=blue!70!black, fill=blue!8,
        minimum width=11mm, minimum height=7mm},
    user/.style={circle, draw=teal!70!black, fill=teal!10, minimum size=8mm},
    item/.style={rectangle, rounded corners=2pt, draw=teal!70!black, fill=teal!10,
        minimum width=13mm, minimum height=7mm},
    hl/.style={draw=red!80!black, very thick},
    incidence/.style={draw=black!65, line width=0.45pt},
    lab/.style={font=\scriptsize, fill=white, inner sep=1pt, text=black},
    redlab/.style={font=\scriptsize, fill=white, inner sep=1pt, text=red!75!black},
    map/.style={-{Latex[length=2.2mm]}, thick, blue!75!black}
]

\begin{scope}[shift={(0,0)}]
  \node[title, anchor=west] at (0.10,4.85)
    {A. PCFSC instance~\eqref{eqn:prize-collection-fractional-set-cover}};

  \node[elem] (e1) at (0.95,3.85) {${e}_1$};
  \node[elem] (e2) at (0.95,2.80) {$e_2$};
  \node[note] at (0.95,1.95) {$\vdots$};
  \node[elem,hl] (ee) at (0.95,1.00) {$e$};
  \node[note] at (0.95,0.15) {$\vdots$};

  \node[set] (S1) at (4.05,4.05) {${S}_1$};
  \node[set] (S2) at (4.05,2.80) {$S_2$};
  \node[note] at (4.05,1.95) {$\vdots$};
  \node[set,hl] (SS) at (4.05,1.00) {$S$};
  \node[note] at (4.05,0.15) {$\vdots$};

  \draw[incidence] (e1) -- (S1);
  \draw[incidence] (e1) -- (S2);
  \draw[incidence] (e2) -- (S2);
  \draw[incidence] (ee) -- (S1);
  \draw[incidence, very thick, red!75!black] (ee) -- (SS)
      node[midway, above, lab] {$e\in S$};

  \node[note, anchor=east] at ($(e1.west)+(-0.10,0)$) {$1$};
  \node[note, anchor=east] at ($(e2.west)+(-0.10,0)$) {$1$};
  \node[note, anchor=east, text=red!75!black] at ($(ee.west)+(-0.10,0)$) {$1$};

  \node[note, anchor=west] at ($(S1.east)+(0.12,0)$) {$c({S}_1)$};
  \node[note, anchor=west] at ($(S2.east)+(0.12,0)$) {$c(S_2)$};
  \node[note, anchor=west, text=red!75!black] at ($(SS.east)+(0.12,0)$) {$c(S)$};

  \node[note, anchor=west] at (0.20,-0.43)
    {penalty $=1$ on each element; edge $\Leftrightarrow$ set membership};

  \begin{scope}[on background layer]
    \node[panel, fit={(-0.25,5.15) (5.55,-0.78)}] {};
  \end{scope}
\end{scope}

\draw[map] (5.80,2.25) -- 
  node[midway, above=2pt, note, align=center, fill=white, inner sep=1pt,
        text=blue!75!black]
    {$e\mapsto$ type $e$\\[-1pt]$S\mapsto$ content family $\ell_S$}
  (7.35,2.25);

\begin{scope}[shift={(7.75,0)}]
  \node[title, anchor=west] at (0.10,4.85)
    {B. Constructed platform instance~\eqref{eqn:problem-buy-or-rent}};

  \node[note, anchor=west] at (0.25,4.42)
    {$\kappa=1, \quad p=1, \quad \sigma=3, \quad \radse =  0,  \quad \lambda^{(e)}=1$};

  \node[user] (u1) at (0.90,3.75) {${e}_1$};
  \node[user] (u2) at (2.10,3.75) {$e_2$};
  \node[note] at (3.10,3.75) {$\cdots$};
  \node[user,hl] (ue) at (4.70,3.75) {$e$};
  \node[note] at (5.85,3.75) {$\cdots$};

  \node[item] (l1) at (0.90,2.00) {$\ell_{{S}_1}$};
  \node[item] (l2) at (2.10,2.00) {$\ell_{S_2}$};
  \node[note] at (3.10,2.00) {$\vdots$};
  \node[item,hl] (lS) at (4.70,2.00) {$\ell_S$};
  \node[note] at (5.85,2.00) {$\vdots$};

  \draw[incidence] (u1) -- (l1);
  \draw[incidence] (u1) -- (l2);
  \draw[incidence] (u2) -- (l2);

  \draw[incidence, very thick, red!75!black] (ue) -- (lS)
      node[pos=.53, right=4pt, redlab, align=left]
      {$\alpha^{(e)}_{\ell_S}=1$\\[-1pt]$u^{(e)}_{\ell_S}=2$};

  \node[note, anchor=west] at (0.35,1.15)
    {edge $e$--$\ell_S$ is present iff $e\in S$};

  \node[note, anchor=west] at (0.35,0.73)
    {for every content family $\ell_S$: $\eta_{\ell_S}=\gamma_{\ell_S}=2c(S)$};

  \node[note, anchor=west] at (0.35,0.31)
    {all other user–-content-family pairs have $\alpha^{(j)}_{\ell}=u^{(j)}_{\ell}=0$};

\begin{scope}[on background layer]
  \node[panel, fit={(-0.25,5.15) (8.15,-0.38)}] {};
\end{scope}
\end{scope}

\end{tikzpicture}%
}
\caption{Instance construction}
\label{fig:np-hard-instance-construction}
\end{figure}

\begin{proof}[Proof outline of Lemma~\ref{lemma-our-prob-harder-than-fractional-cover}]
The proof proceeds as follows.
\begin{enumerate}
    \item Lemma~\ref{lemma:pcfsc-reduction-assumption1} establishes that the constructed instance
    satisfies Assumption~\ref{assump:MNL}.

    \item Lemma~\ref{lemma:pcfsc-reduction-full-equals-sub} shows that the ad side is inactive in
    the constructed instance. Thus, the full problem has the same optimal value as its
    subscription-only restriction.

    \item Lemma~\ref{lemma:pcfsc-reduction-reduced-form} then simplifies the subscription-only
    problem: in an optimal assortment distribution, 
    $P^{(s,e)} = 0$ implies that type $e$ is routed to the empty assortment and
    each user type $e$ only needs to be offered singleton set-items $\{\ell_S\}$
    with $e\in S$.

    \item Lemma~\ref{lemma:pcfsc-reduction-identities-self-contained} computes the induced utility
    and flow under this reduced form. In particular, type $e$ subscribes if and only if
    $\sum_{S\ni e} y_S^{(e)} = 1$, and the buy cost of content family $\ell_S$ becomes
    $c(S)\max_{e\in S} y_S^{(e)}$.

    \item Lemma~\ref{lemma:pcfsc-reduction-objective-self-contained} shows that every used
    set-item is weakly cheaper to buy than to rent. Therefore,   the platform objective can be
    rewritten as
    \[
    \Pi(\bm{y})=|U|-\mathsf{cost}(\bm{y}),
    \]
    where
    \[
    \mathsf{cost}(\bm{y})
    :=
    \sum_{S\in\mathcal S} c(S)\max_{e\in S} y_S^{(e)}
    +
    \sum_{e\in U} \I{u^{(s,e)}<1}.
    \]

    \item Lemma~\ref{lemma:pcfsc-reduction-platform-to-pcfsc-self-contained} maps any reduced-form
    platform assortment distribution $\bm{y}$ to a feasible PCFSC solution by setting
    \[
    x_S := \max_{e\in S} y_S^{(e)},
    \qquad
    \bar z_e := \I{u^{(s,e)}<1},
    \]
    and shows that the PCFSC objective is exactly $\mathsf{cost}(\bm{y})$.

    \item Lemma~\ref{lemma:pcfsc-reduction-pcfsc-to-platform-self-contained} proves the converse:
    from any feasible PCFSC solution $(\bm{x}, \bm{z})$ one can construct a reduced-form subscription assortment distribution
    whose platform value is at least
    \[
    |U| - \left(\sum_{S\in\mathcal S} c(S)x_S + \sum_{e\in U} z_e\right).
    \]

    \item Finally, Lemma~\ref{lemma:pcfsc-reduction-value-identity-self-contained} combines the
    two directions to obtain the exact value identity above, which yields the desired
    polynomial-time reduction from PCFSC to the decision version of problem~\eqref{eqn:problem-buy-or-rent}.
\end{enumerate}
\end{proof}

\begin{lemma}\label{lemma:pcfsc-reduction-assumption1}
The constructed instance satisfies Assumption~\ref{assump:MNL}.
\end{lemma}

\begin{proof}[Proof of Lemma~\ref{lemma:pcfsc-reduction-assumption1}]

Assumption~\ref{assump:MNL}(\ref{assm:alpha-implies-u-support}) is immediate from the construction: every content family with positive attraction for a type has positive utility for that type.
To verify Assumption~\ref{assump:MNL}(\ref{item:opt-profit-positive}), use the zero-cost singleton set guaranteed by the
preprocessing step. Let \(\Tilde{S}_1=\{\Tilde{e}_1\}\) with \(c(\Tilde{S}_1)=0\). Consider the feasible
solution that offers the singleton assortment \(\{\ell_{\Tilde{S}_1}\}\) to type \(\Tilde{e}_1\)
in the subscription mode with probability one, 
and routes every other type to the empty assortment and sets every
ad-side assortment distribution to
the empty assortment.
By construction,
\[
\alpha_{\ell_{\Tilde{S}_1}}^{(\Tilde{e}_1)}=1,
\qquad
u_{\ell_{\Tilde{S}_1}}^{(\Tilde{e}_1)}=2.
\]
Therefore,
\[
x^{(s,\Tilde{e}_1)}(\{\ell_{\Tilde{S}_1}\})
=
\frac{1}{1+1}
=
\frac12,
\]
and
\[
u^{(s,\Tilde{e}_1)}(\{\ell_{\Tilde{S}_1}\})
=
2\cdot \frac12
=
1
=
p.
\]
Since the ad-side assortment is empty, the ad-side utility is zero,  so $P^{(s,\Tilde{e}_1)}=1.$
Moreover, since \(c(\Tilde{S}_1)=0\) and \(\eta_{\ell_S}=\gamma_{\ell_S}=2c(S)\) for
all \(S\in\mathcal S\), we have
\[
\eta_{\ell_{\Tilde{S}_1}}=\gamma_{\ell_{\Tilde{S}_1}}=0.
\]
The platform therefore earns subscription revenue
\[
\lambda^{(\Tilde{e}_1)} pP^{(s,\Tilde{e}_1)}=1
\]
and incurs zero procurement cost. 
Thus, its profit is equal to \(1\), which
implies
\[
R^\star\ge 1>0.
\]
Thus, Assumption~\ref{assump:MNL}(\ref{item:opt-profit-positive}) holds.

Next, for  each $j$, $F\jj$ is
continuous and has bounded density
\[
\phi_F\jj(\chi) = 2\,\I{\chi\in(\frac{1}{2},1)}.
\]
So Assumption~\ref{assump:MNL}(\ref{item:assump-Lip}) holds  with
$W_{F^{(j)}}=2.$
\end{proof}

\begin{lemma}\label{lemma:pcfsc-reduction-full-equals-sub}
Let $\mathrm{OPT}^{\mathrm{full}}$ denote the optimal value of the   constructed instance and
let $\mathrm{OPT}^{\mathrm{sub}}$ denote the optimal value of the
same instance when only the subscription mode is offered.
Then
\[
\mathrm{OPT}^{\mathrm{full}} = \mathrm{OPT}^{\mathrm{sub}}.
\]
\end{lemma}

\begin{proof}[Proof of Lemma~\ref{lemma:pcfsc-reduction-full-equals-sub}]

Throughout this proof, for an assortment $A$ and a type
$e\in U$, $x^{(e)}(A)$ and $u^{(e)}(A)$ denote the click
probability and expected utility induced by offering $A$ to type $e$.
For a mode $m\in\{a,s\}$, the quantities $x^{(m,e)}$ and $u^{(m,e)}$
denote the corresponding averages under the mode-$m$ randomized assortment
distribution $\bm y^{(m,e)}$:
\[
x^{(m,e)}=\sum_{A\in\mathcal A_\kappa}y_A^{(m,e)}x^{(e)}(A),
\qquad
u^{(m,e)}=\sum_{A\in\mathcal A_\kappa}y_A^{(m,e)}u^{(e)}(A).
\]
Thus, for a    type $e\in U$, the superscripts $(a,e)$ and $(s,e)$
refer to the ad-side and subscription-side averages for type $e$.

The inequality $\mathrm{OPT}^{\mathrm{full}} \ge \mathrm{OPT}^{\mathrm{sub}}$ is immediate.
We prove the reverse inequality $\mathrm{OPT}^{\mathrm{full}} \le \mathrm{OPT}^{\mathrm{sub}}$. Fix any feasible assortment distribution $\bm{y}$
for the full model. 
Since $\kappa=1$, every nonempty assortment is a singleton.

\emph{Step 1: every ad-side singleton satisfies $u^{(e)}(A) \le 2x^{(e)}(A)$ for all feasible $A$.}

Fix a user type \(e\in U\). If \(A=\{\ell_S\}\) with \(e\in S\), then
\[
x^{(e)}(\{\ell_S\})=\frac12,
\qquad
u^{(e)}(\{\ell_S\})=1.
\]
If \(A=\{\ell_S\}\) with \(e\notin S\), then
\[
x^{(e)}(\{\ell_S\})=0,
\qquad
u^{(e)}(\{\ell_S\})=0.
\]
The empty assortment also satisfies
\[
x^{(e)}(\emptyset)=0,
\qquad
u^{(e)}(\emptyset)=0.
\]
Combining these cases, every feasible assortment \(A\in\mathcal A_\kappa\)
satisfies
\[
u^{(e)}(A)\le 2x^{(e)}(A).
\]

Therefore, by averaging over the ad-side assortment distribution $\bm y^{(a,e)}$,
\begin{align}
u^{(a,e)}
=
\sum_{A\in\mathcal A_\kappa} y_A^{(a,e)}u^{(e)}(A)
\le
2\sum_{A\in\mathcal A_\kappa} y_A^{(a,e)}x^{(e)}(A)
=
2x^{(a,e)} \qquad \forall e\in U. \label{eqn:sub-matters-step-1}
\end{align}

\emph{Step 2: $\bigl(u^{(s,e)}-p\bigr)_+ = 0$.}
For every user type $e$, 
the only positive-utility subscription-side singletons are
\(\{\ell_S\}\) with \(e\in S\), and each such singleton yields expected
utility \(1\).
 
Thus,
\[
u^{(s,e)} \le 1 = p
\qquad \forall e\in U, %
\]
so
\begin{align}
    \bigl(u^{(s,e)}-p\bigr)_+ = 0
\qquad \forall e\in U. %
\label{eqn:sub-matters-step-2}
\end{align}

\emph{Step 3: no type ever chooses the ad mode.}
Fix $e\in U$. If $x^{(a,e)} > 0$, then necessarily $u^{(a,e)} > 0$, because every singleton
that contributes positive inside-choice probability also contributes positive utility. Using~\eqref{eqn:sub-matters-step-1} and~\eqref{eqn:sub-matters-step-2}, the ad threshold in~\eqref{eqn:padj} becomes
\[
h^{(e)}(\bm y)
=
\frac{\sigma x^{(a,e)}}{u^{(a,e)}}
\ge
\frac{\sigma}{2}
>
1.
\]
Since $F^{(e)}(\chi)=1$ for all $\chi\ge 1$, it follows that
\[
P^{(a,e)} = \bar F^{(e)}\bigl(h^{(e)}\bigr) = 0,
\qquad
P^{(s,e)} = \I{u^{(s,e)}\ge 1}\,F^{(e)}\bigl(h^{(e)}\bigr) = \I{u^{(s,e)}\ge 1}.
\]
If instead $x^{(a,e)}=0$, then also $u^{(a,e)}=0$ by the inequality $u^{(a,e)}\le 2x^{(a,e)}$, so
$h^{(e)}=\infty$ by definition and again
\[
P^{(a,e)} = 0,
\qquad
P^{(s,e)} = \I{u^{(s,e)}\ge 1}.
\]
Thus deleting the ad side leaves every subscription choice probability unchanged and sets every
ad choice probability to zero.

\emph{Step 4: deleting the ad side weakly improves the objective.}
Replace every ad-side routing by the empty assortment. Revenue is unchanged because $\radse =0$ and
all $P^{(s,e)}$ values are unchanged. Procurement cost weakly decreases: for each content family $\ell$, the buy term
$\eta_\ell \max_{e,m} f_\ell^{(m,e)}(\bm{y})$ can only decrease when ad-side flows are removed, and the rental
term
\[
\gamma_\ell \sum_e \lambda^{(e)}
\bigl(
P^{(a,e)} f_\ell^{(a,e)}(\bm{y}) + P^{(s,e)} f_\ell^{(s,e)}(\bm{y})
\bigr)
\]
also weakly decreases because the ad contribution vanishes while the subscription-side terms stay
unchanged. Therefore,   every full-model assortment distribution can be converted into a subscription-only assortment distribution with
objective value at least as large. Thus,
\[
\mathrm{OPT}^{\mathrm{full}} \le \mathrm{OPT}^{\mathrm{sub}}.
\]
Combining the two inequalities proves the claim.
\end{proof}

From now on,
we work in the subscription-only restriction.

\begin{lemma}\label{lemma:pcfsc-reduction-reduced-form}
There exists an optimal subscription-only assortment distribution with the following properties:
\begin{enumerate}
\item whenever $P^{(s,e)} = 0$, type $e$ is routed to the empty assortment; \label{item:zero-prob-implies-zero-assortment}
\item for each user type $e$, the 
nonempty 
assortments offered can be 
restricted to $\set{\set{\ell_S}: S \ni e}$. 
\label{item:restricted-assortment-offer}
\end{enumerate}
\end{lemma}

\begin{proof}[Proof of Lemma~\ref{lemma:pcfsc-reduction-reduced-form}]

Because $\kappa=1$, every nonempty assortment is a singleton.

Whenever $P^{(s,e)}=0$,
replacing all of type $e$'s routing by the empty assortment leaves revenue unchanged and weakly
decreases procurement: buy costs depend on maxima of flows, and rental costs are multiplied by
$P^{(s,e)}$, which is zero. Thus, in an optimal assortment distribution,
we may assume   that
$P^{(s,e)}=0 \Rightarrow y^{(s,e)}=y_{\emptyset}$.
This proves item~(\ref{item:zero-prob-implies-zero-assortment}).

Finally, fix any singleton that has zero attraction and zero utility for type $e$; this includes every
$\{\ell_S\}$ with $e\notin S$.
Moving the mass on
such a singleton to the empty assortment leaves both revenue and procurement unchanged, because it
creates no flow and contributes no utility. Therefore,   in an optimal assortment distribution the only nonempty assortments
offered to type $e$ can be taken to be $\{\ell_S\}$ with $e\in S$. This proves~(\ref{item:restricted-assortment-offer}).
\end{proof}

For an assortment distribution in the reduced form of Lemma~\ref{lemma:pcfsc-reduction-reduced-form}, write
$y_S^{(e)}$ for the probability that type $e$ is offered $\{\ell_S\}$.

\begin{lemma}\label{lemma:pcfsc-reduction-identities-self-contained}
Fix any feasible subscription-only assortment distribution of the reduced form defined in Lemma~\ref{lemma:pcfsc-reduction-reduced-form}. Then, for every user type
$e\in U$,
\begin{align}
    u^{(s,e)} = \sum_{S\ni e} y_S^{(e)},
\qquad
P^{(s,e)} = \I{\sum_{S\ni e} y_S^{(e)} = 1}, \label{eqn:u-sum-of-y}
\end{align}
and, for every set $S\in\mathcal S$,
\begin{align}
    f_{\ell_S}^{(s,e)}(\bm{y}) = \frac{1}{2} y_S^{(e)}\I{e\in S},
\qquad
m_S(\bm{y}) := \max_{e\in U} f_{\ell_S}^{(s,e)}(\bm{y}) = \frac{1}{2}\max_{e\in S} y_S^{(e)}, \label{eqn:m_S-y-def}
\end{align}
so that
\[
\eta_{\ell_S} m_S(\bm{y}) = c(S)\max_{e\in S} y_S^{(e)}.
\]
\end{lemma}

\begin{proof}[Proof of Lemma~\ref{lemma:pcfsc-reduction-identities-self-contained}]
For $e\in U$ and $S\in\mathcal S$,
\[
p_{\ell_S}^{(e)}(\{\ell_S\})
=
\frac{\alpha_{\ell_S}^{(e)}}{1+\alpha_{\ell_S}^{(e)}}
=
\frac{1}{2}\,\I{e\in S}.
\]
Thus,
\[
u^{(e)}(\{\ell_S\})
=
u_{\ell_S}^{(e)}\,p_{\ell_S}^{(e)}(\{\ell_S\})
=
2\,\I{e\in S}\cdot \frac{1}{2}\,\I{e\in S}
=
\I{e\in S}.
\]
Therefore,  
\[
u^{(s,e)}
=
\sum_{S\in\mathcal S} y_S^{(e)} u^{(e)}(\{\ell_S\})
=
\sum_{S\ni e} y_S^{(e)}.
\]
Since the reduced form uses only the empty assortment and singletons, we have
$\sum_{S\in\mathcal S} y_S^{(e)} \le 1$.
Thus, $u^{(s,e)}\in[0,1]$. Because $p=1$,
\[
P^{(s,e)} = \I{u^{(s,e)}\ge 1} = \I{\sum_{S\ni e} y_S^{(e)} = 1}.
\]
Also, recall the definition of $f$ in~\eqref{eqn:y-flow-def},
\begin{align}
    f_{\ell_S}^{(s,e)}(\bm{y})
=
y_S^{(e)} p_{\ell_S}^{(e)}(\{\ell_S\})
=
\frac{1}{2} y_S^{(e)}\I{e\in S}. \label{eqn:f-s-l-s-expr}
\end{align}

Using~\eqref{eqn:f-s-l-s-expr},
we obtain
\[
m_S(\bm{y})
=
\max_{e\in U} f_{\ell_S}^{(s,e)}(\bm{y})
=
\max_{e\in U}\frac{1}{2} y_S^{(e)}\mathcal I(e\in S)
=
\frac{1}{2}\max_{e\in S} y_S^{(e)}.
\]
Since $\eta_{\ell_S}=2c(S)$,
\[
\eta_{\ell_S}m_S(\bm{y})
=
2c(S)\cdot \frac{1}{2}\max_{e\in S} y_S^{(e)}
=
c(S)\max_{e\in S} y_S^{(e)}.
\]
\end{proof}

\begin{lemma}\label{lemma:pcfsc-reduction-objective-self-contained}
Fix any feasible subscription-only assortment distribution 
of the reduced form in
Lemma~\ref{lemma:pcfsc-reduction-reduced-form}. Then every used content family $\ell_S$ is weakly
cheaper to procure by buying than by renting. 
Recall that $\Pi(\bm{y})$ denotes the objective value of the constructed platform instance under assortment distribution $\bm{y}$.
Consequently,
\[
\Pi(\bm{y})
=
\sum_{e\in U} P^{(s,e)}
-
\sum_{S\in\mathcal S} c(S)\max_{e\in S} y_S^{(e)}.
\]
Further, if we define
\begin{align}
    \mathsf{cost}(\bm{y})
:=
\sum_{S\in\mathcal S} c(S)\max_{e\in S} y_S^{(e)}
+
\sum_{e\in U} \I{u^{(s,e)}<1}, \label{eqn:cost-y-def}
\end{align}
then
\[
\Pi(\bm{y}) = |U| - \mathsf{cost}(\bm{y}).
\]
\end{lemma}

\begin{proof}[Proof of Lemma~\ref{lemma:pcfsc-reduction-objective-self-contained}]

Fix $S\in\mathcal S$. If $\max_{e\in S} y_S^{(e)}=0$, then content family $\ell_S$ is unused and contributes
zero under both buying and renting. Suppose now that $\max_{e\in S} y_S^{(e)} > 0$, and choose
\[
e^\star \in \arg\max_{e\in S} y_S^{(e)}.
\]
Then $y_S^{(e^\star)} > 0$. By Lemma~\ref{lemma:pcfsc-reduction-reduced-form} (\ref{item:zero-prob-implies-zero-assortment}), type $e^\star$
cannot have $P^{(s,e^\star)}=0$, because such a type would be routed to the empty assortment.
Thus, $P^{(s,e^\star)}=1$. Using Lemma~\ref{lemma:pcfsc-reduction-identities-self-contained} and $\lambda^{(e)} = 1$ for all $e$,
\[
\sum_e \lambda^{(e)} P^{(s,e)} f_{\ell_S}^{(s,e)}(\bm{y})
\ge
\lambda^{(e\opt)}
P^{(s,e\opt)} f_{\ell_S}^{(s,e\opt)}(\bm{y})
=
f_{\ell_S}^{(s,e^\star)}(\bm{y})
=
\frac{1}{2}\max_{e\in S} y_S^{(e)}
=
m_S(\bm{y}).
\]
Since $\gamma_{\ell_S} \ge \eta_{\ell_S}$,
\[
\gamma_{\ell_S} \sum_e \lambda^{(e)} P^{(s,e)} f_{\ell_S}^{(s,e)}(\bm{y})
\ge
\eta_{\ell_S} m_S(\bm{y}).
\]
So for every used set-item, renting is weakly more expensive than buying. Since the procurement
term uses the minimum of the buy cost and the rental cost, the actual procurement cost of
any used set-item equals its buy cost. 
Therefore,  
using $\lambda^{(e)}=1$ and $p=1$,
\[
\Pi(\bm{y})
=
\sum_{e\in U} P^{(s,e)}
-
\sum_{S\in\mathcal S} \eta_{\ell_S} m_S(\bm{y})
=
\sum_{e\in U} P^{(s,e)}
-
\sum_{S\in\mathcal S} c(S)\max_{e\in S} y_S^{(e)},
\]
where the second equality  follows from Lemma~\ref{lemma:pcfsc-reduction-identities-self-contained}.

Finally, $u^{(s,e)}\in[0,1]$ by Lemma~\ref{lemma:pcfsc-reduction-identities-self-contained}, so
\[
P^{(s,e)} = \I{u^{(s,e)}=1} = 1 - \I{u^{(s,e)}<1}.
\]
Substituting this into the objective formula gives
\[
\Pi(\bm{y})
=
|U|
-
\sum_{S\in\mathcal S} c(S)\max_{e\in S} y_S^{(e)}
-
\sum_{e\in U} \I{u^{(s,e)}<1}
=
|U|-\mathsf{cost}(\bm{y}),
\] where the last equality follows from the definition of $\mathsf{cost}(\bm{y})$ in~\eqref{eqn:cost-y-def}.
\end{proof}

Lemma~\ref{lemma:pcfsc-reduction-platform-to-pcfsc-self-contained} shows that any feasible reduced-form platform solution induces a feasible PCFSC solution with the same cost.

\begin{lemma}\label{lemma:pcfsc-reduction-platform-to-pcfsc-self-contained}
Let $\bm{y}$ be any feasible subscription-only assortment distribution of the reduced form in
Lemma~\ref{lemma:pcfsc-reduction-reduced-form}. Define
\[
x_S := \max_{e\in S} y_S^{(e)},
\qquad
\bar z_e := \I{u^{(s,e)}<1}.
\]
Then $(\bm{x},\bar{\bm z})$ is feasible for PCFSC and
\[
\sum_{S\in\mathcal S} c(S)x_S + \sum_{e\in U} \bar z_e = \mathsf{cost}(\bm{y}),
\]
where $\mathsf{cost}$
is defined in~\eqref{eqn:cost-y-def}.
\end{lemma}

\begin{proof}[Proof of Lemma~\ref{lemma:pcfsc-reduction-platform-to-pcfsc-self-contained}]

Fix $e\in U$. If $\bar z_e = 0$, then $u^{(s,e)} \ge 1$ since $\bar z_e = \I{u^{(s,e)}<1}$, so,
by
$u^{(s,e)} = \sum_{S\ni e} y_S^{(e)}$
in~\eqref{eqn:u-sum-of-y}
in Lemma~\ref{lemma:pcfsc-reduction-identities-self-contained},
\[
\sum_{S\ni e} x_S
\ge
\sum_{S\ni e} y_S^{(e)}
=
u^{(s,e)}
\ge 
1.
\]
If $\bar z_e = 1$, then the PCFSC covering constraint 
$\sum_{S\ni e}x_S + \bar{z}_e \ge 1,~  \forall e\in U$
is automatic. Thus, $(\bm{x},\bar{\bm{z}})$ is feasible.
By definition,
\[
\sum_{S\in\mathcal S} c(S)x_S + \sum_{e\in U} \bar z_e
=
\sum_{S\in\mathcal S} c(S)\max_{e\in S} y_S^{(e)}
+
\sum_{e\in U} \I{u^{(s,e)}<1}
=
\mathsf{cost}(\bm{y}).
\]
\end{proof}

\begin{lemma}\label{lemma:pcfsc-reduction-pcfsc-to-platform-self-contained}
For every feasible PCFSC solution $(\bm{x}, \bm{z})$, there exists a feasible subscription-only assortment distribution $\bm{y}$ in the
reduced form   of Lemma~\ref{lemma:pcfsc-reduction-reduced-form}
such that
\[
\Pi(\bm{y})
\ge
|U| - \sum_{S\in\mathcal S} c(S)x_S - \sum_{e\in U} z_e.
\]
\end{lemma}

\begin{proof}[Proof of Lemma~\ref{lemma:pcfsc-reduction-pcfsc-to-platform-self-contained}]

Fix any feasible PCFSC solution 
$(\bm{x}, \bm{z})$. Construct a subscription-only assortment distribution as follows.
For each $e\in U$:
\begin{itemize}
    \item if $z_e=1$, offer the empty assortment with probability one;
    \item if $z_e=0$, let
 \begin{align}
         C_e := \sum_{S\ni e} x_S \ge 1, \label{eqn:C-e-ge-1}
 \end{align}
    where $C_e \ge 1$ follows from the feasibility of  $(\bm{x}, \bm{z})$,
    and set
    \[
    y_S^{(e)}
    :=
    \begin{cases}
    x_S/C_e, & e\in S,\\
    0, & e\notin S,
    \end{cases}
    \qquad
    y_{\emptyset}^{(e)} := 0.
    \]
\end{itemize}
This assortment distribution is feasible and is already in reduced form: it
only offers the empty assortment or set-items \(\{\ell_S\}\) with \(e\in S\).

If $z_e=0$, then
\[
\sum_{S\ni e} y_S^{(e)}
=
\frac{1}{C_e}\sum_{S\ni e} x_S
=
1,
\]
so Lemma~\ref{lemma:pcfsc-reduction-identities-self-contained} gives $u^{(s,e)}=1$. If $z_e=1$,
then $u^{(s,e)}=0$. Therefore,  
\[
\sum_{e\in U} \I{u^{(s,e)}<1} = \sum_{e\in U} z_e.
\]
Also,
since $C_e \ge 1$ in~\eqref{eqn:C-e-ge-1},
for every $S\in\mathcal S$ and every $e\in S$,
\[
y_S^{(e)} = x_S/C_e \le x_S,
\]
so
\[
\max_{e\in S} y_S^{(e)} \le x_S.
\]
Thus,
\[
\mathsf{cost}(\bm{y})
\le
\sum_{S\in\mathcal S} c(S)x_S + \sum_{e\in U} z_e.
\]
Because the assortment distribution is in reduced form, Lemma~\ref{lemma:pcfsc-reduction-objective-self-contained}
applies, and therefore
\[
\Pi(\bm{y})
=
|U| - \mathsf{cost}(\bm{y})
\ge
|U| - \sum_{S\in\mathcal S} c(S)x_S - \sum_{e\in U} z_e.
\]
\end{proof}

\begin{lemma}\label{lemma:pcfsc-reduction-value-identity-self-contained}
Let $\optp$ denote the optimal value of the constructed platform instance. Then
\[
\optp = |U| - \optpcfsc.
\]
\end{lemma}

\begin{proof}[Proof of Lemma~\ref{lemma:pcfsc-reduction-value-identity-self-contained}]
By Lemma~\ref{lemma:pcfsc-reduction-full-equals-sub}, it is enough to optimize over the
subscription-only restriction. By Lemma~\ref{lemma:pcfsc-reduction-reduced-form}, there exists an
optimal subscription-only assortment distribution $\bm{y}\opt$ in reduced form. Define
\[
x_S := \max_{e\in S} y_S^{\star(e)},
\qquad
\bar z_e := \I{u^{(s,e)}(\bm{y}\opt)<1}.
\]
By Lemma~\ref{lemma:pcfsc-reduction-platform-to-pcfsc-self-contained}, the pair $(\bm{x},\bar{\bm{z}})$ is
feasible for PCFSC and has objective value $\mathsf{cost}(\bm{y}\opt)$. Therefore,  
\[
\optpcfsc \le \mathsf{cost}(\bm{y}\opt).
\]
Using Lemma~\ref{lemma:pcfsc-reduction-objective-self-contained},
\begin{align}
    \optp
=
\Pi(\bm{y}\opt)
=
|U| - \mathsf{cost}(\bm{y}\opt)
\le
|U| - \optpcfsc. \label{eqn:optp-is-small-one-side}
\end{align}

Conversely, let $(\bm{x}^\star,\bm{z}^\star)$ be an optimal PCFSC solution. By
Lemma~\ref{lemma:pcfsc-reduction-pcfsc-to-platform-self-contained}, there exists a feasible
subscription-only assortment distribution $\bm{y}$ in reduced form such that
\[
\Pi(\bm{y})
\ge
|U| - \sum_{S\in\mathcal S} c(S)x_S^\star - \sum_{e\in U} z_e^\star
=
|U| - \optpcfsc.
\]
Thus,
\begin{align}
    \optp \ge |U| - \optpcfsc. \label{eqn:optp-is-large-one-side}
\end{align}
Combining~\eqref{eqn:optp-is-small-one-side}
and~\eqref{eqn:optp-is-large-one-side}
proves the claim.
\end{proof}

\begin{proof}[Proof of Lemma~\ref{lemma-our-prob-harder-than-fractional-cover}]

Lemma~\ref{lemma:pcfsc-reduction-assumption1}
and
Lemma~\ref{lemma:pcfsc-reduction-value-identity-self-contained} show  that from every PCFSC
instance one can construct, in polynomial time, an instance of problem~\eqref{eqn:problem-buy-or-rent} with $\kappa=1$
satisfying Assumption~\ref{assump:MNL} and obeying the exact value identity
\[
\optp = |U| - \optpcfsc.
\]
Equivalently, for every bound $B$,
\[
\optpcfsc \le B
\qquad\Longleftrightarrow\qquad
\optp \ge |U| - B.
\]
Thus the PCFSC decision problem reduces in polynomial time to the decision version of
problem~\eqref{eqn:problem-buy-or-rent}.  
\end{proof}

\subsection{Lemma \ref{lemma:pcfsc-nphard} and its proof}
\label{subsec:pcfsc-nphard-proof}

\begin{lemma}\label{lemma:pcfsc-nphard} 
The decision version of~\eqref{eqn:prize-collection-fractional-set-cover}, with rational input data and a rational
bound, is NP-hard.
\end{lemma}

We consider the decision variant of the prize-collecting fractional set cover formulation~\eqref{eqn:prize-collection-fractional-set-cover}: given a rational instance and a rational bound  $B$, decide whether there exists a
feasible solution with objective value at most $B$.
We will use the following strongly NP-complete feasibility problem that appears 
in the proof of Lemma~1 of \cite{QiuAhDeWo14}.

\begin{proposition}[{\cite{QiuAhDeWo14}, Lemma~1}]
\label{prop:qiu-HI-feasibility}
Let $G=(V,E)$ be an undirected graph and let $q\in\mathbb{N}$. 
Deciding feasibility of the following system is strongly NP-complete:
\begin{align}
2\sum_{v\in V} x_v \;-\; \sum_{v\in V} z_v \;-\; \sum_{e\in E} y_e
&\le q-|E|,
\label{eq:qiu-feas-7}\\
x_i + x_j &\ge y_{ij}, \quad (i,j) \in E
\label{eq:qiu-feas-8}\\
x_i &\ge z_i, \quad i \in V,
\label{eq:qiu-feas-9}\\
x &\in \mathbb{R}^{|V|}_{\ge 0},\quad
y\in\{0,1\}^{|E|},\quad
z\in\{0,1\}^{|V|}.
\label{eq:qiu-feas-dom}
\end{align} 
\end{proposition}

\paragraph{Exact one-to-one mapping to prize-collecting fractional set cover.}
Fix any instance $(G=(V,E),q)$ of Proposition~\ref{prop:qiu-HI-feasibility}.
We construct a prize-collecting fractional set cover instance $(U,\mathcal{S},c,p)$ and bound $B$
such that there is an \emph{explicit one-to-one correspondence} between solutions of
\eqref{eq:qiu-feas-7}--\eqref{eq:qiu-feas-dom} and solutions of \eqref{eqn:prize-collection-fractional-set-cover}
meeting the bound.

\subsubsection*{Construction (universe, sets, costs, penalties, and bound)}
\paragraph{Universe.}
Let
\[
U \ :=\ V\ \cup\ E,
\]
i.e., we create one element for each vertex and one element for each edge.

\paragraph{Sets.}
For each vertex $v\in V$, create a set $S_v\in\mathcal{S}$ defined by
\[
S_v \ :=\ \{v\}\ \cup\ \{e\in E : e \text{ is incident to } v\}.
\]
Thus,
each edge-element $e=\{i,j\}$ belongs to exactly the two sets $S_i,S_j$, while each vertex-element $v$
belongs to exactly one set, namely $S_v$.

\paragraph{Costs and penalties.}
Set uniform costs and penalties
\[
c(S_v)=2 \ \ \forall v\in V,
\qquad
p(u)=1 \ \ \forall u\in U.
\]

\paragraph{Decision bound.}
Set
\[
B \ :=\ q + |V|.
\]

\subsubsection*{The induced PCFSC MIP.}
Writing $x_v$ for the variable $x_{S_v}$, the PCFSC decision instance is:
\begin{equation}\label{eq:PCFSC-from-qiu}
\begin{split}
\min_{\bm{x}, \bm{w}} & \quad  2\sum_{v\in V} x_v \;+\; \sum_{u\in U} w_u\\
\text{s.t. }\
& x_v + w_v \ge 1 \qquad\qquad\qquad\qquad\ \forall v\in V,\\
& x_i + x_j + w_{\{i,j\}} \ge 1 \qquad\qquad\ \forall \{i,j\}\in E,\\
& x_v \ge 0 \qquad\qquad\qquad\qquad\qquad\qquad\ \ \forall v\in V,\\
& w_u \in\{0,1\} \qquad\qquad\qquad\qquad\ \ \forall u\in U.
\end{split}
\end{equation}
This is exactly \eqref{eqn:prize-collection-fractional-set-cover} on the constructed $(U,\mathcal{S},c,p)$,
with the binary variables renamed from $\bm{z}$ to $\bm{w}$ to avoid confusion with Proposition~\ref{prop:qiu-HI-feasibility}.

\begin{lemma}
\label{lem:qiu-to-PCFSC-mapping}
Take any triple $(\bm x,\bm y, \bm z)$ with
$\bm x\in\mathbb{R}^{|V|}_{\ge 0}$, $\bm y\in\{0,1\}^{|E|}$, and $\bm z\in\{0,1\}^{|V|}$.
Define PCFSC variables $(\bm x,\bm w)$ by keeping the $\bm x$-variables the same and complementing the binary variables:
\[
\begin{aligned}
x_{S_v} &:= x_v && \text{for all } v\in V,\\
w_v &:= 1-z_v && \text{for all } v\in V,\\
w_e &:= 1-y_e && \text{for all } e\in E.
\end{aligned}
\]
This mapping is one-to-one, since we can recover $y_e=1-w_e$ for $e\in E$ and $z_v=1-w_v$ for $v\in V$.

Moreover,
\[
(\bm x,\bm y, \bm z)\text{ satisfies \eqref{eq:qiu-feas-8} and \eqref{eq:qiu-feas-9}}
\Longleftrightarrow
(\bm{x}, \bm{w})\text{ is feasible for \eqref{eq:PCFSC-from-qiu}},
\]
and
\[
(\bm x,\bm y, \bm z)\text{ satisfies \eqref{eq:qiu-feas-7}}
\Longleftrightarrow
\text{the PCFSC objective value of }(\bm{x}, \bm{w})\text{ is at most }B=q+|V|.
\]
\end{lemma}

\begin{proof}[Proof of Lemma \ref{lem:qiu-to-PCFSC-mapping}]
Constraints.
Fix $v\in V$. The PCFSC constraint $x_v+w_v\ge 1$ is equivalent (using $w_v=1-z_v$) to $x_v\ge z_v$, which is \eqref{eq:qiu-feas-9}.
Fix an edge $e=\{i,j\}\in E$. The PCFSC constraint $x_i+x_j+w_e\ge 1$ is equivalent (using $w_e=1-y_e$) to $x_i+x_j\ge y_e$, which is \eqref{eq:qiu-feas-8}.

Objective bound.
Using $w_v=1-z_v$ for $v\in V$ and $w_e=1-y_e$ for $e\in E$,
\[
2\sum_{v\in V}x_v + \sum_{u\in U} w_u
=
\Bigl(2\sum_{v\in V}x_v - \sum_{v\in V}z_v - \sum_{e\in E}y_e\Bigr) + (|V|+|E|).
\]
Therefore,
the PCFSC objective is at most $B=q+|V|$ if and only if
\[
2\sum_{v\in V}x_v - \sum_{v\in V}z_v - \sum_{e\in E}y_e \le q-|E|,
\]
which is exactly \eqref{eq:qiu-feas-7}.
\end{proof}

\begin{proof}[Proof of Lemma~\ref{lemma:pcfsc-nphard}]
Proposition~\ref{prop:qiu-HI-feasibility} (from \cite{QiuAhDeWo14}) states that it is strongly NP-complete
to decide whether the feasibility system \eqref{eq:qiu-feas-7}--\eqref{eq:qiu-feas-dom} has a solution.

Given an instance $(G,q)$ of that system, we construct the prize-collecting fractional set cover instance
$(U,\mathcal{S},c,p)$ and bound $B=q+|V|$ as above. By Lemma~\ref{lem:qiu-to-PCFSC-mapping}, there is a
one-to-one correspondence between feasible solutions of \eqref{eq:qiu-feas-7}--\eqref{eq:qiu-feas-dom} and
feasible PCFSC solutions of cost at most $B$. 
Thus, the PCFSC decision problem is NP-hard.
\end{proof}

\subsection{Proof of Theorem~\ref{thm:our-prob-np-hard}} \label{sec:proof-np-hard-using-lemma}

\begin{proof}[Proof of Theorem \ref{thm:our-prob-np-hard}]
This follows by combining Lemma \ref{lemma-our-prob-harder-than-fractional-cover} 
and Lemma \ref{lemma:pcfsc-nphard}.
\end{proof}

\section{Proof of 
Lemma~\ref{lem:sub-zero-p}, 
Lemma~\ref{lemma:ratio-reduction-in-G}, Theorem~\ref{thm:mnl-inner-separated-gap}, and Theorem~\ref{thm:asymptotic-opt}}
\label{sec:proof-of-thm}

We organize the proof section as follows.
Section~\ref{sec:prelim-lemmas} states the preliminary lemmas.
The detailed proofs of the preliminary lemmas are collected in
Sections~\ref{app:proof-detail-prelim-lemma}$-$~\ref{sec:proof-additional-lemmas},  and the proofs of
Lemma~\ref{lem:sub-zero-p}, 
Lemma~\ref{lemma:ratio-reduction-in-G},
Theorem~\ref{thm:mnl-inner-separated-gap}, and
Theorem~\ref{thm:asymptotic-opt} are given in
Section~\ref{sec:proof-main-lemma-thm}.
Throughout Section~\ref{sec:proof-of-thm},
we will fix $(S,\bm{\beta})$ 
and a user type $j$
unless otherwise stated.

\subsection{Preliminary lemmas}
\label{sec:prelim-lemmas}

We first collect several preliminary lemmas
that isolate the main structural ingredients used in the proof. Lemma~\ref{lemma:relaxation-gap}
bounds the loss from replacing the max-type buying cost by the
$\bm{\beta}$-relaxation. Lemma~\ref{lemma:buyset-gap-mnl} bounds the loss from
fixing a candidate buy set $S$, and Corollary~\ref{cor:buyset-threshold}
identifies the buy set that minimizes this bound. Lemma~\ref{lem:sub-zero-p}
shows that the subscription-side utility can be reduced to either $0$ or $p$,
and Lemma~\ref{lem:ratio-interval} characterizes the feasible range of the
ads-side ratio $x^{(a,j)}/u^{(a,j)}$.

\begin{lemma}
\label{lemma:relaxation-gap}
For every feasible assortment distribution $\bm{y}$,
\begin{align} 
G(S,\bm{\beta}, \bm{y}) - R(S, \bm{y})
=
\sum_{\ell \in S} \eta_\ell
\left[
\max_{j \in [J],m \in \set{a,s}} f_\ell^{(m,j)}(\bm{y})
-
\sum_{j \in [J],m \in \set{a,s}} \beta_\ell^{(m,j)} f_\ell^{(m,j)}(\bm{y})
\right]
\ge 0.\label{eqn:gap-def-identity}
\end{align}
Then,
\begin{align}
0  
\le
G(S,\bm{\beta}, \bm{y}) - R(S, \bm{y}) \le 
\sum_{\ell \in S}
\eta_\ell
\left(1-\min_{j \in [J],m \in \set{a,s}}\beta_\ell^{(m,j)}\right)
\max_{j \in [J],m \in \set{a,s}}f_\ell^{(m,j)}(\bm{y})
\le \TOP_{2J}\set{\bm{q}^{S,\bm{\beta}}},
\label{eqn:gap-bound}
\end{align}
where  $\bm{q}^{S,\bm{\beta}}$ 
is defined in~\eqref{eqn:q-S-def}.
In addition, if $\bm{y}\opt$ is any assortment distribution
that attains 
$G(S,\bm{\beta})$ in~\eqref{eqn:G-S-beta-def}, then
\begin{align}
0 \le R(S) - R(S, \bm{y}\opt) \le 
\TOP_{2J}\set{\bm{q}^{S,\bm{\beta}}}.
\label{eqn:gap-bound-by-relaxation}
\end{align}
\end{lemma}

Next, Lemma~\ref{lemma:buyset-gap-mnl}
bounds the gap induced by the choice of the buying set $S$, namely, the difference between the maximum profit attainable when the buying set is restricted to $S$ and the maximum profit attainable under an
optimal buying set $S\opt$.

\begin{lemma}
\label{lemma:buyset-gap-mnl}
For   every buy set $S \subseteq [L]$, 
\begin{equation}\label{eq:buyset-gap-mnl}
R\opt - R(S)
\le
\TOP_{2 J} \set{\bm{w}^{S}},
\end{equation}
where $\bm{w}^S$
is defined in~\eqref{eq:wS-mnl-def}.
\end{lemma}

We can choose a buying set $S$ that minimizes the upper bound in Lemma~\ref{lemma:buyset-gap-mnl}, namely the right-hand side of~\eqref{eq:buyset-gap-mnl}, where $\bm{w}^{S}$ is defined in~\eqref{eq:wS-mnl-def}.

\begin{corollary}
\label{cor:buyset-threshold}
Let
\begin{equation*}
  b_\ell
  :=
  \Bigl(
    \gamma_\ell \sum_{j=1}^J \lambda\jj - \eta_\ell
  \Bigr)_+.
\end{equation*}
A buy set 
$S$
minimizing the upper bound~\eqref{eq:buyset-gap-mnl}
in Lemma~\ref{lemma:buyset-gap-mnl}, denoted by $S\opt$, is obtained   by
\begin{equation*}
  \ell \in S^\star
  \quad\Leftrightarrow\quad
  \eta_\ell \le b_\ell
  \quad\Leftrightarrow\quad
  2\eta_\ell \le \gamma_\ell \sum_{j=1}^J \lambda\jj.
\end{equation*} 
\end{corollary}

Once the subscription side has been reduced to either zero utility or utility exactly equal to the price, the remaining   degree of freedom on the ads  side is the ratio between click mass and utility mass.
The next lemma identifies the exact feasible interval of this ratio 
and shows that every feasible ratio can already be attained by a mixture supported on at most two singleton assortments.

\begin{lemma} 
\label{lem:ratio-interval}

Every feasible ad assortment distribution with $u^{(a,j)} > 0$ induces a ratio
\begin{equation*}
  \frac{x^{(a,j)}}{u^{(a,j)}} \in I\jj,
\end{equation*}
and every $t \in I^{(j)}$ is attainable by an assortment distribution supported on at most two singleton assortments.
\end{lemma}

Next,  in Section~\ref{sec:generic-LP}, 
we analyze the setting in Section~\ref{sec:abstract-LP}.
 Then we apply it to the ads LP~\eqref{eqn:ads-LP-def} and the subscription LP~\eqref{eqn:sub-LP-def} in Sections~\ref{sec:LP-ads} and~\ref{sec:LP-subscription}, respectively.

\subsubsection{A generic one-equality parametric LP primitive}
\label{sec:generic-LP}

For notational simplicity, we first state the one-equality LP primitive
for a fixed value of the parameter and suppress the dependence on
\(\theta\).  
Let
\begin{align}
    F\opt
=
\max\left\{
\sum_{a \in \cA} y_a \phi(a)
:
y \in \Delta(\cA),\
\sum_{a \in \cA} y_a d(a) = 0
\right\}, \label{eqn:F-opt}
\end{align}
and
\begin{align}
    g(\nu)
=
\max_{a \in \cA}
\left\{
\phi(a)+\nu d(a)
\right\} \label{eqn:g-theta-def}
\end{align}
so that
\[
F\opt = \min_{\nu \in \R} g(\nu).
\]

We refer to
\begin{align*}
g(\nu)
:=
\max_{a\in\cA}
\{
\phi(a)+\nu d(a)
\}
\end{align*}
as the \emph{full dual}, because it is the   dual of the one-constraint primal problem over the full action set $\cA$.
For any finite candidate set $\cC \subseteq \cA$, we refer to
\begin{align}
g_{\cC}(\nu)
:=
\max_{a\in\cC}
\{
\phi(a)+\nu d(a)
\}
\label{eqn:g-theta-C-def}
\end{align}
as the corresponding \emph{restricted dual}, because it is the dual of the same one-constraint primal problem restricted to the candidate set $\cC$.

For later reference, for any finite candidate set $\cC \subseteq \cA$, define
\begin{align}
\Phi_{\mc{C}}
:=
\max_{a \in \cC} \phi(a)
-
\min_{a \in \cC} \phi(a),
\qquad
D
:=
\max_{a \in \cA} |d(a)|.
\label{eqn:Phi-G-theta-def}
\end{align}
In the bisection-based primitive below, the initial dual interval is taken as an input,
so only $D$ enters the final approximation bound.

Suppose  we are given a compact interval
\[
\bar{\mc{J}}
=
[L_{0},U_{0}]
\subseteq \R,
\]
and two actions $b^-_{0}, b^+_{0} \in \cA$ such that
\begin{align}
d(b^-_{0}) < 0 < d(b^+_{0}),
\label{eqn:generic-init-sign}
\end{align}
and
\begin{align}
\phi(b^-_{0}) + L_{0} d(b^-_{0})
\ge g(0),
\qquad
\phi(b^+_{0}) + U_{0} d(b^+_{0})
\ge g(0).
\label{eqn:generic-init-domination}
\end{align}

First, in Lemma~\ref{lem:dual-localization},
we show that the above endpoint certificates imply that the full dual can be minimized over the compact interval $\bar{\mc{J}}$.

\begin{lemma}
\label{lem:dual-localization}
 Assume that $\bar{\mc{J}} = [L_{0},U_{0}]$ and
$b^-_{0}, b^+_{0} \in \cA$
satisfy~\eqref{eqn:generic-init-sign}--\eqref{eqn:generic-init-domination}.
Then
\[
\min_{\nu \in \R} g(\nu)
=
\min_{\nu \in \bar{\mc{J}}} g(\nu).
\]
\end{lemma}

We now describe a  bisection routine.
Starting from
\[
L_{0},\ U_{0},\ b^-_{0},\ b^+_{0},
\qquad
\cC^{\mathrm{bis}}_{0}
:=
\{b^-_{0},b^+_{0}\},
\]
we iterate as follows.
At iteration $n \ge 0$, define the midpoint
\[
m_{n}
:=
\frac{L_{n}+U_{n}}{2},
\]
query
\[
a_{n}
\in
\mathcal \argmax_{a\in\cA}
\{
\phi(a)+ m_n d(a)
\},
\]
and enlarge the candidate set by
\[
\cC^{\mathrm{bis}}_{n+1}
:=
\cC^{\mathrm{bis}}_{n}
\cup
\{a_{n}\}.
\]
If $d(a_{n})=0$, the routine terminates.
If $d(a_{n})<0$, we update
\[
L_{n+1}:=m_{n},
\qquad
U_{n+1}:=U_{n},
\qquad
b^-_{n+1}:=a_{n},
\qquad
b^+_{n+1}:=b^+_{n}.
\]
If $d(a_{n})>0$, we update
\[
L_{n+1}:=L_{n},
\qquad
U_{n+1}:=m_{n},
\qquad
b^-_{n+1}:=b^-_{n},
\qquad
b^+_{n+1}:=a_{n}.
\]

For a fixed iteration budget $N \ge 0$, let
$\cC^{\mathrm{bis}}_{N}$
denote the final candidate set returned by this routine after at most $N$ iterations,
and define the restricted primal value
\begin{align}
F^{N}
:=
\max\left\{
\sum_{a \in \cC^{\mathrm{bis}}_{N}} y_a \phi(a)
:
y \in \Delta(\cC^{\mathrm{bis}}_{N}),\
\sum_{a \in \cC^{\mathrm{bis}}_{N}} y_a d(a)=0
\right\}.
\label{eqn:LP-resctricted}
\end{align}

Lemma~\ref{lemma:generic-fixed-theta} 
gives the pointwise guarantee of the maintained-certificate bisection routine.

\begin{lemma}
\label{lemma:generic-fixed-theta}
Assume that
$\bar{\mc{J}} = [L_{0},U_{0}]$ and
$b^-_{0}, b^+_{0} \in \cA$
satisfy~\eqref{eqn:generic-init-sign}--\eqref{eqn:generic-init-domination}.
Run the above bisection routine for at most $N$ iterations,
and let $F^{N}$ be the restricted-LP value in~\eqref{eqn:LP-resctricted}.
If the routine terminates early because
$d(a_{n})=0$ for some iteration $n$, then
\[
F^{N}=F\opt.
\]
Otherwise,
\[
F^{N}
\ge
F\opt
-
\frac{D}{2}
\bigl(U_{0}-L_{0}\bigr) 2^{-N}.
\]
\end{lemma}

\paragraph{Parameter-dependent notation.}
In what follows, we use the generic one-equality LP as a
parameterized primitive.  Thus, the reward, residual, candidate set, and
the constants used in the dual-search certificate may all depend on the
parameter \(\theta\).  To make this dependence explicit, we write the
optimal value of \eqref{eqn:F-opt} as
\begin{align*}
    F\opt(\theta)
:=
\max_{y\in\Delta(\mathcal A)}
\left\{
\sum_{a\in\mathcal A} y_a \phi(a, \theta)
:
\sum_{a\in\mathcal A} y_a d(a, \theta)=0
\right\}. 
\label{eqn:F-opt-para}
\end{align*}
Its dual is
\[
g(\nu, \theta)
:=
\max_{a\in\mathcal A}
\left\{
\phi(a, \theta)+\nu d(a, \theta)
\right\},
\qquad
F\opt(\theta)=\min_{\nu\in\mathbb R} g(\nu, \theta).
\]
Similarly, for a finite candidate set \(C(\theta)\subseteq\mathcal A\), we
define the parameter-dependent reward spread and residual bound by
\[
\Phi_C(\theta)
:=
\max_{a\in C(\theta)}\phi(a, \theta)
-
\min_{a\in C(\theta)}\phi(a, \theta),
\qquad
D(\theta)
:=
\max_{a\in\mathcal A}|d(a, \theta)|.
\]
All bisection intervals,   actions, generated candidate sets,
restricted primal values, and postprocessed solutions below are therefore
understood to be evaluated pointwise at the same parameter value
\(\theta\).  

Lemma~\ref{lemma:grid-generic} is the outer-grid counterpart of Lemma~\ref{lemma:generic-fixed-theta}.

\begin{lemma}
\label{lemma:grid-generic}
Suppose $I$ is a compact interval, $F\opt$ is $W_F$-Lipschitz on $I$, and
$\mathcal G \subseteq I$ is a finite   grid with  
\[
\Delta_{\mathcal G}
:=
\max_{\theta \in I}\min_{\theta_G \in \mathcal G} |\theta-\theta_G|.
\]
For each     $\theta \in \mathcal G$, suppose we are given
an initial interval
\[
\bar{\mc{J}}(\theta)
=
[L_0(\theta),U_0(\theta)],
\]
and actions
$b^-(\theta), b^+(\theta) \in \cA$
satisfying~\eqref{eqn:generic-init-sign}--\eqref{eqn:generic-init-domination}
with $g(\cdot)=g(\cdot,\theta)$.

Define
\[
C^{\mathrm{dual}}
:=
\max_{\theta \in \mathcal G}
\set{
\frac{D(\theta)}{2}
\bigl(U_0(\theta)-L_0(\theta)\bigr)2^{-N(\theta)}}.
\]
Then
\[
\max_{\theta \in \mathcal G} F^{N(\theta)}(\theta)
\ge
\max_{\theta \in I} F\opt(\theta)
-
W_F \Delta_{\mathcal G}
-
C^{\mathrm{dual}},
\] where $F^{N(\theta)}(\theta)$ denotes the   value of~\eqref{eqn:LP-resctricted} pointwise at parameter $\theta$
after $N(\theta)$   bisection iterations.
\end{lemma}

Next, we   apply the above results to 
the ads LP~\eqref{eqn:ads-LP-def}.

\subsubsection{Specialization to the ads LP}
\label{sec:LP-ads}

Recall the definition of $G\adj(t, \mu)$ in~\eqref{eq:ads-fixed-mu}
and $G\adj(t) = 
\min_{\mu \in \R} G\adj(t, \mu)$
in~\eqref{eqn:ads-duality}.
The next result specializes Lemmas~\ref{lem:dual-localization} and~\ref{lemma:generic-fixed-theta} 
to the fixed-$t$ ads LP and 
obtains an explicit interval in model parameters 
using two singleton  assortments.

Fix an interior ratio
$t\in (t_{\min}^{(j)},t_{\max}^{(j)})$.
Choose     content families $\lplus := \lplus\jj,\lminus := \lminus\jj \in[L]$ from~\eqref{eqn:two-prod-pair-def}
with $\alpha\jj_{\lplus} > 0 , \alpha\jj_{\lminus} > 0$
(we drop the dependence on $(j,t)$ for simplicity)
such that
\[
\frac{1}{u_{\lplus }^{(j)}}<t<\frac{1}{u_{\lminus}^{(j)}},
\qquad
A_t^-:=\{\lplus \},\qquad A_t^+:=\{\lminus\}.
\]
Recall from~\eqref{eqn:ads-dual-interval-algo} that
\begin{align}
    \bar{\mc{J}}\adj(t)
=\left[
\frac{-(1+\alpha_{\lplus }^{(j)}) \bar\Phi\adj}
{\alpha_{\lplus }^{(j)}(t u_{\lplus }^{(j)}-1)},
\frac{(1+\alpha_{\lminus}^{(j)}) \bar\Phi\adj}
{\alpha_{\lminus}^{(j)}(1-t u_{\lminus}^{(j)})}
\right],
\label{eqn:J-j-a-def}
\end{align}
where $\bar\Phi\adj
=
\lambda\jj(\rads\sigma+\gamma_{\max})+ \eta_{\max}$ 
is defined in~\eqref{eqn:Phi-ads-sub-def}.

\begin{lemma} 
For any $j$ and $t\in (t_{\min}^{(j)},t_{\max}^{(j)})$,
\begin{align*}
G\adj(t) = 
\min_{\mu \in \R} G\adj(t, \mu) = 
\min_{\mu \in \bar{\mc{J}}\adj(t)} G\adj(t, \mu).
\end{align*}
 \label{lemma-explict-dual-interval-ads}
\end{lemma}

Once the ads dual variable has been localized to a compact interval~\eqref{eqn:J-j-a-def}, each   subproblem~\eqref{eq:ads-fixed-mu}
is  a cardinality-constrained MNL assortment problem with fixed costs. 
The next lemma uses that reduction  to obtain a guarantee for the fixed-$(j,t)$ ads branch.

\begin{lemma}
\label{lem:ads-bisection-guarantee} %
Fix an interior ratio $t \in (t_{\min}^{(j)}, t_{\max}^{(j)})$. \\
Let  $(\ghat(t), \yhat(t) )\gets$~\ref{alg:mnl-fixed-jt-ads}$(j,t,S,\bm{\beta}, \bar{\mc{J}}\adj(t), K)$.
 Then
\begin{align}
    \ghat(t)
\ge
\geadj(t)
-
\frac{D^{(a,j)}}{2}\,
|\bar{\mc{J}}\adj(t)|\,2^{-K},
\label{eqn:ads-opt-gap-in-K}
\end{align} where $\geadj(t)$ is defined  in~\eqref{eqn:ads-LP-def}.
\end{lemma}

\subsubsection{Specialization to the subscription LP}
\label{sec:LP-subscription}

 Recall the definition of $G\subj(t, \xi)$ in~\eqref{eq:subs-dual-split}
 and $G\subj(t) = 
\max\set{\min_{\xi \in \R} G\subj(t, \xi), 0}$
in~\eqref{eqn:sub-duality}.
We next perform  a similar analysis for the subscription branch.

\begin{lemma}
Fix  a ratio
$t\in I^{(j)} \cup \set{\infty}$ and $p\in(0,\uumaxj)$.
Recall
\[
\uumaxj =\max_{B\in A_\kappa}u^{(j)}(B).
\] 
Define
\begin{align}
\bar{\mc{J}}\subj 
=
\left[
-\frac{\bar\Phi\subj}{p},
\frac{\bar\Phi\subj}{\uumaxj-p}
\right],
\label{eqn:J-j-s-def}
\end{align}
where $\bar\Phi\subj$ is defined in~\eqref{eqn:Phi-ads-sub-def}.
Then
\begin{align*}
\min_{\xi \in \R} G\subj(t, \xi) = 
\min_{\xi \in \bar{\mc{J}}\subj} G\subj(t, \xi).
\end{align*}
\label{lemma-explict-dual-interval-subscription}
\end{lemma}

After splitting the subscription dual into an affine term and a reduced assortment-maximization problem, we can apply the same 
strategy as on the ads side. The next lemma records the resulting guarantee for the fixed-$(j,t)$ subscription branch, with the additional affine-shift term made explicit.

\begin{lemma}
\label{lem:subscription-bisection-guarantee} %
Fix $p$  and $t \in I^{(j)} \cup \{\infty\}$.
Let 
$(\uumaxj, B^{(j,+)})
\gets$~\ref{alg:mnl-oracle}$(\bm\alpha^{(j)}, \bm{u}\jj, \kappa)$   and
\\ 
$(\ghat(t), \yhat(t)) \gets$~\ref{alg:mnl-fixed-jt-sub}($j,t,S,\bm{\beta},\bar{\mc{J}}\subj, K, B^{(j,+)},\uumaxj$). 
Suppose $p \in (0,\uumaxj)$. 
Then,
\[
\ghat(t)
\ge
\gesubj(t)
-
\frac{D\subj}{2}
|\bar{\mc{J}}\subj| 2^{-K}.
\]
\end{lemma}

\subsubsection{Additional lemmas}
\label{sec:additional-lemmas}

We now aggregate the fixed-$(j,t)$ ads and subscription guarantees over the sampled ratio grid for a fixed procurement decision $(S,\bm{\beta})$. The next lemma  combines the branchwise dual 
search errors and the subscription affine-shift budget into a single sampled-grid guarantee.

\begin{lemma}
\label{lem:G-gap-combined-ads-sub}
Fix an integer $K_T \ge 2$.
For each type $j\in[J]$, let $T\jj_{K_T}\subseteq I\jj$ be the ratio grid in
Algorithm~\ref{alg:mnl-solver-all}.
Let $\bar{\mc{J}}\adj(\infty) := [0,0]$
and $K\adj(\infty) = 0$.
For each  $t \in T\jj_{K_T} \cup \set{\infty},$
let
$(\ghat_{K_T}^{(a,j)}(t),  \bm y\adj(t) )\gets$~\ref{alg:mnl-fixed-jt-ads}$(j,t,S,\bm{\beta}, \bar{\mc{J}}\adj(t), K\adj(t))$ \\
and 
$(\ghat_{K_T}\subj(t), \bm {y}\subj(t)) \gets$~\ref{alg:mnl-fixed-jt-sub}($j,t,S,\bm{\beta},\bar{\mc{J}}\subj, K\subj, B^{(j,+)},\uumaxj$). 
Define
\[
\ghat\jj_{K_T}(t)
:=
\ghat\adj_{K_T}(t)+\ghat\subj_{K_T}(t).
\]
The subscript $K_T$
indicates dependence on the outer grid size, through the budgets 
$K\adj(t), K\subj$.
Then, for every $j\in[J]$ and every $t\in T\jj_{K_T}\cup \offset$,
\begin{align}
\label{eq:combined-error-at-t}
\ghat_{K_T}^{(j)}(t) \ge \geadj(t) + \gesubj(t) - \errork,
\end{align}
where
\begin{align}
\begin{small}
    \label{eqn:search-error-in-grid}
\errork
:=
\begin{cases}
\dfrac{D\adj}{2}  |\bar{\mc{J}}\adj(t)|  2^{-K\adj(t)},
& t=t_r^{(j)} \in T\jj_{K_T}\cap (t_{\min}\jj,t_{\max}\jj), \\[1.0ex]
0,
& t \in \{t_{\min}\jj,t_{\max}\jj,\infty\},
\end{cases}
+
\begin{cases}
\dfrac{D\subj}{2}  |\bar{\mc{J}}\subj|  2^{-K\subj},
& p\in(0,U_{\max}\jj), \\[1.0ex]
0,
& p\notin(0,U_{\max}\jj),
\end{cases}
\end{small}
\end{align}
where  $|\bar{\mc{J}}\adj(t)|, K\adj(t), |\bar{\mc{J}}\subj|, K\subj$ 
(all depend on $K_T$)
are defined in~\eqref{eqn:ads-dual-interval-width},~\eqref{eqn:inner-grid-size-ads},~\eqref{eqn:sub-dual-interval-width}, and~\eqref{eqn:inner-grid-size-sub},
respectively.
Here $D\adj,D\subj$ are defined in~\eqref{eqn:D-j-a-sdef}.

Consequently,
\begin{align}
\label{eq:grid-hat-error}
\sum_{j=1}^J \max_{t\in T\jj_{K_T} \cup \offset}
\ghat\jj_{K_T}(t)
\ge
\sum_{j=1}^J \max_{t\in T\jj_{K_T} \cup \offset} G\jj(t)
-
\sum_{j=1}^J \max_{t\in T\jj_{K_T} \cup \offset}
\errork .
\end{align}
\end{lemma}

\subsubsection{An explicit Lipschitz constant}
\label{sec:lip-consant}

In this section, we provide an explicit Lipschitz constant so that ~\eqref{eqn:lip-cond-holds} holds.
Let
\begin{align}
\deltad 
\defeq 
\min_{A,B \in \mc{A}_\kappa }\set{
\min\{u\jj(A),u\jj(B)\}
\left|
\frac{x\jj(A)}{u\jj(A)}
-
\frac{x\jj(B)}{u\jj(B)}
\right|: \frac{x\jj(A)}{u\jj(A)} \ne \frac{x\jj(B)}{u\jj(B)},
u\jj(A) > 0, u\jj(B) > 0
}.
\label{eqn:delta-d-def}
\end{align}
If the set over which the minimum is taken is empty, define  $\deltad = \infty$; equivalently, the second term in~\eqref{eqn:L-ratio-def}
is interpreted as $0$.
For fixed \((S,\bm{\beta},\sigma,p)\), define
 \begin{align} 
  \lipsig
  :=
\lambda\jj\sigma W_{F\jj}(\rads\sigma+p+2\gamma_{\max})
+
\frac{2\bar\Phi\adj \umaxj}{\deltad} \quad
\text{ with }
\deltad \text{ defined in~\eqref{eqn:delta-d-def}}.
\label{eqn:L-ratio-def}
\end{align}
 Lemma~\ref{lem:mnl-global-ratio-lipschitz}  
 shows that
\(\lipsig \) is a valid global Lipschitz constant for
the fixed-ratio mapping \(t\mapsto G\jj(t)\).
\begin{lemma} 
\label{lem:mnl-global-ratio-lipschitz}
For every type \(j\) and every
\(t,t'\in I\jj\),
\[
\ABS{G\jj(t)-G\jj(t')}
\le
\lipsig |t-t'|.
\]
\end{lemma}
The constant $\deltad$~\eqref{eqn:delta-d-def}
is computable when $L$ is moderate.
To obtain a valid lower bound on $\deltad$ for larger instances, we  impose the additional
regularity condition.
Note that Algorithm~\ref{alg:mnl-solver-all} does not need $\lipsig$ as part of the input,
and the purpose of computing 
a valid $\lipsig$ is to obtain an explicit convergence guarantee.
\begin{assumption}
\label{assump:rational}
For each type $j$, the model parameters admit a rational representation
\begin{align*}
 \alpha_\ell\jj = \frac{a_\ell\jj}{Q\jj},
 \qquad
 u_\ell\jj = \frac{b_\ell\jj}{L\jj},
 \qquad
 \ell \in [L],
\end{align*}
where $Q\jj, L\jj \in \mathbb{N}$ and $a_\ell\jj, b_\ell\jj \in \mathbb{Z}_{\ge 0}$. 
\end{assumption}

Assumption~\ref{assump:rational} is a technical condition used to obtain a computable lower bound for $\deltad$~\eqref{eqn:delta-d-def}.
Similar arithmetic assumptions appear in the assortment-optimization literature; see, for example,~\citep{SumidaGaRuToDa21}. 
Under Assumption~\ref{assump:rational}, for each type $j$, define
\[
N_{j,\kappa} := \TOP_{\kappa}\{(a_\ell\jj b_\ell\jj)_{\ell=1}^L\} > 0,
\]
and
\[
A_{j,\kappa}^{\mathrm{tot}} := \TOP_{\kappa}\{(a_\ell\jj)_{\ell=1}^L\}.
\]
Set
\begin{align}
\delta\jj := \frac{1}{(Q\jj + A_{j,\kappa}^{\mathrm{tot}}) N_{j,\kappa}} > 0.\label{eqn:delta-j-def}
\end{align}

Next,
we show  $    \deltad  
\ge \delta\jj.$
\begin{lemma} 
\label{lem:kappa-aware-weighted-separation}
Fix a type $j$.  Recalling 
$\deltad$~\eqref{eqn:delta-d-def} and
$\delta\jj$~\eqref{eqn:delta-j-def},
we have
\begin{align*}
\deltad  
\ge \delta\jj.  
\end{align*}
 Thus, in the definition of $\lipsig$~\eqref{eqn:L-ratio-def}, we can replace 
$\deltad$ with $\delta\jj$ and still have  
\begin{align*}
\ABS{G\jj(t)-G\jj(t')}
\le
\lipsig |t-t'|.
\end{align*}
\end{lemma}

\subsection{Proof of results in Section~\ref{sec:prelim-lemmas}}
 \label{app:proof-detail-prelim-lemma}

\begin{proof}[Proof of Lemma~\ref{lemma:relaxation-gap}]
The identity in~\eqref{eqn:gap-def-identity}
is immediate from the definitions of $R(S, \bm{y})$ and $G(S,\bm{\beta}, \bm{y})$.
Also, since $\sum_{j \in [J], m \in \set{a,s}}\beta_\ell^{(m,j)} = 1$ and $\beta_\ell^{(m,j)} \ge 0$,
\begin{align*}
  \sum_{j \in [J], m \in \set{a,s}}\beta_\ell^{(m,j)} f_\ell^{(m,j)}(\bm{y})
  \le
  \max_{j \in [J], m \in \set{a,s}} f_\ell^{(m,j)}(\bm{y}),
\end{align*}
which gives nonnegativity.

For each $\ell \in S$, let
\begin{align*}
  z_\ell(\bm{y}) := \max_{j \in [J], m \in \set{a,s}} f_\ell^{(m,j)}(\bm{y}),
  \qquad
  \underline{\beta}_\ell := \min_{j \in [J], m \in \set{a,s}}\beta_\ell^{(m,j)}.
\end{align*}
Then
\begin{align*}
  \sum_{j \in [J], m \in \set{a,s}}\beta_\ell^{(m,j)}f_\ell^{(m,j)}(\bm{y})
  \ge
  \underline{\beta}_\ell \sum_{j \in [J], m \in \set{a,s}}f_\ell^{(m,j)}(\bm{y})
  \ge
  \underline{\beta}_\ell  z_\ell(\bm{y}),
\end{align*}
so
\begin{align*}
  z_\ell(\bm{y}) - \sum_{j \in [J], m \in \set{a,s}}\beta_\ell^{(m,j)}f_\ell^{(m,j)}(\bm{y})
  \le
  (1-\underline{\beta}_\ell)  z_\ell(\bm{y}).
\end{align*}
Multiplying by $\eta_\ell$ and summing over $\ell \in S$ yields
\begin{align*}
G(S,\bm{\beta}, \bm{y}) - R(S, \bm{y})
  \le
  \sum_{\ell \in S}
  \eta_\ell (1-\underline{\beta}_\ell) z_\ell(\bm{y})
  =
  \sum_{\ell \in S}
  q^{S,\bm{\beta}}_{\ell}  z_\ell(\bm{y}).
\end{align*}

Finally, $0 \le z_\ell(\bm{y}) \le 1$ and
\begin{align}
  \sum_{\ell \in [L]} z_\ell(\bm{y})
  \le
  \sum_{\ell \in [L]} \sum_{j \in [J], m \in \set{a,s}}f_\ell^{(m,j)}(\bm{y})
  \le 
  2J,
  \label{eqn:sum-z-bound}
\end{align}
where the last inequality follows
from~\eqref{eqn:y-sum-less-than-kappa}.
Hence, by the definition of the $\TOP$ operator,
\begin{align*}
  \sum_{\ell \in [L]} q^{S,\bm{\beta}}_{\ell} z_\ell(\bm{y})
  \le
  \TOP_{2J}\set{q^{S,\bm{\beta}}},
\end{align*}
which proves~\eqref{eqn:gap-bound}.

Next,
we have
\begin{align*}
R(S)- R(S, \bm{y}\opt)
\le 
G(S,\bm{\beta}, \bm{y}\opt)  -
R(S, \bm{y}\opt)  \le
 \TOP_{2J}\set{q^{S,\bm{\beta}}}
\end{align*}
where the first inequality follows from~\eqref{eqn:beta-relax-ineq},
which implies that $R(S)\le 
G(S,\bm{\beta})= 
G(S,\bm{\beta}, \bm{y}\opt)$.
The   second inequality follows 
from~\eqref{eqn:gap-bound}.
\end{proof}

\begin{proof}[Proof of Lemma~\ref{lemma:buyset-gap-mnl}]
Fix    an assortment distribution $\bm{y}$. 
Recall we write $ z_\ell(\bm{y}) = \max_{j \in [J], m \in \set{a,s}} f_\ell^{(m,j)}(\bm{y})$.
Let
\begin{align*}
  \Delta(S, \bm{y})
  :=
  R(\bm{y})-R(S, \bm{y}).
\end{align*}

Recall
\begin{align*}
\mathcal{R}_\ell(\bm{y})
&=
\sum_{j=1}^J \lambda\jj
\Bigl(
\padj(\bm{y}) f_\ell^{(a,j)}(\bm{y})
+
\psubj(\bm{y}) f_\ell^{(s,j)}(\bm{y})
\Bigr) 
\end{align*}

By the definitions of $R(\bm{y})$ and $R(S, \bm{y})$, 
\begin{align*}
  \Delta(S, \bm{y})
  =
  \sum_{\ell \in S}
  \Bigl[
    \eta_\ell z_\ell(\bm{y}) - \min\{\eta_\ell z_\ell(\bm{y}), \gamma_\ell \mathcal{R}_\ell(\bm{y})\}
  \Bigr]
  +
  \sum_{\ell \notin S}
  \Bigl[
    \gamma_\ell \mathcal{R}_\ell(\bm{y}) - \min\{\eta_\ell z_\ell(\bm{y}), \gamma_\ell \mathcal{R}_\ell(\bm{y})\}
  \Bigr].
\end{align*}
Using the identity
$x - \min\set{x,y} = (x-y)_+, \forall x,y \in \R$ and $\min\{\eta_\ell z_\ell(\bm{y}), \gamma_\ell \mathcal{R}_\ell(\bm{y})\} \ge 0$,
\begin{equation}\label{eq:buy-gap-itemwise-start}
  \Delta(S, \bm{y})
  \le
  \sum_{\ell \in S} \eta_\ell z_\ell(\bm{y})
  +
  \sum_{\ell \notin S}
  \bigl(\gamma_\ell \mathcal{R}_\ell(\bm{y})-\eta_\ell z_\ell(\bm{y})\bigr)_+.
\end{equation}

Since $\padj(\bm{y})+\psubj(\bm{y})\le 1$ and
$z_\ell(\bm{y})\ge \max\{ f_\ell^{(a,j)}(\bm{y}), f_\ell^{(s,j)}(\bm{y})\}$, we get
\begin{align*}
  \padj(\bm{y}) f_\ell^{(a,j)}(\bm{y})
  +
  \psubj(\bm{y}) f_\ell^{(s,j)}(\bm{y})
  \le z_\ell(\bm{y}),
\end{align*}
Hence,
\begin{align*}
  \mathcal{R}_\ell(\bm{y})
  \le
  \left(\sum_{j=1}^J \lambda\jj  \right) z_\ell(\bm{y}).
\end{align*}
Substituting this into \eqref{eq:buy-gap-itemwise-start} yields
\begin{align*}
\Delta(S, \bm{y})
\le
\sum_{\ell \in [L]} w_\ell^{S} z_\ell(\bm{y}),
\end{align*}
where $w_\ell^{S}$ is defined in \eqref{eq:wS-mnl-def}. 
Next,
from~\eqref{eqn:sum-z-bound}, 
\begin{align*}
\sum_{\ell \in [L]} z_\ell(\bm{y})
\le 
2J.
\end{align*}
In addition,
for any nonnegative weights $w_\ell$, 
for any integer $M \ge 1$,
if $0 \le z_\ell \le 1$ and $\sum_\ell z_\ell \le M$, then
$\sum_\ell w_\ell z_\ell \le \TOP_{M}\set{\bm{w}}$,
so 
\begin{align}
\Delta(S, \bm{y})
\le
\TOP_{2J}\set{\bm{w}^S}. \label{eqn:Delta-S-bound}
\end{align}

Finally, let $\bm{y}\opt$ attain $R\opt=R(\bm{y}\opt)$. 
Since $R(S)\ge R(S, \bm{y}\opt)$,
\begin{align*}
R\opt-R(S)
\le
R(\bm{y}\opt)-R(S, \bm{y}\opt)
=\Delta(S, \bm{y}\opt)
\le
\TOP_{2J}\set{\bm{w}^S},
\end{align*}
where the last inequality follows from~\eqref{eqn:Delta-S-bound}.
This proves~\eqref{eq:buyset-gap-mnl}.
\end{proof}

\begin{proof}[Proof of Corollary~\ref{cor:buyset-threshold}]
For each content family $\ell$, from Lemma~\ref{lemma:buyset-gap-mnl}, the term $w_\ell^S$ in~\eqref{eq:wS-mnl-def} contributes either $\eta_\ell$ if
$\ell \in S$ or $b_\ell$ if $\ell \notin S$. Since the map
$\bm v \mapsto \TOP_{2J}\{\bm v\}$ is coordinatewise nondecreasing on $\mathbb{R}_+^L$,
the minimizing buy set is obtained by choosing, for each $\ell$, the smaller of the two family-level
penalties. Hence,
\begin{align*}
  \ell \in S^\star
  \quad\Longleftrightarrow\quad
\eta_\ell \le b_\ell.
\end{align*}
By definition,
\begin{align*}
b_\ell
=
\Bigl(
\gamma_\ell \sum_{j=1}^J \lambda\jj   - \eta_\ell
\Bigr)_+,
\end{align*}
so $\eta_\ell \le b_\ell$ is equivalent to
$2\eta_\ell \le \gamma_\ell \sum_{j=1}^{J} \lambda\jj  $.  
\end{proof}

\begin{proof}[Proof of Lemma~\ref{lem:ratio-interval}]
For a deterministic nonempty assortment $A$ with $u^{(j)}(A)>0$,
\[
\frac{x^{(j)}(A)}{u^{(j)}(A)}
=
\frac{\sum_{\ell\in A}\alpha_\ell^{(j)}}{\sum_{\ell\in A}\alpha_\ell^{(j)}u_\ell^{(j)}}
=
\sum_{\ell\in A}
\frac{\alpha_\ell^{(j)}u_\ell^{(j)}}{\sum_{k\in A}\alpha_k^{(j)}u_k^{(j)}}
\cdot
\frac{1}{u_\ell^{(j)}}.
\]
By     Assumption~\ref{assump:MNL}(\ref{assm:alpha-implies-u-support}), every $\ell$ with $\alpha_\ell^{(j)}>0$ also satisfies
$u_\ell^{(j)}>0$, so the coefficients above are well defined, nonnegative,
and sum to one. Thus, $\frac{x^{(j)}(A)}{u^{(j)}(A)}$ is a convex combination
of the singleton ratios $\{1/u_\ell^{(j)}: \alpha_\ell^{(j)}>0\}$, and therefore lies in
\[
\left[\frac{1}{\umaxj},\frac{1}{\uminj}\right]
=
I^{(j)}.
\]

Now let $y^{(a,j)}\in\Delta(A_\kappa)$ be an
ad assortment distribution
with $u^{(a,j)}>0$, and define
\[
A_+ := \{A\in A_\kappa : y_A^{(a,j)}>0,\ u^{(j)}(A)>0\}.
\]
Again by Assumption~\ref{assump:MNL}(\ref{assm:alpha-implies-u-support}), if $u^{(j)}(A)=0$, then $x^{(j)}(A)=0$. Hence
assortments with zero utility contribute neither to the numerator nor to the
denominator of $\frac{x^{(a,j)}}{u^{(a,j)}}$, and
\[
\frac{x^{(a,j)}}{u^{(a,j)}}
=
\sum_{A\in A_+}
\frac{y_A^{(a,j)}u^{(j)}(A)}
{\sum_{B\in A_+}y_B^{(a,j)}u^{(j)}(B)}
\cdot
\frac{x^{(j)}(A)}{u^{(j)}(A)}.
\]
Thus $\frac{x^{(a,j)}}{u^{(a,j)}}$ is a convex combination of points in $I^{(j)}$,
so it also lies in $I^{(j)}$.

Conversely, fix any $t\in I^{(j)}$. Choose $\ell_1,\ell_2\in [L]$ such that
$\alpha_{\ell_1}^{(j)}, \alpha_{\ell_2}^{(j)} > 0$ and
\[
\frac{1}{u_{\ell_1}^{(j)}} \le t \le \frac{1}{u_{\ell_2}^{(j)}},
\]
with $\ell_1=\ell_2$ allowed at an endpoint. If $\ell_1=\ell_2$, the singleton
assortment $\{\ell_1\}$ attains $t$. Otherwise define
\[
\rho_i:=\frac{x^{(j)}(\{\ell_i\})}{u^{(j)}(\{\ell_i\})}=\frac{1}{u_{\ell_i}^{(j)}},
\qquad
w :=\frac{\rho_2-t}{\rho_2-\rho_1}\in[0,1],
\]
and
\[
q:=\frac{w u^{(j)}(\{\ell_2\})}
{w u^{(j)}(\{\ell_2\})+(1- w) u^{(j)}(\{\ell_1\})}\in[0,1].
\]
Assign probability $q$ to the singleton assortment $\{\ell_1\}$ and probability
$1-q$ to the singleton assortment $\{\ell_2\}$. Then
\[
\frac{x^{(a,j)}}{u^{(a,j)}}
=
\frac{q x^{(j)}(\{\ell_1\})+(1-q) x^{(j)}(\{\ell_2\})}
{q u^{(j)}(\{\ell_1\})+(1-q) u^{(j)}(\{\ell_2\})}
=
w \rho_1 + (1-w )\rho_2
=
t.
\]
Therefore, every $t\in I^{(j)}$ is attainable by an assortment distribution
supported on at most two singleton assortments.
\end{proof}

\subsection{Proof of results in Section~\ref{sec:generic-LP}}

\begin{proof}[Proof of Lemma~\ref{lem:dual-localization}]
Define the two affine functions
\[
W^-(\nu)
:=
\phi(b^-_0) + \nu d(b^-_0),
\qquad
W^+(\nu)
:=
\phi(b^+_0) + \nu d(b^+_0).
\]
Recall we have
$\bar J
=
[L_0,U_0]$ and 
$g(\nu)
=
\max_{a \in \cA}
\left\{
\phi(a)+\nu d(a)
\right\}$
from~\eqref{eqn:g-theta-def}.
Since $d(b^-_0)<0$, for every $\nu < L_0$,
we have
\[
W^-(\nu)
=
W^-(L_0)
+
d(b^-_0)(\nu-L_0)
>
W^-(L_0)
\ge
g(0)
\ge
F\opt,
\] where $W^-(L_0)
\ge
g(0)$
follows from~\eqref{eqn:generic-init-domination}
and the last inequality follows from $F\opt = \min_{\nu \in \R} g(\nu).$
Since
$g(\nu) \ge W^-(\nu)$ for every $\nu$, it follows that
$g(\nu) > F\opt$ for all $\nu < L_0$.
Thus, no minimizer of $g$ can lie to the left of $L_0$.

Similarly, since $d(b^+_0)>0$, for every $\nu > U_0$ we have
\[
W^+(\nu)
=
W^+(U_0)
+
d(b^+_0)(\nu-U_0)
>
W^+(U_0)
\ge
g(0)
\ge
F\opt.
\]
Again using $g(\nu)\ge W^+(\nu)$, we obtain
$g(\nu) > F\opt$ for all $\nu > U_0$.
Thus, no minimizer can lie to the right of $U_0$.
Therefore, every minimizer of $g$ lies in $\bar J = [L_0,U_0]$, and thus
we complete the proof.
\end{proof}

\begin{proof}[Proof of Lemma~\ref{lemma:generic-fixed-theta}] 

For brevity, write
\[ 
C_n:=C^{\mathrm{bis}}_{n}.
\]

We claim that, for every iteration $n$ before termination,
\begin{align}
&\text{there exists a minimizer of } g   \text{ in } [L_n,U_n],
\label{eq:B8-inv-bracket}\\
&d  (b_n^-)<0<d  (b_n^+),
\label{eq:B8-inv-sign}\\
&\phi  (b_n^-)+L_n d  (b_n^-)\ge F\opt,\qquad
\phi  (b_n^+)+U_n d  (b_n^+)\ge F\opt,
\label{eq:B8-inv-cert}\\
&U_n-L_n=2^{-n}(U_{0}-L_{0}).
\label{eq:B8-inv-width}
\end{align}

The base case follows from Lemma~\ref{lem:dual-localization},
\eqref{eqn:generic-init-sign}, and \eqref{eqn:generic-init-domination},
together with $F\opt=\min_{\nu\in\R}g  (\nu)\le g  (0)$.

Now assume \eqref{eq:B8-inv-bracket}--\eqref{eq:B8-inv-width} hold at iteration $n$,
and suppose the routine has not yet terminated. Let
\[
m_n:=\frac{L_n+U_n}{2},\qquad a_n \in \arg\max_{a\in A}\{\phi  (a)+m_n d  (a)\}.
\] 
Therefore, since \(g  \) is the pointwise maximum of the affine functions $\nu \mapsto \phi  (a)+\nu d  (a),$ $g$ is convex and 
we have
\[
d  (a_n)\in \partial g  (m_n).
\]

If $d  (a_n)<0$, then for every $\nu<m_n$,
\[
g  (\nu)\ge g  (m_n)+d  (a_n)(\nu-m_n)>g  (m_n)\ge F\opt.
\]
Thus, no minimizer lies to the left of $m_n$, so by \eqref{eq:B8-inv-bracket}
some minimizer lies in $[m_n,U_n]$. With the update
\[
L_{n+1}=m_n,\qquad U_{n+1}=U_n,\qquad
b_{n+1}^-=a_n,\qquad b_{n+1}^+=b_n^+,
\]
the bracket property \eqref{eq:B8-inv-bracket}, sign condition \eqref{eq:B8-inv-sign},
and width identity \eqref{eq:B8-inv-width} are immediate, while
\[
\phi  (b_{n+1}^-)+L_{n+1}d  (b_{n+1}^-)
=
\phi  (a_n)+m_n d  (a_n)
=
g  (m_n)\ge F\opt,
\]
so \eqref{eq:B8-inv-cert} also holds.

If $d  (a_n)>0$, the symmetric argument gives, for every $\nu>m_n$,
\[
g  (\nu)\ge g  (m_n)+d  (a_n)(\nu-m_n)>g  (m_n)\ge F\opt,
\]
so no minimizer lies to the right of $m_n$. With the update
\[
L_{n+1}=L_n,\qquad U_{n+1}=m_n,\qquad
b_{n+1}^-=b_n^-,\qquad b_{n+1}^+=a_n,
\]
the same reasoning yields \eqref{eq:B8-inv-bracket}--\eqref{eq:B8-inv-width} at iteration $n+1$.

Thus, \eqref{eq:B8-inv-bracket}--\eqref{eq:B8-inv-width} hold for every iteration prior to termination.

If the routine terminates at some iteration $n\le N$ because $d  (a_n)=0$, then
\[
g  (m_n)=\phi  (a_n)+m_n d  (a_n)=\phi  (a_n),
\]
and for every $\nu\in\R$,  
\[
g  (\nu)\ge \phi  (a_n)+\nu d  (a_n)=\phi  (a_n)=g  (m_n).
\]
Thus, $m_n$ is a global minimizer, so $F\opt=\phi  (a_n)$. Since
$a_n\in C_N$ and $d  (a_n)=0$, the pure action $a_n$ is feasible for the restricted LP,
which gives $F^N\ge F\opt$. The reverse inequality is trivial because $C_N\subseteq A$.
So
\[
F^N=F\opt.
\]

Assume now that no termination occurs in the first $N$ iterations. Define
\begin{align}
    g_{N}^{\mathrm{bis}}(\nu):=g_{C_N}(\nu) = 
\max_{a \in C_N}
\bigl\{
\phi  (a)+\nu d  (a)
\bigr\}, \label{eqn:g-bis-def}
\end{align}
and  
\[
W_N^-(\nu):=\phi  (b_N^-)+\nu d  (b_N^-),\qquad
W_N^+(\nu):=\phi  (b_N^+)+\nu d  (b_N^+).
\]
Since $b_N^-,b_N^+\in C_N$, 
\begin{align}
g_{N}^{\mathrm{bis}}(\nu) \ge \max\{W_N^-(\nu),W_N^+(\nu)\}
\qquad \forall \nu\in\R,
\label{eqn:g-larger-than-two}
\end{align}
where the inequality follows from 
the definition of 
$g_{N}^{\mathrm{bis}}(\nu)$~\eqref{eqn:g-bis-def}.
Moreover, by \eqref{eq:B8-inv-cert},
\[
W_N^-(L_N)\ge F\opt,\qquad W_N^+(U_N)\ge F\opt.
\]

If $\nu\le L_N$, then \eqref{eq:B8-inv-sign} gives $d  (b_N^-)<0$, thus,
\begin{align}
g_{N}^{\mathrm{bis}}(\nu) \ge  
W_N^-(\nu)\ge W_N^-(L_N)\ge F\opt, \nu\le L_N.
\label{eqn:W_N_large_small_nu}
\end{align}
Similarly, if $\nu\ge U_N$, then
\begin{align}
g_{N}^{\mathrm{bis}}(\nu) \ge  
W_N^+(\nu)\ge W_N^+(U_N)\ge F\opt, \nu\ge U_N.
\label{eqn:W_N_large_large_nu}
\end{align}
Finally,   using $|d  (a)|\le D  $ for all $a\in A$,  
\[
W_N^-(\nu)\ge F\opt-D  (\nu-L_N),\qquad
W_N^+(\nu)\ge F\opt-D  (U_N-\nu), \forall \nu\in[L_N,U_N].
\]
Combining the   previous display with~\eqref{eqn:g-larger-than-two},
\[
g_{N}^{\mathrm{bis}}(\nu)\ge
F\opt-D  \min\{\nu-L_N,U_N-\nu\}
\ge
F\opt-\frac{D  }{2}(U_N-L_N), \forall \nu\in[L_N,U_N].
\]
Combining the previous display with~\eqref{eqn:W_N_large_small_nu} and~\eqref{eqn:W_N_large_large_nu},
\[
g_{N}^{\mathrm{bis}}(\nu)\ge
F\opt-\frac{D  }{2}(U_N-L_N)
\qquad \forall \nu\in\R.
\]

By \eqref{eq:B8-inv-sign}, the restricted LP~\eqref{eqn:LP-resctricted}
is feasible, so strong duality gives
\[
F^N=\min_{\nu\in\R} g_{N}^{\mathrm{bis}}(\nu).
\]
Thus, 
\[
F^N\ge F\opt-\frac{D  }{2}(U_N-L_N).
\]
Using \eqref{eq:B8-inv-width},
\[
U_N-L_N=2^{-N}(U_{0}-L_{0}),
\]
and therefore 
\[
F^N\ge
F\opt-\frac{D  }{2}(U_{0}-L_{0})2^{-N}.
\]
\end{proof}

\begin{proof}[Proof of Lemma~\ref{lemma:grid-generic}]
Since $F\opt$ is Lipschitz on the compact interval $I$, it is continuous, so there exists
\[
\theta^\star \in \argmax_{\theta \in I} F\opt(\theta).
\]
Choose a nearest sampled point $\theta_G^\star \in \mathcal G$ such that
\[
|\theta^\star-\theta_G^\star|
\le
\Delta_{\mathcal G}.
\]
By the $W_F$-Lipschitz continuity of $F\opt$,
\[
F\opt(\theta_G^\star)
\ge
F\opt(\theta^\star) - W_F |\theta_G^\star-\theta^\star|
\ge
\max_{\theta \in I} F\opt(\theta) - W_F \Delta_{\mathcal G}.
\]

Now apply Lemma~\ref{lemma:generic-fixed-theta} at  $\theta_G^\star$.
We obtain
\[
F^{N(\theta_G\opt)}(\theta_G^\star)
\ge
F\opt(\theta_G^\star)
-
\frac{D(\theta_G\opt)}{2}
\bigl(U_{\theta_G^\star,0}-L_{\theta_G^\star,0}\bigr)
2^{-N(\theta_G\opt)}.
\]
By the definition of $C^{\mathrm{dual}}$,
\[
\frac{D(\theta_G\opt)}{2}
\bigl(U_{\theta_G^\star,0}-L_{\theta_G^\star,0}\bigr)
2^{-N(\theta_G\opt)}
\le
C^{\mathrm{dual}}.
\]
Therefore,
\[
F^{N(\theta_G\opt)}(\theta_G^\star)
\ge
\max_{\theta \in I} F\opt(\theta)
-
W_F \Delta_{\mathcal G}
-
C^{\mathrm{dual}}.
\]
Finally, since $\theta_G^\star \in \mathcal G$,
\[
\max_{\theta  \in \mathcal G} F^{N(\theta)}(\theta)
\ge
F^{N(\theta_G\opt)}(\theta_G^\star).
\]
This proves the claim.
\end{proof}

\subsection{Proof of results in Section~\ref{sec:LP-ads}}

\begin{proof}[Proof of Lemma~\ref{lemma-explict-dual-interval-ads}]
Recall that the fixed-$t$ ads LP is of the generic one-equality form, with action set
$\mc{A}_\kappa$, reward
\begin{align*}
\phi(A):=G\adj(t,A),
\end{align*}
and residual
\begin{align*}
d(A):=d\adj(t, A).
\end{align*}
In addition,
\begin{align*}
G\adj(t, \mu)
=
\max_{A\in \mc{A}_\kappa}
\left\{
G\adj(t,A)+\mu  d\adj(t, A)
\right\},
\end{align*}
and
\begin{align*}
G\adj(t)
=
\min_{\mu\in\mathbb R} G\adj(t, \mu).
\end{align*}

We first compute the residuals of the singleton  assortments
$A_t^-=\{\lplus\}$ and $A_t^+=\{\lminus\}$.
For any singleton $\{\ell\}$, the MNL formulas give
\begin{align*}
x^{(j)}(\{\ell\})=\frac{\alpha_\ell^{(j)}}{1+\alpha_\ell^{(j)}},
\qquad
u^{(j)}(\{\ell\})=\frac{\alpha_\ell^{(j)}u_\ell^{(j)}}{1+\alpha_\ell^{(j)}}.
\end{align*}
Therefore,
\begin{align*}
d\adj(t, \{\ell\})
=
x^{(j)}(\{\ell\})-t u^{(j)}(\{\ell\})
=
\frac{\alpha_\ell^{(j)}}{1+\alpha_\ell^{(j)}}\bigl(1-t u_\ell^{(j)}\bigr).
\end{align*}
Applying this to $\lplus$ and $\lminus$ yields
\begin{align}
    d\adj(t, A_t^-)
=
\frac{\alpha_{\lplus}^{(j)}}{1+\alpha_{\lplus}^{(j)}}
\bigl(1-t u_{\lplus}^{(j)}\bigr),
\qquad
d\adj(t, A_t^+)
=
\frac{\alpha_{\lminus}^{(j)}}{1+\alpha_{\lminus}^{(j)}}
\bigl(1-t u_{\lminus}^{(j)}\bigr).
\label{eqn:d-ads-explicit}
\end{align}
Since
\begin{align*}
\frac{1}{u_{\lplus}^{(j)}}<t<\frac{1}{u_{\lminus}^{(j)}},
\end{align*}
we have
\begin{align*}
d\adj(t, A_t^-)
=
-\frac{\alpha_{\lplus}^{(j)}}{1+\alpha_{\lplus}^{(j)}}
\bigl(t u_{\lplus}^{(j)}-1\bigr)<0,  \quad
d\adj(t, A_t^+)
=
\frac{\alpha_{\lminus}^{(j)}}{1+\alpha_{\lminus}^{(j)}}
\bigl(1-t u_{\lminus}^{(j)}\bigr)>0.
\end{align*}

Now define the exact fixed-$t$ ads reward spread
\begin{align*}
\Phi_t\adj
:=
\max_{A\in \mc{A}_\kappa} G\adj(t,A)
-
\min_{A\in \mc{A}_\kappa} G\adj(t,A).
\end{align*}
Set
\[
L:=\frac{\Phi_t^{(a,j)}}{d\adj(t, A_t^-)},
\qquad
U:=\frac{\Phi_t^{(a,j)}}{d\adj(t, A_t^+)}.
\]
Then
\[
L d\adj(t, A_t^-)=\Phi_t^{(a,j)},
\qquad
U d\adj(t, A_t^+)=\Phi_t^{(a,j)}.
\]
Thus,
\[
\phi(A_t^-)+L d(A_t^-)
\ge
\min_{A\in A_\kappa}\phi(A)+\Phi_t^{(a,j)}
=
\max_{A\in A_\kappa}\phi(A),
\]
and similarly,
\[
\phi(A_t^+)+U d(A_t^+)\ge \max_{A\in A_\kappa}\phi(A).
\]
Therefore, the hypotheses of
Lemma~\ref{lem:dual-localization}
hold, and
\begin{align}
\min_{\mu\in\mathbb R} G\adj(t, \mu)
=
\min_{\mu\in
\left[
-\frac{\Phi_t^{(a,j)}}{-d\adj(t, A_t^-)},
\frac{\Phi_t^{(a,j)}}{d\adj(t, A_t^+)}
\right]
}
G\adj(t, \mu).
\label{eqn:d-ads-step-1}
\end{align}
It remains only to replace the exact numerator $\Phi_t^{(a,j)}$ by the explicit upper bound $\bar\Phi\adj$. 
For any assortment $A\in \mc{A}_\kappa$, write
\begin{align*}
R(A):=\sum_{\ell\in A\setminus S}\gamma_\ell p_\ell^{(j)}(A),
\qquad
B(A):=\sum_{\ell\in A\cap S}\eta_\ell \beta_\ell^{(a,j)} p_\ell^{(j)}(A) .
\end{align*}
Then
\begin{align*}
G\adj(t,A)
=
\lambda\jj \bar F\jj(\sigma t)\bigl[\rads\sigma x^{(j)}(A)-R(A)\bigr]-B(A).
\end{align*}
Using
\begin{align}
    0\le x^{(j)}(A)\le 1,
\qquad
0\le R(A)\le \gamma_{\max},
\qquad
0\le B(A)\le  \eta_{\max},
\label{eqn:R-B-bound}
\end{align}
together with $0\le \bar F\jj(\sigma t)\le 1$, we obtain
\begin{align*}
-\lambda\jj\gamma_{\max}-\eta_{\max}
\le
G\adj(t,A)
\le
\lambda\jj \rads\sigma.
\end{align*}
Therefore
\begin{align*}
\Phi_t\adj
\le
\lambda\jj \rads\sigma
-
\bigl(-\lambda\jj\gamma_{\max}-\eta_{\max}\bigr)
=
\lambda\jj(\rads\sigma+\gamma_{\max})+\eta_{\max}
=
\bar\Phi\adj.
\end{align*}
Combining the above with~\eqref{eqn:d-ads-explicit}
and~\eqref{eqn:d-ads-step-1}, we finish the proof.
\end{proof}

\begin{proof}[Proof of Lemma~\ref{lem:ads-bisection-guarantee}]
The fixed-$t$ ads LP~\eqref{eqn:ads-LP-def} is of the generic one-equality form, with
action set $\mathcal A_\kappa$, reward map
\[
\phi(A) := G^{(a,j)}(t,A),
\]
residual map
\[
d(A) := d\adj(t, A),
\]
and dual variable $\mu$.
By the preceding lemma,
\[
G^{(a,j)}(t)
=
\min_{\mu \in \R} G\adj(t, \mu)
=
\min_{\mu \in \bar{\mc{J}}\adj(t)} G\adj(t, \mu).
\]

At every midpoint $\bar\mu$, Algorithm~\ref{alg:mnl-fixed-jt-ads} calls Algorithm~\ref{alg:mnl-oracle} on the
reduced revenue vector in~\eqref{eqn:R-C-def-ads}. Hence, the midpoint oracle returns an assortment in
\[
\argmax_{A \in \mathcal A_\kappa}
\bigl\{
G^{(a,j)}(t,A) + \bar\mu  d\adj(t, A)
\bigr\}.
\]

We now verify the initialization conditions of
Lemma~\ref{lemma:generic-fixed-theta}.
From the proof of the preceding lemma,
\[
d\adj(t, A_t^-) < 0 < d\adj(t, A_t^+).
\]
Define the exact fixed-$t$ ads reward spread by
\[
\Phi_t^{(a,j)}
:=
\max_{A \in \mathcal A_\kappa} G^{(a,j)}(t,A)
-
\min_{A \in \mathcal A_\kappa} G^{(a,j)}(t,A).
\]
The proof of   Lemma~\ref{lemma-explict-dual-interval-ads}  shows that
\[
\Phi_t^{(a,j)} \le \bar\Phi\adj.
\]
Recall $\bar{\mc{J}}\adj(t)$ in~\eqref{eqn:J-j-a-def} 
and let
\begin{align*}
\bar{\mc{J}}\adj(t) = [\underline\mu_t^{(a,j)}, \bar\mu_t^{(a,j)}].
\end{align*}
Since
\[
\underline\mu_t^{(a,j)}  d\adj(t, A_t^-)
=
\bar\Phi\adj,
\qquad
\bar\mu_t^{(a,j)}  d\adj(t, A_t^+)
=
\bar\Phi\adj,
\]
we obtain
\[
\phi(A_t^-)+\underline\mu_t^{(a,j)} d(A_t^-)
=
\phi(A_t^-)+\bar\Phi\adj
\ge
\min_{A \in \mathcal A_\kappa}\phi(A) + \Phi_t^{(a,j)}
=
\max_{A \in \mathcal A_\kappa}\phi(A)
=G\adj(t, \mu = 0)
\]
and similarly
\[
\phi(A_t^+) + \bar\mu_t^{(a,j)} d(A_t^+) \ge
G\adj(t, \mu=0).
\]
Thus the initialization conditions of Lemma~\ref{lemma:generic-fixed-theta} hold
with
\[
[L_0,U_0] = \bar{\mc{J}}\adj(t),
\qquad
b_0^- = A_t^-,
\qquad
b_0^+ = A_t^+.
\]

It remains to bound the residual constant.
For every $A \in \mathcal A_\kappa$,
\[
|d(A)|
=
|d\adj(t, A)|
=
|x^{(j)}(A) - t u^{(j)}(A)|
\le
x^{(j)}(A) + t u^{(j)}(A)
\le
1 + t \uumaxj
\le
1 + t_{\max}^{(j)} \uumaxj
=
D^{(a,j)}.
\]

Applying Lemma~\ref{lemma:generic-fixed-theta} yields
\[
\ghat(t)
\ge
\geadj(t)
-
\frac{D^{(a,j)}}{2} 
|\bar{\mc{J}}\adj(t)| 2^{-K}
\]
and this completes the proof.
\end{proof}

\subsection{Proof of results in Section~\ref{sec:LP-subscription}}

\begin{proof}[Proof of Lemma~\ref{lemma-explict-dual-interval-subscription}]

Let $ B_p^- := \emptyset$
and choose 
$ B_p^+ \in \argmax_{B\in \mc{A}_\kappa} u\jj(B)$, so $u\jj(B_p^+)  = \uumaxj$.
Assume now that $p\in(0,\uumaxj)$. Since $u^{(j)}(\emptyset)=0$, we have
\begin{align*}
d\subj(p, B_p^-)
=
d\subj(p, \emptyset)
=
-p
<
0.
\end{align*}
Also, since $u^{(j)}(B_p^+)=\uumaxj$, we have
\begin{align*}
d\subj(p, B_p^+)
=
\uumaxj-p
>
0.
\end{align*}

Now define the exact fixed-$t$ subscription reward spread
\begin{align*}
\Phi_t\subj
:=
\max_{B\in \mc{A}_\kappa}G\subj(t,B)
-
\min_{B\in \mc{A}_\kappa}G\subj(t,B).
\end{align*}
By Lemma~\ref{lem:dual-localization}, there exists a minimizer of problem $\min_{\xi \in \R} G\subj(t, \xi)$ that lies in
\begin{align*}
\left[
-\frac{\Phi_t\subj}{-d\subj(p, B_p^-)},
\frac{\Phi_t\subj}{d\subj(p, B_p^+)}
\right]
=
\left[
-\frac{\Phi_t\subj}{p},
\frac{\Phi_t\subj}{\uumaxj-p}
\right].
\end{align*}

It remains to bound $\Phi_t\subj$ explicitly. For any assortment $B\in \mc{A}_\kappa$, write
\begin{align*}
R(B):=\sum_{\ell\in B\setminus S}\gamma_\ell p_\ell^{(j)}(B),
\qquad
C(B):=\sum_{\ell\in B\cap S}\eta_\ell\beta_\ell^{(s,j)}  p_\ell^{(j)}(B).
\end{align*}
Then
\begin{align*}
G\subj(t,B)
=
\lambda\jj F\jj(\sigma t) (p
- R(B))
-
C(B).
\end{align*}
The term $\lambda\jj F\jj(\sigma t) p$ is constant in $B$, so only the last two terms contribute
to the spread. Since
\begin{align}
    0\le R(B)\le \gamma_{\max},
\qquad
0\le C(B)\le \eta_{\max},
\qquad
0\le F\jj(\sigma t)\le 1, \label{eqn:R-B-bound}
\end{align}
we obtain
\begin{align*}
-\lambda\jj F\jj(\sigma t)\gamma_{\max}-\eta_{\max}
\le
-\lambda\jj F\jj(\sigma t)R(B)-C(B)
\le 0.
\end{align*}
Therefore,
\begin{align}
\Phi_t\subj
\le
\lambda\jj F\jj(\sigma t)\gamma_{\max}+\eta_{\max}
\le
\lambda\jj\gamma_{\max}+\eta_{\max}
=
\bar\Phi\subj. \label{eqn:phi-j-bound}
\end{align}
Thus,
\begin{align*}
\min_{\xi\in\mathbb R}G\subj(t, \xi)
=
\min_{\xi\in \bar{\mc{J}}\subj}G\subj(t, \xi).
\end{align*}
\end{proof}

\begin{proof}[Proof of Lemma~\ref{lem:subscription-bisection-guarantee}]

For $t \in I^{(j)}$, define
\[
c_t := F\jj(\sigma t),
\]
and for the off branch set
\[
c_\infty := 1.
\]
Write
\[
\phi_t(B)
:=
\lambda\jj c_t
\left[
p -
\sum_{\ell \in [L]\setminus S}\gamma_\ell p_\ell^{(j)}(B)
\right]
-
\sum_{\ell \in S}\eta_\ell \beta_\ell^{(s,j)} p_\ell^{(j)}(B),
\]
and
\[
d(B) := d\subj(p, B)=u^{(j)}(B)-p.
\]
Then, for both $t \in I^{(j)}$ and $t=\infty$, the relevant subscription branch is
a generic one-equality LP with dual variable $\xi$ and full dual
\[
g_t(\xi)
=
\max_{B \in \mathcal A_\kappa}
\bigl\{
\phi_t(B)+\xi d(B)
\bigr\}.
\]

For finite $t$, the preceding lemma gives
\[
\min_{\xi \in \R} G\subj(t,\xi)
=
\min_{\xi \in \bar{\mc{J}}^{(s,j)}} G\subj(t,\xi).
\]
For $t=\infty$, the same dual representation holds with the convention $c_\infty=1$,
which is exactly the current off-branch convention.

At every midpoint $\bar\xi$, Algorithm~\ref{alg:mnl-fixed-jt-sub} calls Algorithm~\ref{alg:mnl-oracle}
with the reduced revenue vector in~\eqref{eqn:R-C-def-sub}. 

We next verify the initialization conditions of
Lemma~\ref{lemma:generic-fixed-theta}.
From the proof of the preceding lemma,
\[
d(B_p^-)= -p <0,
\qquad
d(B_p^+) = \uumaxj-p >0.
\]
Define the exact reward spread by
\[
\Phi_t^{(s,j)}
:=
\max_{B \in \mathcal A_\kappa}\phi_t(B)
-
\min_{B \in \mathcal A_\kappa}\phi_t(B).
\]
For finite $t$,~\eqref{eqn:phi-j-bound} in
the proof of Lemma~\ref{lemma-explict-dual-interval-subscription}
shows that
\[
\Phi_t^{(s,j)} \le \bar\Phi\subj.
\]
For $t=\infty$, we have
\[
\phi_\infty(B)
=
\lambda\jj
\left[
p -
\sum_{\ell \in [L]\setminus S}\gamma_\ell p_\ell^{(j)}(B)
\right]
-
\sum_{\ell \in S}\eta_\ell \beta_\ell^{(s,j)} p_\ell^{(j)}(B),
\]
so only the last two terms contribute to the spread, and therefore
\[
\Phi_\infty^{(s,j)}
\le
\lambda\jj \gamma_{\max} + \eta_{\max}
=
\bar\Phi\subj.
\]
Thus, for all $t \in I^{(j)} \cup \{\infty\}$,
\[
\Phi_t^{(s,j)} \le \bar\Phi\subj.
\]

Recall $\bar{\mc{J}}\subj$ in~\eqref{eqn:J-j-s-def}
and let
\begin{align*}
\bar{\mc{J}}\subj
= [\underline\xi\subj, \bar\xi\subj].
\end{align*}

Since
\[
\underline\xi^{(s,j)}  d(B_p^-)=\bar\Phi\subj,
\qquad
\bar\xi^{(s,j)}  d(B_p^+)=\bar\Phi\subj,
\]
we obtain
\[
\phi_t(B_p^-)+\underline\xi^{(s,j)} d(B_p^-)
=
\phi_t(B_p^-)+\bar\Phi\subj
\ge
\min_{B \in \mathcal A_\kappa}\phi_t(B)+\Phi_t^{(s,j)}
=
\max_{B \in \mathcal A_\kappa}\phi_t(B)
=
g_t(0),
\]
and similarly
\[
\phi_t(B_p^+) + \bar\xi^{(s,j)} d(B_p^+) \ge g_t(0).
\]
Hence, the initialization conditions of Lemma~\ref{lemma:generic-fixed-theta} hold
with
\[
[L_0,U_0]=\bar{\mc{J}}^{(s,j)},
\qquad
b_0^- = B_p^-,
\qquad
b_0^+ = B_p^+.
\]

Finally, for every $B \in \mathcal A_\kappa$,
\[
|d(B)|
=
|u^{(j)}(B)-p|
\le
\uumaxj
=
D^{(s,j)}.
\]

To avoid confusion, we use $\gtilde(t)$ to denote the value of $\ghat(t)$
after the postprocessing step in line~\ref{line:V-hat-before-max} and before the clipping step in line~\ref{line:V-hat-after-max}
in Algorithm~\ref{alg:mnl-fixed-jt-sub}.
Applying Lemma~\ref{lemma:generic-fixed-theta} yields
\[
\gtilde(t)
\ge
\min_{\xi \in \R} G\subj(t,\xi)
-
\frac{D^{(s,j)}}{2} 
|\bar{\mc{J}}^{(s,j)}| 2^{-K}.
\]
If a midpoint query returns an assortment with zero residual before iteration $K$, then
$\gtilde(t) = \min_{\xi \in \R} G\subj(t,\xi)$.

The clipped bound follows once more from the $1$-Lipschitz property of
$x \mapsto x_+$:
\[
\ghat(t) =
\max\{0,\gtilde(t)\}
\ge
\max\{0, \min_{\xi \in \R} G\subj(t,\xi)\}
-
\frac{D^{(s,j)}}{2} 
|\bar{\mc{J}}^{(s,j)}| 2^{-K}.
\]
Since $\gesubj(t) = \max\{0,\min_{\xi \in \R} G\subj(t,\xi)\}$, the claim follows.
\end{proof}

\subsection{Proof of results in Section~\ref{sec:additional-lemmas}}
\label{sec:proof-additional-lemmas}

\begin{proof}[Proof of Lemma~\ref{lem:G-gap-combined-ads-sub}]
Fix $j\in[J]$ and $t\in T\jj_{K_T}\cup\offset$.

We first consider the ads branch.

If $t\in T\jj_{K_T}\cap (t_{\min}\jj,t_{\max}\jj)$, then
Lemma~\ref{lem:ads-bisection-guarantee} gives
\[
\ghat\adj_{K_T}(t)
\ge
\geadj(t)
-
\frac{D\adj}{2} |\bar{\mc{J}}\adj(t)| 2^{-K\adj(t)}.
\]

If $t\in\{t_{\min}\jj,t_{\max}\jj\}$, then
since $x^{(j)}(A)/u^{(j)}(A)$ is a convex combination of the singleton ratios $\{1/u_\ell^{(j)} : \alpha_\ell^{(j)} > 0\}$,   every   $A$ 
with $d\adj(t, A)=0$ 
  only uses content families $\ell$ 
  (note that 
  content families with zero attraction affect neither the residual nor the objective and may be deleted)
  with $u_\ell^{(j)} = 1/t$, i.e.,
 $A \subseteq \mc{M}\jj(t)$ where $\mc{M}\jj(t) = \set{\ell \in [L] : u_\ell^{(j)} = 1/t,  \alpha_\ell^{(j)} > 0}$. 
 So~\eqref{eqn:ads-LP-def} reduces to maximizing $G^{(a,j)}(t,A)$ over $A \subseteq \mc{M}\jj(t)$ with $|A| \le \kappa$, 
 which is exactly 
 the problem solved by Algorithm~\ref{alg:mnl-ads-boundary}. Thus, $\ghat\adj_{K_T}(t)=\geadj(t)$, and the ads contribution to $\errork = 0$
 in~\eqref{eqn:search-error-in-grid}
 by definition.

If $t=\infty$, then by convention the ad mode is inactive, so
\[
\ghat\adj_{K_T}(\infty)=0=\geadj(\infty),
\]
and again the ads contribution to $\errork$ is $0$.

Combining these three cases, we obtain for every $t\in T\jj_{K_T}\cup\offset$,
\begin{align}
\label{eq:B14-ads-branch}
\ghat\adj_{K_T}(t)
\ge
\geadj(t)
-
\begin{cases}
\dfrac{D\adj}{2} |\bar{\mc{J}}\adj(t)| 2^{-K\adj(t)},
& t=t_r^{(j)} \in T\jj_{K_T}\cap (t_{\min}\jj,t_{\max}\jj), \\[1.0ex]
0,
& t\in\{t_{\min}\jj,t_{\max}\jj,\infty\}.
\end{cases}
\end{align}

We next consider the subscription branch.

If $p\in(0,U_{\max}\jj)$, then
Lemma~\ref{lem:subscription-bisection-guarantee} gives, for every
$t\in T\jj_{K_T}\cup\offset$,
\[
\ghat\subj_{K_T}(t)
\ge
\gesubj(t)
-
\frac{D\subj}{2} |\bar{\mc{J}}\subj| 2^{-K\subj}.
\]

If $p\notin(0,U_{\max}\jj)$, then
Algorithm~\ref{alg:mnl-fixed-jt-sub} is exact:
it either returns the zero solution when $p>U_{\max}\jj$ or $p=0$, or solves the
boundary branch exactly when $p=U_{\max}\jj$.
Therefore, for every $t\in T\jj_{K_T}\cup\offset$,
\[
\ghat\subj_{K_T}(t)=\gesubj(t),
\]
and the subscription contribution to $\errork$ is $0$.

Hence, for every $t\in T\jj_{K_T}\cup\offset$,
\begin{align}
\label{eq:B14-sub-branch}
\ghat\subj_{K_T}(t)
\ge
\gesubj(t)
-
\begin{cases}
\dfrac{D\subj}{2} |\bar{\mc{J}}\subj| 2^{-K\subj},
& p\in(0,U_{\max}\jj), \\[1.0ex]
0,
& p\notin(0,U_{\max}\jj).
\end{cases}
\end{align}

Adding~\eqref{eq:B14-ads-branch} and~\eqref{eq:B14-sub-branch}, and using
\[
G\jj(t)=\geadj(t)+\gesubj(t),
\]
yields
\[
\ghat\jj_{K_T}(t)
\ge
G\jj(t)-\errork,
\]
which proves~\eqref{eq:combined-error-at-t}.

 Next, we show~\eqref{eq:grid-hat-error}. Define
\[
a_t := \ghat\jj_{K_T}(t),
\qquad
b_t := G\jj(t),
\qquad
c_t := \errork.
\]
Since $a_t \ge b_t-c_t$ for every $t\in T\jj_{K_T}\cup\offset$, we have
\[
\max_{t\in T\jj_{K_T}\cup\offset} a_t
\ge
\max_{t\in T\jj_{K_T}\cup\offset} (b_t-c_t)
\ge
\max_{t\in T\jj_{K_T}\cup\offset} b_t
-
\max_{t\in T\jj_{K_T}\cup\offset} c_t .
\]
Applying this inequality for all $j$ and summing over $j=1,\dots,J$ gives
\eqref{eq:grid-hat-error}.
\end{proof}

\begin{proof}[Proof of Lemma~\ref{lem:mnl-global-ratio-lipschitz}]
If \((\delta\jj)\opt=\infty\), then all deterministic assortments
with positive utility have the same ratio. Since all positive-utility singleton
assortments belong to \(\mathcal A_\kappa\), this implies
\(t^{(j)}_{\min}=t^{(j)}_{\max}\), so \(I^{(j)}\) is a singleton and the claim
is immediate. Thus, we assume \((\delta\jj)\opt<\infty\).

\paragraph{Ads branch.}
For \(A\in\mathcal A_\kappa\), define
\begin{align*}
\Psi_A(t)
&:=G^{(a,j)}(t,A),\\
d(t, A)
&:=x^{(j)}(A)-t u^{(j)}(A).
\end{align*}
Only the factor \(\bar F\jj(\sigma t)\) depends on \(t\). Since
\(\bar F\jj(\sigma t)\) is \(\sigma W_{F\jj}\)-Lipschitz as a function of \(t\), and
\begin{align*}
0\le x^{(j)}(A)\le 1,\qquad
0\le
\sum_{\ell\in A\setminus S}\gamma_\ell p_\ell^{(j)}(A)
\le \gamma_{\max},
\end{align*}
we have, for all \(t,t'\in I^{(j)}\),
\begin{align}
|\Psi_A(t)-\Psi_A(t')|
\le
\lambda\jj\sigma W_{F\jj}
\bigl(\rads\sigma+\gamma_{\max}\bigr)|t-t'|.
\label{eqn:Psi-A-Lip}
\end{align}
Also,
\begin{align*}
|\Psi_A(t)|\le \bar\Phi\adj.
\end{align*}

By Assumption~\ref{assump:MNL}(\ref{assm:alpha-implies-u-support}), if
\(u^{(j)}(A)=0\), then \(x^{(j)}(A)=0\), and hence
\(d(t, A)=0\) and \(\Psi_A(t)=0\) for all \(t\). Thus zero-utility assortments
contribute only the zero candidate.

The fixed-\(t\) ads problem is a linear program over the simplex with one
linear equality constraint. Thus, an optimal extreme point has support size at
most two. Equivalently,
\begin{align*}
G^{(a,j)}(t)
=
\max_{\substack{
A,B\in\mathcal A_\kappa,\ \theta\in[0,1]\\
\theta d(t, A)+(1-\theta)d(t, B)=0
}}
\left\{
\theta\Psi_A(t)+(1-\theta)\Psi_B(t)
\right\}.
\end{align*}
The case \(A=B\) represents a one-support solution.

For every positive-utility assortment \(A\), write
\begin{align*}
\rho_A^{(j)}
:=
\frac{x^{(j)}(A)}{u^{(j)}(A)}.
\end{align*}
Consider any pair \(A,B\) with positive utility and
\(\rho_A^{(j)}<\rho_B^{(j)}\). For
\(t\in[\rho_A^{(j)},\rho_B^{(j)}]\), the unique feasible mixture supported on
\(\{A,B\}\) has weights
\begin{align*}
\omega_A(t)
&=
\frac{
u^{(j)}(B)\bigl(\rho_B^{(j)}-t\bigr)
}{
D_{A,B}(t)
},\\
\omega_B(t)
&=
\frac{
u^{(j)}(A)\bigl(t-\rho_A^{(j)}\bigr)
}{
D_{A,B}(t)
},
\end{align*}
where
\begin{align*}
D_{A,B}(t)
:=
u^{(j)}(B)\bigl(\rho_B^{(j)}-t\bigr)
+
u^{(j)}(A)\bigl(t-\rho_A^{(j)}\bigr).
\end{align*}
Define the corresponding pair value
\begin{align}
V_{A,B}(t)
:=
\omega_A(t)\Psi_A(t)+\omega_B(t)\Psi_B(t). \label{eqn:V-A-B}
\end{align}
One-support positive-utility candidates are included by these formulas through
a zero weight on one of the two supports. If there is only one positive-utility
ratio, we are in the singleton-interval case already handled above.

Let
\begin{align*}
\Delta_{A,B}:=\rho_B^{(j)}-\rho_A^{(j)}.
\end{align*}
For every \(t\in[\rho_A^{(j)},\rho_B^{(j)}]\),
\begin{align}
D_{A,B}(t)
& \ge \min\{u^{(j)}(A),u^{(j)}(B)\} \cdot \bigl(\rho_B^{(j)}-t + t-\rho_A^{(j)}\bigr) \nonumber  \\
&=
\min\{u^{(j)}(A),u^{(j)}(B)\}\Delta_{A,B} \label{eqn:D-A-B-eq}\\
&\ge
(\delta\jj)\opt. \nonumber
\end{align}
Differentiating the weights gives
\begin{align*}
\omega_A'(t)
&=
-\frac{
u^{(j)}(A)u^{(j)}(B)\Delta_{A,B}
}{
D_{A,B}(t)^2
},\\
\omega_B'(t)
&=
\frac{
u^{(j)}(A)u^{(j)}(B)\Delta_{A,B}
}{
D_{A,B}(t)^2
}.
\end{align*}
Using \(u^{(j)}(A),u^{(j)}(B)\le \umaxj \), and the 
inequality that 
$a,b \in [0,M] \Rightarrow ab \le M \cdot \min\set{a,b}$ and~\eqref{eqn:D-A-B-eq}, we have
\begin{align*}
u^{(j)}(A)u^{(j)}(B)\Delta_{A,B}
\le \umaxj  \cdot \min\{u^{(j)}(A),u^{(j)}(B)\}\Delta_{A,B} 
\le 
\umaxj D_{A,B}(t),
\end{align*}
and therefore
\begin{align*}
|\omega_A'(t)|,\ |\omega_B'(t)|
\le
\frac{\umaxj }{(\delta\jj)\opt}.
\end{align*}

Therefore, for any $t,t' \in [\rho_A^{(j)},\rho_B^{(j)}]$,
\begin{align}
\left|\omega_A(t)-\omega_A(t')\right|
\le
\sup_{s\in[\rho_A^{(j)},\rho_B^{(j)}]}
\left|\omega_A'(s)\right| |t-t'|
\le
\frac{\umaxj}{(\delta\jj)\opt} |t-t'|. \label{eqn:omgea-A-Lip}
\end{align}

Thus, for any \(t,t'\in[\rho_A^{(j)},\rho_B^{(j)}]\),
\begin{align*}
|V_{A,B}(t)-V_{A,B}(t')|
&=
\Big|\omega_A(t) (\Psi_A(t)-\Psi_A(t'))
+
\omega_B(t) (\Psi_B(t)-\Psi_B(t'))\\
&\quad
+
(\omega_A(t)-\omega_A(t'))
(\Psi_A(t')-\Psi_B(t')) \Big|\\
&\le
\omega_A(t)|\Psi_A(t)-\Psi_A(t')|
+
\omega_B(t)|\Psi_B(t)-\Psi_B(t')|\\
&\quad
+
|\omega_A(t)-\omega_A(t')| 
|\Psi_A(t')-\Psi_B(t')|\\
&\le
\left[
\lambda\jj\sigma W_{F\jj}
\bigl(\rads\sigma+\gamma_{\max}\bigr)
+
\frac{
2\bar\Phi\adj \umaxj 
}{
(\delta\jj)\opt
}
\right]
|t-t'|.
\end{align*}
Here the equality uses the definition of $V_{A,B}(t)$~\eqref{eqn:V-A-B} and the fact that $\omega_A(t) + \omega_B(t)=1$;
the first inequality uses the  triangle
inequality;
and the last inequality uses~\eqref{eqn:Psi-A-Lip},~\eqref{eqn:omgea-A-Lip},
and
\begin{align*}
|\Psi_A(t')-\Psi_B(t')|
\le
|\Psi_A(t')|+|\Psi_B(t')|
\le
2\bar\Phi\adj.
\end{align*}

Because the set of deterministic-assortment ratios  $\set{\rho_A\jj}$ is finite, we may partition  $I\jj$ at these ratios; on each resulting interval, 
 $G\adj(t)$ is the pointwise maximum of the pair-value functions  
 $V_{A,B}(t)$
 that span that interval, each of which has the common Lipschitz constant above, and hence the same bound holds globally on $I\jj$
by summing across adjacent intervals.
Thus,
\begin{align}
|G^{(a,j)}(t)-G^{(a,j)}(t')|
\le
\left[
\lambda\jj\sigma W_{F\jj}
\bigl(\rads\sigma+\gamma_{\max}\bigr)
+
\frac{
2\bar\Phi\adj \umaxj 
}{
(\delta\jj)\opt
}
\right]
|t-t'|. \label{eqn:G-ad-lip}
\end{align}

\paragraph{Subscription branch.}
For every deterministic subscription assortment \(B\in\mathcal A_\kappa\),
\begin{align*}
G^{(s,j)}(t,B)
=
\lambda\jj F\jj(\sigma t)
\left[
p-\sum_{\ell\in B\setminus S}\gamma_\ell p_\ell^{(j)}(B)
\right]
-
\sum_{\ell\in B\cap S}
\eta_\ell\beta_\ell^{(s,j)}p_\ell^{(j)}(B).
\end{align*}
Since
\begin{align*}
0\le
\sum_{\ell\in B\setminus S}\gamma_\ell p_\ell^{(j)}(B)
\le
\gamma_{\max},
\end{align*}
we have
\begin{align}
\left|
G^{(s,j)}(t,B)
-
G^{(s,j)}(t',B)
\right|
\le
\lambda\jj\sigma W_{F\jj}(p+\gamma_{\max})|t-t'|. 
\label{eqn:G-sub-lip}
\end{align}

The feasible set
\(\set{\bm{z}\in\simplex:d\subj(p,\bm{z})=0}\)
in~\eqref{eqn:sub-LP-def} is independent of \(t\). For every feasible
\(\bm{z}\), linearity and~\eqref{eqn:G-sub-lip} give
\begin{align*}
\left|G\subj(t,\bm{z})-G\subj(t',\bm{z})\right|
&=
\left|
\sum_{B\in\mc{A}_\kappa}z_B
\left(G\subj(t,B)-G\subj(t',B)\right)
\right|\\
&\le
\lambda\jj\sigma W_{F\jj}(p+\gamma_{\max})|t-t'|.
\end{align*}
If this feasible set is empty, then \(G\subj(t)\equiv 0\). Otherwise,
the inner optimal value in~\eqref{eqn:sub-LP-def}, being the pointwise
maximum over this fixed family of functions, has the same Lipschitz
constant. Since the map \(x\mapsto x_+\) is \(1\)-Lipschitz,
\(G\subj(t)\) is also
\(\lambda\jj\sigma W_{F\jj}(p+\gamma_{\max})\)-Lipschitz.

Combining the ads bound~\eqref{eqn:G-ad-lip} with the subscription bound
above gives
\begin{align*}
\left|
G^{(j)}(t)-G^{(j)}(t')
\right|
&\le
\left[
\lambda\jj\sigma W_{F\jj}
\bigl(\rads\sigma+p+2\gamma_{\max}\bigr)
+
\frac{
2\bar\Phi\adj \umaxj 
}{
(\delta\jj)\opt
}
\right]
|t-t'|\\
&=
\lipsig |t-t'|.
\end{align*}
This proves the claim.
\end{proof}

\begin{proof}[Proof of  Lemma~\ref{lem:kappa-aware-weighted-separation}]
For each $A \in \mc{A}_\kappa$, define
\[
m_A := \sum_{\ell \in A} a_\ell\jj,
\qquad
n_A := \sum_{\ell \in A} a_\ell\jj b_\ell\jj.
\]
Under Assumption~\ref{assump:rational}, whenever $u^{(j)}(A)>0$ we have
\[
x^{(j)}(A)=\frac{m_A}{Q\jj+m_A},
\qquad
u^{(j)}(A)=\frac{n_A}{L\jj(Q\jj+m_A)},
\qquad
\frac{x^{(j)}(A)}{u^{(j)}(A)} = L\jj \frac{m_A}{n_A}.
\]
Since $|A| \le \kappa$, all coefficients are nonnegative, and
$A_{j,\kappa}^{\mathrm{tot}}$, $N_{j,\kappa}$ are the sums of the $\kappa$ largest entries,
we have
\[
m_A \le A_{j,\kappa}^{\mathrm{tot}},
\qquad
n_A \le N_{j,\kappa}.
\]
The same bounds hold for $m_B,n_B$.

Now suppose
\[
\frac{x^{(j)}(A)}{u^{(j)}(A)} \neq \frac{x^{(j)}(B)}{u^{(j)}(B)}.
\]
Equivalently,
\[
\frac{m_A}{n_A} \neq \frac{m_B}{n_B},
\]
so $m_A n_B - m_B n_A$ is a nonzero integer, and therefore
\[
|m_A n_B - m_B n_A| \ge 1.
\]
Hence,
\begin{small}
\[
u^{(j)}(A)
\left|
\frac{x^{(j)}(A)}{u^{(j)}(A)}
-
\frac{x^{(j)}(B)}{u^{(j)}(B)}
\right|
=
\frac{n_A}{L\jj(Q\jj+m_A)}
\cdot
L\jj\left|\frac{m_A}{n_A}-\frac{m_B}{n_B}\right|
=
\frac{|m_A n_B - m_B n_A|}{(Q\jj+m_A)n_B}
\ge
\frac{1}{(Q\jj + A_{j,\kappa}^{\mathrm{tot}})N_{j,\kappa}}
=
\delta\jj.
\]
\end{small}
By symmetry,
\[
u^{(j)}(B)
\left|
\frac{x^{(j)}(A)}{u^{(j)}(A)}
-
\frac{x^{(j)}(B)}{u^{(j)}(B)}
\right|
\ge \delta\jj.
\]
Taking the minimum of the two inequalities proves the first claim.
\end{proof}

\subsection{Proof of 
Lemma~\ref{lem:sub-zero-p}, Lemma~\ref{lemma:ratio-reduction-in-G}, Theorem~\ref{thm:mnl-inner-separated-gap}, and Theorem~\ref{thm:asymptotic-opt}}
\label{sec:proof-main-lemma-thm}

In this section, we present proofs of 
Lemma~\ref{lem:sub-zero-p}, 
Lemma~\ref{lemma:ratio-reduction-in-G}, Theorem~\ref{thm:mnl-inner-separated-gap}, and Theorem~\ref{thm:asymptotic-opt} using lemmas introduced previously.
First, we show Lemmas~\ref{lem:sub-zero-p}
and~\ref{lemma:ratio-reduction-in-G}. 

\begin{proof}[Proof of Lemma~\ref{lem:sub-zero-p}]
Let
\[
\yadS := \yadj, \qquad \ysubS := \ysubj.
\] 

Write 
\[
R_a := \sum_{\ell\in [L]\setminus S}\gamma_\ell f_\ell^{(a,j)}(\yadS), \qquad
C_a := \sum_{\ell\in S}\eta_\ell \beta_\ell^{(a,j)} f_\ell^{(a,j)}(\yadS),
\]
\[
R_s := \sum_{\ell\in [L]\setminus S}\gamma_\ell f_\ell^{(s,j)}(\ysubS), \qquad
C_s := \sum_{\ell\in S}\eta_\ell \beta_\ell^{(s,j)}  f_\ell^{(s,j)}(\ysubS).
\]
Also, let $u_a = \uadj, u_s = \usubj, x_a = \xadj$.
Let $P_a,P_s$ be the induced ad/subscription choice probabilities under
$(\yadS,\ysubS)$.

If $p=0$ and $u_s>0$, let
$M_a:=\rads\sigma x_a-R_a$, and let $P_a^{\emptyset}$ denote the
ad-choice probability after replacing $\ysubS$ by $\ynullS$.
By~\eqref{eqn:padj}, this replacement weakly decreases the ad cutoff,
so $P_a^{\emptyset}\ge P_a$; moreover,
$P_s=1-P_a$ by~\eqref{eqn:psubj}. If $M_a\ge0$, then
\begin{align*}
G^{(j)}(\yadS,\ynullS)-G^{(j)}(\yadS,\ysubS)
&=
\lambda\jj(P_a^{\emptyset}-P_a)M_a
+\lambda\jj P_sR_s+C_s
\ge0.
\end{align*}
If instead $M_a<0$, then
\[
G^{(j)}(\yadS,\ysubS)
=
\lambda\jj P_aM_a-\lambda\jj P_sR_s-C_a-C_s
\le0
=
G^{(j)}(\ynullS,\ynullS).
\]
Thus, in either case, the original solution is weakly dominated by one
with $u_s=0=p$. Thus, an optimal solution exists with $u_s=0=p$, and
below we may assume $p>0$.
If $0<u_s<p$, then $P_s=0$ and $(u_s-p)_+=0$. Replacing $\ysubS$ by $\ynullS$
leaves the ads term unchanged and weakly decreases the subscription-side cost,
so such a solution is dominated.

Now suppose $u_s>p$, and set
\[
\theta := \frac{p}{u_s}\in(0,1), \qquad
\tilde \ysubS := \theta \ysubS + (1-\theta)\ynullS.
\]
By linearity,
\[
u^{(s,j)}(\tilde \ysubS)=p, \qquad
\tilde R_s=\theta R_s\le R_s, \qquad
\tilde C_s=\theta C_s\le C_s.
\]

If the ads mode is inactive, i.e.,
$\uadS = 0$, then
\[
G^{(j)}(\ynullS,\tilde \ysubS)
=
\lambda\jj(p-\tilde R_s)-\tilde C_s
\ge
\lambda\jj(p-R_s)-C_s
=
G^{(j)}(\ynullS,\ysubS),
\]
so replacing $\ysubS$ by $\tilde \ysubS$ weakly improves the objective. Thus it
remains to consider the ads-active case.

Assume now that the ads mode is active, i.e.,
$\uadS > 0$. Since $u_s>p$, we have $P_s=1-P_a$. For
the pair $(\yadS,\tilde \ysubS)$, define
\[
\tilde P_a := \bar F\jj\!\left(\frac{\sigma x_a}{u_a}\right), \qquad
\tilde P_s := 1-\tilde P_a.
\]
Lowering the subscription utility from $u_s$ to $p$ weakly increases the
denominator in the ad-threshold expression, so $\tilde P_a\ge P_a$. Define
\[
\Delta := \lambda\jj\bigl(\rads\sigma x_a - R_a - p + R_s\bigr).
\]
A direct expansion gives
\[
G^{(j)}(\yadS,\tilde \ysubS)
-
G^{(j)}(\yadS,\ysubS)
=
\Delta(\tilde P_a-P_a)
+\lambda\jj(R_s-\tilde R_s)\tilde P_s
+(C_s-\tilde C_s).
\]
If $\Delta\ge 0$, the right-hand side is nonnegative, so replacing $\ysubS$ by
$\tilde \ysubS$ weakly improves the objective and yields subscription utility
exactly $p$.

If instead \(\Delta<0\), then, since \(u_s>p\), we have \(P_s=1-P_a\). Thus,
\begin{align*}
G^{(j)}(\yadS, \ysubS)
&=
\lambda\jj P_a\bigl(\rads\sigma x_a-R_a\bigr)
+
\lambda\jj(1-P_a)(p-R_s)
-C_a-C_s \\
&=
-C_a+\lambda\jj(p-R_s)-C_s
+\lambda\jj\bigl(\rads\sigma x_a-R_a-p+R_s\bigr)P_a \\
&=
-C_a+\lambda\jj(p-R_s)-C_s+\Delta P_a \\
&\le
\lambda\jj(p-R_s)-C_s \\
&\le
\lambda\jj(p-\tilde R_s)-\tilde C_s
=
G^{(j)}(\ynullS,\tilde \ysubS),
\end{align*}
where the first 
equality is the definition of $G^{(j)}(\yadS, \ysubS)$,
 the second equality is direct algebra,
 the third equality uses the definition of $\Delta$, the first
inequality uses \(\Delta<0\), \(P_a\ge 0\), and
\(C_a\ge 0\), and the second inequality uses
\(\tilde R_s\le R_s\) and \(\tilde C_s\le C_s\).
Thus, again the original solution is dominated by one with subscription utility
exactly $p$.

Thus, whenever utility level $p$ is attainable on the subscription side (i.e., $\uumaxj \ge p$), there
exists an optimal solution with
\[
u_\star^{(s,j)} \in \{0,p\}.
\]

Finally, suppose no subscription routing attains utility $p$. As we show above, any
subscription assortment distribution that yields utility in the interval $(0,p)$ is
dominated by $\ynullS$. Hence, there exists an optimal solution with
$u_\star^{(s,j)} = 0$,
and we complete the proof.
\end{proof}

\begin{proof}[Proof of Lemma~\ref{lemma:ratio-reduction-in-G}]
By Lemma~\ref{lem:sub-zero-p}, we may restrict attention to solutions with
\begin{align}
    (u^{(s,j)}-p)_+=0. \label{eqn:sub-small-p-ratio-lemma}
\end{align}

We first prove the upper bound:
$G\jj(S, \bm{\beta})
\le
\max \set{
G\jj(\infty),
\max_{t\in I^{(j)}}
\left(
\geadj(t)+\gesubj(t)
\right)
}.$
For any feasible pair
\(\yadj,\ysubj \in \simplex \),
we need to show that
\begin{align}
G^{(j)}\bigl( \yadj, \ysubj\bigr)
\le \max \set{
G\jj(\infty),
\max_{t\in I^{(j)}}
\left(
\geadj(t)+\gesubj(t)
\right)
}.\label{eqn:lemma1-upper-pre}
\end{align}

Consider the following two cases.
\begin{enumerate}
    \item   \(u^{(a,j)}=0\):   
    Assumption~\ref{assump:MNL}(\ref{assm:alpha-implies-u-support}) 
    implies
\(x^{(a,j)}=0\), so the ad mode is inactive. 
Then  the objective
of $G^{(j)}\bigl( \yadj, \ysubj\bigr)$ in~\eqref{eqn:lemma1-upper-pre}
reduces to the subscription-only value generated by   $\ysubj$.
Since 
$G\jj(\infty)$ is the maximum value over all subscription-only feasible choices, the value from this particular  $\ysubj$
is at most $G\jj(\infty)$.
Thus,
\(G^{(j)}\bigl( \yadj, \ysubj\bigr) \le G\jj(\infty)\).
\item   \(u^{(a,j)}>0\): Define
\[
t:=\frac{x^{(a,j)}}{u^{(a,j)}}.
\]
By Lemma~\ref{lem:sub-zero-p}, we may assume that \(u^{(s,j)} \in \set{0, p }\).
By Lemma~\ref{lem:ratio-interval}, 
\(t\in I^{(j)}\), and the ad assortment distribution satisfies
\[
d\adj(t, \yadj)=0.
\]
Thus,
the ad-choice cutoff is \(h^{(j)}=\sigma t\), and
the type-\(j\) objective separates into an ad term and a subscription term. The
ad term is at most \(\geadj(t)\). On the subscription side, either
\(u^{(s,j)}=0\), in which case the subscription contribution is zero, or
\(u^{(s,j)}=p\), in which case the subscription contribution is at most
\(G^{(s,j)}(t)\). Therefore, for this finite ratio \(t\),  
\[
G^{(j)}\bigl( \yadj, \ysubj\bigr) \le 
\geadj(t)+\gesubj(t).
\]
\end{enumerate}
Combining the above two cases, we prove~\eqref{eqn:lemma1-upper-pre}.
Taking $\max_{\yadj,\ysubj \in \simplex }$ on both sides of~\eqref{eqn:lemma1-upper-pre}, we show 
\begin{align}
    G\jj(S, \bm{\beta})
\le
\max\left\{
G\jj(\infty),
\max_{t\in I^{(j)}}
\left(
\geadj(t)+\gesubj(t)
\right)
\right\}. \label{eqn:lemma1-upper}
\end{align}

We next prove the reverse inequality:
$G\jj(S, \bm{\beta})
\ge
\max \set{
G\jj(\infty),
\max_{t\in I^{(j)}}
\left(
\geadj(t)+\gesubj(t)
\right)
}.$

\begin{enumerate}
\item First, we see
\begin{align}
G\jj(S, \bm{\beta}) \ge 
\max_{\ysubj \in \simplex }
\set{G^{(j)}\bigl( \yadj = \ynullS, \ysubj\bigr):
d\subj(p, \ysubj)=0 \text{ or }
\ysubj = \ynullS
} = G\jj(\infty).
\label{eqn:lemma1-lower-case-1}
\end{align}
\item 

Fix any finite 
\begin{align}
t\in I^{(j)}. \label{eqn:t-finite-assumption-in-ratio-lemma}
\end{align}
Since \(\ynullS\) is feasible for the ads
LP~\eqref{eqn:ads-LP-def} and gives zero reward and zero residual, we have
\[
\geadj(t)\ge 0.
\]
\begin{enumerate}
\item If \(\geadj(t)>0\), let 
$\bm y_t^{(a,j)}$ be an optimizer of the ads LP~\eqref{eqn:ads-LP-def}.
This optimizer must have \(u^{(a,j)}>0\); otherwise,
Assumption~\ref{assump:MNL}(\ref{assm:alpha-implies-u-support})  would imply
\(x^{(a,j)}=0\) and the ads reward would be zero. Thus, 
feasibility of $\bm y_t^{(a,j)}$ implies that
\begin{align*}
 \frac{\xadj}{\uadj} = t.
\end{align*}
By~\eqref{eqn:t-finite-assumption-in-ratio-lemma},
we have
\begin{align*}
 \frac{\xadj}{\uadj} = t < \infty
\end{align*} and
by~\eqref{eqn:sub-small-p-ratio-lemma},
the cutoff in~\eqref{eqn:padj}
is \(h^{(j)}(\bm y_t^{(a,j)})=\sigma t\). If
\(G^{(s,j)}(t)>0\), choose an optimizer \(\bm y_t^{(s,j)} = \bm{z}_t\) of the  
subscription LP; otherwise,
set the subscription assortment distribution to \(\bm y_t^{(s,j)} = \ynullS\). 
Thus,
\begin{align}
    G\jj(S, \bm{\beta}) \ge 
 G^{(j)}\bigl( \bm y_t^{(a,j)}, \bm y_t^{(s,j)}\bigr) 
 \ge
\geadj(t)+\gesubj(t).
 \label{eqn:lemma1-lower-case-2-a}
\end{align}
\item It remains to consider the case \(\geadj(t)=0\). If
\(G^{(s,j)}(t) =  0\), then  
we have 
\begin{align}
    \geadj(t)+\gesubj(t) \le 0 \le  G\jj(\infty) \le G\jj(S, \bm{\beta}), \label{eqn:lemma1-lower-case-2-b-1}
\end{align} where the last inequality follows from~\eqref{eqn:lemma1-lower-case-1}.
If
\(G^{(s,j)}(t)>0\), let \(\bm{z}_t\) be an optimizer of the  
subscription LP so that
\begin{align}
    d\subj(p, \bm{z}_t)=0 \label{eqn:z-t-feasible-sub-LP}.
\end{align}
  Write
\[
M_t
:=
p-\sum_{\ell\notin S}\gamma_\ell f_\ell^{(s,j)}(\bm{z}_t),
\qquad
B_t
:=
\sum_{\ell\in S}\eta_\ell\beta_\ell^{(s,j)}f_\ell^{(s,j)}(\bm{z}_t)\ge 0.
\]
Then
\[
G^{(s,j)}(t,\bm{z}_t)
=
\lambda\jj F\jj(\sigma t)M_t-B_t>0.
\]
Thus, \(M_t>0\). 
Therefore,
\begin{align}
G\jj(\infty, \bm{z}_t)
=
\lambda\jj M_t-B_t
\ge
\lambda\jj F\jj(\sigma t)M_t-B_t
=
G^{(s,j)}(t,\bm{z}_t)
=
G^{(s,j)}(t),
\label{eqn:G-off-better-than-subj}
\end{align}
where the first inequality follows from $F\jj(\sigma t) \le $1.
Next,
\begin{align}
    G\jj(S, \bm{\beta}) \ge  
G\jj(\infty)
\ge
G\jj(\infty, \bm{z}_t)
\ge
G^{(s,j)}(t) = \geadj(t)+\gesubj(t),
\label{eqn:lemma1-lower-case-2-b-2}
\end{align}
where the first inequality follows from~\eqref{eqn:lemma1-lower-case-1},
the second inequality follows from $d\subj(p, \bm{z}_t)=0$~\eqref{eqn:z-t-feasible-sub-LP} and
\begin{align*}
        G\jj(\infty)
=
\max\Biggl\{
0, 
\max_{z\in\Delta(\mc{A}_\kappa):
d\subj(p, \bm{z})=0}
G\jj(\infty, \bm{z})
\Biggr\} \ge G\jj(\infty, \bm{z}_t),
\end{align*}
 the third inequality follows 
from~\eqref{eqn:G-off-better-than-subj},
and the   equality follows from  $\geadj(t)=0, G^{(s,j)}(t) > 0$.
\end{enumerate}
\end{enumerate}

Combining~\eqref{eqn:lemma1-lower-case-1},~\eqref{eqn:lemma1-lower-case-2-a},~\eqref{eqn:lemma1-lower-case-2-b-1}, and~\eqref{eqn:lemma1-lower-case-2-b-2}, we have   
\[
G\jj(S, \bm{\beta})
\ge
\max\left\{
G\jj(\infty),
\max_{t\in I^{(j)}}
\left(
\geadj(t)+\gesubj(t)
\right)
\right\}.
\]
Together with the upper bound~\eqref{eqn:lemma1-upper}, this proves~\eqref{eqn:G-can-be-solved-in-ratio}.
\end{proof}

\begin{proof}[Proof of Theorem~\ref{thm:mnl-inner-separated-gap}]

We first prove that
\begin{align}
\begin{split}
&  G(S,\bm{\beta})-\ghat(S,\bm{\beta},K_T) 
\le    
\frac{1}{2(K_T-1)}
\sum_{j=1}^J
\lipsig
\bigl(t_{\max}^{(j)}-t_{\min}^{(j)}\bigr)  + 
\frac{1}{K_T - 1}
\sum_{j=1}^J (\Gamma\adj +\Gamma\subj ),
\end{split}
\label{eqn:thm-relaxed-gap}
\end{align}
then derive~\eqref{eqn:thm-original-gap}.

Recall that
\[
G\jj( t)
=
\geadj(t)+ \gesubj(t),
\qquad t\in I^{(j)}.
\]

Also define 
\begin{align}
\ggrid
:=
\sum_{j=1}^J  
\max_{t\in T\jj_{K_T} \cup \set{\infty}} G\jj( t), \label{eqn:ggrid-def}
\end{align}
where $T\jj_{K_T}$ is defined in~\eqref{eqn:outer-grid-def}.
This quantity  $\ggrid$~\eqref{eqn:ggrid-def}
represents the best value obtained when the ratio \(t\) is
restricted to the grid \(T_{K_T}^{(j)}\), while also allowing the ad-inactive
branch \(t=\infty\). Thus, the difference between \(G(S,\bm{\beta})\) and
\(\ggrid\) is only the ratio-grid
discretization error.

In addition, from~\eqref{eqn:G-can-be-solved-in-ratio},
we have
\[
G(S,\bm{\beta})
=
\sum_{j=1}^J G\jj(S, \bm{\beta})
=
\sum_{j=1}^J 
\max\set{G\jj(\infty),\max_{t\in I^{(j)}} G\jj( t)}.
\]

We now separate the approximation error into two components.  The
quantity \(\ggrid \) is  what we would obtain if the fixed-\((j,t)\) ads and subscription
LPs were solved exactly, but only at the sampled branches
\(T^{(j)}_{K_T}\cup\{\infty\}\).  Thus, $G(S,\beta)-\ggrid$
is purely the outer ratio-grid discretization error.  
Conditional on each $t$, however, Algorithm~\ref{alg:mnl-solver-all}
computes
the fixed-\((j,t)\) LP values only approximately, using finite dual
bisection and postprocessing. 
Lemma~\ref{lem:G-gap-combined-ads-sub} bounds this inner fixed-branch
computational error through the terms \(\errork\). By the branch-selection
rule and the two-case comparison in the proof of
Lemma~\ref{lemma:ratio-reduction-in-G}, applied to the returned branch
solutions, \(\ghat(S,\bm{\beta},K_T)\) weakly dominates the left-hand side
of~\eqref{eq:grid-hat-error}. Hence,~\eqref{eq:grid-hat-error} yields
\[
\ghat(S,\bm{\beta},K_T)
\ge
\ggrid
-
\sum_{j=1}^J \max_{t\in T\jj_{K_T} \cup \offset}
\errork,
\]
where $\errork$
is defined in~\eqref{eqn:search-error-in-grid}.
Therefore,
\begin{align*}
G(S,\bm{\beta})-\ghat(S,\bm{\beta},K_T)
&\le
G(S,\bm{\beta})
-
\ggrid
+
\sum_{j=1}^J \max_{t\in T\jj_{K_T} \cup \offset}
\errork.
\end{align*}

For each $j$, choose $(\tau\jj)\opt \in \argmax_{\tau \in I^{(j)} \cup \{\infty\}} G^{(j)}(\tau)$. If $(\tau\jj)\opt = \infty$, then the discretization error for type $j$ is zero because $\infty \in T\jj_{K_T} \cup \offset$. Otherwise, let $\bar{t}\jj \in T_{K_T}^{(j)}$ be a nearest grid point to $(\tau\jj)\opt$. By~\eqref{eqn:Delta-t-K-def}, $\lvert (\tau\jj)\opt - \bar{t}\jj \rvert \leq \Delta_{t,K_T}^{(j)}$, and hence by~\eqref{eqn:lip-cond-holds},
\[
\max_{\tau \in I^{(j)} \cup \{\infty\}} G^{(j)}(\tau)
-
\max_{\tau \in T\jj_{K_T} \cup \set{\infty}} G^{(j)}(\tau)
\leq
G^{(j)}((\tau\jj)\opt) - G^{(j)}(\bar{t}\jj)
\leq
\lipsig 
\Delta_{t,K_T}^{(j)},
\] where $\lipsig$
is defined in~\eqref{eqn:lip-cond-holds}. 
Summing over $j$ gives
\begin{align}
G(S,\bm{\beta})-\ggrid
 \le
\sum_{j=1}^J
\lipsig\Delta_{t,K_T}^{(j)} = 
\frac{1}{2(K_T-1)}
\sum_{j=1}^J
\lipsig
\bigl(t_{\max}^{(j)}-t_{\min}^{(j)}\bigr),
\label{eqn:thm-gap-step-1}
\end{align}
where the last equality again uses~\eqref{eqn:Delta-t-K-def}.
 
It remains to bound the search-error term. For each type \(j\) and each sampled ratio
\(t=t_r^{(j)}\in T_{K_T}^{(j)}\), Lemma~\ref{lem:G-gap-combined-ads-sub} and the definition~\eqref{eqn:search-error-in-grid}
give
\[
\errork
=
\begin{cases}
\dfrac{D^{(a,j)}}{2} \bigl|\bar{\mc{J}}\adj(t) \bigr| 2^{-K^{(a,j)}(t)}, & t \in (t_{\min}^{(j)},t_{\max}^{(j)}),
\\[0.8em]
0, & t\in\{t_{\min}^{(j)},t_{\max}^{(j)},\infty\},
\end{cases}
\;+\;
\begin{cases}
\dfrac{D^{(s,j)}}{2} \bigl|\bar{\mc{J}}^{(s,j)}\bigr| 2^{-K^{(s,j)}}, & p\in(0,U_{\max}^{(j)}),
\\[0.8em]
0, & p\notin(0,U_{\max}^{(j)}).
\end{cases}
\]

For the ads branch, if \(r\in\{2,\dots,K_T-1\}\), then by 
\eqref{eqn:outer-grid-def} and
\eqref{eqn:ads-dual-interval-algo},
\[
\bigl|\bar{\mc{J}}\adj(t_r^{(j)})\bigr|
=
\bar\Phi^{(a,j)}\frac{K_T-1}{t\jj_{\max} - t\jj_{\min}}
\left(
\frac{\Theta_-^{(a,j)}}{r-1}
+
\frac{\Theta_+^{(a,j)}}{K_T-r}
\right).
\]
By    \eqref{eqn:inner-grid-size-ads},
\[
2^{-K^{(a,j)}(t_r^{(j)})}
\le
\frac{2(t\jj_{\max} - t\jj_{\min})}{(K_T-1)^2}
\left(
\frac{\Theta_-^{(a,j)}}{r-1}
+
\frac{\Theta_+^{(a,j)}}{K_T-r}
\right)^{-1}.
\]
Hence,
\[
\frac{D^{(a,j)}}{2} \bigl|\bar{\mc{J}}\adj(t_r^{(j)})\bigr| 2^{-K^{(a,j)}(t_r^{(j)})}
\le
\frac{D^{(a,j)}\bar\Phi^{(a,j)}}{K_T-1}
=
\frac{\Gamma^{(a,j)}}{K_T-1}.
\]
For \(r\in\{1,K_T\}\), the ads contribution to $\errork$
is \(0\).

For the subscription branch, if \(p\in(0,U_{\max}^{(j)})\), then by \eqref{eqn:sub-dual-interval-algo} and \eqref{eqn:inner-grid-size-sub},
\[
\bigl|\bar{\mc{J}}^{(s,j)}\bigr|=\bar\Phi^{(s,j)}\Theta^{(s,j)},
\qquad
2^{-K^{(s,j)}}\le \frac{2}{(K_T-1)\Theta^{(s,j)}},
\]
so
\[
\frac{D^{(s,j)}}{2} \bigl|\bar{\mc{J}}^{(s,j)}\bigr| 2^{-K^{(s,j)}}
\le
\frac{D^{(s,j)}\bar\Phi^{(s,j)}}{K_T-1}
=
\frac{\Gamma^{(s,j)}}{K_T-1},
\] where the last equality follows from~\eqref{eqn:Gamma-j-a-def}.
If \(p\notin(0,U_{\max}^{(j)})\),   
then the subscription contribution to
$\errork = 0$.

Therefore, for every \(j\),
\[
\max_{t\in T_{K_T}^{(j)}\cup\{\infty\}} \errork
\le
\frac{\Gamma^{(a,j)}+\Gamma^{(s,j)}}{K_T-1}.
\]
Summing the bound in  
Lemma~\ref{lem:G-gap-combined-ads-sub} 
over \(j\) yields
\[
\ggrid -   \ghat(S,\bm{\beta},K_T) \le  
\sum_{j=1}^J \max_{t\in T_{K_T}^{(j)}\cup\{\infty\}} \errork
\le
\frac{1}{K_T-1}\sum_{j=1}^J\bigl(\Gamma^{(a,j)}+\Gamma^{(s,j)}\bigr).
\]

Combining this bound with~\eqref{eqn:thm-gap-step-1},
we obtain
\begin{align*}
    G(S,\bm{\beta})- \ghat(S,\bm{\beta},K_T)
&= \Paran{G(S,\bm{\beta})-  \ggrid}
+
\Paran{\ggrid - \ghat(S,\bm{\beta},K_T)} \\
& \le
\frac{1}{2(K_T-1)}
\sum_{j=1}^J
\lipsig
\bigl(t_{\max}^{(j)}-t_{\min}^{(j)}\bigr)
+
\frac{1}{K_T-1}
\sum_{j=1}^J
\bigl(\Gamma^{(a,j)}+\Gamma^{(s,j)}\bigr),
\end{align*}
which is exactly \eqref{eqn:thm-relaxed-gap}.

We now prove~\eqref{eqn:thm-original-gap}.
By adding and subtracting $R(S)$ and using the relaxation inequality~\eqref{eqn:beta-relax-ineq},
\begin{align*}
R\opt-\rhat(S,\bm{\beta},K_T)
&=
\bigl(R\opt-R(S)\bigr)
+
\bigl(R(S)-\rhat(S,\bm{\beta},K_T)\bigr) \\
&\le
\bigl(R\opt-R(S)\bigr)
+
\bigl(G(S,\bm{\beta})-\ghat(S,\bm{\beta},K_T)\bigr) 
+
\bigl(\ghat(S,\bm{\beta},K_T)-\rhat(S,\bm{\beta},K_T)\bigr).
\end{align*}
By Lemma~\ref{lemma:buyset-gap-mnl},
\[
R\opt - R(S)
\le
\TOP_{2J}\{\bm w^S\},
\]
and by Lemma~\ref{lemma:relaxation-gap} applied to the returned assortment distribution
$\widehat{\bm y}(S,\bm{\beta},K_T)$,
\[
\ghat(S,\bm{\beta},K_T)-\rhat(S,\bm{\beta},K_T) 
\le
G(S,\bm{\beta}, \hat{\bm{y}})
-R(S, \hat{\bm{y}}) \le
\TOP_{2J}\set{\bm q^{S,\bm{\beta}}},
\] where we omit the dependence of $\hat{\bm y}$ on $(S,\bm{\beta},K_T)$ for notational clarity.
Combining these two bounds above 
with~\eqref{eqn:thm-relaxed-gap}, we obtain
\begin{align*}
R\opt-\rhat(S,\bm{\beta},K_T)
&\le
\TOP_{2J}\{\bm w^S\}
+
\TOP_{2J}\set{\bm q^{S,\bm{\beta}}}  \\ 
&+
\frac{1}{2(K_T-1)}
\sum_{j=1}^J
\lipsig
\bigl(t_{\max}^{(j)}-t_{\min}^{(j)}\bigr)  
+
\frac{1}{K_T-1}
\sum_{j=1}^J
(\Gamma\adj+\Gamma\subj)
\end{align*}
and we complete the proof.
\end{proof}

\begin{proof}[Proof of Theorem~\ref{thm:asymptotic-opt}]
First, we show that
\begin{align*}
   R\opt_N \ge N \RL
\qquad \forall N \ge 1.
\end{align*}
Indeed, for any fixed assortment distribution \(\bm{y}\), the revenue term scales
linearly in \(N\), while the procurement cost is sublinear in \(N\):
\[
C_{N\lambda}(\bm{y})
=
\sum_{\ell}
\min\{a_\ell(\bm{y}),N b_\ell(\bm{y})\}
\le
N\sum_{\ell}\min\{a_\ell(\bm{y}),b_\ell(\bm{y})\}
=
N C_{\lambda}(\bm{y}),
\]
where
\begin{align*}
a_\ell(\bm y) & := \eta_\ell \max_{j,m} f_\ell^{(m,j)}(\bm y), 
\\ b_\ell(\bm y) & := \gamma_\ell \sum_j \lambda^{(j)}
\left(
P^{(a,j)}(\bm y) f_\ell^{(a,j)}(\bm y)
+
P^{(s,j)}(\bm y) f_\ell^{(s,j)}(\bm y)
\right).
\end{align*}
Evaluating the size-\(N\) instance at a base-instance optimizer \(\bm{y}_1^\star\)
therefore gives  
\[
R_N\opt\ge N \RL
\qquad \forall N\ge 1.
\]

Fix \(N\ge 1\). Consider the normalized-cost instance with consumer masses
\(\bm{\lambda}\) and buy costs \((\eta_\ell/N)_{\ell\in[L]}\).
For this instance and an assortment distribution $\bm{y}$, let
\[
\widetilde R_{N}(S_{N}\opt, \bm{y}),\qquad
\widetilde G_N(S_N\opt, \bm{\beta}\opt_N, \bm{y})
\]
denote, respectively,    the fixed-buy-set objective with
buy set \(S_{N}\opt\)~\eqref{eqn:S-N-opt-def}  and the \(\bm{\beta}_N\opt\)~\eqref{eqn:uniform-beta}-relaxed objective with
buy set \(S_{N}\opt\). 
Let
\begin{align*}
   \widetilde R_N(\bm{y}) = \max_{S \subseteq [L]}\widetilde R_N(S,\bm{y}).
\end{align*}
Also define the corresponding optimal values
\[
\widetilde R_N\opt
:=
\max_{\bm{y} \in \simplexprod}\widetilde R_N(\bm{y}),
\qquad
\widetilde R_{N}(S_{N}\opt)
:=
\max_{\bm{y} \in \simplexprod}\widetilde R_{N}(S_{N}\opt, \bm{y}),
\]
and
\[
\widetilde G_N
:=
\max_{\bm{y} \in \simplexprod}\widetilde G_N(S_N\opt, \bm{\beta}\opt_N, \bm{y}).
\]
Couple the executions of Algorithm~\ref{alg:mnl-solver-all} on the two
instances by using the same fixed lexicographic tie-breaking rule for
every set-valued choice in the algorithm and its subroutines.
The ratio grids, residual functions, $U_{\max}^{(j)}$, and bisection
budgets coincide. Each branch reward is $N$ times its normalized
counterpart, and each dual-search interval is obtained by multiplying
the normalized interval by $N$. Inductively, corresponding dual
midpoints and reduced-revenue vectors are scaled by $N$, so the same
assortments and candidate sets are generated, and all residual-sign and
clipping decisions coincide. The exact boundary solvers are invariant
under the same positive scaling.
Postprocessing returns the same mixing weights because the residuals are
identical and the rewards are scaled by $N$. Finally, the outer branch
scores are scaled by $N$, so the same branch is selected and the two runs
return the same assortment solution.
Let \(\widehat{\bm{y}}_{N,K_T}\) be the solution returned by
Algorithm~\ref{alg:mnl-solver-all} 
  on this normalized-cost instance. 
We also let
\[
\rhat_{N,K_T}
:=
R_N(S_{N}\opt, \widehat{\bm{y}}_{N,K_T})
\]
denote the value of the same returned assortment distribution when it is evaluated in the original
size-\(N\) instance.

\medskip
\noindent
\paragraph{Step 1: pass to the normalized instance.}
By the definitions, we have
\begin{align*}
R_N(\bm{y})
&=
N \widetilde R_N(\bm{y}), \\
R_N(S_{N}\opt, \bm{y})
& =
N \widetilde R_{N}(S_{N}\opt, \bm{y}), \\
G_N(S_{N}\opt,\bm{\beta}_N\opt, \bm{y})
& =
N \widetilde G_N(S_N\opt, \bm{\beta}\opt_N, \bm{y}).
\end{align*}
Therefore,
\[
R_N\opt=N \widetilde R_N\opt,
\qquad
\rhat_{N,K_T}
=
N \widetilde R_{N}(S_{N}\opt, \widehat{\bm{y}}_{N,K_T}),
\]
so the approximation ratio is unchanged:
\begin{align}
    \frac{\rhat_{N,K_T}}
     {R_N\opt}
=
\frac{\widetilde R_{N}(S_{N}\opt, \widehat{\bm{y}}_{N,K_T})}
     {\widetilde R_N\opt}. \label{eqn:asymptotic-ratio-identity}
\end{align}

Moreover, the threshold buy set~\eqref{eqn:S-N-opt-def} is the same in both formulations, since the
  threshold rule compares buy cost with aggregate rental exposure,
and both sides scale in the same way:
\[
2\eta_\ell
\le
N\gamma_\ell \sum_{j=1}^J \lambda\jj
\iff
2\frac{\eta_\ell}{N}
\le
\gamma_\ell \sum_{j=1}^J \lambda\jj.
\]

\paragraph{Step 2: apply the finite-instance guarantee to the normalized-cost instance.}
For the normalized-cost instance, define
\[
A_{\sigma,p}^{(N)}
:=
\frac{1}{\RL}\Paran{
\frac{1}{2}\sum_{j=1}^J
\lipsig_N
\bigl(t_{\max}^{(j)}-t_{\min}^{(j)}\bigr)
+
\sum_{j=1}^J
\bigl(\Gamma_N^{(a,j)}+\Gamma_N^{(s,j)}\bigr)},
\]
where \(\lipsig_N\), \(\Gamma_N^{(a,j)}\), and \(\Gamma_N^{(s,j)}\) are the constants from Theorem~\ref{thm:mnl-inner-separated-gap} computed for the normalized-cost instance. 

 Let
\[
\widetilde q_{\ell}^{ S_{N}\opt,\bm{\beta}_N\opt}
=
\frac{\eta_\ell}{N}
\Bigl(1-\min_{j,m}(\bm{\beta}_N\opt)_{\ell}^{(m,j)}\Bigr)
\I{\ell\in S_{N}\opt}
\le
\frac{\eta_\ell}{N}.
\]
Therefore,
\begin{align}
    \TOP_{2J}\set{\widetilde{\bm q}^{ S_{N}\opt,\bm{\beta}_N\opt}}
\le
\frac{\TOP_{2J}\{(\eta_\ell)_{\ell\in[L]}\}}{N}. \label{eqn:top-q-bound}
\end{align}
 Define $\widetilde {\bm{w}}^{ S_{N}\opt}$ as
\[
\widetilde w_{\ell}^{ S_{N}\opt}
=
\begin{cases}
\frac{\eta_\ell}{N}, & \ell\in S_{N}\opt,\\[2mm]
\left(\gamma_\ell\sum_{j=1}^J\lambda\jj-\frac{\eta_\ell}{N}\right)_+, & \ell\notin S_{N}\opt.
\end{cases}
\]
By the threshold characterization of \(S_{N}\opt\), the second line is at
most \(\eta_\ell/N\). Thus,
\begin{align*}
    \widetilde w_{\ell}^{ S_{N}\opt}\le \frac{\eta_\ell}{N}
\qquad\forall \ell,
\end{align*}
and therefore,
\begin{align}
    \TOP_{2J}\set{\widetilde{\bm w}^{ S_{N}\opt}}
\le
\frac{\TOP_{2J}\{(\eta_\ell)_{\ell\in[L]}\}}{N}.  \label{eqn:top-w-bound}
\end{align}

Applying Theorem~\ref{thm:mnl-inner-separated-gap} to the normalized-cost instance with
\[
S=S_N\opt,\qquad \bm{\beta}=\bm{\beta}_N\opt,\qquad K_T\ge 2,
\]
yields
\begin{align}
\widetilde R_{N}(S_{N}\opt, \widehat{\bm{y}}_{N,K_T}) \ge
\widetilde R_N\opt
-  \TOP_{2J}\set{\widetilde{\bm q}^{ S_{N}\opt,\bm{\beta}_N\opt}}
 - \TOP_{2J}\set{\widetilde{\bm w}^{ S_{N}\opt}}
 - \frac{A_{\sigma,p}^{(N)} \RL}{K_T - 1}.
 \label{eqn:asymptotic-bound-step-1}
\end{align}

\paragraph{Step 3: compare \(A_{\sigma,p}^{(N)}\) with \(A_{\sigma,p}\).}
Relative to the base-mass instance used to define $A_{\sigma,p}$, the normalized-cost instance differs only through the replacement $\eta_\ell \mapsto \eta_\ell/N$. Thus,
$\eta_{\max}$ is replaced by $\eta_{\max}/N \le \eta_{\max}$. By the monotonicity assumption in~\eqref{eqn:lip-cond-holds},
\[
\lipsig_N \le \lipsig.
\]

By the explicit formulae in~\eqref{eqn:Gamma-j-a-def},
\[
\Gamma_N^{(a,j)} \le \Gamma^{(a,j)},
\qquad
\Gamma_N^{(s,j)} \le \Gamma^{(s,j)}.
\]

Therefore,
$A_{\sigma,p}^{(N)} \le A_{\sigma,p}$. 
 By~\eqref{eqn:asymptotic-bound-step-1},~\eqref{eqn:top-q-bound}
and~\eqref{eqn:top-w-bound},
\begin{align*}
\widetilde R_{N}(S_{N}\opt, \widehat{\bm{y}}_{N,K_T}) 
& \ge
\widetilde R_N\opt
 - \frac{A_{\sigma,p} \RL}{K_T-1}
-  \TOP_{2J}\set{\widetilde{\bm q}^{ S_{N}\opt,\bm{\beta}_N\opt}}
 - \TOP_{2J}\set{\widetilde{\bm w}^{ S_{N}\opt}}\\
  & \ge
\widetilde R_N\opt
- \frac{A_{\sigma,p} \RL}{K_T-1}
-
\frac{2\TOP_{2J}\{(\eta_\ell)_{\ell\in[L]}\}}{N}.
\end{align*}

\medskip
\noindent
\paragraph{Step 4: divide by the linear lower bound.}
Since
\[
\widetilde R_N\opt
=
\frac{R_N\opt}{N}
\ge
\RL,
\]
we obtain
\[
\frac{\widetilde R_{N}(S_{N}\opt, \widehat{\bm{y}}_{N,K_T})}
     {\widetilde R_N\opt}
\ge
1
-
\frac{A_{\sigma,p}}
     {K_T-1}
-
\frac{2\TOP_{2J}\{(\eta_\ell)_{\ell\in[L]}\}}
     {N \RL}.
\]
Using the ratio identity from~\eqref{eqn:asymptotic-ratio-identity},
\[
\frac{\rhat_{N,K_T}}
     {R_N\opt}
=
\frac{\widetilde R_{N}(S_{N}\opt, \widehat{\bm{y}}_{N,K_T})}
     {\widetilde R_N\opt},
\]
and therefore
\[
\frac{\rhat_{N,K_T}}
     {R_N\opt}
\ge
1
-
\frac{A_{\sigma,p}}
     {K_T-1}
-
\frac{2\TOP_{2J}\{(\eta_\ell)_{\ell\in[L]}\}}
     {N \RL}.
\]
Recalling the definition
\[
B_{\sigma,p}
:=
\frac{2\TOP_{2J}\{(\eta_\ell)_{\ell\in[L]}\}}{\RL},
\]
this is exactly the claimed bound.
\end{proof}
 
\section{Algorithm details}

\subsection{Algorithm~\ref{alg:mnl-sub-boundary}}
\label{subsec:sub-boundary}

First, we present 
Algorithm~\ref{alg:exact-card-mnl}
as a subroutine used in
Algorithm~\ref{alg:mnl-sub-boundary}.
Algorithm~\ref{alg:exact-card-mnl} is adapted from 
Algorithm~\ref{alg:mnl-oracle}.
 To avoid confusion with $L$, the number of content families, we denote the element index set in Algorithm~\ref{alg:exact-card-mnl} by $\mc{L}$.
For a vector $\bm v$,
we let
\(\topc(\bm v)\) denote the subset of its indices
obtained by sorting elements in
nonincreasing order of \(v_e\), breaking ties according to any fixed order, and
selecting the first \(c\) elements.
At a high level, it constructs the breakpoint set $\mathcal B$ of pairwise
intersections of the adjusted-score lines $g_e(\cdot)=a_e(r_e-\cdot)$ and
sweeps over the intervals induced by $\mathcal B$, within which the top-$c$
set is unchanged.
Note that line~\ref{line:exactcard-no-distinct-case}
of Algorithm~\ref{alg:exact-card-mnl} 
handles the special
case in which the breakpoint set $\mathcal B$ is empty. This occurs
when $a_e=a_f$ for all $e,f$. In this case, 
the ordering of the
adjusted scores $g_e(v)=a_e(r_e-v)$ does not change with $v$. Thus, it
suffices to evaluate the adjusted scores at $v=0$. 
It is easy to check that 
an optimal exact-cardinality assortment is obtained by
selecting the top $c$ elements of $\{a_e r_e\}$.
The displayed pseudocode is conceptual; the stated
$O(|\mathcal L|^2\log|\mathcal L|)$ 
complexity corresponds to the standard
line-sweep implementation, which maintains the ordering of the adjusted-score
lines and the current top-$c$ set incrementally across the sorted breakpoints,
rather than recomputing $\topc(\bm v)$ from scratch on each interval~\citep{RusmevichientongShSh10}.

\begin{algorithm}[H]
\renewcommand{\thealgorithm}{\texttt{ExactCardMNL}}
\caption{Exact-cardinality MNL solver}
\label{alg:exact-card-mnl}
\begin{algorithmic}[1]
\STATE \textbf{Input:} weights $(a_e)_{e\in \mc{L}}$, 
revenues $(r_e)_{e\in \mc{L}}$, exact cardinality $c\in\{0,\dots,|\mc{L}|\}$
\IF{$c=0$}
    \STATE \textbf{return} $\emptyset$
\ENDIF
\FOR{each $e\in \mc{L}$}
    \STATE    $g_e(\cdot)\gets a_e(r_e-\cdot)$
\ENDFOR
\STATE  
\[
\mathcal B \gets \left\{
\frac{a_e r_e-a_f r_f}{a_e-a_f}
:\ e,f\in \mc{L},\ e\neq f,\ a_e\neq a_f
\right\}.
\]
\STATE Let $B_1<\cdots<B_m$ be the distinct values in $\mathcal B$
\IF{$m=0$}
    \STATE   $E\opt \gets \topc(\set{a_e r_e}_{e\in \mc{L}})$ \label{line:exactcard-no-distinct-case}
    \STATE \textbf{return} $E\opt$
\ENDIF
\STATE Pick $v < B_1$, 
$\mc C \gets \{\topc(\{g_e(v)\}_{e\in \mc{L}}) \}$
 \STATE $B_{m+1} \gets \infty$
\FOR{$k=2,\dots,m + 1$}
    \STATE Pick any $v_k \in (B_{k-1}, B_k) , E_k \gets\topc(\{g_e(v_k)\}_{e\in \mc{L}})$ 
    \STATE Update $\mc C \gets \mc C\cup\{E_k\}$
\ENDFOR
\STATE \textbf{return}  
\[
E\opt\in
\argmax_{E\in \mc C}
\frac{\sum_{e\in E} a_e r_e}{1+\sum_{e\in E} a_e}.
\]
\end{algorithmic}
\end{algorithm}

Next, we describe Algorithm~\ref{alg:mnl-sub-boundary}, which builds on
Algorithm~\ref{alg:exact-card-mnl}. Recall that the notation
\(G\subj(t,\bm z)\) in~\eqref{eqn:sub-LP-def} suppresses its dependence on the
subscription price \(p\). In the boundary case considered here, this suppressed
price is evaluated at $U\opt := U_{\max}^{(j)}.$
Accordingly, we define the boundary subscription LP value by
\begin{align}
G\subj(U\opt,t)
:=
\max_{\bm z \in \simplex}
\set{
\left. G\subj(t,\bm z)\right|_{p=U\opt}
:
d\subj(U\opt, \bm z)=0
}.
\label{eqn:G-sub-def-U-opt}
\end{align}
For a deterministic assortment \(B\in\mathcal A_\kappa\), we also write
\begin{align}
G\subj(U\opt,t,B)
:=
\left. G\subj(t,\bm e_B)\right|_{p=U\opt},
\label{eqn:G-sub-boundary-det}
\end{align}
where \(\bm e_B\) denotes the unit vector on assortment \(B\).

\begin{algorithm}[H]
\renewcommand{\thealgorithm}{\texttt{SubscriptionBoundarySolver}}
\caption{Subscription boundary solver at $p=U_{\max}^{(j)}$}
\label{alg:mnl-sub-boundary}
\begin{algorithmic}[1]
\STATE \textbf{Input:} $j,t,S,\sigma,\bm\beta,U_{\max}^{(j)}$
\FOR{each $\ell\in[L]$}
    \STATE $s_\ell \gets \alpha_\ell^{(j)}\bigl(u_\ell^{(j)}-U_{\max}^{(j)}\bigr)$
\STATE $\rho_\ell \gets \lambda\jj F\jj(\sigma t)\gamma_\ell \I{\ell\notin S}
+ \eta_\ell \beta_\ell^{(s,j)} \I{\ell \in S}$
\ENDFOR
\STATE Sort $\{s_\ell\}_{\ell\in[L]}$ in weakly decreasing order
\STATE $q \gets |\{\ell\in[L]: s_\ell>0\}|$
\IF{$q\ge \kappa$}
\STATE $\tau \gets$ the $\kappa$-th largest value of $\bm s$
\STATE $H \gets \{\ell\in[L]: s_\ell>\tau\}$,
$T \gets \{\ell\in[L]: s_\ell=\tau\}$,
$c \gets \kappa-|H|$,
$D_H \gets 1+\sum_{h\in H}\alpha_h^{(j)}$,
$M_H \gets \dfrac{\sum_{h\in H}\alpha_h^{(j)}\rho_h}{D_H}$,
\FOR{each $e\in T$}
\STATE $a_e \gets \alpha_e^{(j)}/D_H, r_e \gets M_H-\rho_e$
\ENDFOR
\STATE Compute
$E\opt \gets$~\ref{alg:exact-card-mnl}$((a_e)_{e\in T},(r_e)_{e\in T},c)$
\STATE $B \gets H\cup E\opt$
\ELSE \STATE
$P \gets \{\ell\in[L]: s_\ell>0\}$,
$Z \gets \{\ell\in[L]: s_\ell=0\}$,
$h \gets \kappa-|P|$,
$D_P \gets 1 + \sum_{\ell \in P} \alpha_\ell^{(j)},$
$M_P \gets \dfrac{\sum_{\ell \in P}\alpha_\ell^{(j)}\rho_\ell}{D_P}$
\FOR{each $e\in Z$}
\STATE $a_e \gets \alpha_e^{(j)}/D_P$,
$r_e \gets M_P-\rho_e$
\ENDFOR
\STATE $(V\opt ,E\opt)\gets$~\ref{alg:mnl-oracle} $((a_e)_{e\in Z},(r_e)_{e\in Z},  h)$
\STATE $B\gets P\cup E\opt$
\ENDIF
\STATE  %
$\yhat(A)  \gets \I{A=B}, \forall A \in \mc{A}_\kappa$ 
\STATE \textbf{Return:} $G^{(s,j)}(U_{\max}^{(j)},t,B)$ (defined in~\eqref{eqn:G-sub-boundary-det}),
$\yhat$
\end{algorithmic}
\end{algorithm}

\begin{lemma}
\label{lem:exact-card-mnl}
Fix a finite set $\mc{L}$, 
weights $a_e \ge 0$,
revenues $r_e\in\mathbb R$, and
$c\in\{0,\dots,|\mc{L}|\}$. Algorithm~\ref{alg:exact-card-mnl} returns a set
$E\opt\subseteq \mc{L}$, $|E\opt|=c$, that solves
\[
\max_{E\subseteq \mc{L},\ |E|=c}
\frac{\sum_{e\in E} a_e r_e}{1+\sum_{e\in E} a_e}.
\]
Moreover, it can be implemented in $O(|\mc{L}|^2\log |\mc{L}|)$ time.
\end{lemma}

\begin{proof}[Proof of Lemma~\ref{lem:exact-card-mnl}]
The proof follows the classical fractional-programming transformation, together with the standard adjusted-score
sweep used in cardinality-constrained MNL assortment optimization~\citep{RusmevichientongShSh10}. 
We give   an argument for completeness.

For any feasible set $E\subseteq \mc{L}$ with $|E|=c$, define
\[
A_E:=\sum_{e\in E} a_e r_e,\qquad
B_E:=\sum_{e\in E} a_e,\qquad
\frac{\sum_{e\in E} a_e r_e}{1 + \sum_{e\in E} a_e }:=\frac{A_E}{1+B_E}.
\]
For $\nu\in\mathbb R$, define the parametric score
\[
\Psi_\nu(E):= \sum_{e\in E} a_e r_e-\nu \Paran{1+\sum_{e\in E} a_e}.
\]
Using $g_e(\nu)=a_e(r_e-\nu)$, we obtain
\[
\Psi_\nu(E)=  -\nu+\sum_{e\in E} g_e(\nu).
\]
Thus, for fixed \(\nu\), maximizing \(\Psi_\nu(E)\) over all \(E\subseteq \mc{L}\)
with \(|E|=c\) is equivalent to choosing the \(c\) largest values among
\(\{g_e(\nu)\}_{e\in \mc{L}}\). The only values of \(\nu\) at which this ordering
can change are pairwise score-crossing points
\begin{align}
    \mathcal B
:=
\left\{
\frac{a_e r_e-a_f r_f}{a_e-a_f}
:
e,f\in \mc{L},\ e\neq f,\ a_e\neq a_f
\right\}. \label{eqn:breakpoint-def}
\end{align}
On each connected component of \(\mathbb R\setminus\mathcal B\), the ordering
of the affine scores \(\{g_e(\nu)\}_{e\in \mc{L}}\) is constant; hence, under the
fixed tie-breaking rule, the top-\(c\) set is constant on that interval and is
examined by the sweep.

Now let
\[
\nu\opt:=\max_{E\subseteq \mc{L},  |E|=c}\frac{\sum_{e\in E} a_e r_e}{1 + \sum_{e\in E} a_e }.
\]
Since $1+B_E>0$ for every feasible $E$,
\[
\Psi_\nu(E)=(1+B_E)\left(\frac{\sum_{e\in E} a_e r_e}{1 + \sum_{e\in E} a_e }-\nu\right).
\]
Therefore,
\[
\max_{E: |E|=c} \Psi_{\nu\opt}(E) =
\max_{E: |E|=c} \set{\sum_{e\in E} a_e r_e-\nu\opt \Paran{1+\sum_{e\in E} a_e}}=  
\max_{E: |E|=c} \set{
-\nu\opt 
+\sum_{e\in E} g_e(\nu\opt)} = 
0,
\]
and every maximizer of $\Psi_{\nu\opt}$ is optimal for the original ratio
problem.

If \(\nu\opt \notin \mc{B}\), where $\mc{B}$ is defined 
in~\eqref{eqn:breakpoint-def}, then it lies in one interval of
\(\mathbb R\setminus \mathcal B\).
On that interval, the ordering of the affine scores 
$g_e(\cdot)=a_e(r_e-\cdot)$
is constant. Thus, the top-\(c\) set at
\(\nu\opt\) is exactly the representative top-\(c\) set recorded by the sweep for that interval.

Now suppose \(\nu\opt\in\mathcal B\). Let \(E^-\) be a top-\(c\) set immediately to the left of
\(\nu\opt\), which the algorithm records. We claim that \(E^-\) is also a top-\(c\) set at
\(\nu\opt\). If not, then there must exist \(e\in E^-\) and \(f\notin E^-\) such that
\[
g_f(\nu\opt)>g_e(\nu\opt).
\]
By continuity of the affine functions \(g_e\) and \(g_f\), this strict inequality also holds for all
\(\nu\) sufficiently close to \(\nu\opt\) from the left. But then \(E^-\) could not have been a
top-\(c\) set immediately to the left of \(\nu\opt\), a contradiction. Therefore,
\(E^-\) is a
top-\(c\) set at \(\nu\opt\).

Thus, whether or not \(\nu\opt\) is a breakpoint, the sweep records a set that maximizes
\(\Psi_{\nu\opt}\). By the preceding paragraph, such a set is optimal for the original ratio
problem. Since the algorithm returns the recorded candidate with the largest ratio value, the
returned set is optimal.

There are $O(|\mc{L}|^2)$ pairwise breakpoints in~\eqref{eqn:breakpoint-def}. Sorting them costs
$O(|\mc{L}|^2\log |\mc{L}|)$. 
A standard sweep  maintains the line
order and the current top-$c$ set within the same bound. Thus, the total
running time is $O(|\mc{L}|^2\log |\mc{L}|)$.
\end{proof}

\begin{lemma} 
\label{lem:sub-boundary-poly-simple}
Fix $t\in I^{(j)}\cup\{\infty\}$ and set $p=U_{\max}^{(j)}$. 
Let
\begin{align*}
 G\subj(p,t,B)
&  =
\lambda\jj F\jj(\sigma t)
\Biggl[
p-
\sum_{\ell\in [L]\setminus S}\gamma_\ell p_\ell^{(j)}(B)
\Biggr]
-
\sum_{\ell\in S}\eta_\ell \beta_\ell\subj p_\ell^{(j)}(B).
\end{align*}
Under the
convention $F\jj(\sigma\infty)=1$,~\ref{alg:mnl-sub-boundary}($j,t,S,\sigma,\bm\beta,U_{\max}^{(j)}$)
returns an assortment
\[
\tilde{B} \in
\argmax\Bigl\{
G^{(s,j)}(\uumaxj, t,B):
B\in\mathcal A_\kappa,\ u^{(j)}(B)=U_{\max}^{(j)}
\Bigr\}.
\]
Moreover, after $U_{\max}^{(j)}$ is known, the algorithm runs in $O(L^2 \log L)$ time.
\end{lemma}

\begin{proof}[Proof of Lemma~\ref{lem:sub-boundary-poly-simple}]
Let $U\opt:=U_{\max}^{(j)}$, and define
\begin{align*}
s_\ell &:=\alpha_\ell^{(j)}(u_\ell^{(j)}-U\opt), \qquad \ell\in[L].
\end{align*}
Recalling that
$u\jj(B) = \frac{\sum_{\ell \in B} 
\alpha_\ell^{(j)}u_\ell^{(j)}
}{1 + \sum_{\ell \in B} 
\alpha_\ell^{(j)}}$, we have
\begin{align}
\sum_{\ell \in B} s_\ell
& =
\left(1 + \sum_{k \in B} \alpha_k^{(j)}\right)
\bigl(u^{(j)}(B) - U\opt\bigr) + U\opt. 
\label{eqn:sum-B-in-U}
\end{align}
Let
\begin{align*}
        M\jj := \{B \in \mathcal{A}_\kappa : u^{(j)}(B) = U\opt\}
\end{align*} denote  the boundary-feasible assortments.
Since $u^{(j)}(B) \le U\opt$ for every $B \in \mathcal{A}_\kappa$,
from~\eqref{eqn:sum-B-in-U}, we have 
$\sum_{\ell \in B} s_\ell \le U\opt$ for any $B$ and $\sum_{\ell \in B} s_\ell = U\opt$ if and only if $u^{(j)}(B) = U\opt$.
Therefore,
\begin{align}
M\jj = \argmax_{B \in \mathcal{A}_\kappa} \sum_{\ell \in B} s_\ell.
\label{eqn:M-j-argmax}
\end{align}
Moreover, recall that
$q = |\{\ell\in[L]: s_\ell>0\}|$ is
the number of content families with strictly positive transformed score.
We will show that
\begin{align}
    M\jj=
\begin{cases}
\{H\cup E:E\subseteq T,\ |E|=c\}, & q\ge \kappa,\\[3pt]
\{P\cup E:E\subseteq Z,\ |E|\le h\}, & q<\kappa,
\end{cases} \label{eqn:M-j-expression}
\end{align}
with $H =  \{\ell\in[L]: s_\ell>\tau\},T = \{\ell\in[L]: s_\ell=\tau\},c = \kappa-|H|$ and $P = \{\ell\in[L]: s_\ell>0\},Z = \{\ell\in[L]: s_\ell=0\},h$ exactly as constructed by Algorithm~\ref{alg:mnl-sub-boundary}.
Here  $\tau$ is 
the $\kappa$-th largest value among  $\bm{s}$.

We now justify~\eqref{eqn:M-j-expression}. 
If $q\ge \kappa$, then there are at least $\kappa$ content families with positive score $s$. 
From~\eqref{eqn:M-j-argmax}, any
$B\in M\jj$ must satisfy $|B|=\kappa$; otherwise, if $|B|<\kappa$, we could add a content family with
positive score and strictly increase $\sum_{\ell\in B}s_\ell$, contradicting the optimality of
$B\in \argmax_{B\in\mathcal A_\kappa}\sum_{\ell\in B}s_\ell$. Moreover, every content family in
$H=\{\ell\in[L]: s_\ell>\tau\}$ must belong to $B$. Indeed, if some $h\in H$ were excluded from
$B$, then since $\tau$ is the $\kappa$-th largest score, the set $B$ would have to contain some
content family $e\notin H$ with $s_e\le \tau < s_h$; replacing $e$ by $h$ would strictly increase
$\sum_{\ell\in B}s_\ell$, again a contradiction. Therefore,
every maximizer has the form
$B=H\cup E$, where $E\subseteq T$ and $|E|=c$, with
$T=\{\ell\in[L]: s_\ell=\tau\}$ and $c=\kappa-|H|$.
On the other hand, if $q<\kappa$, then there are fewer than $\kappa$ content families with positive score. In this case,
every
$B\in M\jj$ must contain $P=\{\ell\in[L]: s_\ell>0\}$, since omitting any $\ell \in P$ would leave room
to add $\ell$ and thereby strictly increase the objective. After including all content families in $P$, any
additional selected content family must come from $Z=\{\ell\in[L]: s_\ell=0\}$, because adding a zero-score
content family
does not change the objective, whereas adding a negative-score content family would decrease it. Thus,
we do not need to include all elements of $Z$; rather, a maximizer may include any subset
$E\subseteq Z$ with $|E|\le h$, where $h=\kappa-|P|$. Therefore every maximizer has the form
$B=P\cup E$, where $E\subseteq Z$ and $|E|\le h$.

Since every deterministic assortment satisfies $u^{(j)}(B)\le U\opt$, any
feasible solution of the active boundary LP must be supported on $M\jj$. Because
the LP objective is linear in the assortment distribution, it is enough to maximize
$G^{(s,j)}(U\opt,t,B)$ 
over $B\in M\jj$.

Define
\[
\rho_\ell
:=
\lambda\jj F\jj(\sigma t)\gamma_\ell \I{\ell\notin S }
+
\eta_\ell\beta_\ell^{(s,j)} \I{\ell \in S}.
\]
Then for every deterministic assortment $B$,
\begin{align}
    G^{(s,j)}(U\opt,t,B)
=
\lambda\jj F\jj(\sigma t)U\opt
-
\frac{\sum_{\ell\in B}\alpha_\ell^{(j)}\rho_\ell}
{1+\sum_{\ell\in B}\alpha_\ell^{(j)}}. \label{eqn:G-exp-in-B-boundary}
\end{align}

We now treat the two cases.
\begin{enumerate}
\item \emph{Case 1: $q\ge \kappa$.} 
From~\eqref{eqn:M-j-expression}, every feasible boundary assortment in~\eqref{eqn:M-j-argmax}
has the form $B=H\cup E$ with $E\subseteq T$
and $|E|=c$. Let
\[
D_H:=1+\sum_{h\in H}\alpha_h^{(j)}, \qquad
M_H:=\frac{\sum_{h\in H}\alpha_h^{(j)}\rho_h}{D_H}.
\]
From~\eqref{eqn:G-exp-in-B-boundary}, 
direct algebra shows that

\begin{align}
    G^{(s,j)}(U\opt,t,H\cup E)
=
\lambda\jj F\jj(\sigma t)U\opt - M_H
+
\frac{\sum_{e\in E} a_e r_e}{1+\sum_{e\in E} a_e},
 \label{eqn:G-exp-in-E-H-boundary}
\end{align}
where
\[
a_e:=\frac{\alpha_e^{(j)}}{D_H}, \qquad r_e:=M_H-\rho_e.
\]
Therefore, maximizing $G^{(s,j)}(U\opt,t,B)$ over
$B\in M\jj$ is equivalent to solving
\[
\max_{E\subseteq T,\ |E|=c}
\frac{\sum_{e\in E} a_e r_e}{1+\sum_{e\in E} a_e}.
\]
By correctness of Algorithm~\ref{alg:exact-card-mnl} (Lemma~\ref{lem:exact-card-mnl}), it returns an optimal
$E\opt$.
Thus, $\Tilde{B}=H\cup E\opt$ is optimal over $M\jj$.
\item \emph{Case 2: $q<\kappa$.}
Every feasible boundary assortment has the form $B=P\cup E$ with $E\subseteq Z$
and $|E|\le h$. Let
\[
D_P:=1+\sum_{\ell \in P}\alpha_\ell^{(j)}, \qquad
M_P:=\frac{\sum_{\ell \in P}\alpha_\ell^{(j)}\rho_\ell}{D_P}.
\]
Similar to~\eqref{eqn:G-exp-in-E-H-boundary},
\[
G^{(s,j)}(U\opt,t,P\cup E)
=
\lambda\jj F\jj(\sigma t)U\opt - M_P
+
\frac{\sum_{e\in E} a_e r_e}{1+\sum_{e\in E} a_e},
\]
where
\[
a_e:=\frac{\alpha_e^{(j)}}{D_P}, \qquad r_e:=M_P-\rho_e.
\]
Thus, maximizing over $M\jj$ is exactly a standard cardinality-constrained MNL
problem on the optional set $Z$ with capacity $h$. Thus, the
Algorithm~\ref{alg:mnl-oracle} call returns an optimal $E\opt$, and
$\Tilde{B}=P\cup E\opt$ is optimal over $M\jj$.
\end{enumerate}

In both cases, Algorithm~\ref{alg:mnl-sub-boundary}
returns an assortment maximizing
$G^{(s,j)}(U\opt,t,B)$ over all $B\in\mathcal A_\kappa$ satisfying
$u^{(j)}(B)=U\opt$. 
Therefore,
 the returned assortment distribution, which places unit mass on
$\Tilde{B}$, is an optimal feasible solution of the active fixed-$(j,t)$
subscription LP.

For the running time, computing the scores and sorting them takes $O(L\log L)$.
In the branch $q<\kappa$, there is one Algorithm~\ref{alg:mnl-oracle} call on $|Z|$ content families.
In the branch $q\ge\kappa$, there is one Algorithm~\ref{alg:exact-card-mnl} call on $|T|$
content families, costing $O(|T|^2\log |T|)$. Since $|T|,|Z|\le L$, the claimed
$O(L^2 \log L)$
  bound follows.
\end{proof}

\subsection{Algorithm~\ref{alg:lp-postprocessing}}
\label{subsec:andrew-monotone-chain}

We now discuss Algorithm~\ref{alg:lp-postprocessing}, which is the
postprocessing routine used to solve the restricted LP over the candidate set
\(C\). The algorithm is a specialization of Andrew's monotone-chain convex
hull algorithm~\citep{Andrew79}, which takes
\(O(|C|\log |C|)\) time.
Algorithm~\ref{alg:lp-postprocessing} first sorts the actions by their
\(d\)-values and removes duplicates, keeping only the action with the largest
\(\phi\)-value for each repeated \(d\)-value. It then makes a single
left-to-right pass through this sorted list while maintaining a list \(H\).
Whenever the last three actions \(b,c,e\) in \(H\) violate the required
strictly decreasing slope condition, the middle action \(c\) is removed.
Intuitively, such a point \(c\) lies on or below the line segment between
\(b\) and \(e\), so it is never useful for maximizing the objective and can
be safely discarded. The global first and last actions in the sorted list are
never removed, since the algorithm only deletes middle elements of triples.
Therefore, because the input contains actions with $d$-values
on both sides of \(0\), the
final list \(H\) still crosses \(d=0\), so a feasible solution can always be
extracted. The decreasing-slope property ensures that \(H\) is exactly the
upper hull of the points \((d(a),\phi(a))\).

\begin{algorithm}[h]
\renewcommand{\thealgorithm}{\texttt{Postprocessing}}
\caption{Restricted-LP postprocessing}
\label{alg:lp-postprocessing}
\begin{algorithmic}[1]
\STATE \textbf{Input:}  Candidate set \(C\subseteq A\) containing actions \(a^-,a^+\in C\) such that
$
d(a^-)\le 0\le d(a^+).
$
\STATE For actions \(x,y\) with \(d(x)\neq d(y)\), define
$
s(x,y)
:=
\frac{\phi(y)-\phi(x)}{d(y)-d(x)}.
$

\STATE Remove duplicate \(d\)-values from \(C\): for each value \(r\), keep one action
$
a(r)\in \arg\max\{\phi(a):a\in C,\ d(a)=r\}.
$
Let \(a_1,\dots,a_n\) be the remaining actions sorted so that
$
d(a_1)<d(a_2)<\cdots<d(a_n).
$

\STATE Set \(H\gets ()\) and \(y^\star_a\gets 0\) for all \(a\in A\).

\FOR{\(i=1,\dots,n\)}
  \STATE Append \(a_i\) to \(H\).
  \WHILE{\(H\) has last three elements \(b,c,e\) satisfying \(s(b,c)\le s(c,e)\)}
    \STATE Remove \(c\) from \(H\).
  \ENDWHILE
\ENDFOR

\STATE Let \(H=(h_1,\dots,h_m)\).

\IF{there exists \(k\in\{1,\dots,m\}\) such that \(d(h_k)=0\)}
  \STATE Set
  $
  y^\star_{h_k}\gets 1,
  \qquad
   F\opt \gets \phi(h_k).
  $
\ELSE
  \STATE Let \(k\) be the unique index such that
  $
  d(h_k)<0<d(h_{k+1}).
  $
  \STATE Set
  $
  a^- \gets h_k,
  \qquad
  a^+ \gets h_{k+1}.
  $
  \STATE Define
  $
  \lambda^- \gets
  \frac{d(a^+)}{d(a^+)-d(a^-)},
  \qquad
  \lambda^+ \gets
  \frac{-d(a^-)}{d(a^+)-d(a^-)}.
  $
  \STATE  
  $
  y^\star_{a^-}\gets \lambda^-,
  \qquad
  y^\star_{a^+}\gets \lambda^+,
  $
  and
  $
  F\opt\gets
  \lambda^- \phi(a^-)+\lambda^+ \phi(a^+).
  $
\ENDIF

\RETURN \(( F\opt, y^\star)\)
\end{algorithmic}
\end{algorithm}

\section{Complexity analysis}
\label{sec:complexity}
In this section, we quantify the running time of
Algorithm~\ref{alg:mnl-solver-all}. We first give a generic bound for arbitrary
outer ratio grids and arbitrary bisection budgets. We then specialize the bound
to the theorem-driven grids and budgets used in
Theorem~\ref{thm:mnl-inner-separated-gap}. The same runtime bound applies to
the scaled instances in Theorem~\ref{thm:asymptotic-opt}, since market scaling
changes objective coefficients but not the worst-case complexity.

We keep the dependence on $J$, $K_T$, and $L$ explicit through
$\toracle(L)$. Constants depending only on fixed primitive magnitudes and on the
cardinality cap $\kappa$ are suppressed in the complexity expressions.

\paragraph{Complexity of Algorithm~\ref{alg:mnl-solver-all}.}
The next lemma records a generic upper bound under arbitrary ratio grids
$\{T\jj\}_{j=1}^J$ and arbitrary bisection budgets for the fixed-$(j,t)$
subsolvers.
Throughout this section, we use the convention $0 \log 0 := 0.$

\begin{lemma}
\label{lem:complexity-generic}
For each user type $j$, let $T^{(j,+)} := T\jj \cup \{\infty\}$.
For each   $t\in T\jj\cap (t_{\min}^{(j)},t_{\max}^{(j)})$,
recall that $K\adj(t)$ denotes the number of  bisection
iterations used by Algorithm~\ref{alg:mnl-fixed-jt-ads}.
For each branch index $\tau\in T^{(j,+)}$,   $K\subj(\tau)$ is the number of
bisection iterations used by
Algorithm~\ref{alg:mnl-fixed-jt-sub}.

After precomputing
$B^{(j, +)}\in \argmax_{B\in\mathcal A_\kappa} u^{(j)}(B)$ once for each type
$j$ for which $p\in(0,U_{\max}^{(j)})$ using Algorithm~\ref{alg:mnl-oracle}, the overall running time of
Algorithm~\ref{alg:mnl-solver-all} is
\begin{align}
T_{\mathrm{inner}}
&=
O\Biggl(
\sum_{j=1}^J
\Bigl[
( 1 + 
|T\jj|) 
\toracle(L)
\notag\\
&\qquad
+
\sum_{t\in T\jj\cap (t_{\min}^{(j)},t_{\max}^{(j)})}
\Bigl(
K\adj(t) \toracle(L)
+
K\adj(t) \log K\adj(t)
\Bigr)
\notag\\
&\qquad
+
\sum_{\tau\in T^{(j,+)}}
\Bigl(
K\subj(\tau) \toracle(L)
+
K\subj(\tau) \log K\subj(\tau)
\Bigr)
\Bigr]
\Biggr),
\label{eqn:complexity-all}
\end{align}
up to lower-order $O\!\left(J + \sum_{j=1}^J |T\jj|\right)$ aggregation costs.
Here $\toracle(L)$ denotes the cost of one call to
Algorithm~\ref{alg:mnl-oracle} or~\ref{alg:exact-card-mnl}. 
The term $(1+|T\jj|)\toracle(L)$ absorbs the
one-time precomputation and boundary or non-bisection branch evaluations.
\end{lemma}

\begin{proof}[Proof of Lemma~\ref{lem:complexity-generic}]
Fix a type $j$ and an interior  
$t\in T\jj\cap (t_{\min}^{(j)},t_{\max}^{(j)})$.
Each iteration of the  bisection routine in
Algorithm~\ref{alg:mnl-fixed-jt-ads} makes   one call to
Algorithm~\ref{alg:mnl-oracle}, so the ads branch contributes
$K\adj(t)$ oracle calls.
Its candidate set contains the initial certificate assortments, the empty
assortment, and at most one newly queried assortment per iteration. Thus, its
final size is $O(K\adj(t))$. Algorithm~\ref{alg:lp-postprocessing} costs
$O \left(K\adj(t) \log K\adj(t)\right)$ time.

For the interior subscription case $p\in(0,U_{\max}^{(j)})$, exactly the same
reasoning applies at any $\tau\in T^{(j,+)}$. Each  bisection
iteration in Algorithm~\ref{alg:mnl-fixed-jt-sub} contributes one oracle call,
so the subscription branch contributes $K\subj(\tau)$ oracle calls. The final
candidate set has size $O(K\subj(\tau)+1)$, and Algorithm~\ref{alg:lp-postprocessing} costs $O\left(K\subj(\tau) \log K\subj(\tau)\right)$ time. 
Non-bisection
subscription cases only reduce the postprocessing burden and are covered by the
same upper bound, while their oracle-scale boundary evaluations are absorbed by
$|T\jj|\toracle(L)$ in~\eqref{eqn:complexity-all}.

The one-time precomputation of $B^{(j, +)}$  in line~\ref{line:one-time-U-max-compute}
contributes at most one call to
Algorithm~\ref{alg:mnl-oracle} per type. Summing these contributions over all
$j\in[J]$ and all relevant branches proves~\eqref{eqn:complexity-all}. Finally,
selecting the best branch and aggregating over types costs only
$O\!\left(J+\sum_{j=1}^J |T\jj|\right)$ time.
\end{proof}

\medskip

We now specialize~\eqref{eqn:complexity-all} to the theorem-driven budgets of
Theorem~\ref{thm:mnl-inner-separated-gap}. For notational convenience, we write
$K^{(s,j)}$ for the subscription bisection budget defined by
\eqref{eqn:inner-grid-size-sub}; this quantity is independent of the ratio
index $r$.

\begin{lemma}
\label{lem:complexity-thm-callcount}
Fix $K_T\ge 2$ and run Algorithm~\ref{alg:mnl-solver-all} with the  outer ratio grids from
Theorem~\ref{thm:mnl-inner-separated-gap}, namely
$|T\jj| \le K_T$ for every $j$, ads bisection budgets $K\adj(t\jj_r)$ from
\eqref{eqn:inner-grid-size-ads}, and subscription bisection budgets
$K^{(s,j)}$ from \eqref{eqn:inner-grid-size-sub}.
Then the total number of calls to 
Algorithms~\ref{alg:mnl-oracle}
and~\ref{alg:exact-card-mnl} (which have the same complexity)
satisfies
\begin{align}
\noracle(K_T)
=
O(JK_T)
+
\sum_{j=1}^J
\left[
\sum_{r=2}^{K_T-1} K\adj(t\jj_r)
+
(K_T+1)K^{(s,j)}
\right],
\label{eqn:complexity-thm-callcount}
\end{align}
where $K\adj(t\jj_r)$ 
is defined in~\eqref{eqn:inner-grid-size-ads}
and $K^{(s,j)}$
is defined in~\eqref{eqn:inner-grid-size-sub}.
\end{lemma}

\begin{proof}[Proof of Lemma~\ref{lem:complexity-thm-callcount}]
For the ads branch, Algorithm~\ref{alg:mnl-fixed-jt-ads} makes one oracle call
per bisection iteration at each interior grid point
$t_r^{(j)}\in\{t_2^{(j)},\dots,t_{K_T-1}^{(j)}\}$, and no bisection iterations
at the two endpoints. Thus, the interior ads branches contribute at most
$\sum_{r=2}^{K_T-1} K\adj(t\jj_r)$ oracle calls for type $j$.
In addition, the two endpoint boundary branches contribute only $O(1)$ oracle
calls per type, see Algorithm~\ref{alg:mnl-ads-boundary}.
For a degenerate type, there are no interior ad grid points and no ad bisection.

For the subscription branch, Algorithm~\ref{alg:mnl-fixed-jt-sub} is evaluated
at each point
 of $T\jj$
and once more on the off branch
$\tau=\infty$. After the one-time precomputation of $B^{(j, +)}$, the
interior case $p\in(0,U_{\max}^{(j)})$ contributes at most $K^{(s,j)}$ oracle
calls from the bisection loop at each invocation. Thus, these subscription
branches contribute $(K_T+1)K^{(s,j)}$ oracle calls for type $j$. If instead
$p=U_{\max}^{(j)}$, the boundary solver is invoked at each branch; its total
additional oracle contribution is at most $O(K_T)$ for type $j$ and is absorbed
into the additive $O(JK_T)$ term.

Summing over $j=1,\dots,J$ proves~\eqref{eqn:complexity-thm-callcount}.
\end{proof}

\begin{lemma}
\label{lem:complexity-thm-callcount-order}
Under the   outer grids and bisection budgets used in
Theorem~\ref{thm:mnl-inner-separated-gap},
\begin{align}
\noracle(K_T)=O(JK_T\log K_T).
\label{eqn:complexity-thm-callcount-theta}
\end{align}
Moreover,
\begin{align}
T_{\mathrm{inner}}
=
O\!\left(
JK_T\log K_T\, \toracle(L)
+
JK_T(\log K_T \log  \log K_T) 
\right).
\label{eqn:complexity-thm-runtime}
\end{align}
With the direct implementation $\toracle(L)=O(L^2\log L)$, this gives
\[
T_{\mathrm{inner}}=\widetilde O(JK_TL^2).
\]
\end{lemma}

\begin{proof}[Proof of Lemma~\ref{lem:complexity-thm-callcount-order}]
Fix a type $j$ and abbreviate $K:=K_T$. By
\eqref{eqn:inner-grid-size-ads}, for each interior index
$r=2,\dots,K-1$, all primitive quantities appearing in the ads bisection budget
are fixed constants, while $r-1\ge 1$ and $K-r\ge 1$. Thus, the argument of the
logarithm in \eqref{eqn:inner-grid-size-ads} is at most $O(K^2)$, uniformly over
$r$. Therefore
\[
K\adj(t\jj_r)=O(\log K)
\qquad\text{uniformly over }r=2,\dots,K-1.
\]
Summing over the $K-2$ interior grid points gives
\[
\sum_{r=2}^{K-1} K\adj(t\jj_r)=O(K\log K).
\]

For the subscription branch, \eqref{eqn:inner-grid-size-sub} gives
\[
K^{(s,j)} = O(\log K),
\]
with $K^{(s,j)}=0$ outside the interior subscription case
$p\in(0,U_{\max}^{(j)})$. 
Thus,
\[
(K+1)K^{(s,j)}=O(K\log K).
\]
Combining these two bounds with
Lemma~\ref{lem:complexity-thm-callcount} and summing over $j=1,\dots,J$ proves
\eqref{eqn:complexity-thm-callcount-theta}.

It remains to consider postprocessing. Under the theorem-driven budgets, every
candidate set generated by a bisection branch has size $O(\log K_T)$.  
Algorithm~\ref{alg:lp-postprocessing}  costs $O(\log K_T \log \log K_T)$. There are $O(JK_T)$ ads and
subscription branch evaluations, including the off branch, so the total
postprocessing contribution is
\[
O\!\left(JK_T\log K_T \log \log K_T\right).
\]
Substituting the oracle-call bound and this postprocessing bound into
Lemma~\ref{lem:complexity-generic} proves~\eqref{eqn:complexity-thm-runtime}.
The final displayed bound follows from $\toracle(L)=O(L^2\log L)$ and the
convention that $\widetilde O(\cdot)$ hides polylogarithmic factors.
\end{proof}

\section{Exact Explicit-Assortment Mixed-Integer Formulation under Uniform Ad Tolerance}
\label{sec:MIP}

We specialize to the case in which, for every user type $j \in [J]$, the ad-tolerance variable satisfies
$\chi\jj \sim \mathrm{Unif}[a\jj,b\jj]$ with $0<a\jj<b\jj$, and $F\jj$ denotes the corresponding cdf.
After explicit enumeration of the feasible assortment family,
\[
\mathcal{A}_\kappa := \{A \subseteq [L] : |A| \le \kappa\},
\]
the formulation below is an exact explicit-assortment mixed-integer nonconvex quadratic formulation of
problem~\eqref{eqn:problem-buy-or-rent}. The implications used below are compact notation for the corresponding
disjunctive constraints. Under the uniform specification, it is cleaner to eliminate the cutoff variable $h\jj$
and write the ad-choice block directly in the variables $\rho\jj := \sigma x^{(a,j)}$ and
$w_{+}\jj  := \bigl(u^{(a,j)}-(u^{(s,j)}-p)_+\bigr)_+$.  

For each user type $j \in [J]$, assortment $A \in \mathcal{A}_\kappa$, and content family $\ell \in [L]$, precompute
\[
\bar{x}_{jA} := x^{(j)}(A),
\qquad
\bar{u}_{jA} := u^{(j)}(A),
\qquad
\bar{p}_{jA\ell} := p^{(j)}_\ell(A).
\]
Then the induced quantities in~\eqref{eqn:x-u-def} and the content family flows in~\eqref{eqn:y-flow-def}
become linear in the assortment-mixture variables.

\subsection{Decision Variables}

\textbf{Assortment-mixture variables.}
For each $j \in [J]$ and $A \in \mathcal{A}_\kappa$,
\[
y_A^{(a,j)} \ge 0,
\qquad
y_A^{(s,j)} \ge 0.
\]

\textbf{Aggregate choice, utility, and flow variables.}
For each $j \in [J]$,
\[
x^{(a,j)} \ge 0,
\qquad
u^{(a,j)} \ge 0,
\qquad
u^{(s,j)} \ge 0,
\qquad
\rho\jj \ge 0,
\]
and for each $j \in [J]$ and $\ell \in [L]$,
\[
f_\ell^{(a,j)} \ge 0,
\qquad
f_\ell^{(s,j)} \ge 0.
\]

\textbf{Mode-choice variables.}
For each $j \in [J]$,
\[
0 \le P^{(a,j)} \le 1,
\qquad
0 \le P^{(s,j)} \le 1.
\]

\textbf{Positive-part and regime variables.}
For each $j \in [J]$,
\[
\iota\jj,\ s\jj,\ d^{(j, L)},\ d^{(j, I)},\ d^{(j, R)} \in \{0,1\},
\]
\[
w_{+,1}\jj ,\ w_{-,1}\jj ,\ w_{+}\jj ,\ w_{-}\jj  \ge 0,
\qquad
c\jj \ge 0,
\qquad
v\jj \ge 0.
\]
Here $\iota\jj$ indicates whether the subscription utility reaches the price threshold, $s\jj$ indicates whether the
ad branch is active, and $d^{(j, L)},d^{(j, I)},d^{(j, R)}$ select the left, interior, and right regions of the uniform
truncation when the ad branch is active.

\textbf{Auxiliary product variables.}
For each $j \in [J]$,
\[
0 \le q\jj \le \sigma,
\]
and for each $j \in [J]$ and $\ell \in [L]$,
\[
0 \le g_\ell^{(a,j)} \le 1,
\qquad
0 \le g_\ell^{(s,j)} \le 1.
\]
We use $q\jj$ for $\sigma x^{(a,j)} P^{(a,j)}$, and $g_\ell^{(a,j)}$, $g_\ell^{(s,j)}$ for
$P^{(a,j)} f_\ell^{(a,j)}$ and $P^{(s,j)} f_\ell^{(s,j)}$, respectively.

\textbf{Procurement variables.}
For each $\ell \in [L]$,
\[
0 \le z_\ell \le 1,
\qquad
b^{\text{buy}}_\ell \ge 0,
\qquad
b^{\text{rental}}_\ell \ge 0,
\qquad
e_\ell \ge 0,
\qquad
\psi_\ell \in \{0,1\}.
\]

\subsection{Objective Function}

The fixed-$(\sigma,p)$ objective in~\eqref{eqn:problem-buy-or-rent} becomes
\[
\max \; \sum_{j \in [J]} \lambda\jj \bigl(\rads 
q\jj + p P^{(s,j)}\bigr)
- \sum_{\ell \in [L]} e_\ell.
\]

\subsection{Valid Upper Bounds}

Let
\[
M\jj := U^{(j)}_{\max} + p,
\qquad
Q_\ell := \max\!\Bigl\{\eta_\ell,\ \gamma_\ell \sum_{j \in [J]} \lambda\jj\Bigr\}.
\]
The bound $M\jj$ is valid for all four positive-part auxiliaries because
$u^{(a,j)},u^{(s,j)} \in [0,U^{(j)}_{\max}]$. The bound $Q_\ell$ is valid for both $b^{\text{buy}}_\ell$ and $b^{\text{rental}}_\ell$ because
$z_\ell \le 1$ and, for every $j \in [J]$,
\[
\begin{aligned}
g_\ell^{(a,j)} + g_\ell^{(s,j)}
&= P^{(a,j)}f_\ell^{(a,j)} + P^{(s,j)}f_\ell^{(s,j)} \\
&\le P^{(a,j)} + P^{(s,j)} \le 1.
\end{aligned}
\]
Consequently,
\[
0 \le b^{\text{rental}}_\ell = \gamma_\ell \sum_{j \in [J]} \lambda\jj \bigl(g_\ell^{(a,j)} + g_\ell^{(s,j)}\bigr)
\le \gamma_\ell \sum_{j \in [J]} \lambda\jj.
\]

\subsection{Constraint Definitions}

\paragraph{Assortment-mixture simplex constraints.}
For every $j \in [J]$,
\begin{align*}
\sum_{A \in \mathcal{A}_\kappa} y_A^{(a,j)} &= 1, \\
\sum_{A \in \mathcal{A}_\kappa} y_A^{(s,j)} &= 1.
\end{align*}

\paragraph{Linear aggregation identities.}
The explicit-assortment versions of~\eqref{eqn:x-u-def} and~\eqref{eqn:y-flow-def} are
\begin{align*}
x^{(a,j)} &= \sum_{A \in \mathcal{A}_\kappa} \bar{x}_{jA} \, y_A^{(a,j)}, && \forall j \in [J], \\
\rho\jj &= \sigma x^{(a,j)}, && \forall j \in [J], \\
u^{(a,j)} &= \sum_{A \in \mathcal{A}_\kappa} \bar{u}_{jA} \, y_A^{(a,j)}, && \forall j \in [J], \\
u^{(s,j)} &= \sum_{A \in \mathcal{A}_\kappa} \bar{u}_{jA} \, y_A^{(s,j)}, && \forall j \in [J], \\
f_\ell^{(a,j)} &= \sum_{A \in \mathcal{A}_\kappa : \ell \in A} \bar{p}_{jA\ell} \, y_A^{(a,j)}, && \forall j \in [J],\ \forall \ell \in [L], \\
f_\ell^{(s,j)} &= \sum_{A \in \mathcal{A}_\kappa : \ell \in A} \bar{p}_{jA\ell} \, y_A^{(s,j)}, && \forall j \in [J],\ \forall \ell \in [L].
\end{align*}

\paragraph{Exact subscription-threshold block.}
We first encode
$w_{+,1}\jj  = (u^{(s,j)}-p)_+$
and use an auxiliary   variable to enforce the convention
$\iota\jj = 1$ if and only if $u^{(s,j)} \ge p$.
For every $j \in [J]$,
\begin{align*}
u^{(s,j)} - p &= w_{+,1}\jj  - w_{-,1}\jj , \\
0 \le w_{+,1}\jj  &\le M\jj \iota\jj, \\
0 \le w_{-,1}\jj  &\le M\jj (1-\iota\jj), \\
w_{-,1}\jj  \, c\jj &= 1-\iota\jj.
\end{align*}
If $\iota\jj=1$, then $w_{-,1}\jj =0$ and hence $u^{(s,j)}-p = w_{+,1}\jj  \ge 0$.
If $\iota\jj=0$, then $w_{-,1}\jj c\jj = 1$ forces $w_{-,1}\jj >0$, so
$u^{(s,j)}-p = -w_{-,1}\jj  < 0$.
Therefore,
the block is exact and already resolves the tie $u^{(s,j)}=p$ in favor of subscription.

\paragraph{Exact activation block for the ad branch.}
Next encode
$w_{+}\jj  = \bigl(u^{(a,j)}-w_{+,1}\jj \bigr)_+$
and use a second witness variable to distinguish exactly between the active and inactive ad branches.
For every $j \in [J]$,
\begin{align*}
u^{(a,j)} - w_{+,1}\jj  &= w_{+}\jj  - w_{-}\jj , \\
0 \le w_{+}\jj  &\le M\jj s\jj, \\
0 \le w_{-}\jj  &\le M\jj (1-s\jj), \\
w_{+}\jj  \, v\jj &= s\jj.
\end{align*}
If $s\jj=1$, then $w_{+}\jj v\jj=1$ forces $w_{+}\jj >0$, so the ad branch is active.
If $s\jj=0$, then $w_{+}\jj =0$ and
$u^{(a,j)}-w_{+,1}\jj  = -w_{-}\jj  \le 0$,
which is exactly the inactive branch.

\paragraph{Exact ad-choice block when $\chi\jj \sim \mathrm{Unif}[a\jj,b\jj]$.}
The binaries $d^{(j, L)}, d^{(j, I)}, d^{(j, R)}$ select the left, interior, and right regions of the uniform truncation.
For every $j \in [J]$,
\begin{align*}
d^{(j, L)} + d^{(j, I)} + d^{(j, R)} &= s\jj, \\
0 \le P^{(a,j)} &\le s\jj.
\end{align*}
Then impose the following branch constraints:
\begin{align*}
d^{(j, L)} = 1 &\Rightarrow \rho\jj \le a\jj w_{+}\jj , \qquad P^{(a,j)} = 1, \\
d^{(j, I)} = 1 &\Rightarrow a\jj w_{+}\jj  \le \rho\jj \le b\jj w_{+}\jj , \qquad
(b\jj-a\jj) w_{+}\jj  P^{(a,j)} + \rho\jj = b\jj w_{+}\jj , \\
d^{(j, R)} = 1 &\Rightarrow \rho\jj \ge b\jj w_{+}\jj , \qquad P^{(a,j)} = 0.
\end{align*}
This is the exact graph of the ad-probability map under the original model. In particular, the off branch
correctly leaves $\rho\jj$ unrestricted; indeed $w_{+}\jj =0$ does not force $x^{(a,j)}=0$ in the original model.
\begin{itemize}
  \item if $s\jj=0$, then $w_{+}\jj =0$ and $P^{(a,j)}=0$;
  \item if $s\jj=1$ and $\rho\jj \le a\jj w_{+}\jj $, then $P^{(a,j)}=1$;
  \item if $s\jj=1$ and $a\jj w_{+}\jj  \le \rho\jj \le b\jj w_{+}\jj $, then
  $P^{(a,j)} = \dfrac{b\jj - \rho\jj/w_{+}\jj }{b\jj-a\jj}$;
  \item if $s\jj=1$ and $\rho\jj \ge b\jj w_{+}\jj $, then $P^{(a,j)}=0$.
\end{itemize}
At the two boundaries $\rho\jj = a\jj w_{+}\jj $ and $\rho\jj = b\jj w_{+}\jj $, two adjacent regimes may both be
feasible, but they imply the same ad probability ($1$ at the left boundary and $0$ at the right boundary),
so exactness is unaffected. Hence no cutoff variable and no activation tolerance are needed in the uniform case.

\paragraph{Exact subscription probability and product definitions.}
For every $j \in [J]$ and $\ell \in [L]$,
\begin{align*}
P^{(s,j)} &= \iota\jj \bigl(1-P^{(a,j)}\bigr), \\
q\jj &= \rho\jj P^{(a,j)}, \\
g_\ell^{(a,j)} &= P^{(a,j)} f_\ell^{(a,j)}, \\
g_\ell^{(s,j)} &= P^{(s,j)} f_\ell^{(s,j)}.
\end{align*}
The equality for $P^{(s,j)}$ is exact because, under the uniform specification,
$F\jj(h\jj) = 1-P^{(a,j)}$ on the active branch and equals $1$ on the inactive branch.

\paragraph{Procurement block.}
For every $\ell \in [L]$,
\begin{align*}
z_\ell &\ge f_\ell^{(a,j)}, && \forall j \in [J], \\
z_\ell &\ge f_\ell^{(s,j)}, && \forall j \in [J], \\
b^{\text{buy}}_\ell &= \eta_\ell z_\ell, \\
b^{\text{rental}}_\ell &= \gamma_\ell \sum_{j \in [J]} \lambda\jj \bigl(g_\ell^{(a,j)} + g_\ell^{(s,j)}\bigr).
\end{align*}
Thus,
$z_\ell \ge \max_{j \in [J],\, m \in \{a,s\}} f_\ell^{(m,j)}$.
Since the objective weakly decreases in $z_\ell$, there exists an optimal solution
$(\bm{z}\opt, \bm{f}\opt)$ satisfying
\[
z\opt_\ell = \max_{j \in [J],\,m \in \{a,s\}} f_\ell^{\star,(m,j)},
\qquad \forall \ell \in [L].
\]
To impose $e_\ell = \min\{b^{\text{buy}}_\ell,b^{\text{rental}}_\ell\}$ exactly, use
\begin{align*}
e_\ell &\le b^{\text{buy}}_\ell, \\
e_\ell &\le b^{\text{rental}}_\ell, \\
e_\ell &\ge b^{\text{buy}}_\ell - Q_\ell (1-\psi_\ell), \\
e_\ell &\ge b^{\text{rental}}_\ell - Q_\ell \psi_\ell, \\
\psi_\ell &\in \{0,1\}.
\end{align*}
These constraints impose the procurement minimum exactly. If $\psi_\ell=1$, they imply
\[
e_\ell=b^{\text{buy}}_\ell \le b^{\text{rental}}_\ell,
\]
whereas if $\psi_\ell=0$, they imply
\[
e_\ell=b^{\text{rental}}_\ell \le b^{\text{buy}}_\ell.
\]
Hence,
\[
e_\ell=\min\{b^{\text{buy}}_\ell,b^{\text{rental}}_\ell\},
\]
which is exactly the family-level buy-versus-rent procurement term in~\eqref{eqn:cost-content}.

\subsection{Specialization to the Fixed-Buy-Set Problem}

If one wishes to solve the fixed-buy-set problem~\eqref{eqn:R-S-problem-def} for a given set
$S \subseteq [L]$, first remove the four inequalities used to impose
\[
e_\ell=\min\{b^{\text{buy}}_\ell,b^{\text{rental}}_\ell\}
\]
and drop the binary variable $\psi_\ell$

for every $\ell \in [L]$. One may then retain
$e_\ell$, $b^{\text{buy}}_\ell$, and $b^{\text{rental}}_\ell$, together with their defining equations, and impose
\begin{align*}
e_\ell &= b^{\text{buy}}_\ell, && \forall \ell \in S, \\
e_\ell &= b^{\text{rental}}_\ell, && \forall \ell \in [L] \setminus S.
\end{align*}
Alternatively, one may eliminate $e_\ell$, $b^{\text{buy}}_\ell$, and $b^{\text{rental}}_\ell$ and replace
$\sum_{\ell \in [L]}e_\ell$ in the objective directly by
\[
\sum_{\ell \in S} \eta_\ell z_\ell
+
\sum_{\ell \in [L]\setminus S} \gamma_\ell
\sum_{j \in [J]} \lambda\jj \bigl(g_\ell^{(a,j)} + g_\ell^{(s,j)}\bigr),
\]
which is exactly the fixed-$S$ objective in~\eqref{eqn:R-S-def}.

\section{Additional Experimental Details}
\label{app:additional-experimental-details}

This appendix reports additional details for the numerical experiments in
Section~\ref{sec:exp-convergence} and
Section~\ref{sec:preference-concentration}. 
We first give the   data details 
of
the baseline instance used in Section~\ref{sec:exp-convergence}. 
We then report the returned
assortment distributions for the market-scaling and preference-concentration
experiments. Each 
table entry lists the assortments in the support of the
returned randomized assortment and the corresponding mixing probabilities.

\subsection{Data details}

Table~\ref{tab:appendix-baseline-alpha} reports the baseline attraction matrix.
Table~\ref{tab:appendix-baseline-utility} repeats the same matrix as the utility
parameters because the numerical study sets $u_\ell^{(j)}=\alpha_\ell^{(j)}$.
Table~\ref{tab:appendix-baseline-costs} reports the rental and buy costs, and
Table~\ref{tab:appendix-preference-targets} reports the high-attraction targets used
for the moving shared families in the preference-concentration experiment.

\begin{table}[t]
\centering
\scriptsize
\resizebox{\textwidth}{!}{%
\begin{tabular}{c cccccccccc}
\hline
Type \(j\) & 1 & 2 & 3 & 4 & 5 & 6 & 7 & 8 & 9 & 10 \\
\hline
1 & 3.11 & 0.79 & 0.63 & 0.79 & 0.66 & 1.28 & 1.40 & 1.27 & 1.31 & 1.16 \\
2 & 0.75 & 3.12 & 0.67 & 0.76 & 0.66 & 1.29 & 1.19 & 1.27 & 1.21 & 1.23 \\
3 & 0.75 & 0.66 & 3.09 & 0.80 & 0.79 & 1.37 & 1.31 & 1.23 & 1.20 & 1.44 \\
4 & 0.70 & 0.62 & 0.72 & 3.27 & 0.72 & 1.43 & 1.16 & 1.31 & 1.29 & 1.17 \\
5 & 0.73 & 0.77 & 0.72 & 0.65 & 3.30 & 1.30 & 1.30 & 1.38 & 1.19 & 1.40 \\
\hline
\end{tabular}%
}
\caption{Baseline attraction parameters \(\alpha_{\ell}^{(j)}\).}
\label{tab:appendix-baseline-alpha}
\end{table}

\begin{table}[t]
\centering
\scriptsize
\resizebox{\textwidth}{!}{%
\begin{tabular}{c cccccccccc}
\hline
Type \(j\) & 1 & 2 & 3 & 4 & 5 & 6 & 7 & 8 & 9 & 10 \\
\hline
1 & 3.11 & 0.79 & 0.63 & 0.79 & 0.66 & 1.28 & 1.40 & 1.27 & 1.31 & 1.16 \\
2 & 0.75 & 3.12 & 0.67 & 0.76 & 0.66 & 1.29 & 1.19 & 1.27 & 1.21 & 1.23 \\
3 & 0.75 & 0.66 & 3.09 & 0.80 & 0.79 & 1.37 & 1.31 & 1.23 & 1.20 & 1.44 \\
4 & 0.70 & 0.62 & 0.72 & 3.27 & 0.72 & 1.43 & 1.16 & 1.31 & 1.29 & 1.17 \\
5 & 0.73 & 0.77 & 0.72 & 0.65 & 3.30 & 1.30 & 1.30 & 1.38 & 1.19 & 1.40 \\
\hline
\end{tabular}%
}
\caption{Baseline utility parameters \(u_{\ell}^{(j)}\). Under the numerical
specification \(u_{\ell}^{(j)}=\alpha_{\ell}^{(j)}\), this matrix is identical to
Table~\ref{tab:appendix-baseline-alpha}.}
\label{tab:appendix-baseline-utility}
\end{table}

\begin{table}[t]
\centering
\scriptsize
\resizebox{\textwidth}{!}{%
\begin{tabular}{c l c c c}
\hline
Content family \(\ell\) & Role & Rental cost \(\gamma_\ell\) & Buy multiplier \(\eta_\ell/\gamma_\ell\) & Buy cost \(\eta_\ell\) \\
\hline
1  & Niche for type 1       & 0.85 & 5.50 & 4.68 \\
2  & Niche for type 2       & 0.83 & 5.50 & 4.56 \\
3  & Niche for type 3       & 0.82 & 5.50 & 4.51 \\
4  & Niche for type 4       & 0.88 & 5.50 & 4.84 \\
5  & Niche for type 5       & 0.86 & 5.50 & 4.73 \\
6  & Reusable content family & 0.93 & 3.00 & 2.79 \\
7  & Reusable content family & 0.89 & 3.00 & 2.67 \\
8  & Reusable content family & 0.90 & 3.00 & 2.70 \\
9  & Reusable content family & 0.87 & 3.00 & 2.61 \\
10 & Reusable content family & 0.90 & 2.30 & 2.07 \\
\hline
\end{tabular}%
}
\caption{Baseline procurement costs.}
\label{tab:appendix-baseline-costs}
\end{table}

\begin{table}[t]
\centering 
\begin{tabular}{c c c}
\hline
Content family \(\ell\) & Shared portfolio & Target attraction \(h_\ell\) \\
\hline
6 & \(H_1\) & 3.21 \\
7 & \(H_1\) & 3.27 \\
8 & \(H_2\) & 2.91 \\
9 & \(H_2\) & 3.28 \\
\hline
\end{tabular}
\caption{High-attraction targets used in the preference-concentration experiment.}
\label{tab:appendix-preference-targets}
\end{table}

\subsection{Solution details}

We now report the returned assortment distributions. 
Table~\ref{tab:appendix-market-scaling-distributions} reports the market-scaling
experiment from Section~\ref{sec:exp-convergence}. 
Table~\ref{tab:appendix-preference-concentration-distributions}
reports the preference-concentration
experiment from Section~\ref{sec:preference-concentration}. 
Within each mode and type, support assortments are ordered by their probability
mass from largest to smallest.
In addition, 
Tables~\ref{tab:appendix-market-scaling-rental-flow} and~\ref{tab:appendix-preference-concentration-rental-flow} report the corresponding family-level realized rental flows.

\begin{table}[t]
\centering
\scriptsize
\resizebox{\textwidth}{!}{%
\begin{tabular}{c c c l l}
\hline
\(N\) & Type \(j\) & \(S_N\opt\) & Ad-mode assortment & Subscription-mode assortment \\
\hline
1 & 1 & \(\{10\}\) & \(\{1,7,10\}\): 0.875; \(\{1,10\}\): 0.125 & \(\{1\}\): 0.765; \(\emptyset\): 0.235 \\
1 & 2 & \(\{10\}\) & \(\{2,9,10\}\): 0.910; \(\{2,10\}\): 0.090 & \(\{2\}\): 0.762; \(\emptyset\): 0.238 \\
1 & 3 & \(\{10\}\) & \(\{3,7,10\}\): 0.951; \(\{3,10\}\): 0.049 & \(\{3\}\): 0.771; \(\emptyset\): 0.229 \\
1 & 4 & \(\{10\}\) & \(\{4,9,10\}\): 0.996; \(\{8,9,10\}\): 0.004 & \(\{4\}\): 0.719; \(\emptyset\): 0.281 \\
1 & 5 & \(\{10\}\) & \(\{5,8,10\}\): 0.942; \(\{5,10\}\): 0.058 & \(\{5\}\): 0.711; \(\emptyset\): 0.289 \\
\hline
2 & 1 & \(\{6,7,8,9,10\}\) & \(\{1,7,9\}\): 0.976; \(\{6,7,9\}\): 0.024 & \(\{1\}\): 0.765; \(\emptyset\): 0.235 \\
2 & 2 & \(\{6,7,8,9,10\}\) & \(\{2,6,10\}\): 0.973; \(\{6,8,10\}\): 0.027 & \(\{2\}\): 0.762; \(\emptyset\): 0.238 \\
2 & 3 & \(\{6,7,8,9,10\}\) & \(\{3,6,10\}\): 0.965; \(\{3,10\}\): 0.035 & \(\{3\}\): 0.771; \(\emptyset\): 0.229 \\
2 & 4 & \(\{6,7,8,9,10\}\) & \(\{4,6,8\}\): 0.986; \(\{6,8,9\}\): 0.014 & \(\{4\}\): 0.719; \(\emptyset\): 0.281 \\
2 & 5 & \(\{6,7,8,9,10\}\) & \(\{5,8,10\}\): 0.961; \(\{7,8,10\}\): 0.039 & \(\{5\}\): 0.711; \(\emptyset\): 0.289 \\
\hline
5 & 1 & \(\{1,\ldots,10\}\) & \(\{1,7,9\}\): 0.890; \(\{1,7\}\): 0.110 & \(\{1\}\): 0.765; \(\emptyset\): 0.235 \\
5 & 2 & \(\{1,\ldots,10\}\) & \(\{2,6,8\}\): 0.924; \(\{2,6\}\): 0.076 & \(\{2\}\): 0.762; \(\emptyset\): 0.238 \\
5 & 3 & \(\{1,\ldots,10\}\) & \(\{3,6,10\}\): 0.965; \(\{3,10\}\): 0.035 & \(\{3\}\): 0.771; \(\emptyset\): 0.229 \\
5 & 4 & \(\{1,\ldots,10\}\) & \(\{4,6,8\}\): 0.986; \(\{6,8,9\}\): 0.014 & \(\{4\}\): 0.719; \(\emptyset\): 0.281 \\
5 & 5 & \(\{1,\ldots,10\}\) & \(\{5,8,10\}\): 0.942; \(\{5,10\}\): 0.058 & \(\{5\}\): 0.711; \(\emptyset\): 0.289 \\
\hline
10 & 1 & \(\{1,\ldots,10\}\) & \(\{1,7,9\}\): 0.890; \(\{1,7\}\): 0.110 & \(\{1\}\): 0.765; \(\emptyset\): 0.235 \\
10 & 2 & \(\{1,\ldots,10\}\) & \(\{2,6,8\}\): 0.971; \(\{2\}\): 0.029 & \(\{2\}\): 0.762; \(\emptyset\): 0.238 \\
10 & 3 & \(\{1,\ldots,10\}\) & \(\{3,6,10\}\): 0.965; \(\{3,10\}\): 0.035 & \(\{3\}\): 0.771; \(\emptyset\): 0.229 \\
10 & 4 & \(\{1,\ldots,10\}\) & \(\{4,6,8\}\): 0.986; \(\{6,8,9\}\): 0.014 & \(\{4\}\): 0.719; \(\emptyset\): 0.281 \\
10 & 5 & \(\{1,\ldots,10\}\) & \(\{5,8,10\}\): 0.942; \(\{5,10\}\): 0.058 & \(\{5\}\): 0.711; \(\emptyset\): 0.289 \\
\hline
100 & 1 & \(\{1,\ldots,10\}\) & \(\{1,7,9\}\): 0.890; \(\{1,7\}\): 0.110 & \(\{1\}\): 0.765; \(\emptyset\): 0.235 \\
100 & 2 & \(\{1,\ldots,10\}\) & \(\{2,6,8\}\): 0.971; \(\{2\}\): 0.029 & \(\{2\}\): 0.762; \(\emptyset\): 0.238 \\
100 & 3 & \(\{1,\ldots,10\}\) & \(\{3,6,10\}\): 0.965; \(\{3,10\}\): 0.035 & \(\{3\}\): 0.771; \(\emptyset\): 0.229 \\
100 & 4 & \(\{1,\ldots,10\}\) & \(\{4,6,8\}\): 0.986; \(\{6,8,9\}\): 0.014 & \(\{4\}\): 0.719; \(\emptyset\): 0.281 \\
100 & 5 & \(\{1,\ldots,10\}\) & \(\{5,8,10\}\): 0.978; \(\{5\}\): 0.022 & \(\{5\}\): 0.711; \(\emptyset\): 0.289 \\
\hline
\end{tabular}}
\caption{Returned assortment distributions corresponding to
Table~\ref{tab:sec51-main} at \(K_T=129\) under the numerical specification
\(u_{\ell}^{(j)}=\alpha_{\ell}^{(j)}\). Within each mode and type, support
assortments are ordered from largest to smallest probability mass.}
\label{tab:appendix-market-scaling-distributions}
\end{table}

\begin{table}[t]
\centering
\scriptsize
\resizebox{\textwidth}{!}{%
\begin{tabular}{c c c c l l}
\hline
\(\rho\) & Type \(j\) & Cluster & \(S\opt(\rho)\) & Ad-mode assortment & Subscription-mode support \\
\hline
0.0 & 1 & \(C_1\) & \(\{10\}\) & \(\{1,7,10\}\): 0.875; \(\{1,10\}\): 0.125 & \(\{1\}\): 0.765; \(\emptyset\): 0.235 \\
0.0 & 2 & \(C_1\) & \(\{10\}\) & \(\{2,9,10\}\): 0.910; \(\{2,10\}\): 0.090 & \(\{2\}\): 0.762; \(\emptyset\): 0.238 \\
0.0 & 3 & \(C_1\) & \(\{10\}\) & \(\{3,7,10\}\): 0.951; \(\{3,10\}\): 0.049 & \(\{3\}\): 0.771; \(\emptyset\): 0.229 \\
0.0 & 4 & \(C_2\) & \(\{10\}\) & \(\{4,9,10\}\): 0.996; \(\{8,9,10\}\): 0.004 & \(\{4\}\): 0.719; \(\emptyset\): 0.281 \\
0.0 & 5 & \(C_2\) & \(\{10\}\) & \(\{5,8,10\}\): 0.942; \(\{5,10\}\): 0.058 & \(\{5\}\): 0.711; \(\emptyset\): 0.289 \\
\hline
0.1 & 1 & \(C_1\) & \(\{10\}\) & \(\{1,7,10\}\): 0.881; \(\{1,9,10\}\): 0.119 & \(\{1\}\): 0.765; \(\emptyset\): 0.235 \\
0.1 & 2 & \(C_1\) & \(\{10\}\) & \(\{2,7,10\}\): 0.992; \(\{2,10\}\): 0.008 & \(\{2\}\): 0.762; \(\emptyset\): 0.238 \\
0.1 & 3 & \(C_1\) & \(\{10\}\) & \(\{3,7,10\}\): 0.873; \(\{3,10\}\): 0.127 & \(\{3\}\): 0.771; \(\emptyset\): 0.229 \\
0.1 & 4 & \(C_2\) & \(\{10\}\) & \(\{4,9,10\}\): 0.854; \(\{4,10\}\): 0.146 & \(\{4\}\): 0.719; \(\emptyset\): 0.281 \\
0.1 & 5 & \(C_2\) & \(\{10\}\) & \(\{5,8,10\}\): 0.982; \(\{5,10\}\): 0.018 & \(\{5\}\): 0.711; \(\emptyset\): 0.289 \\
\hline
0.2 & 1 & \(C_1\) & \(\{6,7,10\}\) & \(\{1,6,7\}\): 0.957; \(\{6,7,10\}\): 0.043 & \(\{1\}\): 0.765; \(\emptyset\): 0.235 \\
0.2 & 2 & \(C_1\) & \(\{6,7,10\}\) & \(\{2,6,7\}\): 0.962; \(\{6,7,10\}\): 0.038 & \(\{2\}\): 0.762; \(\emptyset\): 0.238 \\
0.2 & 3 & \(C_1\) & \(\{6,7,10\}\) & \(\{3,6,7\}\): 0.997; \(\{3\}\): 0.003 & \(\{3\}\): 0.771; \(\emptyset\): 0.229 \\
0.2 & 4 & \(C_2\) & \(\{6,7,10\}\) & \(\{4,6,10\}\): 0.989; \(\{6,7,10\}\): 0.011 & \(\{4\}\): 0.719; \(\emptyset\): 0.281 \\
0.2 & 5 & \(C_2\) & \(\{6,7,10\}\) & \(\{5,7,10\}\): 0.965; \(\{6,7,10\}\): 0.035 & \(\{5\}\): 0.711; \(\emptyset\): 0.289 \\
\hline
0.3 & 1 & \(C_1\) & \(\{6,7,10\}\) & \(\{1,6,7\}\): 0.921; \(\{6,7,10\}\): 0.079 & \(\{1\}\): 0.765; \(\emptyset\): 0.235 \\
0.3 & 2 & \(C_1\) & \(\{6,7,10\}\) & \(\{2,6,7\}\): 0.911; \(\{6,7,10\}\): 0.089 & \(\{2\}\): 0.762; \(\emptyset\): 0.238 \\
0.3 & 3 & \(C_1\) & \(\{6,7,10\}\) & \(\{3,6,7\}\): 0.960; \(\{6,7,10\}\): 0.040 & \(\{3\}\): 0.771; \(\emptyset\): 0.229 \\
0.3 & 4 & \(C_2\) & \(\{6,7,10\}\) & \(\{4,6,10\}\): 0.989; \(\{6,7,10\}\): 0.011 & \(\{4\}\): 0.719; \(\emptyset\): 0.281 \\
0.3 & 5 & \(C_2\) & \(\{6,7,10\}\) & \(\{5,7,10\}\): 0.965; \(\{6,7,10\}\): 0.035 & \(\{5\}\): 0.711; \(\emptyset\): 0.289 \\
\hline
0.4 & 1 & \(C_1\) & \(\{6,7,8,9,10\}\) & \(\{1,6,7\}\): 0.844; \(\{6,7,9\}\): 0.156 & \(\{1\}\): 0.765; \(\emptyset\): 0.235 \\
0.4 & 2 & \(C_1\) & \(\{6,7,8,9,10\}\) & \(\{2,6,7\}\): 0.945; \(\{6,7,10\}\): 0.055 & \(\{2\}\): 0.762; \(\emptyset\): 0.238 \\
0.4 & 3 & \(C_1\) & \(\{6,7,8,9,10\}\) & \(\{3,6,7\}\): 0.721; \(\{6,7,10\}\): 0.279 & \(\{3\}\): 0.771; \(\emptyset\): 0.229 \\
0.4 & 4 & \(C_2\) & \(\{6,7,8,9,10\}\) & \(\{4,8,9\}\): 0.996; \(\{4,9\}\): 0.004 & \(\{4\}\): 0.719; \(\emptyset\): 0.281 \\
0.4 & 5 & \(C_2\) & \(\{6,7,8,9,10\}\) & \(\{5,8,9\}\): 0.963; \(\{8,9,10\}\): 0.037 & \(\{5\}\): 0.711; \(\emptyset\): 0.289 \\
\hline
0.5 & 1 & \(C_1\) & \(\{6,7,8,9,10\}\) & \(\{6,7,9\}\): 0.677; \(\{1,6,7\}\): 0.323 & \(\{1\}\): 0.765; \(\emptyset\): 0.235 \\
0.5 & 2 & \(C_1\) & \(\{6,7,8,9,10\}\) & \(\{6,7,10\}\): 0.566; \(\{2,6,7\}\): 0.434 & \(\{2\}\): 0.762; \(\emptyset\): 0.238 \\
0.5 & 3 & \(C_1\) & \(\{6,7,8,9,10\}\) & \(\{6,7,10\}\): 0.837; \(\{3,6,7\}\): 0.163 & \(\{3\}\): 0.771; \(\emptyset\): 0.229 \\
0.5 & 4 & \(C_2\) & \(\{6,7,8,9,10\}\) & \(\{4,8,9\}\): 0.989; \(\{6,8,9\}\): 0.011 & \(\{4\}\): 0.719; \(\emptyset\): 0.281 \\
0.5 & 5 & \(C_2\) & \(\{6,7,8,9,10\}\) & \(\{5,8,9\}\): 0.942; \(\{8,9,10\}\): 0.058 & \(\{5\}\): 0.711; \(\emptyset\): 0.289 \\
\hline
\end{tabular}}
\caption{Returned assortment distributions corresponding to
Table~\ref{tab:pref-concentration}, evaluated at \(K_T=129\) under the numerical
specification \(u_{\ell}^{(j)}(\rho)=\alpha_{\ell}^{(j)}(\rho)\). Within each mode
and type, support assortments are ordered from largest to smallest probability mass.}
\label{tab:appendix-preference-concentration-distributions}
\end{table}

\begin{table}[t]
\centering
\scriptsize
\resizebox{\textwidth}{!}{%
\begin{tabular}{c c c c c c c c c c c c c c}
\hline
\(N\) & \(S_N\opt\)
& \(Q^{\mathrm{rent}}_1\)
& \(Q^{\mathrm{rent}}_2\)
& \(Q^{\mathrm{rent}}_3\)
& \(Q^{\mathrm{rent}}_4\)
& \(Q^{\mathrm{rent}}_5\)
& \(Q^{\mathrm{rent}}_6\)
& \(Q^{\mathrm{rent}}_7\)
& \(Q^{\mathrm{rent}}_8\)
& \(Q^{\mathrm{rent}}_9\)
& \(Q^{\mathrm{rent}}_{10}\)
& \(Q^{\mathrm{rent}}\)
& \(\bar Q^{\mathrm{rent}}\) \\
\hline
1 & \(\{10\}\)
& 0.487 & 0.491 & 0.465 & 0.487 & 0.476
& 0 & 0.344 & 0.176 & 0.340 & 0
& 3.267 & 0.653 \\
2 & \(\{6,7,8,9,10\}\)
& 0.907 & 0.930 & 0.919 & 0.929 & 0.906
& 0 & 0 & 0 & 0 & 0
& 4.591 & 0.459 \\
5 & \(\{1,\ldots,10\}\)
& 0 & 0 & 0 & 0 & 0
& 0 & 0 & 0 & 0 & 0
& 0 & 0 \\
10 & \(\{1,\ldots,10\}\)
& 0 & 0 & 0 & 0 & 0
& 0 & 0 & 0 & 0 & 0
& 0 & 0 \\
100 & \(\{1,\ldots,10\}\)
& 0 & 0 & 0 & 0 & 0
& 0 & 0 & 0 & 0 & 0
& 0 & 0 \\
\hline
\end{tabular}}

\caption{Realized rental flow into each content family under market scaling, summed
over user types. The values are computed from the returned assortment distributions in
Table~\ref{tab:appendix-market-scaling-distributions} at \(K_T=129\). The final two
columns report total rental flow \(Q^{\mathrm{rent}}\) and rental flow per arrival
\(\bar Q^{\mathrm{rent}}\)~\eqref{eqn:total-rental-flow}, respectively.}
\label{tab:appendix-market-scaling-rental-flow}
\end{table}

\begin{table}[t]
\centering
\scriptsize
\resizebox{\textwidth}{!}{%
\begin{tabular}{c c c c c c c c c c c c c c}
\hline
\(\rho\) & \(S\opt(\rho)\)
& \(Q^{\mathrm{rent}}_1\)
& \(Q^{\mathrm{rent}}_2\)
& \(Q^{\mathrm{rent}}_3\)
& \(Q^{\mathrm{rent}}_4\)
& \(Q^{\mathrm{rent}}_5\)
& \(Q^{\mathrm{rent}}_6\)
& \(Q^{\mathrm{rent}}_7\)
& \(Q^{\mathrm{rent}}_8\)
& \(Q^{\mathrm{rent}}_9\)
& \(Q^{\mathrm{rent}}_{10}\)
& \(Q^{\mathrm{rent}}\)
& \(\bar Q^{\mathrm{rent}}\) \\
\hline
0.0 & \(\{10\}\)
& 0.487 & 0.491 & 0.465 & 0.487 & 0.476
& 0 & 0.344 & 0.176 & 0.340 & 0
& 3.267 & 0.653 \\
0.1 & \(\{10\}\)
& 0.457 & 0.462 & 0.456 & 0.493 & 0.462
& 0 & 0.596 & 0.199 & 0.198 & 0
& 3.323 & 0.665 \\
0.2 & \(\{6,7,10\}\)
& 0.381 & 0.390 & 0.397 & 0.475 & 0.459
& 0 & 0 & 0 & 0 & 0
& 2.102 & 0.420 \\
0.3 & \(\{6,7,10\}\)
& 0.341 & 0.342 & 0.353 & 0.475 & 0.459
& 0 & 0 & 0 & 0 & 0
& 1.971 & 0.394 \\
0.4 & \(\{6,7,8,9,10\}\)
& 0.291 & 0.327 & 0.252 & 0.399 & 0.388
& 0 & 0 & 0 & 0 & 0
& 1.657 & 0.331 \\
0.5 & \(\{6,7,8,9,10\}\)
& 0.113 & 0.147 & 0.066 & 0.378 & 0.363
& 0 & 0 & 0 & 0 & 0
& 1.067 & 0.213 \\
\hline
\end{tabular}}

\caption{Realized rental flow into each content family along the
preference-concentration path, summed over user types. The values are computed from the
returned assortment distributions in
Table~\ref{tab:appendix-preference-concentration-distributions} at \(K_T=129\). The
final two columns report total rental flow \(Q^{\mathrm{rent}}\) and rental flow per
arrival \(\bar Q^{\mathrm{rent}}\)~\eqref{eqn:total-rental-flow}, respectively.}
\label{tab:appendix-preference-concentration-rental-flow}
\end{table}

\end{appendices}

\bibliography{reference.bib}

\begin{thebibliography}{37}
\providecommand{\natexlab}[1]{#1}
\providecommand{\url}[1]{\texttt{#1}}
\expandafter\ifx\csname urlstyle\endcsname\relax
  \providecommand{\doi}[1]{doi: #1}\else
  \providecommand{\doi}{doi: \begingroup \urlstyle{rm}\Url}\fi

\bibitem[Amaldoss et~al.(2021)Amaldoss, Du, and Shin]{AmaldossDuSh21}
Wilfred Amaldoss, Jinzhao Du, and Woochoel Shin.
\newblock Media platforms’ content provision strategies and sources of
  profits.
\newblock \emph{Marketing Science}, 40\penalty0 (3):\penalty0 527--547, 2021.

\bibitem[Amaldoss et~al.(2024)Amaldoss, Du, and Shin]{AmaldossDuSh24}
Wilfred Amaldoss, Jinzhao Du, and Woochoel Shin.
\newblock Pricing strategy of competing media platforms.
\newblock \emph{Marketing Science}, 43\penalty0 (3):\penalty0 488--505, 2024.
\newblock \doi{10.1287/mksc.2021.0092}.

\bibitem[Andrew(1979)]{Andrew79}
A.M. Andrew.
\newblock Another efficient algorithm for convex hulls in two dimensions.
\newblock \emph{Information Processing Letters}, 9\penalty0 (5):\penalty0
  216--219, 1979.
\newblock ISSN 0020-0190.

\bibitem[Barnes-Schuster et~al.(2002)Barnes-Schuster, Bassok, and
  Anupindi]{BarnesSchusterBaAn02}
Dawn Barnes-Schuster, Yehuda Bassok, and Ravi Anupindi.
\newblock Coordination and flexibility in supply contracts with options.
\newblock \emph{Manufacturing \& Service Operations Management}, 4\penalty0
  (3):\penalty0 171--207, 2002.
\newblock \doi{10.1287/msom.4.3.171.7754}.

\bibitem[{Bilibili Inc.}(2025)]{Bilibili25}
{Bilibili Inc.}
\newblock {Annual Report on {F}orm 20-{F} for the Fiscal Year Ended December
  31, 2024}.
\newblock Investor relations / U.S. Securities and Exchange Commission filing,
  2025.
\newblock URL
  \url{https://ir.bilibili.com/media/reeb2oqf/2024-annual-report-on-form-20-f.pdf}.
\newblock Accessed 2026-02-18.

\bibitem[Cai and Spulber(2023)]{CaiSp23}
Gaoyang Cai and Daniel~F. Spulber.
\newblock The freemium pricing strategy and the opportunity cost of time.
\newblock Available at SSRN, 2023.
\newblock URL \url{https://ssrn.com/abstract=4516760}.

\bibitem[Carroni and Paolini(2020)]{CarroniPa20}
Elias Carroni and Dimitri Paolini.
\newblock Business models for streaming platforms: Content acquisition,
  advertising and users.
\newblock \emph{Information Economics and Policy}, 52:\penalty0 100877, 2020.

\bibitem[Dean and Morgenstern(2022)]{DeanMo22}
Sarah Dean and Jamie Morgenstern.
\newblock Preference dynamics under personalized recommendations.
\newblock In \emph{Proceedings of the 23rd ACM Conference on Economics and
  Computation}, page 795–816, 2022.

\bibitem[DeValve and Peke\v{c}(2022)]{DeValvePe22}
Levi DeValve and Sa\v{s}a Peke\v{c}.
\newblock Optimal price/advertising menus for two-sided media platforms.
\newblock \emph{Operations Research}, 70\penalty0 (3):\penalty0 1629--1645,
  2022.
\newblock \doi{10.1287/opre.2021.2230}.

\bibitem[Dinitz and Gupta(2013)]{DinitzGu13}
Michael Dinitz and Anupam Gupta.
\newblock Packing interdiction and partial covering problems.
\newblock In \emph{Proceedings of the 16th International Conference on Integer
  Programming and Combinatorial Optimization}, IPCO'13, page 157–168, 2013.
\newblock \doi{10.1007/978-3-642-36694-9_14}.
\newblock URL \url{https://doi.org/10.1007/978-3-642-36694-9_14}.

\bibitem[Fan et~al.(2007)Fan, Kumar, and Whinston]{FanKuWh07}
Ming Fan, Subodha Kumar, and Andrew Whinston.
\newblock Selling or advertising: Strategies for providing digital media
  online.
\newblock \emph{Journal of Management Information Systems}, 24\penalty0
  (3):\penalty0 143–166, 2007.
\newblock ISSN 0742-1222.

\bibitem[Feldman and Topaloglu(2015)]{FeldmanTo15}
Jacob~B. Feldman and Huseyin Topaloglu.
\newblock Capacity constraints across nests in assortment optimization under
  the nested logit model.
\newblock \emph{Operations Research}, 63\penalty0 (4):\penalty0 812--822, 2015.
\newblock \doi{10.1287/opre.2015.1383}.

\bibitem[Gallego et~al.(2015)Gallego, Ratliff, and Shebalov]{GallegoRaSh15}
Guillermo Gallego, Richard Ratliff, and Sergey Shebalov.
\newblock A general attraction model and sales-based linear program for network
  revenue management under customer choice.
\newblock \emph{Operations Research}, 63\penalty0 (1):\penalty0 212--232, 2015.
\newblock \doi{10.1287/opre.2014.1328}.

\bibitem[Goldfarb and Tucker(2011)]{GoldfarbTu11}
Avi Goldfarb and Catherine~E Tucker.
\newblock Privacy regulation and online advertising.
\newblock \emph{Management Science}, 57\penalty0 (1):\penalty0 57--71, 2011.

\bibitem[Halbheer et~al.(2014)Halbheer, Stahl, Koenigsberg, and
  Lehmann]{HalbheerStLe14}
Daniel Halbheer, Florian Stahl, Oded Koenigsberg, and Donald~R. Lehmann.
\newblock Choosing a digital content strategy: How much should be free?
\newblock \emph{International Journal of Research in Marketing}, 31\penalty0
  (2):\penalty0 192--206, 2014.

\bibitem[Iyengar et~al.(2025)Iyengar, Ma, and Sethuraman]{IyengarMaSe25}
Garud Iyengar, Yuanzhe Ma, and Jay Sethuraman.
\newblock User engagement maximization through preference switching.
\newblock Available at SSRN, 2025.
\newblock URL \url{https://ssrn.com/abstract=5230517}.

\bibitem[Jiang et~al.(2025)Jiang, Cui, He, Li, and Tian]{Jiangetal25}
Zhong‐Zhong Jiang, Shiliang Cui, Na~He, Kunyang Li, and Lin Tian.
\newblock Contract selection and piracy surveillance for video platforms in the
  age of social media.
\newblock \emph{Production and Operations Management}, 34\penalty0
  (12):\penalty0 3775--3792, 2025.
\newblock \doi{10.1111/poms.13990}.

\bibitem[Kawaguchi et~al.(2021)Kawaguchi, Uetake, and
  Watanabe]{KawaguchiUeWa21}
Kohei Kawaguchi, Kosuke Uetake, and Yasutora Watanabe.
\newblock Designing context-based marketing: Product recommendations under time
  pressure.
\newblock \emph{Management Science}, 67\penalty0 (9):\penalty0 5642--5659,
  2021.

\bibitem[Kumar and Sethi(2009)]{KumarSe09}
Subodha Kumar and Suresh~P. Sethi.
\newblock Dynamic pricing and advertising for web content providers.
\newblock \emph{European Journal of Operational Research}, 197\penalty0
  (3):\penalty0 924--944, 2009.

\bibitem[Lambrecht and Misra(2017)]{LambrechtMi17}
Anja Lambrecht and Kanishka Misra.
\newblock Fee or free: When should firms charge for online content?
\newblock \emph{Management Science}, 63\penalty0 (4):\penalty0 1150--1165,
  2017.
\newblock \doi{10.1287/mnsc.2015.2383}.

\bibitem[Li et~al.(2024)Li, Balasubramanian, Chen, and Pang]{LiBaChPa24}
Xin Li, Hari Balasubramanian, Yan Chen, and Chuan Pang.
\newblock Managing conflicting revenue streams from advertisers and subscribers
  for online platforms.
\newblock \emph{European Journal of Operational Research}, 314\penalty0
  (1):\penalty0 241--254, 2024.

\bibitem[Lin(2020)]{Lin20}
Song Lin.
\newblock Two-sided price discrimination by media platforms.
\newblock \emph{Marketing Science}, 39\penalty0 (2):\penalty0 317--338, 2020.

\bibitem[Manieri et~al.(2025)Manieri, Falsone, and Prandini]{ManieriFaPr25}
Lucrezia Manieri, Alessandro Falsone, and Maria Prandini.
\newblock {DualBi: A dual bisection algorithm for non-convex problems with a
  scalar complicating constraint}.
\newblock \emph{Automatica}, 175:\penalty0 112198, 2025.

\bibitem[Pei et~al.(2011)Pei, Simchi-Levi, and Tunca]{PeiSiTu11}
Pamela Pen-Erh Pei, David Simchi-Levi, and Tunay~I. Tunca.
\newblock Sourcing flexibility, spot trading, and procurement contract
  structure.
\newblock \emph{Operations Research}, 59\penalty0 (3):\penalty0 578--601, 2011.
\newblock \doi{10.1287/opre.1100.0905}.

\bibitem[Qiu et~al.(2014)Qiu, Ahmed, Dey, and Wolsey]{QiuAhDeWo14}
Feng Qiu, Shabbir Ahmed, Santanu~S. Dey, and Laurence~A. Wolsey.
\newblock Covering linear programming with violations.
\newblock \emph{INFORMS Journal on Computing}, 26\penalty0 (3):\penalty0
  531--546, 2014.
\newblock \doi{10.1287/ijoc.2013.0582}.

\bibitem[Rusmevichientong et~al.(2010)Rusmevichientong, Shen, and
  Shmoys]{RusmevichientongShSh10}
Paat Rusmevichientong, Zuo-Jun~Max Shen, and David~B. Shmoys.
\newblock Dynamic assortment optimization with a multinomial logit choice model
  and capacity constraint.
\newblock \emph{Operations Research}, 58\penalty0 (6):\penalty0 1666--1680,
  2010.
\newblock \doi{10.1287/opre.1100.0866}.

\bibitem[Salganik et~al.(2006)Salganik, Dodds, and Watts]{SalganikDoWa06}
Matthew~J Salganik, Peter~Sheridan Dodds, and Duncan~J Watts.
\newblock Experimental study of inequality and unpredictability in an
  artificial cultural market.
\newblock \emph{science}, 311\penalty0 (5762):\penalty0 854--856, 2006.

\bibitem[Sato(2019)]{Sato19}
Susumu Sato.
\newblock Freemium as optimal menu pricing.
\newblock \emph{International Journal of Industrial Organization}, 63:\penalty0
  480--510, 2019.

\bibitem[{Spotify Technology S.A.}(2025)]{Spotify25}
{Spotify Technology S.A.}
\newblock Annual report 2024.
\newblock Annual report (Form 20-F), 2025.
\newblock URL
  \url{https://s29.q4cdn.com/175625835/files/doc_financials/2024/ar/Annual-Report-2024.pdf}.
\newblock Accessed 2026-02-18.

\bibitem[Sumida et~al.(2021)Sumida, Gallego, Rusmevichientong, Topaloglu, and
  Davis]{SumidaGaRuToDa21}
Mika Sumida, Guillermo Gallego, Paat Rusmevichientong, Huseyin Topaloglu, and
  James Davis.
\newblock Revenue-utility tradeoff in assortment optimization under the
  multinomial logit model with totally unimodular constraints.
\newblock \emph{Management Science}, 67\penalty0 (5):\penalty0 2845--2869,
  2021.
\newblock \doi{10.1287/mnsc.2020.3657}.

\bibitem[T{\aa}g(2009)]{Tag09}
Joacim T{\aa}g.
\newblock Paying to remove advertisements.
\newblock \emph{Information Economics and Policy}, 21\penalty0 (4):\penalty0
  245--252, 2009.
\newblock ISSN 0167-6245.

\bibitem[Talluri and van Ryzin(2004)]{TalluriRy04}
Kalyan Talluri and Garrett van Ryzin.
\newblock Revenue management under a general discrete choice model of consumer
  behavior.
\newblock \emph{Management Science}, 50\penalty0 (1):\penalty0 15--33, 2004.
\newblock \doi{10.1287/mnsc.1030.0147}.

\bibitem[Terui et~al.(2011)Terui, Ban, and Allenby]{TeruiBaAl11}
Nobuhiko Terui, Masataka Ban, and Greg~M. Allenby.
\newblock The effect of media advertising on brand consideration and choice.
\newblock \emph{Marketing Science}, 30\penalty0 (1):\penalty0 74--91, 2011.

\bibitem[Tsay(1999)]{Tsay99}
Andy~A. Tsay.
\newblock The quantity flexibility contract and supplier-customer incentives.
\newblock \emph{Management Science}, 45\penalty0 (10):\penalty0 1339--1358,
  1999.
\newblock \doi{10.1287/mnsc.45.10.1339}.

\bibitem[Wang(2013)]{Wang13}
Ruxian Wang.
\newblock Assortment management under the generalized attraction model with a
  capacity constraint.
\newblock \emph{Journal of Revenue and Pricing Management}, 12\penalty0
  (3):\penalty0 254--270, 2013.

\bibitem[Wei and Nault(2014)]{WeiNa14}
Xueqi~(David) Wei and Barrie~R. Nault.
\newblock Monopoly versioning of information goods when consumers have group
  tastes.
\newblock \emph{Production and Operations Management}, 23\penalty0
  (6):\penalty0 1067--1081, 2014.

\bibitem[Weyl(2010)]{Weyl10}
E~Glen Weyl.
\newblock A price theory of multi-sided platforms.
\newblock \emph{American economic review}, 100\penalty0 (4):\penalty0
  1642--1672, 2010.

\end{thebibliography}
 
\end{document}